\documentclass[10pt, journal, cspaper, compsoc]{IEEEtran}
\usepackage{amsmath}
\usepackage{amssymb}
\usepackage{algorithm}
\usepackage{graphicx}
\usepackage{subfigure}
\usepackage{mdwlist}
\usepackage{algpseudocode}
\usepackage{caption}
\usepackage[left=1.2cm, right=1.2cm, top=1.2cm, bottom=1.2cm]{geometry} 
\newtheorem{lemma}{Lemma}
\newtheorem{theorem}{Theorem}
\newtheorem{proof}{Proof}
\newtheorem{definition}{Definition}

\begin{document}

\title{Socio-Spatial Group Queries for Impromptu Activity Planning}
\author{Chih-Ya Shen, De-Nian Yang,\IEEEmembership{~Senior Member,~IEEE}, Liang-Hao Huang, \\ Wang-Chien Lee,\IEEEmembership{~Member,~IEEE}, and Ming-Syan Chen,\IEEEmembership{~Fellow,~IEEE}%
\IEEEcompsocitemizethanks{
\IEEEcompsocthanksitem \noindent C.-Y. Shen and M.-S. Chen are with the Research Center for Information Technology Innovation, Academia Sinica, Taipei, Taiwan. E-mail: chihya@citi.sinica.edu.tw, mschen@citi.sinica.edu.tw.
\IEEEcompsocthanksitem D.-N. Yang and L.-H. Huang are with the Institute of Information Science, Academia Sinica, Taipei, Taiwan. E-mail: dnyang@iis.sinica.edu.tw, lhhuang@iis.sinica.edu.tw.
\IEEEcompsocthanksitem W.-C. Lee is with the Department of Computer Science and Engineering, Pennsylvania State University, PA, USA. E-mail: wlee@cse.psu.edu.}}

\IEEEcompsoctitleabstractindextext{
\begin{abstract}
The development and integration of social networking services and smartphones have made it easy for individuals to
organize impromptu social activities anywhere and anytime. 
Main challenges arising in organizing impromptu activities are mostly due to the requirements of making timely invitations in accordance with the potential activity locations, corresponding to the locations of and the relationships among the candidate attendees.
Various combinations of candidate attendees and activity locations create a large solution space. Thus, in this paper, we propose Multiple Rally-Point Social Spatial Group Query (MRGQ), to select an appropriate activity location for 
a group of nearby attendees with tight social relationships. We first consider a special case of MRGQ, namely the Socio-Spatial Group Query (SSGQ), to determine a set of socially acquainted attendees while minimizing the total spatial distance to a specific activity location. We prove that SSGQ is NP-hard and formulate an Integer Linear Programming optimization model for SSGQ. We then develop an efficient algorithm, called SSGS, 
which employs effective pruning techniques to reduce the running time to determine the optimal solution. 
%We also propose a new index structure, called Social R-Tree, to further improve the efficiency. 
Moreover, we propose a heuristic algorithm for SSGQ to efficiently produce
good solutions. We next consider the more general MRGQ. Although MRGQ is NP-hard, the number of attendees in practice is usually small enough such that an optimal solution can be found efficiently.
Therefore, we first propose an Integer Linear Programming optimization model for MRGQ. We then design an efficient algorithm, called MAGS, which employs effective search space exploration and pruning strategies 
to reduce the running time for finding the optimal solution. We also propose to further optimize efficiency by indexing
the potential activity locations. A user study demonstrates the strength of using SSGS and MAGS over manual
coordination in terms of both solution quality and efficiency. Experimental results on real datasets show that our algorithms can process SSGQ and MRGQ efficiently and significantly outperform other baseline algorithms, 
including one based on the commercial parallel optimizer IBM CPLEX.

%Owing to the arrival of the mobile computing age, organizing impromptu social activities outside office or home becomes feasible. Challenges faced in organizing impromptu activities are the requirements of
%making timely invitations in accordance with the locations of candidate
%attendees and the social relationship among them. 
%In this paper, we propose Socio-Spatial Group Query (SSGQ) to
%select a group of nearby attendees with tight social relation. 
%We prove that processing SSGQ is NP-hard and develop an efficient algorithm, called SSGSelect, 
%which employs effective pruning
%techniques to reduce the running time for finding the optimal solution. We
%also propose a new index structure, called Social R-Tree, to further improve the
%efficiency. Moreover, we propose an efficient heuristic algorithm for SSGQ to find good
%solutions very efficiently. Moreover, we extend SSGQ to consider multiple activity locations and formulate
%Multi Rally-Point Socio-Spatial Group Query (MRGQ). We design an efficient algorithm, called MAGS, which employs effective pruning strategies to reduce the running time for finding the optimal solution. We also propose to index
%the potential activity locations to further optimize the efficiency.
%We conduct a user study
%that demonstrates the strength of SSGSelect and MAGS over manual
%coordination in terms of both solution quality and efficiency. Experimental results
%on real datasets show that our algorithms can process SSGQ and MRGQ
%efficiently and outperform other baseline algorithms significantly, 
%including the commercial parallel optimizer IBM CPLEX.
\end{abstract}
}

\maketitle

\section{Introduction\label{Intro}}

\baselineskip=11.3pt
The successful development and integration of social networking services and smartphones have driven the recent emergence of location-based social networking (LBSN) services. 
Such services, including applications on Foursquare, Meetup, Facebook, and Google+, allow users to connect with friends, comment on events and places (e.g., restaurants, theaters, stores, etc.), and share their happenings and current locations. This availability of users' locations and their social
information allows mobile users to instantly organize impromptu social activities anywhere anytime.

As an LBSN application, an impromptu activity planning service needs to account for both spatial and social factors. In other words, both the locations and friends considered need to be suitable for the activity, i.e., the location should be close to the participants so that they arrive in a timely manner, and the invited friends should already be acquainted with each other to ensure comity. 
%\footnote{Other factors, such as the preference on locations, also play important roles.  Later in Section \ref{Baseline} we present an efficient way to incorporate the preference factor in the algorithm design.}
Thus, a major challenge
for impromptu activity planning lies in factoring in the distances from invitees' current locations to the activity locations, along with their shared social connectivity. Note that close friends may not be located near
a specific activity location, while friends near a potential activity location may not enjoy tight social relationships. Moreover, 
when the number of candidate attendees increases, or when the number of activity locations grows, selecting the most
suitable attendees and activity location becomes tedious and time-consuming. Therefore, impromptu activity planning would benefit significantly from efficient
query processing algorithms that automatically recommend both attendees and an activity location.

\begin{figure}[tbp]
\centering
\includegraphics[scale=0.33]{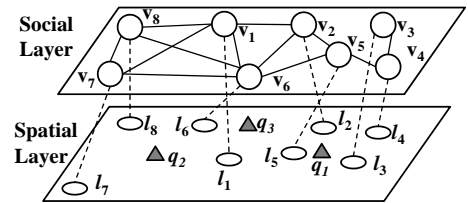}
\caption{Finding friends for impromptu social activity.}
\vspace{-0.7cm}
\label{FIG_IlluExample}
\end{figure}

\noindent \textbf{Motivating Example.} The interplay of social relationships among activity attendees and the activity locations creates significant
challenges for the organization of impromptu social activities. Figure \ref{FIG_IlluExample} shows a 
database of 8 candidate attendees $\{v_{1},..,v_{8}\}$ with three potential activity locations $Q=\{q_{1}, q_{2}, q_{3}\}$. 
The social relationships among the candidate attendees are captured as a social graph (shown as the social layer in the figure), while
the locations of the candidate attendees are shown as the spatial layer. Given a desired group size, $4$, and a social constraint where each attendee can only be unfamiliar with at most $1$ other attendee, an approach to select a group and the corresponding activity location with minimized total spatial distance is to issue a 4-nearest neighbor (4NN) query on each activity location. In the result, we obtain $F_{1}=\{v_{2}, v_{3}, v_{4},v_{5}\}$ with the activity location $q_{1}$. 
However, in this case, $F_{1}$ does not satisfy the required social constraint
because both $v_{2}$ and $v_{3}$ are unacquainted with more than $1$ other group member. Instead, if we focus on social tightness, we obtain group 
$F_{2}=\{v_{1},v_{6},v_{7},v_{8}\}$ with activity location $q_{2}$, where each attendee is familiar with all the other members. 
However, this group incurs a large spatial distance and thus is not suitable for an impromptu activity. In contrast, $F_{3}=\{v_{1},v_{2},v_{5},v_{6}\}$ with activity location $q_{3}$ is probably the most suitable solution because each attendee in $F_{3}$ is unacquainted with no more than $1$ other group member while incurring a small total spatial distance to $q_{3}$. 

%Recently, Socio-Spatial Group Query (SSGQ) has been proposed \cite{YSLC12}, which finds a socially tight group for an impromptu social gathering. Given a social network and current locations of candidate attendees, a specified activity location, and a social constraint $k$, SSGQ finds the optimal group of $p$ attendees with the minimum total spatial distance to the activity location, while each attendee in average can be unfamiliar with at most other $k$ members in the group. In SSGQ, only a single activity location is considered. However, in many scenarios, an impromptu activity organizer may have a number of alternative 
%activity locations (e.g., movie theaters or chain restaurants) in mind. Moreover, the social constraint in SSGQ may potentially result in imbalanced acquaintance relations of the selected group, i.e., some group members are mutually familiar, while some have very few acquaintances. 

In this paper, we propose a new query, namely Multiple Rally-Point Social Spatial Group Query (MRGQ), to determine a suitable activity location and a socially acquainted group which minimizes the total spatial distance to the activity location.
MRGQ seeks a set of most-suitable 
attendees with a corresponding activity location by considering both social and spatial factors of impromptu activity planning. 
MRGQ is beneficial for real social network applications (e.g., Facebook) and can integrate with group buying websites (e.g., Groupon) to provide \textit{social-aware} location-based advertisements. We will discuss these issues in Section \ref{scenarios}.
Here, we assume that the service provider has access to the users' underlying social relationships along with their current locations. Let $G=(V,E)$ be a social graph, where each vertex $v\in V$ is associated with a location $l_{v}$, and two mutually acquainted vertices $u$ and $v$ are connected by an edge $e_{u,v}$. Given a set of potential activity locations $Q=\{q_1,..,q_n\}$, the planned number of activity attendees $p$, the number of unacquainted people each attendee may have $k$, and the maximum spatial distance $t$ (i.e., spatial radius) from the chosen activity location to each of the selected attendees, MRGQ aims to find a set of $p$ attendees from the social graph and an activity location $q^*$ from the potential activity location list, such that the total distance from each attendee to the activity location $q^*$ is minimal, and the distance from each attendee to the activity location $q^*$ is bounded by $t$.\footnote{In most cases a user can specify $p$ and $Q$ according to the motivation of the corresponding group activity, such as a "buy three and get one free" coupon in a chain restaurant. While it may be more difficult for a user to specify the exact values of $k$ and $t$, one promising way is to let the user select the ranges of the two parameters. Accordingly, the algorithm returns multiple solutions with different $k$ and $t$ so that the user can choose the most desirable one.} Notice that MRGQ includes a social constraint (i.e., $k$) to ensure the familiarity 
between each attendee, i.e., each attendee can be unfamiliar with at most $k$ other people in the selected group. By setting $k$, 
the coordinator can freely adjust the social atmosphere of the activity to accommodate different types of social activities.
Formally, MRGQ is formulated as follows.

\vspace{+5pt}
\noindent \textbf{Problem: }Multiple Rally-Point Social Spatial Group Query (MRGQ).

\noindent \textbf{Given: }A social graph $G=(V,E)$, location $l_{v}$ for each $v\in V$, 
the number of attendees $p$, the set of potential activity locations $Q$, the familiarity constraint $k$, and the spatial radius $t$.

\noindent \textbf{Objective: }$MRGQ(p,Q,k,t)$ finds $\langle F, q^* \rangle$ where $F\subseteq V$, $q^*\in Q$, 
such that $|F|=p$, $\sum_{v\in F}d_{v,q^*}$ is minimal\footnote{$d_{v,q^*}$ is the spatial distance from $v$ to $q^*$.}, $d_{v,q^*}\leq t$, and $unfamiliar(v,F)$ $\leq k$\footnote{The number of vertices in $F$ which share no edge with $v$.}, $\forall v \in F$. 
%and each $v$ can only share no edge with at most $k$ other vertices in $F$.
\vspace{+5pt}

A straightforward approach for processing MRGQ is to enumerate all possible groups of $p$ attendees for each activity location and eliminate those not satisfying the constraints on social familiarity and spatial radius. Then, this approach  returns the pair of group and activity location which incur the minimum total spatial distance. This straightforward approach
needs to enumerate $|Q|\cdot C^{|V|}_{p}$ candidate pairs of groups and locations, entailing an enormous search space. Indeed, as we 
show in the next section, MRGQ is NP-hard. However, as the size of $p$ is relatively
small in most practical impromptu activity scenarios, the problem can be solved efficiently. 
By carefully exploring the social and spatial constraints in MRGQ, we develop several processing strategies to obtain the optimal solution efficiently. We systematically examine the search space to avoid examining all combinations of candidate attendees and the activity locations. We incrementally select attendees with the corresponding activity location by giving priority to those attendees (i) who are close to an activity location, and (ii) who are close friends. Obtaining a group which satisfies
both (i) and (ii) is non-trivial because an algorithm that addresses (i) should simultaneously choose suitable attendees and the nearest activity location. However, while achieving (i) may quickly obtain a group with small total spatial distance, it does not always result in a feasible group that satisfies the familiarity constraint. Alternatively, we can address (ii) by prioritizing the search for a group of attendees who know each other well. However, the group may not have the minimum spatial distance to the closest activity location. 
In summary, efficiently processing MRGQ requires carefully designed algorithms to select the attendees along with their 
nearby activity location while simultaneously satisfying the familiarity constraint.

To efficiently process MRGQ, we propose to index the attendees' locations and the activity locations. In addition, we design effective strategies for traversing the search space, including Socio-Spatial Ordering and All-Pair Distance Ordering, as well as a number of search space pruning rules, including Inner-Triangle Distance Pruning, Outer-Triangle Distance Pruning, Activity Location Distance Pruning, and Familiarity Pruning, to reduce the processing time. During the selection of the attendees and the activity location, we address both the spatial distance among the candidate locations, and from attendees to activity locations. Meanwhile, the social connectivity of the attendees is also carefully explored. As such, we effectively prune redundant search space to find the optimal solution efficiently.

The contributions of this paper are summarized as follows.

\begin{itemize*}
\item We identify the organization of impromptu social activities as a new social networking application and formulate a novel query, 
MRGQ, to obtain the optimal set of invitees and a suitable activity location. MRGQ is unique because it specifies the familiarity constraint among the invitees. We prove that the problem is NP-hard and inapproximable within any factor.

\item We consider a special case of MRGQ, namely SSGQ, for considering only a single activity location. 
We prove that SSGQ is NP-hard and propose SSGS with various
strategies for finding the optimal solution efficiently. 
%In addition, we design a new index structure, namely SR-Tree, together with
%Joint Insertion, to simultaneously select multiple vertices at each iteration
%to improve efficiency in finding the optimal solution. 
%
%\item 
In addition, we propose a heuristic algorithm for SSGQ, namely SSGMerge, which effectively exploits the structures of intermediate solutions, to obtain good solutions in polynomial time. We also propose an Integer Linear Programming (ILP) optimization model for SSGQ and demonstrate that SSGS outperforms ILP.

\item To efficiently process MRGQ, we propose to index the locations of candidate attendees and the activity locations and propose an efficient algorithm, namely MAGS, which enables various search space traversing and pruning strategies to find the optimal solution efficiently. We also propose an Integer Linear Programming (ILP) optimization model for MRGQ and demonstrate that MAGS outperforms ILP, even if it runs on a commercial integer programming optimizer with parallel computation.

\item 
We conduct a user study with 206
people. The results demonstrate that our proposed algorithms significantly outperform
manual coordination in terms of both solution quality and efficiency for both SSGQ and MRGQ. 
We also implement SSGQ in Facebook.

\item
We evaluate the performance of the proposed algorithms by conducting extensive experiments on real datasets. Experimental results manifest that SSGS and SSGMerge require much less time than the ILP optimization model with the commercial parallel optimizer IBM CPLEX \cite{CPLEX}. Likewise, for MRGQ, MAGS outperforms the baseline algorithms in terms of both solution quality and efficiency, and is much more efficient than the ILP optimization model.
\end{itemize*}

The rest of this paper is summarized as follows. Section \ref{analysis} analyzes MRGQ and proves that it is NP-hard. Section \ref{RelatedWork} introduces the related works. 
Section \ref{SSGQ} studies a special case of MRGQ, namely SSGQ and details the proposed algorithms. 
%Section \ref{SSGQ} proposes algorithm SSGSelect and 
%Section \ref{EnhancedSSGSelect} enhances it with the proposed SR-Tree. 
%Section \ref{Heuri_Section} details the heuristic algorithm design for SSGQ. 
Section \ref{MRGQ} details the proposed algorithm to efficiently process MRGQ.
Section \ref{Exp} shows the results of our user study and experiments. Finally, Section \ref{Conclu} concludes this paper.

\section{Problem Analysis and Applications}\label{analysis}

An MRGQ includes four parameters, i.e., $p$, $Q$, $k$ and $t$, which respectively
determine the size of the answer group, activity locations, familiarity
constraint and spatial radius of the query, and all of which have
a significant impact on processing strategies. First, as
the size of group, $p$, increases, the solution space (which consists of all
candidate groups) grows rapidly. While we prove that processing MRGQ is
an NP-hard problem and thus very challenging,
it can still be processed efficiently since the size of $p$ is usually small in most
practical cases. Second, candidate attendees located close to a candidate activity
location $q_{i}$ could be prioritized for processing, as the search criteria aim to
minimize the total spatial distance from the selected attendees to $q_{i}$.
As the size of $Q$ increases, the search space also grows. Third, $k$
dictates the tightness of social relationships among members in the invited group. 
A smaller $k$ in MRGQ indicates that candidate attendees with tighter
social relationships should be given priority.
Finally, $t$ reflects the need to avoid selecting candidates that are
unacceptably far away from the selected activity location. These spatial and
familiarity constraints can be employed for pruning of unqualified candidate groups.
In the following, we first analyze the hardness of MRGQ and then discuss concrete application scenarios for MRGQ. 

\subsection{Problem Analysis}

We prove that MRGQ is NP-hard and inapproximable within any factor, i.e., no approximation algorithm exists for MRGQ. 

\begin{theorem} \label{MRGQ_NPHARD}
MRGQ is NP-hard and is inapproximable within any factor unless $P=NP$.
\end{theorem}

\begin{proof}
We prove that MRGQ is NP-hard with the reduction from $p$-clique. Decision
problem $p$-clique, given a graph $G_{c}$, determines whether the graph
contains a clique, i.e., a complete graph of $p$ vertices and with an edge
connecting every two vertices. In MRGQ, let $G=G_{c}$, $k=0$, $t=\infty $%
, $Q=\{q\}$ and $d_{v,q}=1$ for every vertex $v$ $\in V$. We first prove the
necessary condition. If $G_{c}$ contains a $p$-clique, there must exist a
group with the same vertices in the $p$-clique such that every person has
social relationship with all the other attendees in the group, and the total
spatial distance is $p$. We then prove the sufficient condition. If $G$ in
MRGQ has a group of size $p$ and $k=0$, $G_{c}$ in problem $p$%
-clique must contain a solution of size $p$, too. Therefore, MRGQ is NP-hard.

We prove the inapproximability of MRGQ with a gap-introducing reduction from the $p$-clique
problem. Given a graph $G_{c}$, the decision problem $p$-clique determines
whether the graph contains a clique of size $p$, i.e., a complete graph of $%
p $ vertices with an edge connecting every two vertices. For any instance of
the $p$-clique problem in graph $G_{c}$, we construct an instance of MRGQ as
follows. The input graph of MRGQ, $G$, is constructed by adding a complete
graph $K_{p}$ with $p$ vertices to $G_{c}$, i.e., $G=G_{c}\cup K_{p}$, where
each vertex $v\in K_{p}$ connects to every vertex $u\in G_{c}$. We set $%
Q=\{q\}$, where $q$ is any spatial object, and the spatial distance from
each vertex $u\in G_{c}$ to $q$ is set to $1$, i.e., $d_{u,q}=1,\forall u\in
G_{c}$. By contrast, the spatial distance from each vertex $v\in K_{p}$ to $%
q $ is set to an arbitrary value $l$ much larger than $p$, i.e., $%
d_{v,q}=l,\forall v\in K_{p}$. Moreover, $k=0$ and $t=\infty $ in MRGQ. Now,
if there is a $p$-clique in $G_{c}$, there exists a feasible solution of
MRGQ, i.e., $F\subseteq G_{c}$ in $G$, with the total spatial distance as $%
\sum_{v\in F}d_{v,q}=p$ (i.e., $F\cap K_{p}=\varnothing $). If no $p$-clique exists 
in $G_{c}$, MRGQ has at least one feasible solution, such as $%
K_{p}$, but it is not possible to extract a feasible solution from $%
G_{c}$ alone. Therefore, the optimal solution $F$ returned by MRGQ must include at
least one vertex in $K_{p}$ with a total spatial distance of $\sum_{v\in
F}d_{v,q}\geq (p-1+l)>p$ (i.e., $F\cap K_{p}\neq \varnothing $). MRGQ cannot
be approximated within any factor smaller than $(p-1+l)/p$; otherwise, the
approximation algorithm could solve the $p$-clique decision problem since it
can distinguish the two cases in MRGQ. Since $l$ can be set as an arbitrary
value much larger than $p$, MRGQ cannot be approximated within any ratio.
The theorem follows.
\end{proof}
%
%We also propose an Integer Linear Programming (ILP) optimization model for MRGQ which, via a commercial solver, such as CPLEX \cite{CPLEX}, can obtain the optimal solution. Please refer to the online version \cite{OnlineVersion} for the detailed descriptions of the ILP model.

%%ILP start here
We also propose an Integer Linear Programming (ILP) optimization model for MRGQ which, via a commercial solver, such as CPLEX \cite{CPLEX}, can obtain the optimal solution. 
We first define a number of decision variables in the ILP
formulation. Let binary variable $\phi _{u}$ denote whether vertex $u$ is in 
$F$. Let binary variable $\pi _{q}$ denote whether activity location $q$ is
chosen in the solution. When $u$ is an attendee and thus joins $F$, let
integer variable $\mu _{u}$ denote the number of attendees in $F$ not
acquainted with $u$, $\mu _{u}\geq 0$. Let variable $\delta _{u}$ denote the
distance from $u$ to the activity location if $u$ is selected in $F$, $%
\delta _{u}\geq 0$; otherwise, $\delta _{u}=0$. The problem is to minimize
the total spatial distance from the selected activity location to the
attendees, i.e., $\min \sum_{u\in V}\sum_{q\in Q}\pi _{q}\phi _{u}d_{u,q}$.
However, this simple formula does not serve well as the objective function because it
is not linear. On the other hand, the formula, $\min \sum_{u\in F}\delta _{u}$, also does not serve well as the objective function since $F$ is unknown. Therefore, we
formulate the objective function of MRGQ as follows. 
\begin{equation*}
\min \sum_{u\in V}\delta _{u}.
\end{equation*}%
This objective function can correctly find out the total spatial distance
from the selected activity location to the attendees since only $\delta _{u}$
of each attendee $u$ in $F$ will be assigned a non-zero value, as shown in
the constraint (9) detailed later. In other words, $\delta _{u}$ will be $0$
in the objective function if $u$ is not an attendee.

The ILP formulation for MRGQ is equipped with the following constraints.

\begin{center}
\begin{tabular}{cr}
$\sum_{u\in V}\phi _{u}=p,$ & (A) \\ 
$\sum_{q\in Q}\pi _{q}=1,$ & (B) \\ 
$\left( p-1\right) \phi _{u}-\sum_{v\in N_{u}}\phi _{v}\leq \mu _{u},$ \ \ $%
\forall u\in V$ & (C) \\ 
$\sum_{u\in V}\mu _{u}\leq kp,$ & (D) \\ 
$d_{u,q}(\phi _{u}+\pi _{q}-1)\leq \delta _{u},$ $\forall u\in V,\forall
q\in Q$ & (E) \\ 
$\delta _{u}\leq t,$ \ \ $\forall u\in V$ & (F)%
\end{tabular}
\end{center}

In the above, constraint $(A)$ guarantees that exactly $p$ vertices are
selected in solution set $F$, while constraint $(B)$ states that only one
location is selected for the activity. Constraints $(C)$ and $(D)$ specify
the familiarity condition. Specifically, if $u$ participates in $F$, i.e., $%
\phi _{u}=1$, this constraint becomes $\left( p-1\right) -\sum_{v\in
N_{u}}\phi _{v}\leq \mu _{u}$. In other words, the left-hand-side (LHS) of constraint $%
(C)$ is identical to the number of attendees in $F$ not knowing $u$, and
constraint $(D)$ enforces that the total number of unfamiliar attendees
not to exceed $kp$.

Constraint $(E)$ assigns $\delta _{u}$ as $d_{u,q}$ if $u$ and $q$ are
chosen as an attendee and the activity location, respectively. More
specifically, $\phi _{u}$ and $\pi _{q}$ are both $1$ in this case, and
constraint $(E)$ thus becomes $d_{u,q}\leq \delta _{u}$. Since the objective
function is a minimization function, $\delta _{u}$ will be assigned as $d_{u,q}$
in the optimal solution. On the other hand, if $u$ is not an
attendee, or if $q$ is not the activity location, constraint $(E)$ becomes $%
0\leq \delta _{u}$, and thus non-restrictive to $\delta _{u}$. Therefore,  $%
\delta _{u}$ will be $0$ in the objective function if $u$ is not an
attendee. Constraint $(F)$ ensures that the spatial distance from each
attendee to the activity location not to exceed spatial radius $t$.

We have the following observations from the above constraints.

\begin{enumerate}
\item Constraint $(C)$ cannot be substituted with $\left( p-1\right) \phi
_{u}-\sum_{v\in N_{u}}\phi _{v}=\mu _{u}$. Otherwise, if $u$ does not join $F
$, i.e., $\phi _{u}=0$, this constraint becomes $-\sum_{v\in N_{u}}\phi
_{v}=\mu _{u}$. Therefore, constraint $(D)$ cannot correctly sum up the
number of unfamiliar attendees in $F$, because it considers every person $u$
in $V$. To address this issue, an approach is to replace constraint $(D)$
with $\sum_{u\in F}\mu _{u}\leq kp$, such that only the attendees in $F$
will be considered. However, constraint $(D)$ in this case becomes non-linear
because the set $F$ also needs to be decided too. In contrast, the proposed
constraints $(C)$ and $(D)$ can effectively avoid the above issue. When $%
\phi _{u}=0$, constraint $(C)$ becomes $-\sum_{v\in N_{u}}\phi _{v}\leq \mu
_{u}$, which allows $\mu _{u}$ to be $0$ for constraint $(D)$, such that we
are able to sum up $\mu _{u}$ of every person in $V$, even when $u$ is not
in $F$. Note that $\mu _{u}$ is also allowed to be assigned larger than the
LHS of constraint $(C)$. However, if constraint $(D)$ still holds when $%
\left( p-1\right) -\sum_{v\in N_{u}}\phi _{v}<\mu _{u}$, it guarantees that
assigning $\mu _{u}=\left( p-1\right) -\sum_{v\in N_{u}}\phi _{v}$ also
leads to a solution that does not contradict $(C)$, because the LHS of $(C)$
becomes smaller in this case. Therefore, the familiarity condition can be
enforced with the design of $\mu _{u}$ together with constraints $(C)$ and $%
(D)$. Similarly, constraint $(E)$ cannot be replaced with $d_{u,q}(\phi
_{u}+\pi _{q}-1)=\delta _{u}$.

\item The complexity of this formulation (correlated to the number of
integral decision variables) can be significantly reduced by relaxing the
integrality constraint that enforces $\mu _{u}$ to be a non-negative
integer. In this case, $\mu _{u}$ can be any non-negative real number, and
the number of integer variables in this formulation are significantly
reduced. This formulation in this case is still correct because $\phi _{u}$
in the objective function still needs to be an integer variable. In
addition, for any solution with $\mu _{u}$ not an integer number, replacing $%
\mu _{u}$ with the largest integer number not exceeding $\mu _{u}$ must also
be a feasible solution, since the LHS of constraint $(C)$ needs to be an
integer number.
\end{enumerate}

%%ILP for MRGQ ends here

\subsection{Application Scenarios} \label{scenarios}
We discuss the reasons why MRGQ is beneficial for real social applications, such as Facebook and Groupon.

1) The initiator is a person included in the solution group.
The proposed MRGQ can be employed in various online social network
applications, e.g., Facebook, to initiate impromptu activities. Facebook's \textit{Event}
function allows a user to initiate an activity by specifying the
location and invitees. However, it may be difficult for the initiator to select a set of invitees with tight social
relationships in real time, and the multiple candidate locations, e.g., branches in a popular chain restaurant, may make it difficult for the initiator to manually select a suitable location and the corresponding attendees. 
%Thus, a two-phase
%approach is usually adopted to find a sub-optimal solution through manual
%coordination, not to mention finding a $<$group, place$>$ solution based on both social tightness and
%spatial closeness. 
If MRGQ can be integrated with Facebook, the initiator only needs to specify a set of
candidate activity locations along with the query parameters to quickly identify the invitees and a suitable activity location. 

%On the other hand, in Eventbrite and Meetup, an activity initiator can
%create an event, i.e., a face-to-face activity, by specifying the activity
%location and the event description. Eventbrite and Meetup then send activity
%invitations to the candidates in close proximity to the activity location. Currently, 
%Eventbrite and Meetup only consider the spatial proximity of potential attendees, but not their
%social tightness. MRGQ can be exploited by Eventbrite and Meetup to improve the current services to
%find the best location and potential group for the event.

2) The initiator is \textit{not} a person and thus not included in the solution group.
In addition, deal-of-the-day services such as Groupon, can also benefit 
from MRGQ. Currently, Groupon recommends offered deals (e.g., coupons) to users according to their
preferences or purchase histories. To take advantage of a given deal, a customer may need to organize a certain number of friends (e.g., "buy three get one"), and may be less inclined to buy the coupon if identifying a likely group poses difficulty. 
To address this issue, Groupon can exploit MRGQ to provide \textit{social-aware} location-based advertisement. For example,
to promote a chain restaurant, Groupon can identify
groups with tight social relationships and thus identify branches suitable for each group. The social recommendation can
be attached in the location-based advertisement to increase the chance of
the customer purchasing the coupon. In this case, Groupon is an initiator not included in the solution group.

\vspace{-10pt}
\section{Related Work}

\label{RelatedWork}
\baselineskip=11pt
Some LBSN applications, e.g., Meetup, have been available for activity
coordination for some time. However, they are designed mainly for periodical meetings, 
e.g., a reading club or a user group for 3D printing. 
In this paper, we emphasize the scenarios of impromptu social activities 
where the time and effort for organizing an activity need to be minimized. 
As manual identification of candidate attendees, a common practice today, is tedious and time-consuming, we argue and show in this paper that, MRGQ is very useful
for such scenarios as it recommends a group of suitable attendees and an activity location by taking both the social and spatial factors into account.

Researches on finding groups of socially connected
members, e.g., team formation \cite{LLT09}\cite{LS10}, community search \cite%
{SG10}, Social-Temporal Group Query \cite{YCLC11} and Circle of Friend Query \cite{LCHJ12}, 
%and Nearest Star Group Query \cite{APP13} 
have been reported in
the literature. Nevertheless, their research context and objectives are
totally different from our research goal, i.e., exploring both the spatial and social
dimensions in finding a group of friends and a location for an impromptu activity. Specifically,
team formation \cite{LLT09}\cite{LS10} finds a group of experts with the
required skills, while aiming to minimize the communication cost between
these experts. Community search \cite{SG10} finds a compact community that
contains particular members, aiming to minimize the total degree in the
community. Social-Temporal Group Query \cite{YCLC11} checks the available
times of attendees to find the group with the most suitable activity time. 
Circle of Friend Query \cite{LCHJ12} finds a group of friends by considering
their social and spatial properties. The friends are not grouped to specific
activity locations because no activity location is given in this query, and this query
thus is not suitable for impromptu activity planning. 
%Nearest Star Group Query \cite{APP13}
%finds a group of people who are connected through a common friend and
%are nearby a given location. However, this query only requires the group to 
%be connected through a specific person, but does not consider the familiarity
%between the remaining group members.

Relevant to our work, spatial queries for selecting a set of spatial points,
aiming to minimize the total spatial distance, have been proposed for
various scenarios~\cite{PSTM04,TPS02,GZLC09,GZ09}. However, in these works,
the (social) connectivity among the spatial points is not considered.
Specifically, given two sets of points $P$ and $Q$, together with the number
of points to be selected $k$, Group Nearest Neighbor Query \cite{PSTM04}
finds a set of $k$ points in $P$\ such that the total spatial distance of
the points to all points in $Q$ is minimized. On the other hand, for a line
segment and a set of points, Continuous Nearest Neighbor Search \cite{TPS02}
returns the nearest neighbor of each point on the line segment. 
Meanwhile, Continuous Visible Nearest Neighbor Queries \cite{GZLC09} and Continuous Obstructed 
Nearest Neighbor Query \cite{GZ09} extend Continuous
Nearest Neighbor Search \cite{TPS02} by incorporating the obstacles in the
problem designs, which may affect the visibility or distance between two
points and lead to different results. 
Therefore, the above-mentioned queries focus only on the
spatial dimension and thereby are not applicable to our scenario of LBSN
applications. %impromptu activity planning 

To the best knowledge of the authors, researches on finding groups that
consider constraints in both the spatial and social dimensions just started.
Our work examines the interplay in both social and
spatial dimensions, with an objective to find a group of mutually familiar
attendees such that the total spatial distance to an activity location is
minimized. We envisage that our research result can be
employed in various LBSN applications for group recommendation.

\vspace{-10pt}
\section{Socio-Spatial Group Query (SSGQ)}
\label{SSGQ}
\baselineskip=11.2pt
The challenges for processing MRGQ lie in the interplay of social and spatial dimensions, along with the large solution space. In this section, we first consider a relaxed version of MRGQ with single activity location, i.e., Socio-Spatial Group Query (SSGQ). We formulate SSGQ and propose an Integer Linear Programming (ILP) optimization model for SSGQ, which acts as a baseline for comparison with the proposed algorithms for SSGQ. We then propose an algorithm, called \textit{SSGS}, to efficiently process SSGQ. We also propose a heuristic algorithm for SSGQ, namely \textit{SSGMerge}, to find good solutions very efficiently. 

%\subsection{Problem Formulation and Integer Linear Programming (ILP) Model for SSGQ}
Specifically, SSGQ is formally defined as follows.

\noindent \textbf{Problem:} Socio-Spatial Group Query (SSGQ).

\noindent \textbf{Given:} A social graph $G=(V,E)$, location $l_{v}$ for each $v\in V$, 
and an $SSGQ(p,q,k,t)$ where $p$ is the number of attendees, $q$ is the activity location, 
$k$ is the familiarity constraint, and $t$ is the spatial radius.

\noindent \textbf{Objective:} To find a set $F\subseteq V$ where $|F|=p$ and 
minimize the total spatial distance from $F$ to $q$, i.e., $\sum_{v\in F}d_{v,q}$,
where $d_{v,q}\leq t,\forall v\in F$, and $unfarmiliar(v,F)$ $\leq k$\footnote{The average number of vertices in $F$ sharing no edge with $v$.}, $\forall v \in F$.

\begin{theorem}\label{SSGQ_NPHARD}
SSGQ is NP-hard.
\end{theorem}

\begin{proof}
%Please refer to the online version \cite{OnlineVersion}.
We prove that SSGQ is NP-hard with the reduction from $p$-clique. Decision
problem $p$-clique is given a graph $G_{c}$ to find whether the graph
contains a clique, i.e., a complete graph with an edge connecting every two
vertices, with $p$ vertices. In SSGQ, we let $G=G_{c}$, $k=0$, $t=\infty $, and $%
d_{v,q}=1 $ for every vertex $v$ $\in V$. We first prove the necessary
condition. If $G_{c}$ contains a $p$-clique, there must exist a group with
the same vertices in the $p$-clique such that every person has social
relationship with all the other attendees of the group, and the total spatial
distance is $p$. We then prove the sufficient condition. If $G$ in SSGQ
contains a group with the size as $p$ and $k$ as $0$, $G_{c}$ in problem $p$%
-clique must contain a solution with size $p$, too. The theorem follows.
\end{proof}

%We also propose an Integer Linear Programming (ILP) optimization model for SSGQ which will serve as a baseline to compare with the proposed algorithms for SSGQ. 
%
%Due to the space constraints, please refer to the online version \cite{OnlineVersion} for the detailed descriptions of the ILP model.

%%ILP for SSGQ starts here

In the following, we present an Integer Linear Programming (ILP)
optimization model for SSGQ. We first define a number of decision variables in
the formulation. Let binary variable $\phi _{u}$ denote whether vertex $u$
is in $F$. When $u$ joins $F$, let integer variable $\mu _{u}$ denote the
number of attendees in $F$ not acquainted with $u$, $\mu _{u}\geq 0$. The
problem is to minimize the total spatial distance from each vertex in $F$ to $%
q $, i.e., 
\begin{equation*}
\min \sum_{u\in V}d_{u}\phi _{u}
\end{equation*}%
s.t.

\begin{center}
\begin{tabular}{cr}
$\sum_{u\in V}\phi _{u}=p,$ \ \ $\forall u\in V$ & (G) \\ 
$d_{u}\phi _{u}\leq t,$ \ \ $\forall u\in V$ & (H) \\ 
$\left( p-1\right) \phi _{u}-\sum_{v\in N_{u}}\phi _{v}\leq \mu _{u},$ \ \ $%
\forall u\in V$ & (I) \\ 
$\sum_{u\in V}\mu _{u}\leq kp.$ & (J)%
\end{tabular}
\end{center}

In the above, constraint $(G)$ guarantees that exactly $p$ vertices are selected in
solution set $F$, while constraint $(H)$ ensures that the spatial distance
from each selected attendee to $q$ does not exceed spatial radius $t$.
Constraints $(I)$ and $(J)$ specify the familiarity condition.
Specifically, if $u$ participates in $F$, i.e., $\phi _{u}=1$, this
constraint becomes $\left( p-1\right) -\sum_{v\in N_{u}}\phi _{v}\leq \mu
_{u}$. In other words, the left-hand-side (LHS) of $(I)$ is identical to the
number of attendees in $F$ not knowing $u$, and constraint $(J)$ enforces
that the total number of unfamiliar attendees must not exceed $kp$. 

We make the following observations from the above constraints.

\begin{enumerate}
\item Constraint $(I)$ cannot be substituted with $\left( p-1\right) \phi
_{u}-\sum_{v\in N_{u}}\phi _{v}=\mu _{u}$. Otherwise, if $u$ does not join $%
F $, i.e., $\phi _{u}=0$, this constraint becomes $-\sum_{v\in N_{u}}\phi
_{v}=\mu _{u}$. Therefore, constraint $(J)$ cannot correctly sum up the
number of unfamiliar attendees in $F$, because it considers every person $u$
in $V$. To address this issue, an approach is to replace constraint $(J)$
with $\sum_{u\in F}\mu _{u}\leq kp$, such that only the attendees in $F$
will be considered. However, this constraint in this case becomes non-linear
because the set $F$ also needs to be decided too. In contrast, the proposed
constraints $(I)$ and $(J)$ can effectively avoid the above issue. When $%
\phi _{u}=0$, constraint $(I)$ becomes $-\sum_{v\in N_{u}}\phi _{v}\leq \mu
_{u}$, which allows $\mu _{u}$ to be $0$ for constraint $(J)$, such that we
are able to sum up $\mu _{u}$ of every person in $V$, even when $u$ is not
in $F$. Note that $\mu _{u}$ is also allowed to be assigned larger than the
LHS of constraint $(I)$. However, if constraint $(J)$ still holds when $%
\left( p-1\right) -\sum_{v\in N_{u}}\phi _{v}<\mu _{u}$, it guarantees that
assigning $\mu _{u}=\left( p-1\right) -\sum_{v\in N_{u}}\phi _{v}$ also
leads to a solution that does not contradict $(J)$, because the LHS of $(J)$
becomes smaller in this case. Therefore, the familiarity condition can be
enforced with the design of $\mu _{u}$ together with constraints $(I)$ and $%
(J)$.

\item The complexity of this formulation (correlated to the number of
integral decision variables) can be significantly reduced by relaxing the
integrality constraint that enforces $\mu _{u}$ to be a non-negative
integer. In this case, $\mu _{u}$ can be any non-negative real number, and
the number of integer variables in this formulation are significantly
reduced. This formulation in this case is still correct because $\phi _{u}$
in the objective function still needs to be an integer variable. In
addition, for any solution with $\mu _{u}$ not an integer number, replacing $%
\mu _{u}$ with the largest integer number not exceeding $\mu _{u}$ must also
be a feasible solution, since the LHS of constraint $(I)$ needs to be an
integer number.
\end{enumerate}

%%ILP for SSGQ ends here

\vspace{-10pt}
\subsection{Algorithm Design for SSGQ}
\label{SSGQ_Algo}

Despite only considering a single activity location, processing SSGQ is still challenging since we need to 
account for the interplay between both social and spatial factors, which necessitates a systematic approach for group formation. 
Therefore, in this section, we propose an algorithm, called \textit{SSGS}, to efficiently process SSGQ. 
SSGS adopts a branch-and-bound group formation process to form feasible groups, i.e., those that consist of $p$ members and satisfy the query constraints. The basic idea is to maintain an intermediate group $S_{I}$ and incrementally add a candidate member from the remaining set of candidates, $S_{R}$, based on some ordering strategies to traverse the space of group formation. Given a candidate attendee set $V$ and the activity location $q$,
SSGS initializes $S_{I}=\varnothing $ and $S_{R}$ as the candidate attendees within the spatial radius of $q$. At each subsequent iteration, SSGS moves a candidate attendee from $S_{R}$ into $S_{I}$ until $S_{I}$ becomes a feasible solution. If $S_{I}$ is disqualified during the process, SSGS backtracks to the previous step to choose another candidate attendee from $S_{R}$. When $S_{I}$ becomes feasible, SSGS saves it as the current best solution and backtracks to previous step to continue finding better groups. 
Obviously this process is slow, so the key issue is how to devise a traverse ordering strategy to quickly find a feasible group and devise effective rules to prune redundant groups.

One approach is to use an R-tree which indexes the locations of candidates to provide guidance, and select a candidate from $S_R$
with the shortest spatial distance to the activity location, which is referred as \textit{Distance Ordering}. As such, we can use the spatial properties derived via the maximum bounding rectangles (MBR) in the R-tree and the constraints of SSGQ to prune unqualified candidates and thus reduce the search space. Another approach aims to quickly form a feasible group with small total spatial distance to the activity location for distance-based pruning adopting a Socio-Spatial Ordering, which prioritizes the growth of an intermediate group based on its social tightness. 
Recall that Distance Ordering
first expands $S_{I}$ with the individuals closest to the query point $q$.
For example, consider Figure \ref{fig:FIG_SSO_graph} as the input social graph (the
number besides each node indicates the spatial distance to $q$), where $p=3$
and $k=0$. Figure \ref{fig_do} presents the expansion of $S_{I}$ with only
Distance Ordering, and the number besides each node in the branch-and-bound
tree represents the expansion sequence. As shown in Figure \ref{fig_do}, the
expansion sequence of these nodes is sorted according to the spatial
distance to the query point. The leaf nodes in the branch-and-bound
tree (i.e., the groups of $p$ individuals) can be created according to the
total spatial distance, i.e., a group with a smaller total spatial distance
is generated earlier. However, employing only
the Distance Ordering strategy is not always good because it ignores the
social constraint of the generated groups. As a result, most groups generated at the
early stage, e.g., $\{a,b,c\}$, $\{a,b,d\}$, $\{a,b,e\}$, and $\{a,b,f\}$, 
do not satisfy the familiarity constraint
(i.e., $k=0$) even though they are the top-4 groups with the smallest total
spatial distances. 

To address the weakness of Distance Ordering, we
combine the social connectivity and spatial distance to identify an
intermediate group to be expanded in the next step. Intuitively, when an individual $v$ is chosen by
Distance Ordering, we move it into $S_{I}$ only when $v$ also satisfies the
social condition specified in Eq. (\ref{SSO_eqn}). This social condition
ensures that $S_{I}$ together with $v$ leads to a group with the attendees
familiar with each other. If $v$ does not follow the above social
condition, we find another individual $u$ with Distance Ordering that
satisfies the social condition. As such, both spatial and social factors
are taken into account in Socio-Spatial Ordering. 

More specifically, to ensure that the social connectivity of each selected individual $v$ to the vertices in $S_{I}$ is good, 
a simple approach is to
ensure that $v$ can be selected only when the number of edges between $v$
and the vertices in $S_{I}$ exceeds a given threshold. With a larger
threshold, a candidate attendee that is familiar with more attendees
currently in $S_{I}$ is inclined to be chosen. Nevertheless, parameter $k$ is not
examined for the current attendees in $S_{I}$ when $v$ is added. Consequently, some
attendees in this case may not have a sufficient number of neighbors
in $S_{I}$. By contrast, SSGS selects $v$ only when $v$ satisfies Eq. (\ref%
{SSO_eqn}). Specifically, as Eq. (\ref{SSO_eqn}) assumes that $v$ is
added to $S_{I}$, SSGS examines whether the social connectivity of the
new group $S_{I}\cup \{v\}$ is sufficient according to the criterion $k$.
Let $F(S_{I})$ denote the average number of acquainted members in $S_{I}$,
i.e., $F(S_{I})=\frac{1}{|S_{I}|}\sum_{v\in S_{I}}|N_{v}\cap S_{I}|$, where $%
N_{v}$ is the set of neighbors of $v$ in $V$. Individual $v$ is added to $S_{I}$ if it satisfies Eq. (\ref{SSO_eqn}) as
follows,

\vspace{-10pt}
\begin{equation}
F(S_{I}\cup \{v\})\geq |S_{I}\cup \{v\}|-\frac{\theta \left( |S_{I}\cup
\{v\}|\right) }{p-1}-1  \label{SSO_eqn}
\end{equation}

where $\theta $ here is a dynamically adjusted parameter and set as $%
k$ initially. Intuitively, when $k=p-1$, the activity allows all attendees
to be mutually unfamiliar. In this case, Distance Ordering is the best
strategy. In fact, Eq. (\ref{SSO_eqn}) in this situation becomes $%
F(S_{I}\cup \{v\})\geq -1$, and Socio-Spatial Ordering here is identical to
Distance Ordering. In another extreme case where $k=0$, Eq. (\ref{SSO_eqn})
becomes $F(S_{I}\cup \{v\})\geq |S_{I}\cup \{v\}|-1$, implying that each attendee in $%
S_{I}\cup \{v\}$ needs to be acquainted with all the others in $S_{I}\cup
\{v\}$.

It is worth noting that Eq. (\ref{SSO_eqn}) incorporates the
dynamically adjusted parameter $\theta $. Instead of including $k$ directly,
it properly handles other cases with $0<k<p-1$. When $k=0$, if no
vertex from $S_{R}$ satisfies Eq. (\ref{SSO_eqn}), it is not necessary to
add any individual $v$ from $S_{R}$ to $S_{I}$ because every solution
growing from $S_{I}\cup \{v\}$ does not follow the familiarity constraint. When $%
k>0$, if no individual from $S_{R}$ satisfies Eq. (\ref{SSO_eqn}), it does
not imply that every solution growing from $S_{I}\cup \{v\}$ does not have
sufficient social connectivity. In contrast, it is possible to find an
individual $v$ in $S_{R}$ and a solution growing from $S_{I}\cup \{v\}$ when
other vertices added later bring a sufficient number of edges to the
solution. Therefore, for $k>0$, Socio-Spatial Ordering sets $\theta $ as $k$
initially and increases $\theta $ if no vertex from $S_{R}$ can satisfy Eq. (%
\ref{SSO_eqn}), until at least one vertex follows Eq. (\ref{SSO_eqn}) and
thereby is able to be selected for $S_{I}$. Notice that Eq. (\ref%
{SSO_eqn}) first maintains a high criterion for the social connectivity by
setting $\theta $ as $k$, in order to prioritize a vertex leading to
sufficient social connectivity. If no vertex from $S_{R}$ can satisfy such a
high criterion, Eq. (\ref{SSO_eqn}) increases $\theta $ to avoid filtering
out any feasible solution. Thus, any vertex in $S_{R}$ that did not satisfy
Eq. (\ref{SSO_eqn}) previously will be examined later with a large $\theta $
accordingly.

Figure \ref{fig:FIG_SSO_tree} presents an illustrative example of Socio-Spatial Ordering
with $p=3$ and $k=0$ for the graph in Figure \ref{fig:FIG_SSO_graph}. The exploration of
Socio-Spatial Ordering is shown as the solid line in Figure \ref%
{fig:FIG_SSO_tree}. In this example, $\theta =0$ and $S_{I}=\phi $ initially. Since $a
$ is the vertex with the minimum spatial distance to $q$, and $F(\phi \cup
\{a\})=\frac{0}{1}\geq 1-1-\frac{0\cdot 1}{2}$ satisfies Eq. (\ref{SSO_eqn}%
), SSGS moves vertex $a$ from $S_{R}$ to $S_{I}$ first and lets $S_{I}=\{a\}$%
. However, $F(S_{I}\cup \{b\})=\frac{0}{2}<2-1-\frac{0\cdot 2}{2}$ does not
satisfy Eq. (\ref{SSO_eqn}). Therefore, SSGS examines vertex $c$ and finds
out that $F(S_{I}\cup \{c\})=\frac{1}{2}\cdot 2=1\geq 2-1-\frac{0\cdot 2}{2}$
satisfies Eq. (\ref{SSO_eqn}). Therefore, vertex $c$ is moved into $S_{I}$,
and now $S_{I}=\{a,c\}$. We then expand $S_{I}$ by choosing vertex $d$, and $%
S_{I}=\{a,c,d\}$ now is a feasible solution. In contrast, Distance Ordering
selects vertex $b$ after vertex $a$ (as shown in the dashed-line in Figure %
\ref{fig:FIG_SSO_tree}) and then sequentially constructs four intermediate groups 
$\{a,b,c\}$, $\{a,b,d\}$, $\{a,b,e\}$, and $\{a,b,f\}$. Unfortunately, none
of these meets the familiarity constraint. As shown, this example
illustrates that it is desirable to jointly consider spatial and social
domains in order to find a feasible solution for SSGS earlier, because the obtained
feasible solution is a key factor for the pruning strategy introduced below.
\begin{figure}[tp]
\centering
\subfigure[Graph.] {\  
\includegraphics[scale=0.2]{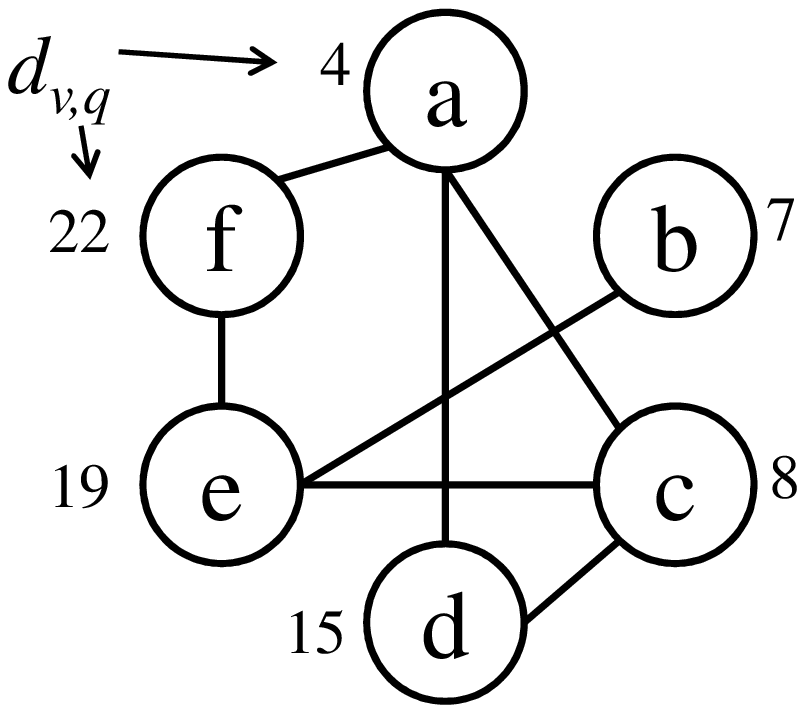} \label{fig:FIG_SSO_graph}} 
\subfigure[][Distance Ordering.] {\  \includegraphics[scale=0.25] {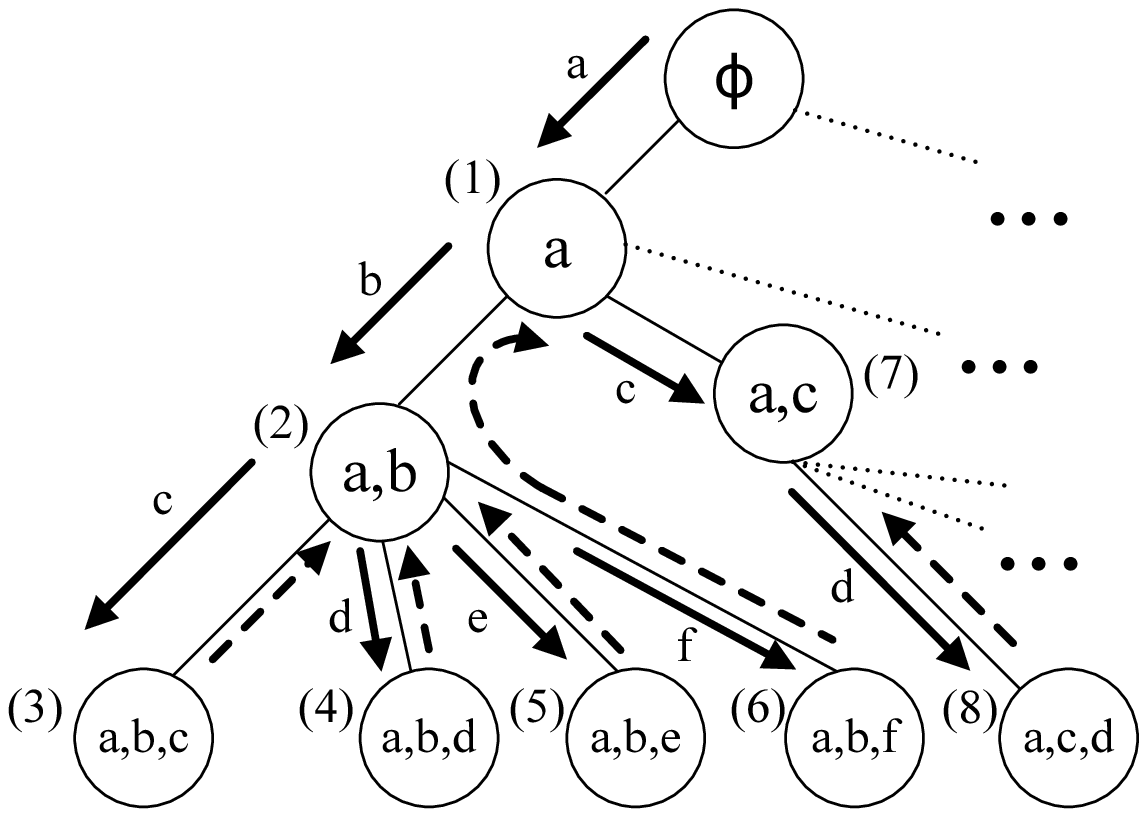}
\label{fig_do} } 
\subfigure[Socio-Spatial Ordering.] {\
\includegraphics[scale=0.25] {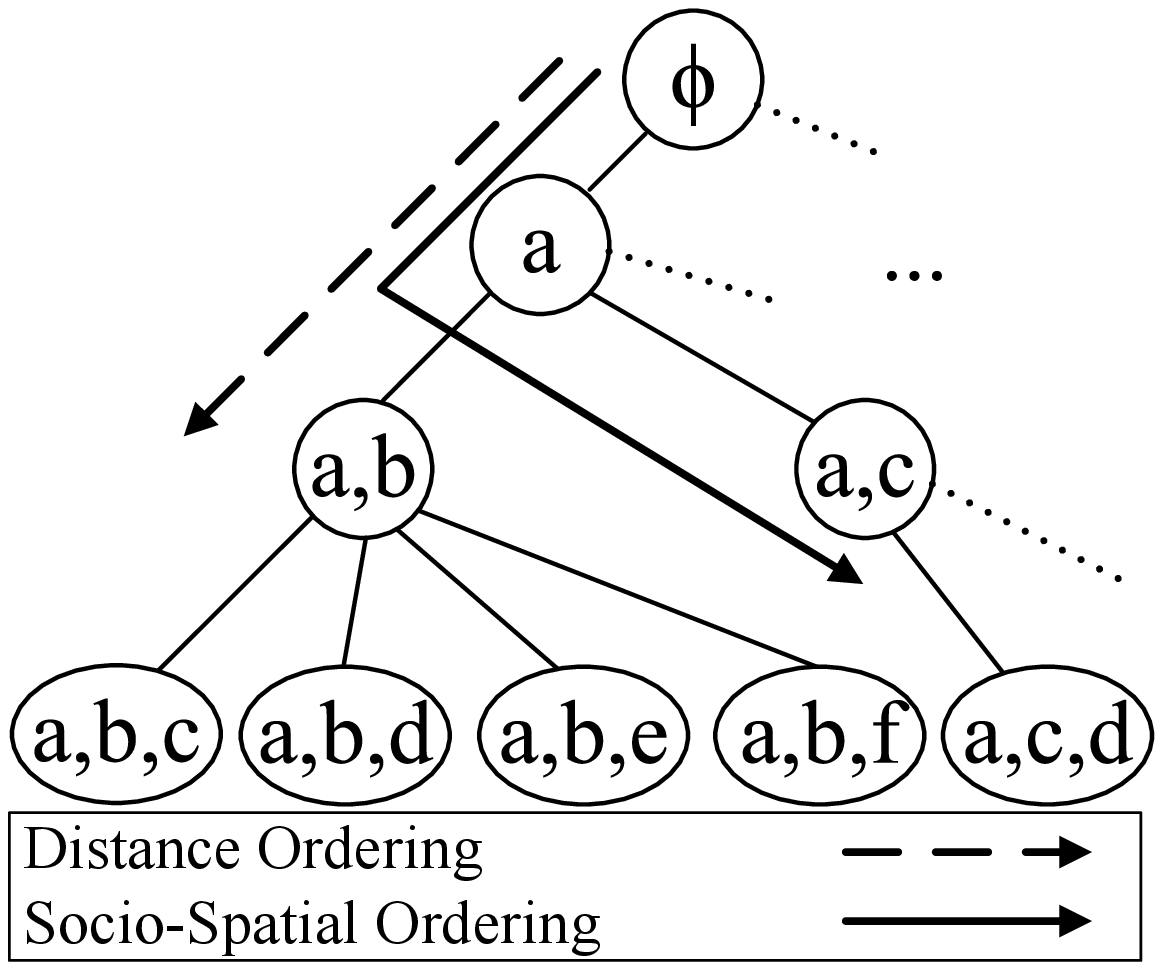} \label{fig:FIG_SSO_tree}}  %\vspace{-15pt} 
\caption{Example of ordering strategies.}
\vspace{-18pt}
\label{fig_sso}
\end{figure}

\subsection{Pruning Strategies for SSGS} \label{pruning_SSGQ}
We also propose two pruning rules, namely Familiarity Pruning and Distance Pruning, which effectively filter out unqualified intermediate groups. The idea of Familiarity Pruning is to derive an upper bound on the number of acquaintances each member may have after new members are included into $S_{I}$. Similarly, Distance Pruning identifies a lower bound on the total spatial distance of each group grown from $S_{I}$. SSGS stops processing $S_{I}$ and backtracks if the current $S_{I}$ is pruned by Familiarity Pruning or Distance Pruning.

\textbf{Familiarity Pruning.} Specifically, the edges in any solution growing from $S_{I}$ can be divided
into three categories: 1) $E_{I}$: the set of edges connecting any two
vertices in $S_{I}$, 2) $E_{R}:$ the set of edges connecting any two
vertices selected from $S_{R}$, and 3) $E_{IR}$: the set of edges connecting
any two vertices in $S_{I}$ and the vertices selected from $S_{R}$.
Apparently, $\left \vert E_{I}\right \vert = \frac{1}{2}\sum
\nolimits_{v\in S_{I}}|N_{v}^{I}|$, where $N_{v}^{I}$ is the set of
acquainted neighbors of $v$ in $S_{I}$. Since the selected vertices in $%
S_{R} $ are not clear, a good way is to find an upper bound on $%
\left
\vert E_{R}\right \vert $, i.e., $\frac{1}{2}(p-|S_{I}|)\max_{v\in
S_{R}}|N_{v}^{R}|$, where $N_{v}^{R}$ is the set of acquainted neighbors of $%
v$ in $S_{R}$. It is an upper bound because the vertex with the maximum
degree in $S_{R}$ is identified, and $(p-|S_{I}|)$ vertices are selected
from $S_{R}$. Similarly, an upper bound on $\left \vert E_{IR}\right \vert = \sum_{v\in S_{I}}|InterEdge(v)|$, where $InterEdge(v)$ is the set of
edges connecting $v$ in $S_{I}$ to any vertices in $S_{R}$.

Notice that the number of edges in a feasible solution is half of the total degree
of all the vertices in the solution. Therefore,
with the above three categories of edges, Familiarity Pruning stops
processing $S_{I}$ when the following condition holds,%\vspace{-3pt} 
\vspace{-5pt}
\begin{eqnarray}
%&&\frac{1}{p}\left[ \sum \nolimits_{v\in
%S_{I}}|N_{v}^{I}|+(p-|S_{I}|)\max_{v\in S_{R}}|N_{v}^{R}|\right. \nonumber \\
%&&\left. +2\cdot \sum_{v\in S_{I}}|InterEdge(v)|\right] \nonumber \\
%&<&(p-k-1).
&&\frac{1}{p}\left[ \sum \nolimits_{v\in
S_{I}}|N_{v}^{I}|+(p-|S_{I}|)\max_{v\in S_{R}}|N_{v}^{R}|\right. \nonumber \\
&&\left. +2\cdot \sum_{v\in S_{I}}|InterEdge(v)|\right]<(p-k-1).
\vspace{-12pt}
\end{eqnarray}%
In the above inequality, the left-hand-side is an upper bound on the average number of
attendees acquainted to each person in any feasible solution growing from $%
S_{I}$. The condition states that, on average, each attendee is acquainted with
fewer than $p-k-1$ other attendees. Familiarity Pruning stops processing $%
S_{I}$ and backtracks if solutions
growing from $S_{I}$ via the exploration of $S_{R}$ do not satisfy
the familiarity constraint.

For the social graph in Figure \ref{fig:FIG_SSO_graph} with $p=3$ and $k=0$, if $S_{I}=\{b,d\}$ and $S_{R}=\{e,f\}$, SSGS stops processing $S_{I}$ and backtracks because $\frac{1}{3}(0+1\cdot 1+2\cdot 1)=1<(3-0-1)=2$. In other words, moving any vertex from $S_{R}$ to $S_{I}$ will never generate a feasible solution following the familiarity constraint.

\textbf{Distance Pruning.}
For a given $S_{I}$, $p-|S_{I}|$ vertices must
be selected from $S_{R}$ to $S_{I}$. Apparently, further
processing of $S_{I}$ is unnecessary if $S_{I}$ and the $p-|S_{I}|$ vertices
with the shortest spatial distances to $q$ have a total distance larger than $D$, where $D$ is the best solution value obtained so far. Therefore, Distance Pruning identifies a lower bound and stops
processing $S_{I}$ when the following condition holds,
\vspace{-3pt} 
\begin{equation}
\sum \nolimits_{u\in S_{I}}d_{u,q}+(p-|S_{I}|)d_{v_{\min },q}\geq D
\vspace{-3pt}
\end{equation}%
where the first term is the total spatial distance from the vertices in $%
S_{I}$ to $q$. For $S_{R}$, only the vertex $v_{\min }$ with the smallest
spatial distance to $q$ is accessed here, and $(p-|S_{I}|)d_{v_{\min },q}$
represents a lower bound on the total spatial distance for the above $%
p-|S_{I}|$ vertices in $S_{R}$. 

Consider the social graph in Figure \ref{fig:FIG_SSO_graph} with $p=3$ as an example. After a feasible solution $\{a,c,d\}$ is explored, its total spatial distance 27 is assigned to $D$. When SSGS considers $S_{I}=\{a,b\}$ and $S_{R}=\{e,f\}$, since $\sum \nolimits_{u\in S_{I}}d_{u,q}+(p-|S_{I}|)d_{v_{\min },q}=11+1\cdot 19=30>27$, Distance Pruning removes states $\{a,b,e\}$ and $\{a,b,f\}$, stops processing $S_I=\{a,b\}$, and backtracks to the previous state accordingly.

\subsection{Heuristic Algorithm for SSGQ\label{Heuri_Section}}

As proved earlier, processing SSGQ is an NP-hard problem. In fact, even when the
spatial distance is the same for every candidate, the problem is still NP-hard due
to the familiarity constraint required to address. Therefore, in the
following, we propose an efficient heuristic algorithm to obtain good
solutions very efficiently. SSGS employs the branch-and-bound framework
to incrementally improve the solution and find the optimal solution. A
straightforward approach to develop heuristic algorithms for SSGQ is to stop the branch-and-bound
search after the $i$-th feasible solution is obtained. However, it suffers
two main drawbacks: 1) the running time is still not constrained in polynomial
time, and 2) this approach only maintains the current minimal total spatial
distance for solution space pruning but ignores the possibility to exploit
many intermediate solutions to further improve the efficiency. To
effectively address the above two issues, we propose an algorithm, named 
\textit{SSGMerge}, which effectively utilizes the structures of
intermediate solutions to generate a good feasible solution in polynomial
time. The idea is to iteratively merge good socially tight groups with
small spatial distances in different intermediate solutions obtained in
earlier iterations. 

Figure \ref{FIG_Exp_MergeHeuri} presents a social network and a snapshot of the
branch-and-bound tree with $p=4$, $k=1$, and $\theta =1$
for Socio-Spatial Ordering. After the first feasible solution $\{a,c,d,e\}$
is obtained, $\{a,d\}$ and $\{b,d\}$ are pruned accordingly by Distance
Pruning,
while $\{a,c,d,f\}$ and $\{a,c,e,f\}$ have higher total spatial distances
than $\{a,c,d,e\}$. If the straightforward heuristic approach stops here, $%
\{a,b,c,d\}$ and $\{a,b,c,e\}$, both extending from $\{a,b\}$ and enjoying
smaller total spatial distances, unfortunately are not to be
discovered. In contrast, since $\{a,b,d\}$ and $\{a,c,d\}$ incurs small spatial distances, 
and if other candidates later join these two groups, they will become socially dense
groups with small spatial distances. A promising idea is to merge the two groups into 
$\{a,b,c,d\}$ (similarly, merging $\{a,b,c\}$ and $\{a,c,e\}$ results into $\{a,b,c,e\}$). 
Based on this idea, we design a systematic approach to choose a set of
suitable groups for constructing a good feasible solution. 

%\textbf{
%The concept of SSGMerge is similar to Joint Insertion with SR-Tree, i.e., 
%we expand the group of attendees with a set of good candidates simultaneously, 
%which can be regarded as traversing a shortcut on the branch-and-bound tree.
%Generally speaking, SR-Tree with Joint Insertion
%provides a general framework for different query parameters, while SSGMerge
%attempts to optimize the efficiency by finding good feasible solutions for the given query parameters.
%In SR-Tree and Joint Insertion, each cluster is a set of socially 
%tight and spatially close candidates. Since the query parameters are not known
%during the construction of SR-Tree, the spatial clusters (represented by MBRs) are created to include people located spatially close 
%to each other, while multiple social clusters are constructed on the same set of candidates
%to accommodate different social constraints. 
%Note that, as explained in Section \ref{SSOrder}, Socio-Spatial Ordering and
%Intra-Familiarity Condition in SSGSelect effectively take both social
%and spatial dimensions into consideration in creating multiple socially dense groups with small
%total spatial distances as different intermediate solutions. 
%Therefore, merging these intermediate solutions is an efficient way toward
%finding good feasible solutions. 
%}

%[FIG2]
\begin{figure}[tp]
\centering
%\hspace{-1cm}
%\vspace{+10pt}
\includegraphics[scale=0.32]{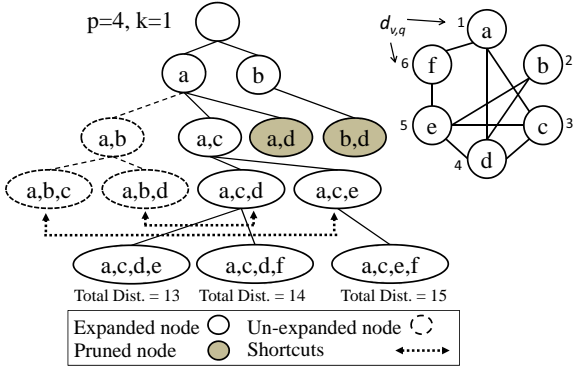}
\caption{Example of SSGMerge.}
\label{FIG_Exp_MergeHeuri}
%\vspace{-10pt}
\end{figure}
%\vspace{-6pt}

The intermediate solutions
expanded according to Socio-Spatial Ordering are created and tailored for each query with the specific parameters and the activity location.
Therefore, it is more efficient for SSGMerge to process 
intermediate solutions directly. 
%To support our idea, we keep track of the early-expanded intermediate solutions, and
%merge them into a better solution. Recall that these early-expanded
%solutions are expanded according to Socio-Spatial Ordering, which means they
%are socially tight and spatially close to the rally point. Therefore, there
%is a high probability that merging multiple intermediate solutions can lead to
%a better feasible solution. 
Given the group size $p$ of SSGQ, we maintain a
set of $p$ intermediate solution queues $\{U_{1},...,U_{p}\}$, 
where each element in $U_{j}$ is an intermediate solution $S_{I}$ with $j$ attendees. 
To prioritize the intermediate solutions in $U_{|S_{I}|}$ with high social tightness and
small spatial distance, we sort the intermediate solutions in $U_{|S_{I}|}$
with a ranking function $R(\cdot )$ based on Socio-Spatial Ordering,

%Each time
%when an intermediate solution $S_{I}$ is constructed, we insert $S_{I}$ into 
%$U_{|S_{I}|}$ to record it. However, recording all the intermediate
%solutions may incur huge storage and computation overhead. Therefore, we use
%a parameter $\lambda $ to control the size of each intermediate solution
%queue. When the size of $U_{|S_{I}|}$ exceeds $\lambda $, $%
%U_{|S_{I}|}$ is trimmed. To maintain the intermediate solutions in $%
%U_{|S_{I}|}$ with high social tightness and small spatial distance, we sort
%the intermediate solutions in $U_{|S_{I}|}$ with a ranking function based on
%Socio-Spatial Ordering, i.e., we maintain $\lambda $ intermediate solutions
%with smallest ranking values. The ranking function $R(\cdot )$ is given as
%follows.

\begin{equation}
R(S_{I})=p\cdot t \cdot \overline{\theta }+\sum_{v\in S_{I}}d_{v,q} \nonumber
\end{equation}%
where $\overline{\theta }\geq k$ is set to the minimum $\theta$
that $S_{I}$ satisfies Socio-Spatial Ordering, and $d_{v,q}$ is the
spatial distance from a candidate $v$ to the activity location $q$. 
%and $t \cdot $ controls the ranking weighting of social and spatial factors of $S_{I}$. 
The ranking function ranks $S_{I}$ based on its $\overline{\theta }$ value and the total spatial distance to $q$, and it gives a smaller score to the $S_{I}$ which has tighter social relationship and closer to $q$. 
%Since $t$ is the maximum value from a selected attendee in $S_I$ to the rally point, $pt\cdot \overline{k}$ is regarded as a threshold for assigning each $S_I$ to a group according to different
%social tightness, while $\sum_{v\in S_{I}}d_{v,q}$ ranks each $S_I$ within each group.
Consider an example with the social graph shown in Figure \ref%
{FIG_Exp_MergeHeuri}. Assume $p=4$, $k=1$, $t=100$, $\overline{S}_{I}=\{a,b,c\}$, and
$\widehat{S}_{I}=\{c,d,e\}$. Then, $R(\overline{S}%
_{I})=4\cdot 100\cdot 2+6=806$ while $R(\widehat{S}_{I})=4\cdot 100\cdot 1+12=412$. 
Therefore, SSGMerge is inclined to choose $\widehat{S}_{I}$.
%The intermediate solutions of Figure \ref{FIG_Exp_MergeHeuri}, i.e., the ones with solid circles and $|S_I|<p$, 
%are shown without underscores in Table \ref{Table_Heur}, while the merged groups constructed
%by SSGMerge are shown in bold with underscores in Table \ref{Table_Heur}. 
%As shown in Table \ref{Table_Heur}, 
%the solutions constructed by SSGMerge, i.e., $\{a,b,c,d\}$ and $\{a,b,c,e\}$, are better than 
%the first feasible solution obtained by SSGSelect, i.e., $\{a,b,d,e\}$, in Figure \ref{FIG_Exp_MergeHeuri}.

\begin{table}[tp]
\centering
\caption{Intermediate solutions maintained and constructed.}
\label{Table_Heur}%
\begin{tabular}{ll}
\hline \hline
& Maintained and constructed solutions \\ \hline \hline
$U_{1}$ & $\{a\}$,$\{b\}$ \\ \hline
$U_{2}$ & $\{a,c\}$, \underline{$\mathbf{\{a,b\}}$} \\ \hline
$U_{3}$ & $\{a,c,d\}$,$\{a,c,e\}$,\underline{$\mathbf{\{a,b,c\}}$} \\ \hline
$U_{4}$ & $\{a,c,d,e\}$,\underline{$\mathbf{\{a,b,c,d\}}$},\underline{$\mathbf{\{a,b,c,e\}}$}
\\ \hline
\end{tabular}%
\end{table}

%SSGMerge takes another parameter, $w$, to indicate when the branch-and-bound search is terminated.
%When the stop criteria of the branch-and-bound search is met, i.e., $w$ % 
%attempts for expanding new nodes in the branch-and-bound tree, 
%we then start merging these intermediate
%solutions to obtain better feasible solutions. 
Given the set of $p$ intermediate solution queues $\{U_{1},...,U_{p}\}$, 
the basic idea is to merge different pairs of small groups into larger ones.
That is, for each $U_{i}$, we merge each pair of small intermediate groups
$\widehat{S}_{I}$, $\overline{S}_{I}\in U_{i}$ into a new
intermediate solution $\widetilde{S}_{I}$, i.e., $\widetilde{S}_{I}=\widehat{%
S}_{I}\cup \overline{S}_{I}$, and store $\widetilde{S}_{I}$ in the
corresponding $U_{|\widetilde{S}_{I}|}$. If there are more than $\lambda $
intermediate solutions in $U_{i}$, after inserting merged intermediate solutions, 
$U_{i}$ maintains the $\lambda $ intermediate solutions with the smallest ranking value according to $R(\cdot )$.
In other words, $\lambda $ here is a filtering parameter for controlling the quality and the number of 
the intermediate solutions in each $U_{i}$.
Therefore, by first setting $i$ as $1$ and increasing $i$ by $1$ at each iteration, we can incrementally
construct new intermediate solutions. Finally, we extract the feasible
solution which incurs the minimum spatial distance from $U_{p}$ and return
it as the solution. 

More importantly, SSGMerge employs a pruning strategy to reduce the number of intermediate solutions under examination. 
When SSGMerge merges the intermediate solutions in $U_{i}$, an intermediate solution $\widehat{S}_{I}$
can be discarded if the following condition holds,

\begin{equation}
\sum_{v\in \widehat{S}_{I}}d_{v,q}+(p-|\widehat{S}_{I}|)\min_{i\leq j\leq p-1}\mu _{j}\geq D. \nonumber \label{Heuri_DC}
\end{equation}
In the above condition, $\mu _{j}$ is the minimum spatial distance of the candidates existing in $U_{j}$, and
$D$ is the currently best solution value. 
It measures the minimum increment of the spatial distance of
$\widehat{S}_{I}$ when $\widehat{S}_{I}$ is merged with others and becomes a feasible solution.
If this condition holds, any feasible solution expanded from $\widehat{S}_{I}$ 
(i.e., having $\widehat{S}_{I}$ as a subset), will never become a better solution, and thus $\widehat{S}_{I}$
can be safely discarded. 
%On the other hand,
%condition (6) estimates the maximum possible numbers of acquainted attendees after $\widehat{S}_{I}$
%grows into a group of $p$ attendees. If condition (6) does not hold, $\widehat{S}_{I}$ will never
%be a subset of any feasible solutions.

In Table \ref{Table_Heur}, the merged solutions constructed by SSGMerge are shown in bold with underscores. 
For example, $\{a,b\}$ is
constructed by merging $\{a\}$ and $\{b\}$ in $U_{1}$, and $\{a,b,c,d\}$ can
be constructed by merging $\{a,c,d\}$ and $\{a,b,c\}$ in $U_{3}$, 
where $\{a,b,c\}$ is the combination of $\{a,c\}$ and $\{a,b\}$ in $U_{2}$.
After the merging process is completed, we extract the feasible solution
with the minimum spatial distance from $U_{p}$, i.e., $U_{4}$ in Table \ref%
{Table_Heur}, which is $\{a,b,c,d\}$. As compared to the best feasible
solution $\{a,c,d,e\}$ obtained in Figure \ref{FIG_Exp_MergeHeuri}, the total
spatial distance of $\{a,b,c,d\}$ is $10$, which is smaller than that of $%
\{a,c,d,e\}$, i.e., $13$.

SSGMerge involves two parameters, $w$ and $\lambda$, and terminates the search process after $w$ states have been
generated in the branch-and-bound tree\footnote{Detailed settings of $w$ and $\lambda$ will be presented in the experimental results.}. SSGMerge then refines the solutions with the above merge approach. By effectively restricting the number
of generated intermediate solutions, SSGMerge can efficiently construct a good feasible solution according to the following theorem.

\begin{theorem} The running time of SSGMerge is $O(p\lambda ^{2}\log p\lambda +wp^{2}+|V|(\log |V|)^{2})$.
%, if the candidates are uniformly distributed in the plane, and the worst-case running time of SSGMerge is bounded
%in polynomial time.
\end{theorem}

\begin{proof}
SSGMerge first generates $w$ nodes in the branch-and-bound tree before it 
merges those intermediate solutions and creates feasible solutions. Three operations are performed:
1) Socio-Spatial Ordering, 2) Distance Pruning, and 3) Familiarity Pruning. 
Socio-Spatial Ordering includes Distance Ordering and the checking of Eq. (1) in Section 4.1.
Distance browsing strategy, i.e., iteratively extracting the candidate attendee with the minimum spatial distance to $q$ from R-Tree, in Distance Ordering is performed $w$ times.
Therefore, in the worst case, the number of R-Tree leaf node access is $O(|V|)$, and the traversal from the root
to a leaf node of R-Tree incurs $O(\log |V|)$ R-Tree internal node access. 
Since each R-Tree node access incurs $O(1)$ time for distance computation, 
the time of R-Tree node access is $O(|V|\log |V|)$. 
The priority queue maintained for Distance Ordering takes $O(\log s)$ time for each insertion and deletion operation, where $s$
is the size of the priority queue. Since there are $O(|V|\log |V|)$ elements inserted into the priority queue, 
and the insertion cost of each element is $O(\log (|V|\log |V|))$ (in worst case, the size of the priority queue is $O(|V|\log |V|)$). Therefore, the total cost is $O(|V|\log |V|)\cdot O(\log (|V|\log |V|))=O(|V|(\log |V|)^{2})$.

For checking Eq. (1) of Socio-Spatial Ordering, since the size of $S_{I}$ does not exceed $p$, it requires $O(p^{2})$ time to compute $F(S_{I}\cup \{v\})$ for each examination, i.e., examining if a vertex ${v}$ can be included in the current
$S_{I}$. Therefore, checking Eq. (1) for $w$ times takes $O(wp^{2})$ time.

Familiarity Pruning is performed in $O(wp^{2})$ time for $w$ examinations. 
Distance Pruning at each time examines the first element of the priority queue and 
the total spatial distance in $S_{I}$ with $O(p)$ time. Therefore, Distance Pruning takes 
$O(wp)$ time for $w$ examinations in the branch-and-bound tree.
In summary, the time complexity of $w$ attempts for including a node into $S_{I}$ is 
$O(wp^{2})+O(wp)+O(|V|(\log |V|)^{2})=O(wp^{2}+|V|(\log |V|)^{2})$.

On the other hand, when SSGMerge merges intermediate solutions, at each $U_{i}$, 
it first ranks the intermediate solutions in $U_{i}$ with the ranking function
and then discards those with ranks higher than $\lambda $. This step takes $O(p\lambda^{2} \log p\lambda )$ time
because in the worst case, each merged intermediate solution in $U_{i},$ $1\leq i\leq j$ is inserted into $U_{j}$.
It costs $O(\lambda ^{2})$ time for SSGMerge to combine each pair of intermediate solutions within each $U_{i}$ 
for $1\leq i\leq p$, including checking the pruning condition. Therefore, it takes $O(p\lambda ^{2})$ time for merging the intermediate solutions.
Overall, the running time for SSGMerge is $O(p\lambda ^{2}\log p\lambda +wp^{2}+|V|(\log |V|)^{2})$.
\end{proof}

Please note that $|V|$ in the complexity comes from the worst case of R-Tree distance browsing.
However, with the assumption of uniform distribution of the candidates' locations, the expected time of R-Tree 
distance browsing becomes $O(w\log |V|\cdot \log(w\log |V|))$.  More importantly, 
the experimental results manifest that SSGMerge is much faster than SSGS 
because SSGMerge effectively merges intermediate solutions into good feasible solutions to avoid examining the large search space.

\vspace{-5pt}
\section{Algorithm Design for MRGQ}\label{MRGQ}
\baselineskip=11.2pt
In this section, we turn our attention to Multiple Rally-Point Social Spatial Group Query (MRGQ), which finds 1) the most suitable activity location from a set of candidate locations and 2) a socially acquainted group with the minimal total spatial distance to the activity location. More specifically, MRGQ aims to find a pair $\langle F, q^* \rangle$, where $F$ is a socially acquainted group of $p$ people satisfying the familiarity constraint, and $q^* \in Q$ is a location in $Q$ such that $\langle F, q^* \rangle$ incurs the minimum total spatial distance. MRGQ is more difficult than SSGQ since different candidate social groups are closer to different locations, which need to be carefully considered as well.

To address the issue of multiple candidate locations, a straightforward approach is to repeat the SSGS algorithm $|Q|$ times to sequentially find the best group for each location. Nevertheless, this straightforward approach is not efficient because a spatial correlation may exist among multiple activity locations and thus can be exploited. 
In addition, it is desirable to design some effective index structures 
to facilitate efficient traversal and pruning of the search space. In this work, we propose to index the candidates with an R-Tree, while indexing the activity locations with a BallTree \cite{SO89}. Accordingly, we design new ordering strategies to quickly identify an activity location near an intermediate group of candidates satisfying the familiarity constraint and pruning strategies to avoid generating redundant
$\langle \widehat{F},q_{i}\rangle $ pairs, where $\widehat{F}$ is a group of $p$ candidates satisfying the familiarity constraint. Moreover, two effective strategies for traversing the search space are proposed, including All-Pair Distance Ordering and Single-Reference Distance Ordering. Processing time is also improved by introducing a number of new search space pruning rules, including Inner-Triangle Distance Pruning, Outer-Triangle Distance Pruning, and Activity Location Distance Pruning. In summary, during the process of selecting attendees and an activity location, we exploit both the spatial distances among different candidate locations as well as the distances from attendees to activity locations to effectively prune redundant search space to efficiently find the optimal solution. 

In Section \ref{analysis}, we present an Integer Linear Programming (ILP)
formulation for MRGQ which can obtain an optimal solution via a commercial solver, such as the IBM CPLEX \cite{CPLEX} parallel optimizer, one of the fastest commercial parallel solvers. However, as shown in Section \ref{Exp}, this still requires an unacceptable amount of time to find the optimal solution because MRGQ needs to simultaneously process the spatial and social dimensions. Therefore, in Section \ref{MAGS}, we design a new algorithm to efficiently process MRGQ.

\vspace{-10pt}
\subsection{Baseline Algorithms for MRGQ}\label{Baseline}
The baseline algorithms are extensions of SSGS mentioned in Section \ref{SSGQ}. While Socio-Spatial Ordering and Distance Pruning remain the same, we extend Familiarity Pruning introduced in Section \ref{pruning_SSGQ} to tailor the familiarity constraint for MRGQ. Specifically, %let $N_{v}$ represent the neighbors of $v$, 
if one of the following conditions holds, Familiarity Pruning stops moving any candidates into $S_{I}$, and the algorithm backtracks to the previous step to consider other candidate attendees.

\vspace{-13pt}
\begin{align}
|S_{I}|-\min_{v\in S_{I}}|N_{v}\cap S_{I}|>k+1 \text{, or} \label{MRGQ_SP_1}\\
\sum_{v\in S_{R}}|S_{R}\cap N_{v}|<(p-|S_{I}|)(p-|S_{I}|-k-1), \label{MRGQ_SP_2}
\end{align}
where $N_v$ is the set of neighbors of $v$ in $V$.

In Eq. (\ref{MRGQ_SP_1}), $\min_{v\in S_{I}}|N_{v}\cap S_{I}|$
represents the minimum number of neighbors for each individual $v$ in $S_{I}$%
. In other words, $|S_{I}|-\min_{v\in S_{I}}|N_{v}\cap S_{I}|-1$ is the
maximum number of unacquainted members for $v$ in $S_{I}$, and $-1$
is incorporated above to exclude $v$ herself. If $|S_{I}|-\min_{v\in
S_{I}}|N_{v}\cap S_{I}|-1>k$, at least one individual in $S_{I}$ has more
than $k$ unacquainted members in $S_{I}$. This situation violates the
familiarity constraint. Therefore, the pruning strategy holds since any
group growing from the current $S_{I}$ will never satisfy the familiarity
constraint. 

Eq. (\ref{MRGQ_SP_1}) considers the vertex degrees of the
individuals in $S_{I}$. In contrast, the pruning condition specified in Eq. (%
\ref{MRGQ_SP_2}) considers the degrees of the individuals that have not been
moved into $S_{I}$, i.e., those individuals that are in $S_{R}$. In the
right-hand-side (RHS) of Eq. (\ref{MRGQ_SP_2}), $(p-|S_{I}|)$ is the number
of individuals that need to be moved from $S_{R}$ to $S_{I}$. On the other
hand, for any solution group that satisfies the familiarity constraint, the
degree of each member is at least $(p-k-1)$ in the group. Therefore, if $%
S_{R}$ has an individual $u$ with the number of neighbors in $S_{R}$ smaller
than $(p-|S_{I}|-k-1)$, $S_{I}$ will never grow into a feasible solution
when $u$ is selected into $S_{I}$. In other words, if the total number of
neighbors that all individuals in $S_{R}$ have (i.e., $\sum_{v\in
S_{R}}|S_{R}\cap N_{v}|$) is smaller than $(p-|S_{I}|)(p-|S_{I}|-k-1)$,
selecting any $(p-|S_{I}|)$ individuals from $S_{R}$ into $S_{I}$ will never
generate a feasible solution, and thus this intermediate group can be trimmed accordingly.

For example, if $p=4$, $k=0$ and the social graph is shown in Figure \ref{fig:FIG_SSO_graph}. 
If $S_I=\{a,e\}$, then this $S_I$ can be pruned by Eq. (\ref{MRGQ_SP_1}) since $2-0>0+1$ holds, i.e., at least one vertex in current $S_I$ does not have enough friends to satisfy the familiarity constraint. Similarly, if $p=5,k=0$ and $S_I=\{a\}$, $S_R=\{b,c,d,e,f\}$, $S_I$ can also be pruned by Eq. (\ref{MRGQ_SP_2}) because $1+2+1+3+1<(5-1)(5-1-0-1)$ holds, i.e., the candidates in $S_R$ do not provide sufficient social tightness for the current $S_I$ to satisfy the familiarity constraint.

In the following, we introduce two baseline algorithms, namely SSP and SFGP.

\noindent \textbf{Sequential SSGQ Processing (SSP).}
As discussed earlier, 
an intuitive approach for answering MRGQ is to sequentially invoke algorithm SSGS for each activity location. However, even though the intermediate best solution 
can be exploited to prune inferior solutions not yet examined, 
this approach still incurs a huge query processing cost because it does not simultaneously trim multiple activity locations. Therefore, we improve SSP to SFGP as follows.

\noindent \textbf{Sequential Feasible Groups Processing (SFGP).}
In contrast to SSP that sequentially explores $|Q|$ 
branch-and-bound trees (i.e., one for each activity location), SFGP\
constructs only one branch-and-bound tree to facilitate joint exploration of the spatial and social dimensions. 
In addition to $S_{I}$ and $S_{R}$, for each node in the tree, SFGP also maintains a set $Q_{I}$ of remaining activity locations that need to be explored.
Initially, setting $S_{I}=\varnothing $, $S_{R}=V$, and $Q_{I}=Q$, SFGP first finds a reference activity location  $q_{ref}\in Q_{I}$ to guide the exploration, where $q_{ref}$ is the closest location to a candidate attendee $u\in S_R$ (i.e., $q_{ref}$ and $u$ are the spatially closest pair). As such, $q_{ref}$ can lead to a smaller total spatial distance in early stages of SFGP. %(the selection of $q_{ref}$ will be detailed later).
Afterwards, SFGP moves candidates from $S_{R}$ into $S_{I}$ according to Socio-Spatial Ordering (introduced in Section \ref{SSGQ_Algo}) based on $q_{ref}$. After moving a candidate from $S_{R}$ into $S_{I}$, SFGP determines whether $S_{I}$ can be pruned by Familiarity Pruning mentioned in Eqs. (\ref{MRGQ_SP_1}) and (\ref{MRGQ_SP_2}). If $S_I$ is pruned by Familiarity Pruning, SFGP stops moving candidates into the current $S_{I}$ and backtracks because the current $S_I$ cannot grow into any feasible solutions. Moreover, each time a candidate is moved into $S_{I}$, SFGP examines each activity location $q_{i} \in Q_{I}$ with the Distance Pruning condition (introduced in Section \ref{pruning_SSGQ}). An activity location $q_{i}$ will be removed from $Q_{I}$ if it is distant from most members in $S_{I}$ (i.e., $q_{i}$ is pruned by Distance Pruning). While expanding $S_I$, if $Q_{I}$ becomes empty (i.e., all activity locations in $Q_{I}$ are pruned), SFGP stops the expansion and backtracks.

When $S_{I}$ contains exactly $p$ candidates and satisfies the familiarity constraint, 
SFGP computes the spatial distances from $S_{I}$ to each activity location in $Q_I$, 
and extracts the activity location $q \in Q_I$ which incurs the minimum spatial distance to $S_I$. 
If the spatial distance from $S_I$ to $q$ is smaller than the current minimum distance $D$,
SFGP records $\left\langle S_{I},q \right\rangle$, updates $D$ and backtracks to examine other possible solutions. When the search space is explored, SFGP outputs the recorded best solution and the corresponding activity location.

\begin{figure}[tbp]
\centering
\subfigure[Social graph and spatial distances.]
{\	\includegraphics[scale=0.3]{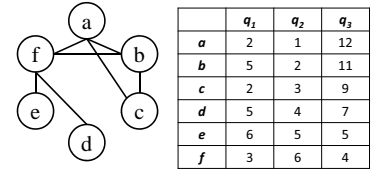} \label{FIG_SFGP_input}}
\subfigure[Branch-and-bound tree of SFGP.]
{\	\includegraphics[scale=0.3]{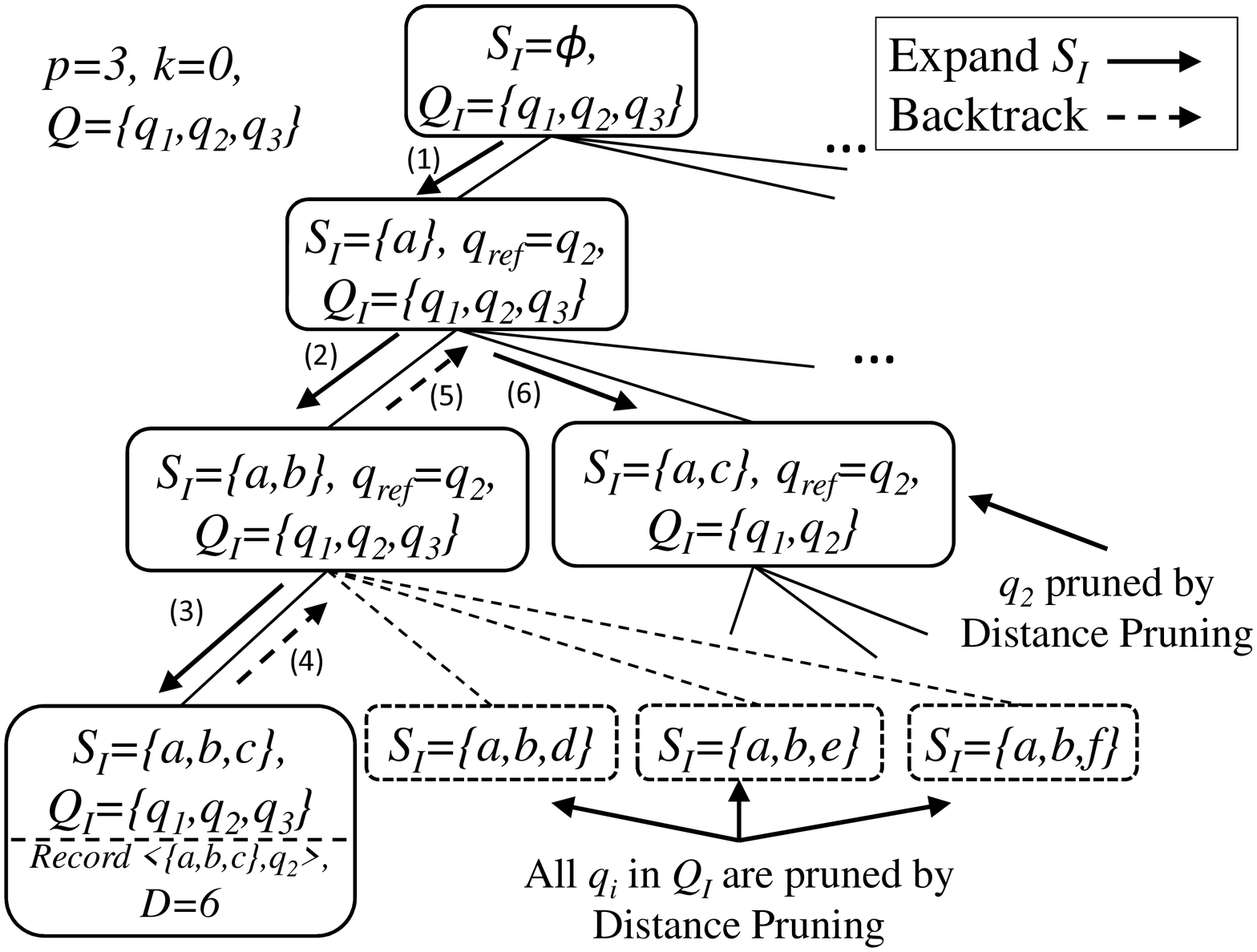}
\label{FIG_SFGP} }
\caption{Example of SFGP.}
\label{FIG_Example_SFGP}
\end{figure}

Figure \ref{FIG_SFGP} presents an example of SFGP to show that the size of $Q_{I}$
can rapidly decrease when a few more candidates are moved into $S_{I}$. The social network and the corresponding spatial distances to each $q_i \in Q_I$ are shown in Figure \ref{FIG_SFGP_input}. Assume $p=3$ and $k=0$, at the beginning, $S_I=\varnothing$, $Q_I=\{q_1,q_2,q_3\}$ and $S_R=V$. In step (1), SFGP first identifies $q_{ref}=q_2$ because $q_2$ and candidate attendee $a$ are the spatially closest pair while $a \in S_R$, and then SFGP moves $a$ from $S_R$ into $S_I$. In step (2), SFGP moves $b$ from $S_R$ into $S_I$. Note that $b$ is the candidate attendee in $S_R$ who is closest to $q_{ref}$. Here, moving $b$ into $S_I$ follows Socio-Spatial Ordering (SSO), and $S_I \cup \{b\}$ is not pruned by Familiarity Pruning. In step (3), SFGP moves $c$ into $S_I$ where $S_I = \{a,b,c\}$ satisfies the familiarity constraint. SFGP then scans over the activity locations in $Q_I$ and extracts $q_2$ because $q_2$ incurs the minimum total spatial distance to $S_I$. SFGP updates the currently best solution $\left\langle \{a,b,c\},q_2\right\rangle $ and its distance value (i.e., $D=6$) and backtracks to the previous state as step (4), i.e., $S_{I}=\{a,b\}$ and $Q_I=\{q_1,q_2,q_3\}$. SFGP then discovers that by applying Distance Pruning, all the activity locations in $Q_I$ can be removed, i.e., moving $d$, $e$ and $f$ into $S_I$ does not generate a better solution. Therefore, SFGP stops expanding the current $S_I$ and backtracks through step (5). Now, $S_I=\{a\}$ and SFGP moves $c$ into $S_I$ in step (6). In this case, $q_3$ can be removed from $Q_I$ because Distance Pruning indicates that $q_3$ will never lead to any better solutions given the current $S_I$. Therefore, SFGP only needs to examine $q_1$ and $q_2$ in the future expansion of the current $S_I$. During the process, if SFGP finds a feasible solution with a distance better than $D$, it records the solution and update $D$. SFGP repeats the above procedures and returns the best solution after the search is complete.

As compared to SSP, SFGP jointly examines the activity locations and candidate attendees, and employs Distance Pruning to effectively remove the activity locations that do not lead to better solutions. It then utilizes Familiarity Pruning to discard the intermediate groups that cannot grow into feasible solutions. Moreover, SFGP avoids the repeated explorations of different social groups, i.e., the same social group may be generated and examined for $|Q|$ times in SSP. As shown in Section \ref{Exp}, SFGP outperforms SSP. However, after carefully examining SFGP, we still find a number of areas that can be further improved, and thus propose a more efficient algorithm as detailed below.

\vspace{-10pt}
\subsection{Algorithm MAGS for MRGQ}\label{MAGS}

Although SFGP is able to prune redundant activity locations, 
it relies on sequential scans over $Q_{I}$ to determine whether a location in $Q_{I}$ can be safely pruned. 
Therefore, for every $q_i$ in $Q_I$, SFGP has to calculate a 
lower bound on the total spatial distance of the feasible solution generated 
from $S_{I}$ and $q_{i}$ according to Distance Pruning.
On the other hand, identifying $q_{ref}$ needs a scan over the activity locations in $Q_{I}$. Moreover, the selected $q_{ref}$ may not always be good because SFGP decides $q_{ref}$ before the first candidate attendee is moved into $S_I$, instead of adaptively changing $q_{ref}$ as $S_I$ grows.

To address these issues, we propose an algorithm, namely
\textit{Multiple Activity-Location Group Selection (MAGS)}, to efficiently process MRGQ. Similar to SFGP, MAGS 
processes multiple activity locations simultaneously. However, MAGS incorporates the following new ideas: a) an index of activity locations, b) new distance ordering strategies, including Single-Reference Distance Ordering and All-Pair Distance Ordering, and c) new distance pruning strategies, including Activity Location Distance Pruning, Outer-Triangle Distance Pruning and Inner-Triangle Distance Pruning. 
Using an index for the activity
locations avoids sequential scans of the activity locations in $Q_{I}$ (i.e., for the selection of $q_{ref}$ and pruning of unnecessary locations).
The new distance ordering strategies obtain $q_{ref}$ more efficiently and enable $q_{ref}$ to change during the expansion of $S_I$. As a result, feasible solutions with smaller total spatial distances can be obtained more effectively.
Moreover, the new distance pruning strategies exploit the interplay between $S_{I}$ and the activity locations, as well as the mutual distances of different locations,
to effectively and simultaneously prune multiple activity locations.

\vspace{-10pt}
\subsection{Indexing the Activity Locations}

%\begin{figure}[tbp]
%\centering
%\subfigure[Center and radius of a ball.] {\  \includegraphics[scale=0.25]{Fig_Ball.eps} \label{FIG_Ball}}
%\subfigure[The structure of BallTree.] {\  \includegraphics[scale=0.25]{Fig_BallTree.eps }
%\label{FIG_BallTree_Example}} %\label{fig:FIG_SCO}
%\caption{Example of balls and BallTree.}
%\vspace{-5pt}
%\end{figure}

\baselineskip=11.3pt
As previously mentioned, SFGP incurs many sequential scans over the activity
locations due to Distance Pruning, i.e., each time a candidate is moved into 
$S_{I}$, $Q_{I}$ needs to be scanned to determine whether some activity locations can
be pruned. Moreover, as SFGP extracts $q_{ref}\in Q_{I}$ and $u\in
S_{R}$ at the beginning, $q_{ref}$ is not always the
closest activity location for $S_{I}$ to be expanded afterward, especially when $%
S_{I}$ does not include $u$. 
Therefore, the proposed All-Pair Distance Ordering (APDO) is designed to
dynamically select $q_{ref}\in Q_{I}$ and $u\in S_{R}$ according to the
current $S_{I}$ (as described in Section 5.4). More specifically, the next
attendee $u$ that will be moved to $S_{I}$ and the corresponding $q_{ref}$
need to minimize the total spatial distance from $S_{I}\cup \{u\}$ to $%
q_{ref}$, i.e., $\min_{u\in S_{R},q_{ref}\in Q_{I}}\left\{
d_{u,q_{ref}}+\sum_{v\in S_{I}}d_{v,q_{ref}}\right\} $. Equipped with APDO,
MAGS finds good feasible solutions more quickly and prunes
search space with distance pruning strategies. However, this approach needs $O(|V|)$
sequential scans over $Q_I$ before a new candidate attendee is identified and moved into $S_I$.
%SFGP incurs many sequential scans over the activity locations due to 1) Distance Pruning - each time a candidate is moved into $S_I$, $Q_I$ is scanned to check if some activity locations can be pruned; 2) extraction of $q_{ref}$ - a sequential scan over $Q_{I}$ is performed to find $q_{ref}$ and the closest candidate $u$. Furthermore, if $q_{ref}$ is changed during the expansion of $S_I$ (as to be described in Section \ref{DistOdrSOSA}), the number of sequential scans grows even more rapidly. 

\begin{figure}[tbp]
\centering
\subfigure[Indexed with balls.]
{\	\includegraphics[scale=0.2]{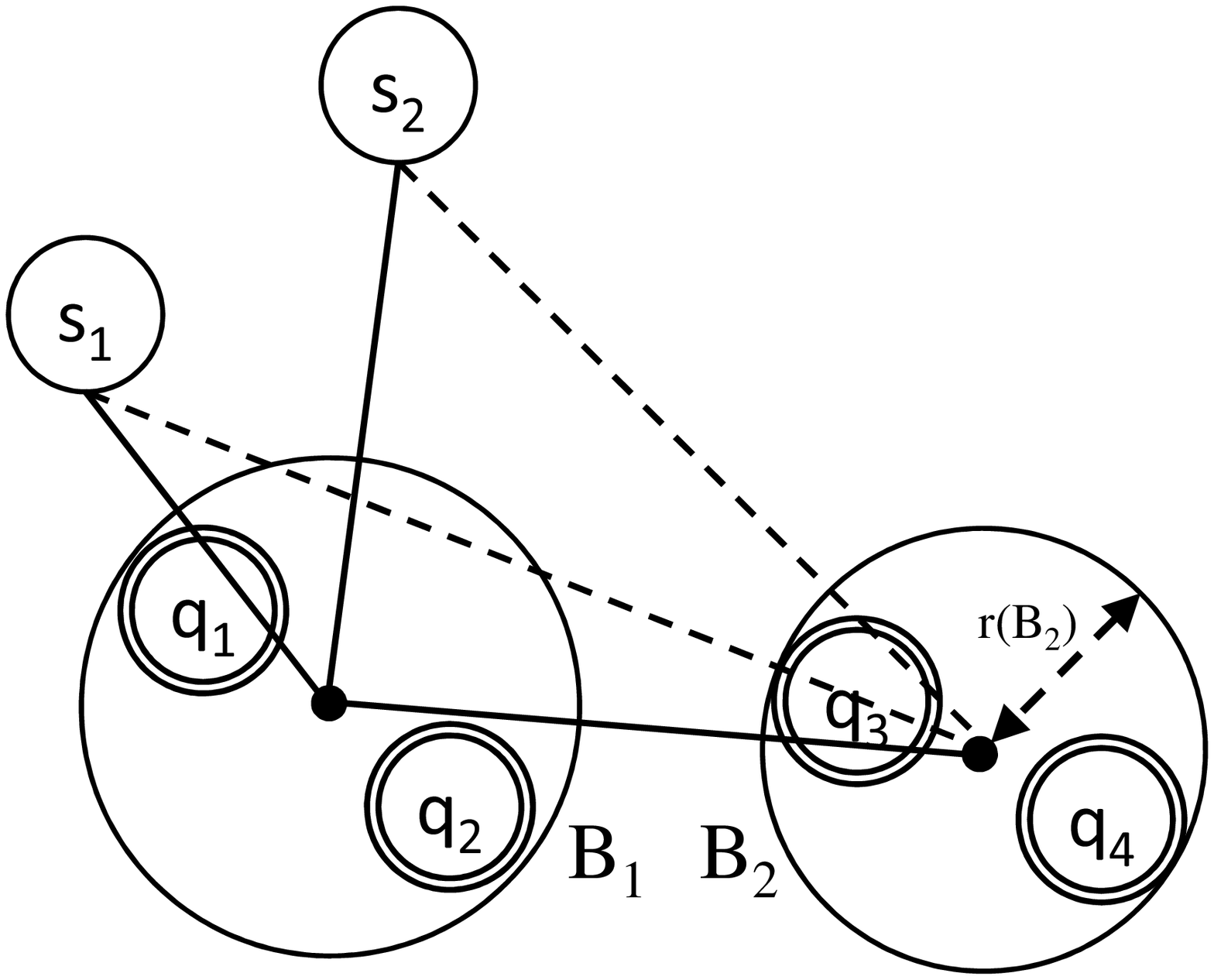}
\label{FIG_WhyTryBallTree} }
\hspace{+15pt}
\subfigure[Indexed with MBRs.]
{\	\includegraphics[scale=0.2]{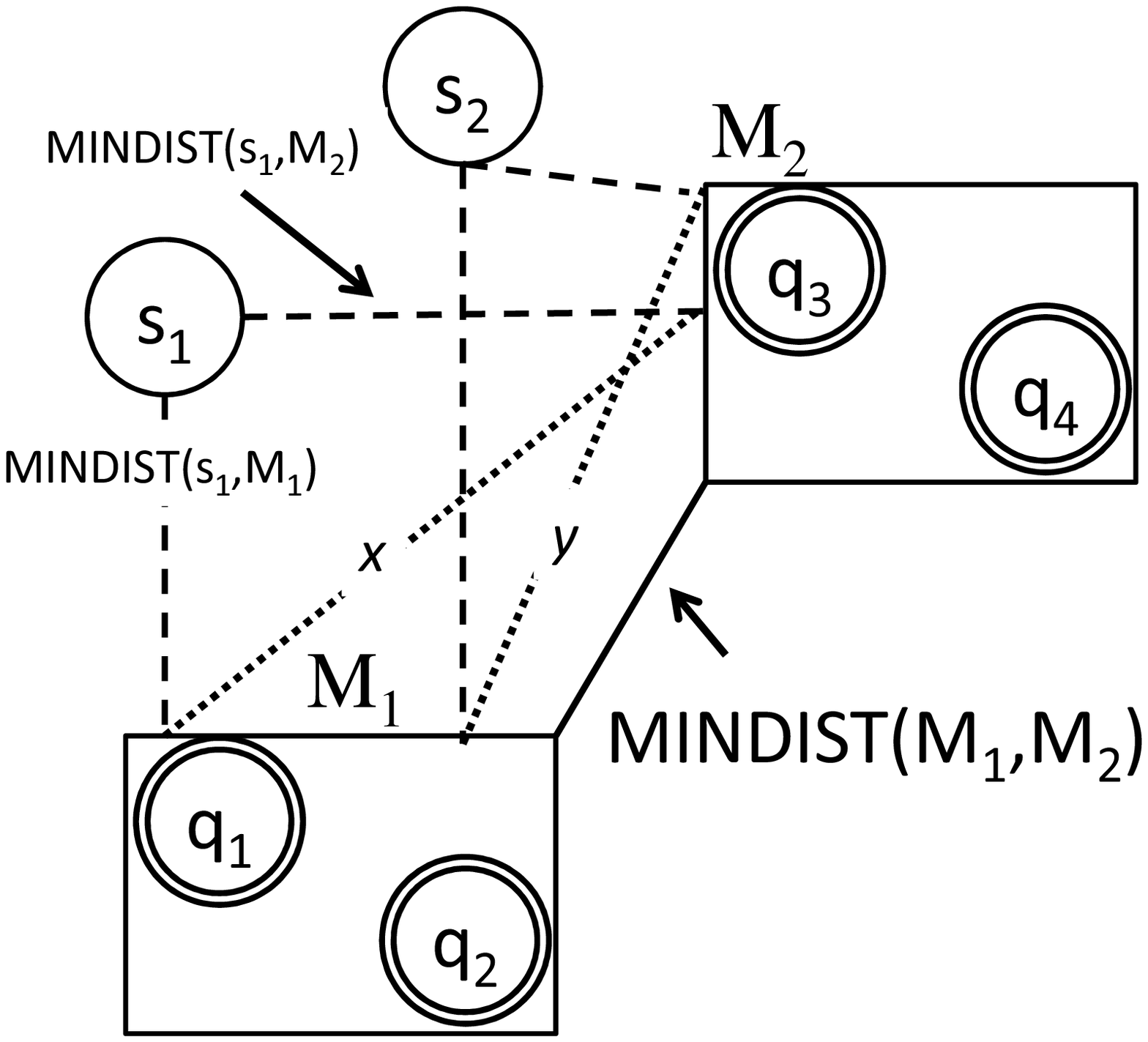} \label{FIG_WhyNotRTree}}
\caption{Comparisons of R-Tree and BallTree.}
\vspace{-20pt}
\label{FIG_IndexSelection}
\end{figure}

One way to avoid sequential scans over $Q_I$ is to index the activity locations in an index structure. This may facilitate rapid estimation of the spatial distances from activity candidates to potential activity locations and thus allow distance pruning strategies to immediately remove redundant activity locations from $Q_I$. With such an index structure, triangular inequality may be exploited in distance pruning strategies to further reduce distance computations (detailed later). Although the index structure has to be constructed at runtime, it can be reused many times in query processing.

We adopt BallTree \cite{SO89} to index the activity locations. In BallTree, each activity location $q_{i}\in Q$ is stored 
as a leaf node, and each internal node in BallTree is the smallest ball covering all the children balls. Here, a ball $B$ is associated with its center $ctr(B)$ and radius $r(B)$. The distance lower bound from a candidate $u$ to a ball $B$ on 2D space can be computed as $MINDIST(u,B)=d_{u,ctr(B)}-r(B)$.
The leaves of the BallTree are the
activity locations, while the internal nodes in the tree corresponds to a ball containing multiple activity locations. 

%We adopt BallTree \cite{SO89} to index the activity locations. In BallTree, each activity location $q_{i}\in Q$ is stored 
%as a leaf node, and each internal node in BallTree is the smallest ball covering all the children balls. Here, a ball $B$ is associated with its center $ctr(B)$ and radius $r(B)$, as illustrated in Figure \ref{FIG_Ball}. The distance lower bound from a candidate $u$ to a ball $B$ on 2D space can be computed as $MINDIST(u,B)=d_{u,ctr(B)}-r(B)$.
%The leaves of the BallTree in Figure \ref{FIG_BallTree_Example} are the
%activity locations, while the internal nodes in the tree corresponds to a ball containing multiple activity locations. 

BallTree enables the removal of many unqualified locations at once, as illustrated in Figure \ref{FIG_WhyTryBallTree}. 
To simultaneously explore and prune multiple activity locations, 
a lower bound on the total spatial distance from $S_I=\{s_1,s_2\}$ to a ball, e.g., $B_1$, can be derived. If this distance lower bound exceeds the currently best solution value $D$, it assures that no activity location in $B_1$ will produce a better solution with any social group grown from $S_I$. Thus, all activity locations in $B_1$ can be safely pruned. In Figure \ref{FIG_WhyTryBallTree}, $\sum_{s_{i}\in S_{I}}MINDIST(s_{i},B_{1})$ serves as a lower bound on the total spatial distance from $S_{I}$ to $q_{1}$ and $q_{2}$. Moreover, we can employ triangular inequality to avoid the distance computation of 
$\sum_{s_{i}\in S_{I}}MINDIST(s_i,B_2)$, i.e., $\sum_{s_{i}\in S_{I}}MINDIST(s_{i},B_{2})=\sum_{s_{i}\in
S_{I}}(d_{s_{i},ctr(B_{2})}-r(B_2))<|S_{I}|\cdot d_{ctr(B_{1}),ctr(B_{2})}+\sum_{s_{i}\in S_{I}}d_{s_{i},ctr(B_{1})}-|S_I|\cdot r(B_{2})
$. Therefore, only the distance from $ctr(B_{1})$ to $ctr(B_{2})$ needs to be computed, together with $\sum_{s_{i}\in S_{I}}d_{s_{i},ctr(B_{1})}$, to derive a lower bound on the spatial distance from $S_{I}$ to $B_{2}$. In summary, instead of invoking sequential scans which need $|S_I|\cdot m$ distance computations to find the total spatial distances from $S_I$ to $m$ activity locations, indexing activity locations in BallTree requires only $|S_I|+(n-1)$ distance computations, where $n$ is the number of balls. 

%In addition to index the candidate attendees in an R-Tree (similar to SFGP), MAGS also indexes the activity locations.

An alternative index is R-Tree, but we argue that BallTree is more suitable for indexing activity locations here. Figure \ref{FIG_WhyNotRTree} illustrates an example where the activity locations are indexed in an R-Tree. As shown, minimum bounding rectangles (MBRs) are used to provide boundary information over locations inside them. In Figure \ref{FIG_WhyNotRTree}, $\sum_{s_{i}\in S_{I}}MINDIST(s_{i},M_{1})$ serves as a lower bound on the total spatial distance from $S_{I}$ to $q_{1}$ and $q_{2}$, where $MINDIST(s_{i},M_{1})$ denotes the minimum 
distance from $s_i$ to MBR $M_1$. However, it is difficult to employ triangular inequality with R-Tree to quickly obtain a lower bound on $\sum_{s_{i}\in S_{I}}MINDIST(s_{i},M_{2})$. As shown in Figure \ref{FIG_WhyNotRTree}, where $MINDIST(s_1,M_1)+x>MINDIST(s_1,M_2)$ holds, the inequality  
$MINDIST(s_1,M_1)+MINDIST(M_1,M_2)>MINDIST(s_1,M_2)$ is not guaranteed to hold because $MINDIST(M_1,M_2)\leq x$.
Therefore, it is necessary to compute $MINDIST(s_{1},M_{2})$ and $MINDIST(s_{2},M_{2})$ directly,
incurring $|S_{I}|\cdot \hat{n}$ on-line distance computations to derive all lower bounds, where $\hat{n}$ is the number of MBRs. In contrast, BallTree needs only $|S_I|+(n-1)$ distance computations with $n$ balls. Therefore, BallTree is preferable to R-Tree in our MAGS design.

%To address this issue, a promising way is to index the activity locations as balls, as shown in Figure \ref{FIG_WhyTryBallTree}, 
%where the radius of ball $B_{2}$ is $r(B_2)$, and $ctr(B_{1})$, $ctr(B_2)$ are the centers of balls $B_{1}$ and $B_2$, respectively.
%For a ball $B$ and any point $s$ on the 2D space, the lower bound of
%distance from $s$ to any activity location within $B$ is $MINDIST(s,B)=d_{s,ctr(B)}-r(B)$.  
%Thus, a lower bound on the total spatial distance from $S_{I}$ to $q_{3}$ or $q_{4}$ is 
%$\sum_{s_{i}\in S_{I}}MINDIST(s_{i},B_{2})=\sum_{s_{i}\in
%S_{I}}(d_{s_{i},ctr(B_{2})}-r(B_2))<|S_{I}|\cdot d_{ctr(B_{1}),ctr(B_{2})}-\sum_{s_{i}\in S_{I}}d_{s_{i},ctr(B_{1})}-|S_I|\cdot r(B_{2})
%$.
%
%Therefore, it only needs to compute the distance from $ctr(B_{1})$ to $ctr(B_{2})$, together with $\sum_{s_{i}\in S_{I}}d_{s_{i},ctr(B_{1})}$, for deriving a lower bound on the spatial distance from $S_{I}$ to $B_{2}$. With balls, 
%instead of MBRs in R-Tree, now it is more efficient since only $|S_I|+(n-1)$ distance computations are involved,
%where $n$ is the number of balls. 
%
%\baselineskip=11.5pt
%In BallTree, each activity location $q_{i}$ is stored 
%as a leaf node, and each internal node in BallTree is the smallest ball covering all the children balls.
%Here, a ball $B$ is associated with its center $ctr(B)$ and radius $r(B)$%
%, as illustrated in Figure \ref{FIG_Ball}. The leaves of the BallTree in Figure \ref{FIG_BallTree_Example} are the
%activity locations, while the interior nodes in the tree corresponds to a ball containing multiple activity locations. 

BallTree brings two advantages to MAGS:
1) BallTree enables the design of efficient distance ordering strategies. By traversing
both R-Tree (for indexing candidate attendees) and BallTree (for indexing activity locations), 
our proposed distance ordering strategies avoid redundant examinations of candidate attendees and activity locations
to extract the reference activity location $q_{ref}$.
The new distance ordering strategies, combined with the original Socio-Spatial Ordering mentioned in Section \ref{SSGQ_Algo}, 
are promising to find good feasible solutions quickly and prune redundant search space effectively. 
2) BallTree enables distance-based pruning of activity locations at once in the early
stages. Moreover, the lower bound on the total spatial distance from a set of balls to $S_{I}$ can be quickly obtained to facilitate pruning.
%\footnote{Although BallTree needs to be constructed at runtime, the experimental results manifest that the query
%computation time can be significantly reduced when BallTree is exploited.}. 
In the following, we first propose two distance ordering strategies and then introduce the distance pruning strategies based on R-Tree and BallTree.

\vspace{-7pt}
\subsection{Distance Ordering}
\label{DistOdrSOSA}

%While Socio-Spatial Ordering in SSGS is applicable to MAGS, its design does not consider selections of activity locations. Here we propose two new distance ordering strategies for MAGS, namely Single-Refernce Distance Ordering (SRDO) and All-Pair Distance Ordering (APDO). It adaptively changes the optimal activity location according to different $S_{I}$, and always chooses the best activity location when a new attendee is included into $S_{I}$ to minimize the total spatial distance from $S_{I}$ to the new reference activity location $q_{ref}$. In addition, we also propose a simplified version of APDO, namely Single-Reference Distance Ordering (SRDO), which selects the reference activity location only when selecting the first candidate attendee of each intermediate group. SRDO is regarded as a baseline in our experiments. 

While Socio-Spatial Ordering in SSGS is applicable to MAGS, its design does not consider selections of activity locations. Here we propose two new distance ordering strategies for MAGS: 
(1) Single-Reference Distance Ordering (SRDO). It selects the activity location along with the first candidate attendee, $v_{seed}$, for $S_{I}$. Note that the total spatial distance of the
feasible solutions obtained by SRDO may not be minimal since only a single location $q_{ref}$ is fixed as a reference. (2) All-Pair Distance Ordering (APDO). It adaptively changes the optimal activity location according to different $S_{I}$, and always chooses the best activity location when a new attendee is included into $S_{I}$ to minimize the total spatial distance from $S_{I}$ to the new reference activity location $q_{ref}$.

%%[Move to Appendix]
\vspace{+3pt}
\noindent \textbf{Single-Reference Distance Ordering (SRDO). }
At the beginning, i.e., $S_{I}=\varnothing $, SRDO starts by selecting a seed candidate $v_{seed}$ and a reference activity location $q_{ref}$ such that $d_{v_{seed},q_{ref}}$ is minimal. However, to avoid excessive distance computations, 
we fix $q_{ref}$ as $S_{I}$ grows. 
While SRDO requires later examination of other activity locations, the minimized distance may effectively eliminate consideration of many potential activity locations. 
To efficiently obtain $v_{seed}$ and $q_{ref}$, we traverse R-Tree
(indexing the candidate attendees) and BallTree (indexing activity locations) simultaneously,
to reduce the number of distance computations. 
To further improve the efficiency, a distance lower bound from any
candidate within an MBR $M_{i}$ to any activity location within a ball $B_{j}$%
, $MINDIST(M_{i},B_{j})$, is derived as $MINDIST(M_{i},B_{j})=MINDIST(M_{i},ctr(B_{j}))-r(B_{j})$, where $MINDIST(M_{i},ctr(B_{j}))$ is the minimum distance from $M_{i}$
to the center of $B_{j}$, and $r(B_{j})$ is the radius of $B_{j}$.
$MINDIST(M_{i},B_{j})$ represents a distance lower bound from any
candidate within $M_{i}$ to any activity location in $B_{j}$, which is particularly useful to determine redundant examinations of candidate attendees and activity locations located in distant MBRs and balls in R-Tree and BallTree.

More specifically, SRDO maintains two lists, $U_{R}$ and $U_{B}$, to record the
traversal status of R-Tree and BallTree. 
Initially, we insert the root of
R-Tree into $U_{R}$ and the root of BallTree into $U_{B}$. Then, at each
stage, we find the MBR $M_{i}$ in $U_{R}$ and the ball $B_{j}$ in $U_{B}$
that incur the minimum $MINDIST(M_{i},B_{j})$. If $M_{i}$ is not a leaf node
in R-Tree, we pop $M_{i}$ from $U_{R}$ and insert its children back into $%
U_{R}$, while a non-leaf node in BallTree is performed similarly. If the extracted $M_{i}$
and $B_{j}$ are both leaf nodes, they are assigned as $v_{seed}$ and $q_{ref}$, respectively.  
Note that the entries in $U_R$ and $U_B$ are popped in accordance with the shortest distance between 
them, $v_{seed}$ and $q_{ref}$ are indeed the closet attendee-location pair. 
Each candidate attendee $v_{i}$ and activity location $q_j$ in any other MBR $M_i$ and ball $B_j$ must incur a larger spatial distance since $MINDIST(M_{i},B_{j})$ is a lower bound, and $d_{v_{seed}, q_{ref}} \leq MINDIST(M_{i},B_{j}) \leq d_{v_{i},q_{j}}$.
Therefore, this approach effectively avoids examining
attendees and locations that are mutually distant because their corresponding MBRs and balls
will never be extracted from the lists. Moreover, if $d_{v_{seed},q_{ref}}>t$ (where $t$ is the spatial radius), MAGS can stop since there is no feasible solution in this case. 

\begin{figure}[tbp]
\centering
\includegraphics[scale=0.3]{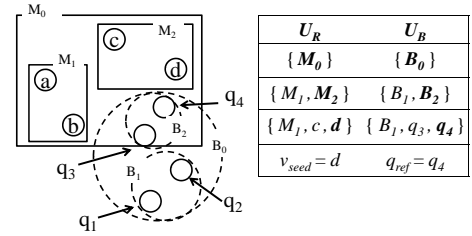}
\caption{Example of SRDO.}
\vspace{-5pt}
\label{FIG_Example_SRDO}
\end{figure}

Figure \ref{FIG_Example_SRDO} presents an illustrative example for SRDO. Assume there are four
candidates $\{a,b,c,d\}$ indexed by an R-Tree and four activity locations $%
\{q_{1},q_{2},q_{3},q_{4}\}$ indexed by a BallTree. To find $v_{seed}$ and $%
q_{ref}$, we first insert the root of R-Tree, $M_{0}$, into $U_{R}$, and
insert the root of BallTree, $B_{0}$, into $U_{B}$. There is only one
element in each list, and $MINDIST(M_{0},B_{0})=0$ since they overlap. Thus, SRDO
extracts $M_{0}$ and $B_{0}$ and insert their children into $U_{R}$ and $%
U_{B} $, respectively. Now, $U_{R}=\{M_{1},M_{2}\}$ and $U_{B}=\{B_{1},B_{2}%
\}$. SRDO then extracts $M_{2}$ and $B_{2}$ from each list since $%
MINDIST(M_{2},B_{2})$ is the smallest one. Afterwards, we insert the children of $M_{2}$
and $B_{2}$ into the lists, respectively, and now $U_{R}=\{M_{1},c,d\}$ and 
$U_{B}=\{B_{1},q_{3},q_{4}\}$. SRDO finds that $d$ and $q_{4}$ 
incur the minimum spatial distance and assigns $v_{seed}$ as $d$ and $q_{ref}$ as $q_{4}$.

Once $v_{seed}$ and $q_{ref}$ are extracted, $q_{ref}$ in SRDO is fixed. The candidate attendees chosen later still need to follow Socio-Spatial Ordering to maintain the required social tightness of $S_{I}$, and Familiarity Pruning is employed to prune the intermediate solutions that will not become feasible groups. Moreover, distance pruning strategies based on R-Tree and BallTree are employed to remove activity locations (detailed later) that will never produce a better solution.

\vspace{+3pt}
\noindent \textbf{All-Pair Distance Ordering (APDO).}
With SRDO, as $S_{I}$ grows, the $q_{ref}$ selected initially may not be the eventual activity location with the minimum
total distance to $S_{I}$. Figure \ref%
{FIG_Example_APDO} presents an illustrative example with $p=3$ and $\{a,b,c,d\}$ as the candidates indexed by
R-Tree, while $\{q_{1},q_{2},q_{3},q_{4}\}$ are activity locations indexed by
BallTree. SRDO finds $d$ for $v_{seed}$ and $q_{4}$ for $q_{ref}$.
%since $d_{v_{seed},q_{4}}$ incurs the minimum spatial distance
%among all possible pairs of attendees and activity locations. 
Thereafter, $a$ and $c$ are moved into $S_{I}$.
However, since $a$ and $c$ are distant from $q_{4}$, the solution obtained by SRDO, i.e., 
$\left\langle \{a,c,d\},q_4\right\rangle $, incurs a large total spatial distance. In contrast, a better feasible solution is
$\left\langle \{a,b,d\},q_3\right\rangle $, which greatly reduces the
total spatial distance. Therefore, we propose All-Pair Distance Ordering (APDO), 
to select proper candidates from $S_{R}$ and adaptively switch $q_{ref}$ to the most suitable activity location.

\begin{figure}[tbp]
\centering
\includegraphics[scale=0.3]{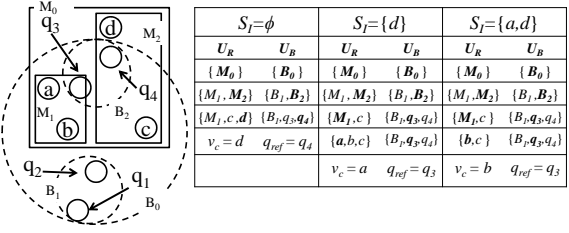}
\caption{Example of APDO.}
\vspace{-20pt}
\label{FIG_Example_APDO}
\end{figure}

We propose All-Pair Distance Ordering (APDO) to select proper candidates from $S_{R}$ and adaptively switch $q_{ref}$ to the most suitable activity location. More specifically, APDO simultaneously chooses $q_{ref}$ and a candidate attendee $v_{c}$ to expand $%
S_{I}$ at each iteration, such that the total spatial distance from $S_{I}$ to the selected $q_{ref}$ is minimized, i.e.,
\vspace{-5pt}
\begin{equation}
\min_{v_{c}\in S_{R},q_{ref}\in Q_{I}}\left \{ d_{v_{c},q_{ref}}+\sum_{v\in
S_{I}}d_{v,q_{ref}}\right \}. 
\vspace{-5pt}
\end{equation}

A straightforward approach to select $v_{c}$ and $q_{ref}$ is to scan over 
the entire sets of $S_{R}$ and $Q_I$. 
However, this approach requires $\left(
|S_{R}|+|S_{I}|\right) \cdot |Q_I|$ distance computations when we move a
candidate into $S_{I}$. To reduce this overhead,
we traverse both R-Tree and BallTree simultaneously, to reduce
unnecessary distance computations.

Two lists $U_{R}$ and $U_{B}$ are maintained during the traversal of 
R-Tree and BallTree. At each stage, MBR\ $M_{i}$
and ball $B_{j}$ are extracted from $U_{R}$ and $U_{B}$ based on the following score function.
\vspace{-8pt}
\begin{eqnarray}
\nonumber
&&\min_{M_{i}\in U_{R},B_{j}\in U_{B}}\left \{ \sum_{v\in
\nonumber
S_{I}}MINDIST(v,B_{j})+ \right. \\
&&\left. MINDIST(M_{i},B_{j})\right \} \label{EQU_APDO}
\end{eqnarray}%
where $MINDIST(v,B_{j})=d_{v,ctr(B_{j})}-r(B_{j})$ and $MINDIST(M_{i},B_{j})$ cannot exceed $t$.
In Eq. (\ref{EQU_APDO}), the first term represents the minimum total spatial distance from $S_I$ to 
any activity location within $B_{j}$, while the second term represents the
minimum spatial distance from a candidate attendee in $M_i$ to an activity location in $B_j$. 

%The above equation asks
%to find the MBR $M_{i}$ and ball $B_{j}$ where the sum of total spatial
%distance of $S_{I}$ to $B_{j}$ augmented by the spatial distance from $M_{i}$
%to $B_{j}$ is minimum. Therefore, $M_{i}$ here gives the hint that it may
%contain $v_{c}$, and $B_{j}$ here gives the hint that it may contain $q_{ref}$. 

After extracting $M_i$ and $B_j$ from Eq. (\ref{EQU_APDO}), if $M_{i}$ is not a leaf node on R-Tree, 
we pop it from $U_{R}$
and insert its children into $U_{R}$. Similarly, if $B_{j}$ is a non-leaf
node on BallTree, we also pop it from $U_{B}$ and insert its children into $%
U_{B}$. As such, APDO extracts $v_c$ and $q_{ref}$ without accessing the candidate attendees and activity locations distant from each other. We repeat the above procedure until $M_{i}$ and $B_{j}$
are both leaf nodes and $M_{i}\in S_{R}$. 
Finally, we move $v_{c}$ from $S_{R}$ into $S_{I}$ and continue the branch-and-bound
search. Moreover, during the above procedure, if $MINDIST(v_{c},B_{i})>t$ for a ball $B_{i}$,
all the activity locations within $B_{i}$ can be removed from $Q_I$ since no 
activity locations in $B_{i}$ satisfies the spatial radius constraint.
APDO iteratively extracts $v_c$ and $q_{ref}$ which incur the minimum spatial distance so as to avoid the situation where $q_{ref}$ is only close to a small number of candidate attendees but distant from the others. Moreover, APDO also allows for the early pruning of activity locations that are distant from the candidate attendees. This effectively reduces computation overhead when performing distance pruning strategies afterwards. 

Figure \ref{FIG_Example_APDO} presents an example with four candidates $\{a,b,c,d\}$ 
indexed by R-Tree and four activity locations $%
\{q_{1},q_{2},q_{3},q_{4}\}$ indexed by BallTree. 
%Initially, APDO behaves the same as SRDO since $S_{I}=\varnothing $ (see the first column in the table). 
Initially, when $S_{I}=\varnothing$, APDO finds the first $v_c$ and the corresponding $q_{ref}$ as follows (see the first column in the table). 
APDO first inserts the root of R-Tree, $M_{0}$, into $U_{R}$, and
inserts the root of BallTree, $B_{0}$, into $U_{B}$. There is only one element in each list, and $MINDIST(M_{0},B_{0})=0$ since they overlap. 
Thus, APDO extracts $M_{0}$ and $B_{0}$ and inserts their children into $U_{R}$ and $U_{B} $, respectively. 
Now, $U_{R}=\{M_{1},M_{2}\}$ and $U_{B}=\{B_{1},B_{2}\}$. APDO then extracts $M_{2}$ and $B_{2}$ from each list since 
$MINDIST(M_{2},B_{2})$ is the smallest one. Afterwards, we insert the children of $M_{2}$ and $B_{2}$ into the lists, respectively, and now $U_{R}=\{M_{1},c,d\}$ and $U_{B}=\{B_{1},q_{3},q_{4}\}$. APDO finds that $d$ and $q_{4}$ incur the minimum spatial distance, and the first candidate to be moved into $S_{I}$ is $d$ with the corresponding $q_{ref}$ as $q_{4}$. 

To choose the second candidate $v_{c}$ and update $q_{ref}$ (see the second column in the table), we insert the roots $M_{0}$ and $B_{0}$
of the R-Tree and BallTree into $U_{R}$ and $U_{B}$, respectively. Then $M_{0}$ and $B_{0}$
are extracted to insert their children, i.e., $U_{R}=\{M_{1},M_{2}\}$ and $U_{B}=\{B_{1},B_{2}\}$. 
Since $M_{2}$ and $B_{2}$
minimize Eq. (\ref{EQU_APDO}), their children are inserted, i.e., $U_{R}=\{M_{1},c\}$ and 
$U_{B}=\{B_{1},q_{3},q_{4}\}$. Note that here $d$ is not inserted into $U_{R}$
since it is not within $S_{R}$. Now, $M_{1}$ and $q_{3}$ minimize 
Eq. (\ref{EQU_APDO}) since $MINDIST(M_{1},q_{3})=0$, i.e., $M_1$ and $q_3$ overlap. Therefore, $%
M_{1} $ is popped from $U_{R}$ with its children inserted back into $U_{R}$.
Thus, $U_{R}=\{a,b,c\}$ and $U_{B}=\{B_{1},q_{3},q_{4}\}$. Among them, $%
\sum_{v\in S_{I}}MINDIST(v,q_{3})+MINDIST(a,q_{3})$ is the minimum.
In other words, $v_{c}=a$ and $q_{ref}=q_{3}$. It is worth noting that at this
stage, $q_{ref}$ changes from $q_{4}$ to $q_{3}$ since $q_{3}$ incurs a
smaller total spatial distance to $S_{I}\cup \{a\}$. Therefore, the second candidate to be moved into $S_I$ is $a$. The third column in the table details the extraction of the next $v_c$ and the corresponding $q_{ref}$, where $v_c=b$ and $q_{ref}=q_3$. After $b$ is moved into $S_I$, $S_{I}=\{a,b,d\}$, $q_{ref}=q_{3}$ is the first feasible solution. In addition, APDO does not need to
examine the children of $B_{1}$ since they are far away from the candidates.

\vspace{-15pt}
\subsection{Distance Pruning Strategies}
\baselineskip=11.3pt
To avoid examining redundant activity locations, 
a simple approach is to apply Distance Pruning (see Section \ref{pruning_SSGQ}) to derive the lower bounds on the total spatial distance from $S_{I}$
to each activity location. If the lower bound is larger than the currently best
solution value, the activity location can be safely discarded from
future expansions of $S_{I}$. However, the above approach is computation intensive because the total distance from each attendee in $S_{I}$ to each activity location needs to be obtained. In the following, we introduce a number of new pruning strategies designed to boost the efficiency in trimming redundant search space when a new attendee is added to $S_I$. 

%Given a ball $B_y$ and the currently best solution value $D$, the new distance pruning strategies employed by MAGS to safely prune activity locations in $B_y$ from $Q_I$ follow the following rule: $\Lambda(S_I,B_y)+(p-|S_I|)\cdot \Lambda(B_y,S_R)\geq D$, where $\Lambda(S_I,B_y)$ is a lower bound on the total spatial distance from $S_I$ to $B_y$ and $(p-|S_I|)\cdot \Lambda(B_y,S_R)$ is a lower bound on the total spatial distance from selecting $p-|S_I|$ attendees in $S_R$ to $B_y$. Therefore, techniques for deriving these two lower bounds are needed. 
%We first propose \textit{Outer-Triangle Distance Pruning (OTDP)} and \textit{Inner-Triangle Distance Pruning (ITDP)} with triangular inequality. OTDP computes a distance lower bound with the aid of a ball that has been examined, and ITDP derives the lower bound with only the attendees inside $S_I$. Furthermore, ITDP and OTDP derive $(p-|S_I|)\cdot \Lambda(B_y,M_i)$ with the aid of BallTree and R-Tree. We also extend the original Distance Pruning in Section \ref{SSGQ_Algo} to \textit{Activity Location Distance Pruning (ALDP)}, which derives the lower bounds with the help of R-Tree and BallTree to facilitate pruning of activity locations in balls simultaneously.

We first propose \textit{Outer-Triangle Distance Pruning (OTDP)} and \textit{Inner-Triangle Distance Pruning (ITDP)} to derive the distance lower bounds with triangular inequality, which incur only small computation overhead. We then propose \textit{Activity Location Distance Pruning (ALDP)}, which derives the lower bounds with the help of R-Tree and BallTree to facilitate pruning of activity locations in balls simultaneously. In the following, we first discuss OTDP and ITDP for pruning single locations (point versions). This is then extended to pruning balls of locations (ball versions). Since points can be viewed as degenerated balls, the point versions of OTDP and ITDP can be treated as special cases of ball versions. 

\vspace{+3pt}
\noindent \textbf{Outer-Triangle Distance Pruning.} 
The strategy is to derive a lower bound on the total
spatial distance from $S_{I}$ to an activity location $q_{y}$ according to the total spatial distance from
$S_{I}$ to another activity location $q_{x}$ derived before.
Here, \textit{Outer-Triangle} indicates that the derivation of triangular inequality is through activity locations, i.e., outside $S_I$. On the other hand, \textit{Inner-Triangle Distance Pruning} (which will be detailed later), derives the distance lower bounds with triangular inequality purely based on the attendees in $S_I$.

Consider an activity location $q_y$ under examination. Let $q_x$ be an examined location, $d_{s_{i},q_{x}}$ denote the
spatial distance from an attendee $s_{i}\in S_{I}$ to $q_{x}$, and $d_{q_{x},q_{y}}$ denote the spatial distance from $q_{x}$ to $q_{y}$. 
As shown in Figure \ref{FIG_OuterTrianglePruning}, the lower bound on the spatial distance from $s_i$ to $q_y$ can be derived 
as $d_{q_x,q_y}-d_{s_i,q_x}<d_{s_i,q_y}$ according to triangular inequality. 
Therefore, a lower bound on the total spatial distance from $S_{I}$ to $q_{y}$ could be computed as
$\sum \nolimits_{i=1}^{|S_{I}|}(d_{q_x,q_y}-d_{s_{i},q_{x}}) = |S_{I}|\cdot d_{q_{x},q_{y}}-\sum \nolimits_{i=1}^{|S_{I}|}d_{s_{i},q_{x}}$. On the other hand, to compose a group with exactly $p$ attendees, MAGS needs to select the remaining $p-|S_{I}|$ attendees from $S_R$ into $S_I$. A lower bound on the total spatial distance of these $p-|S_I|$ attendees to $q_y$ is $(p-|S_I|)\cdot d_{v_{\min },q_{y}}$, where $d_{v_{\min },q_{y}}$ denotes the minimum spatial distance from $q_{y}$ to any candidates in $S_{R}$. Therefore, let $D$ denote the currently best solution value, the following lemma specifies Outer-Triangle Distance Pruning.

\begin{figure}[tbp]
\centering
\subfigure[For points.] {\
\includegraphics[scale=0.31]{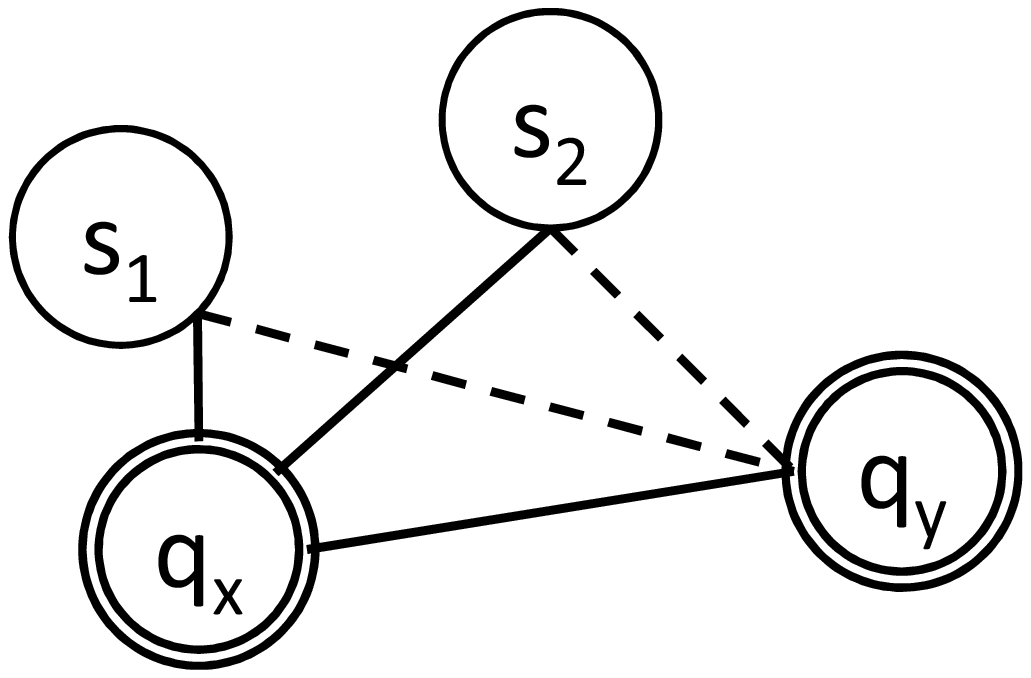}
\label{FIG_OuterTrianglePruning}}
\subfigure[For balls.]{\
\includegraphics[scale=0.31]{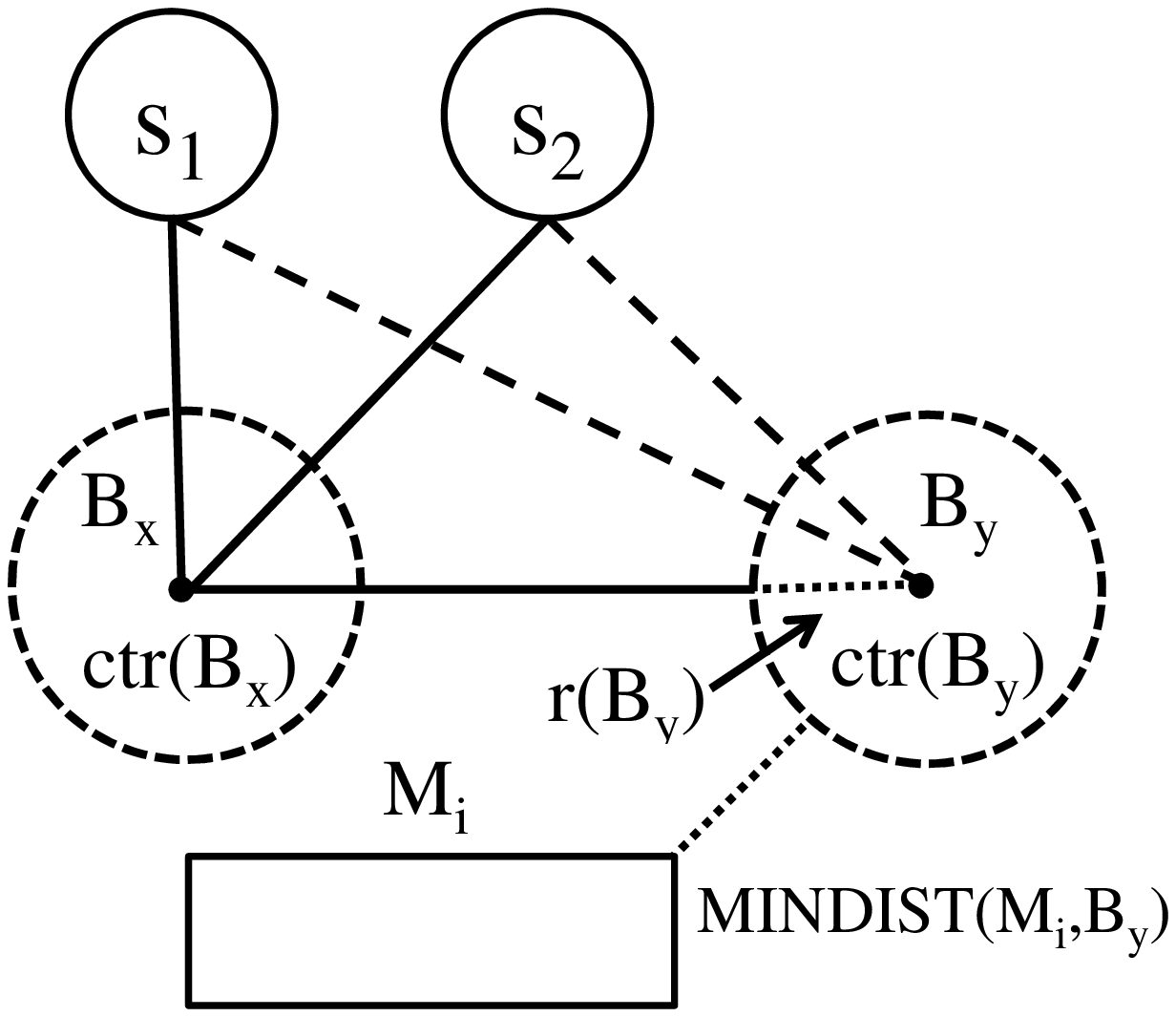}
\label{FIG_Ext_OTDP}}
\caption{Outer-Triangle Distance Pruning.}
\vspace{-20pt}
\end{figure}

\begin{lemma} \label{Lemma_OTDP}
If $|S_{I}|\cdot
d_{q_{x},q_{y}}-\sum \nolimits_{i=1}^{|S_{I}|}d_{s_{i},q_{x}}+(p-|S_{I}|)\cdot d_{v_{\min },q_{y}}\geq D$, %and
%$d_{q_{x},q_{y}}\geq \max_{s_{i}\in S_{I}}d_{s_{i},q_{x}}$ hold,
$q_{y}$ never produces a better solution for any set of candidates expanded from $S_{I}$.
\end{lemma}

\begin{proof}
In the above inequality, the first two terms represent a
lower bound on the total spatial distance from $S_{I}$ to $q_{y}$, and the
third term is a lower bound on the total spatial distance from $q_{y}$ to any $%
(p-|S_{I}|)$ candidates in $S_{R}$.
As shown in Figure \ref{FIG_OuterTrianglePruning}, 
from triangular inequality, if $d_{q_{x},q_{y}}\geq
\max_{s_{i}\in S_{I}}d_{s_{i},q_{x}}$, then $%
d_{q_{x},q_{y}}-d_{s_{i},q_{x}}<d_{s_{i},q_{y}}$, $\forall 1\leq i\leq
|S_{I}|$ must hold. Therefore, $\sum%
\nolimits_{i=1}^{|S_{I}|}(d_{q_{x},q_{y}}-d_{s_{i},q_{x}})\leq
\sum_{i=1}^{|S_{I}|}d_{s_{i},q_{y}}$, which can be written as $|S_{I}|\cdot
d_{q_{x},q_{y}}-\sum \nolimits_{i=1}^{|S_{I}|}d_{s_{i},q_{x}}\leq \sum_{i=1}^{|S_{I}|}d_{s_{i},q_{y}}$. 
Note that $d_{q_{x},q_{y}}\geq
\max_{s_{i}\in S_{I}}d_{s_{i},q_{x}}$ is necessary to be satisfied, otherwise the left-hand-side of $d_{q_{x},q_{y}}-d_{s_{i},q_{x}}<d_{s_{i},q_{y}}$ is not guaranteed to be a non-negative value, and  $\sum
\nolimits_{i=1}^{|S_{I}|}(d_{q_{x},q_{y}}-d_{s_{i},q_{x}})$ is not able to act as a lower bound on $\sum_{i=1}^{|S_{I}|}d_{s_{i},q_{y}}$.
On the other hand, $(p-|S_{I}|)d_{v_{\min },q_{y}}$ is a lower bound on the total spatial from $(p-|S_{I}|)$ candidates in $S_{R}$ to activity location $q_{y}$. Therefore, $|S_{I}|\cdot
d_{q_{x},q_{y}}-\sum%
\nolimits_{i=1}^{|S_{I}|}d_{s_{i},q_{x}}+(p-|S_{I}|)d_{v_{\min },q_{y}}$ is
a lower bound on the total spatial distance from any set of $p$ candidates expanded
from $S_{I}$ to $q_{y}$. In summary, if Outer-Triangle
Distance Pruning condition holds, the total spatial distance from any set of
$p$ candidates expanded from $S_{I}$ to $q_{y}$ always exceeds or equals to $D$. 
\end{proof}

Since $\sum \nolimits_{i=1}^{|S_{I}|}d_{s_{i},q_{x}}$ is computed when
we access $q_{x}$, we only need to compute $d_{q_{x},q_{y}}$
instead of each $d_{s_{i},q_{y}}$. More importantly, it is possible to improve Outer-Triangle Distance Pruning from a single location to a ball of locations, 
as shown in Figure \ref{FIG_Ext_OTDP},
to prune multiple redundant activity locations in the early stages of MAGS. 

Specifically, when we consider two balls $B_x$ and $B_y$ instead of two locations $q_x$ and $q_y$, a lower bound on the spatial distance from $s_i\in S_I$ to any location in $B_y$ can be computed as
$d_{ctr(B_{x}),ctr(B_{y})}-d_{s_i,ctr(B_{x})}-r(B_{y})$, as shown in Figure \ref{FIG_Ext_OTDP}. Therefore, a lower bound on the total spatial distance from $S_I$ to any location in $B_y$ is $|S_{I}|\cdot d_{ctr(B_{x}),ctr(B_{y})}-\sum \nolimits_{s_i\in S_{I}}d_{s_i,ctr(B_{x})}-|S_{I}|\cdot r(B_{y})$. Moreover, MAGS also derives a lower bound of the remaining $p-|S_I|$ attendees as $(p-|S_{I}|)\min_{M_i\in U_{R}}MINDIST(M_{i},B_{y})$, where $\min_{M_i\in U_{R}}MINDIST(M_{i},B_{y})$ denotes a lower bound on the spatial distance from the locations in $B_y$ to its closest candidate attendees (i.e., $M_{i}$).
In summary, given the currently best solution value $D$, 
ball $B_{x}$ and $S_{I}$, a ball $B_{y}$ can be pruned according to the following lemma.

\begin{lemma} \label{Lemma_Ext_OTDP}
%\textbf{Lemma 3.} 
If $|S_{I}|\cdot d_{ctr(B_{x}),ctr(B_{y})}-\sum \nolimits_{s_i\in
S_{I}}d_{s_i,ctr(B_{x})}-|S_{I}|\cdot r(B_{y})+
(p-|S_{I}|)\min_{M_i\in U_{R}}MINDIST(M_{i},B_{y})\geq D$,
%and $\max_{s_i\in S_{I}}d_{s_i,ctr(B_{x})}<d_{ctr(B_{x}),ctr(B_{y})}$ hold, 
the activity locations within ball $B_{y}$ never
produce a better solution for any set of $p$ candidates expanded from $S_{I}$.
\end{lemma}

%\begin{proof}
%Please refer to Lemma 2 in Appendix B.
%\end{proof}

\begin{proof}
As shown in Figure \ref{FIG_Ext_OTDP}, if $\max_{s_i\in
S_{I}}d_{s_i,ctr(B_{x})}<d_{ctr(B_{x}),ctr(B_{y})}$, according to triangular inequality, $%
d_{ctr(B_{x}),ctr(B_{y})}-d_{s_{i},ctr(B_{x})}<d_{s_{i},ctr(B_{y})}$ must
hold. Therefore, $\sum \nolimits_{s_i\in
S_{I}}(d_{ctr(B_{x}),ctr(B_{y})}-d_{s_i,ctr(B_{x})})=|S_{I}|\cdot
d_{ctr(B_{x}),ctr(B_{y})}-\sum \nolimits_{s_i\in S_{I}}d_{s_i,ctr(B_{x})}$ is a
lower bound on $\sum \nolimits_{s_i\in S_{I}}d_{s_i,ctr(B_{y})}$, i.e., the total
distance from $S_{I}$ to $ctr(B_{y})$. In addition, since $%
d_{s_i,ctr(B_{y})}-r(B_{y})$ is a lower bound on the distance from $s_i$ to any activity location
in $B_{y}$, the lower bound on the total spatial distance
from any activity location within $B_{y}$ to $S_{I}$ is thus $|S_{I}|\cdot
d_{ctr(B_{x}),ctr(B_{y})}-\sum \nolimits_{s_i\in
S_{I}}d_{s_i,ctr(B_{x})}-|S_{I}|\cdot r(B_{y})$. Moreover, the lower bound on
the total spatial distance from any activity location within $B_{y}$ to
$(p-|S_{I}|)$ candidates in $S_R$ is $(p-|S_{I}|)\min_{M_i\in
U_{R}}MINDIST(M_{i},B_{y})$. Therefore, if the condition of Outer-Triangle Distance Pruning holds, 
any activity location within $%
B_{y}$ never produces a better solution by incorporating any $(p-|S_{I}|)$
candidates into $S_{I}$, and $B_{y}$ thus can be safely pruned. 
\end{proof}

Compared with Lemma \ref{Lemma_OTDP}, Lemma \ref{Lemma_Ext_OTDP} further aggregates the distance computation by including multiple balls in BallTree. Outer-Triangle Distance Pruning is performed each time $S_I$ is expanded. 
Note that $\sum \nolimits_{s_i\in S_{I}}d_{s_i,ctr(B_{x})}$
%and $\max_{s_i\in S_{I}}d_{s_i,ctr(B_{x})}$ have 
is computed when accessing ball $B_{x}$. Therefore, Outer-Triangle Distance Pruning is able to prune multiple balls without recomputing $\sum \nolimits_{s_i\in S_{I}}d_{s_i,ctr(B_{x})}$ each time.
%Therefore, Outer-Triangle Distance Pruning decides if $B_{y}$ can be
%trimmed without the need to find the spatial distance from each $s_i\in S_{I}$ to $ctr(B_{y})$. 

\vspace{+3pt}
\noindent \textbf{Inner-Triangle Distance Pruning (ITDP).} 
Outer-Triangle Distance Pruning derives the distance lower bounds based on the distance from $S_{I}$ to another previously-calculated activity location. On the other hand,
the idea of Inner-Triangle Distance Pruning is that, when the attendees in $S_I$ are sparser, the total spatial distance from $S_I$ to some activity locations may also increase. Therefore,
Inner-Triangle Distance Pruning removes redundant activity locations by deriving the lower bounds of the total spatial distance from attendees in $S_{I}$ to activity locations, based on the spatial distances of attendees in $S_I$. 

\begin{figure}[tbp]
\centering
\subfigure[For points.] {\
\includegraphics[scale=0.31]{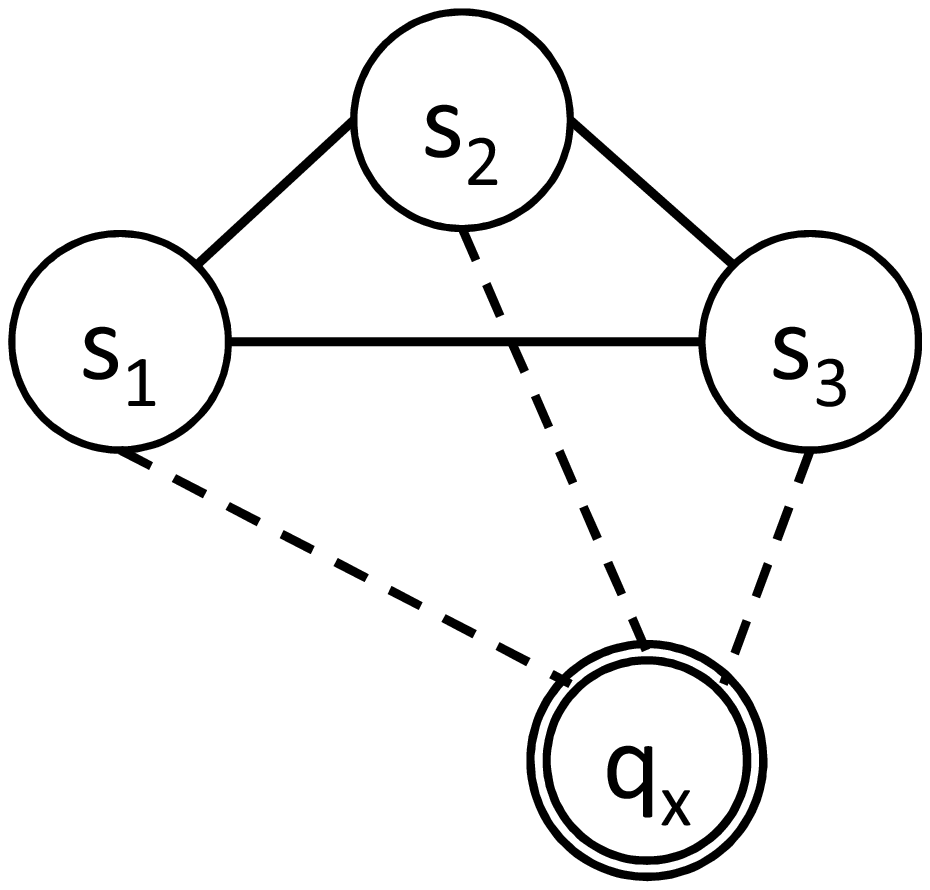}
\label{FIG_InnerTrianglePruning}}
\hspace{+15pt}
\subfigure[For balls.] {\
\includegraphics[scale=0.31]{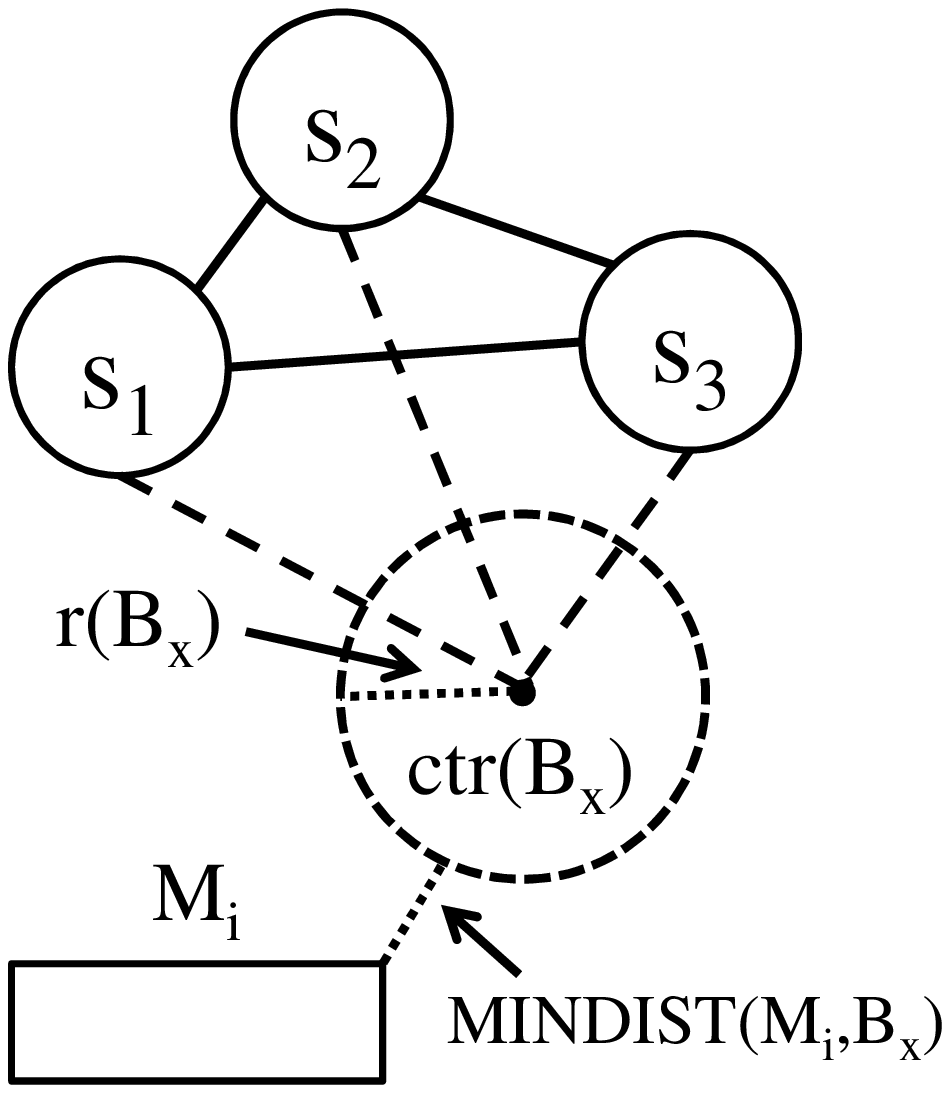}
\label{FIG_Ext_ITDP}}
\caption{Inner-Triangle Distance Pruning.}
\vspace{-20pt}
\end{figure}

%Consider a simple case of only two attendees in $S_I$, 
%i.e., $s_1,s_2 \in S_I$, and a location $q_x$. According to triangular inequality, we have $d_{s_1,q_x}+d_{s_2,q_x}>d_{s_1,s_2}$, which implies that, when the distance from $s_1$ to $s_2$ (i.e., $d_{s_1,s_2}$) is large, the total spatial distance from $S_I$ to $q_x$ (i.e., $d_{s_1,q_x}+d_{s_2,q_x}$) is even larger. 

Figure \ref{FIG_InnerTrianglePruning} shows a case where $S_I$ contains three attendees. 
In this case, the distance among each pair of attendees in $S_I$, i.e., $d_{s_i,s_j}$ (solid lines) is used to derive a lower bound on the total spatial distance from $s_i$, $s_j$ to $q_x$ (dotted lines), i.e., $d_{s_i,q_x}+d_{s_j,q_x}>d_{s_i,s_j}$. Therefore, Figure \ref{FIG_InnerTrianglePruning} shows a set of lower bounds on the spatial distance from $s_i$ to any location $q_x$: 1)
$d_{s_1,q_x}+d_{s_2,q_x}>d_{s_1,s_2}$, 2) $d_{s_1,q_x}+d_{s_3,q_x}>d_{s_1,s_3}$, and 3) $d_{s_2,q_x}+d_{s_3,q_x}>d_{s_2,s_3}$. Summing them up, we have a lower bound on the total spatial distance from $S_I$ to $q_x$, i.e., $(|S_I|-1)\cdot \sum_{i=1}^{|S_I|}d_{s_i,q_x}>\sum_{i=1}^{|S_I|-1} \sum_{j=i+1}^{|S_I|}d_{s_i,s_j}$. On the other hand, since MAGS needs to move other $p-|S_I|$ attendees to $S_I$, a lower bound on the total spatial distance from them to $q_x$ is $(|p-|S_I|)\cdot d_{v_{\min },q_{x}}$, where $d_{v_{\min },q_{x}}$ denotes the minimum spatial distance from $q_{x}$ to any candidates in $S_{R}$. Therefore, Inner-Triangle Distance Pruning is specified in the following lemma.

\begin{lemma} \label{Lemma_ITDP}
%\textbf{Lemma 2}. 
If $\left( \frac{1}{|S_{I}|-1}\sum_{i=1}^{|S_{I}|-1}%
\sum_{j=i+1}^{|S_{I}|}d_{s_{i},s_{j}}\right) +(p-|S_{I}|)d_{v_{\min
},q_{x}}\geq D$ holds, $q_{x}$ never produces a better solution for any set of $p$ candidates
expanded from $S_{I}$.
\end{lemma}

%\begin{proof}
%Please refer to Lemma 3 in Appendix B.
%\end{proof}

\begin{proof}
The first term of the above inequality is a lower bound on the total spatial
distance from $S_{I}$ to $q_{x}$, and the
second term is a lower bound on the total spatial distance to $q_{x}$ from any $(p-|S_{I}|)$ candidates in $S_{R}$.
From triangular inequality, we have $%
d_{s_{i},q_{x}}+d_{s_{j},q_{x}}>d_{s_{i},s_{j}}$, $\forall 1\leq i,j\leq
|S_{I}|$ and $i\neq j$, such as $%
d_{s_{1},q_{x}}+d_{s_{2},q_{x}}>d_{s_{1},s_{2}}$ and $%
d_{s_{2},q_{x}}+d_{s_{1},q_{x}}>d_{s_{2},s_{1}}=d_{s_{1},s_{2}}$ as shown in Figure \ref{FIG_InnerTrianglePruning}. 
Therefore, $(|S_{I}|-1)\sum_{i=1}^{|S_{I}|}d_{s_{i},q_{x}}+(|S_{I}|-1)%
\sum_{j=1}^{|S_{I}|}d_{s_{j},q_{x}}>2\cdot
\sum_{i=1}^{|S_{I}|-1}\sum_{j=i+1}^{|S_{I}|}d_{s_{i},s_{j}}$, where $%
(|S_{I}|-1)\sum_{i=1}^{|S_{I}|}d_{s_{i},q_{x}}=(|S_{I}|-1)%
\sum_{j=1}^{|S_{I}|}d_{s_{j},q_{x}}$. 
Consequently, $(|S_{I}|-1)%
\sum_{i=1}^{|S_{I}|}d_{s_{i},q_{x}}>\sum_{i=1}^{|S_{I}|-1}%
\sum_{j=i+1}^{|S_{I}|}d_{s_{i},s_{j}}$, and we have $%
\sum_{i=1}^{|S_{I}|}d_{s_{i},q_{x}}>\frac{1}{|S_{I}|-1}%
\sum_{i=1}^{|S_{I}|-1}\sum_{j=i+1}^{|S_{I}|}d_{s_{i},s_{j}}$. 
In other words, 
$\frac{1}{|S_{I}|-1}\sum_{i=1}^{|S_{I}|-1}%
\sum_{j=i+1}^{|S_{I}|}d_{s_{i},s_{j}}$ acts as a lower bound on the total spatial
distance from $q_{x}$ to $S_{I}$. On the other hand, $%
(p-|S_{I}|)d_{v_{\min },q_{x}}$ is a lower bound on the total spatial to
$q_{x}$ from any $(p-|S_{I}|)$ candidates in $S_{R}$.
Therefore, $\left( \frac{1}{|S_{I}|-1}\sum_{i=1}^{|S_{I}|-1}%
\sum_{j=i+1}^{|S_{I}|}d_{s_{i},s_{j}}\right) +(p-|S_{I}|)d_{v_{\min },q_{x}}$
is a lower bound on the total spatial distance from any set of $p$ candidates
expanded from $S_{I}$ to the activity location $q_{x}$. Therefore, if the condition of 
Inner-Triangle Distance Pruning holds, the total spatial distance
from any set of $p$ candidates expanded from $S_{I}$ to $q_{x}$ must equal to or exceed $D$. 
\end{proof}

Note that $|S_{I}|-1$ must be included in the denominator to prevent overestimation of duplicated distance $d_{s_{i},s_{j}}$. In addition, the first term can be constructed incrementally as $S_{I}$ expands, which does not require recomputation at each iteration. Therefore, Inner-Triangle Distance Pruning can be performed efficiently.

It is more efficient to trim off multiple unnecessary activity locations all together. Since $\left( \frac{1}{|S_{I}|-1}\sum_{i=1}^{|S_{I}|-1}
\sum_{j=i+1}^{|S_{I}|}d_{s_{i},s_{j}}\right)$ is a lower bound on the total spatial distance from $S_I$ to a point (the center of ball $B_{x}$ of locations), we can subtract this term with $|S_{I}|\cdot r(B_{x})$ to obtain a lower bound on the total spatial distance from $S_I$ to any location in $B_{x}$, as
shown in Figure \ref{FIG_Ext_ITDP}. 
Moreover, similar to OTDP, we can replace $(p-|S_{I}|)d_{v_{\min },q_{x}}$ in Lemma \ref{Lemma_ITDP} by its lower bound $(p-|S_{I}|)\min_{M_i\in U_{R}}MINDIST(M_{i},B_{x})$. Therefore, given a ball $B_{x}$, all activity locations within $B_{x}$ can be safely pruned according to the following lemma.

\begin{lemma} \label{Lemma_Ext_ITDP}
If $\left( \frac{1}{|S_{I}|-1}\sum_{i=1}^{|S_{I}|-1}
\sum_{j=i+1}^{|S_{I}|}d_{s_{i},s_{j}}\right) -|S_{I}|\cdot r(B_{x})+ 
(p-|S_{I}|)\min_{M_i\in U_{R}}MINDIST(M_{i},B_{x})\geq D$ holds, 
any activity location within ball $B_{x}$ never
produces a better solution expanded from $S_{I}$.
\end{lemma}

\begin{proof}
As illustrated in Figure \ref{FIG_Ext_ITDP} and pointed out 
in Lemma \ref{Lemma_ITDP}, $\left( \frac{1}{|S_{I}|-1}\sum_{i=1}^{|S_{I}|-1}%
\sum_{j=i+1}^{|S_{I}|}d_{s_{i},s_{j}}\right) $ is a lower bound on the total
spatial distance from $S_{I}$ to $ctr(B_{x})$. $|S_I|\cdot r(B_x)$ is necessary to be incorporated to ensure that, $\left( \frac{1}{|S_{I}|-1}%
\sum_{i=1}^{|S_{I}|-1}\sum_{j=i+1}^{|S_{I}|}d_{s_{i},s_{j}}\right)
-|S_{I}|\cdot r(B_{x})$ is a lower bound on the total spatial
distance from $S_{I}$ to any activity location within $B_{x}$. On the other
hand, $(p-|S_{I}|)\min_{M_i\in U_{R}}MINDIST(M_{i},B_{x})$ represents a
lower bound on the total spatial distance from $(p-|S_{I}|)$ candidates in $S_R$ to $B_x$.
Therefore, when the above condition holds, any activity location within $B_{x}$
never produces a better solution with any group of $p$ candidates expanded from $S_I$. 
Thus, $B_{x}$ can be safely pruned. 
\end{proof}

%\begin{proof}
%Please refer to Lemma 4 in Appendix B.
%\end{proof}
%
\vspace{+3pt}
\noindent \textbf{Activity Location Distance Pruning. } 
Activity Location Distance Pruning exploits the MBRs in R-Tree and the balls in BallTree, 
to quickly filter out unqualified activity locations. 

%Given $S_{I}$, the currently best
solution value $D$ and ball $B_{x}$, Activity Location Distance Pruning jointly considers 
a lower bound from the attendees in $S_I$ to $B_{x}$ and a lower bound from  
$(p-|S_I|)$ remaining candidates in $S_R$ to $B_{x}$. If the sum of the two lower bounds 
exceeds $D$, it concludes that $S_I$ and all activity locations within $B_x$ never
produce a better solution with the total spatial distance smaller than $D$. Specifically, 
Activity Location Distance Pruning is based on Lemma \ref{Lemma_ALDP} below.

\begin{figure}[tbp]
\centering
\includegraphics[scale=0.35]{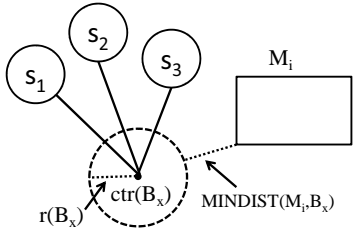}
\caption{Activity Location Distance Pruning.}
\vspace{-20pt}
\label{FIG_Ext_RPDP}
\end{figure}

\begin{lemma} \label{Lemma_ALDP}
If $\sum \nolimits_{s_i\in S_{I}}MINDIST(s_i,B_{x})+ 
(p-|S_{I}|)\min_{M_i\in U_{R}}MINDIST(M_{i},B_{x})\geq D$ holds, 
the activity locations within ball $B_{x}$ do not
produce a better solution than the current solution corresponding to $D$.
\end{lemma}
%
%\begin{proof}
%Please refer to the proof of Lemma 5 in Appendix G.
%\end{proof}
\begin{proof}
As shown in Figure \ref{FIG_Ext_RPDP}. 
$MINDIST(s_i,B_{x})=d_{s_i,B_{x}}-r(B_{x})$ is a lower bound on the distance from $s_i$
to any activity location within $B_{x}$. 
Thus, $\sum \nolimits_{s_i\in S_{I}}MINDIST(s_i,B_{x})$ represents a lower bound on the total
spatial distance from $S_{I}$ and any activity locations within Ball $B_{x}$,
and $(p-|S_{I}|)\min_{M_i\in U_{R}}MINDIST(M_{i},B_{x})$ represents a lower
bound on the total spatial distance from $(p-|S_{I}|)$ remaining candidates to any activity locations in $B_{x}$. 
Therefore, when the above condition holds, any activity location within $B_{x}$ never produces
a better solution when any $(p-|S_{I}|)$ candidates are selected into $S_{I}$. 
Therefore, $B_{x}$ can be safely pruned.
\end{proof}

Although the above strategy is simple, it still incurs high
computation overhead because $|S_{I}|$ distance computations performed
for each ball $B_{x}$ to find the lower bound. In other words, the above strategy incurs $|S_{I}|\cdot n$ distance computations, where $n$ is the 
number of balls. In the following, therefore, we propose two strategies which 
utilize the information of activity locations and the relationship of attendees in $S_{I}$ 
to reduce the number of distance computations.

Here we briefly analyze the above distance pruning strategies. 
Let $m$ denote the number of distance computations for $(p-|S_{I}|)\cdot \min_{M_i\in U_{R}}MINDIST(M_{i},B_{x})$.
Activity Location Distance Pruning incurs the highest computation overhead, 
i.e., $(n\cdot |S_{I}|+m)$, as distance computations are required for $n$ balls
(for each ball, it derives $MINDIST(s_i,B_x),\forall s_i\in S_I$). 
On the other hand, Outer-Triangle Distance Pruning incurs $(|S_I|+n+m-1)$ distance computations, for $n$ balls in the worst case, including
$|S_I|$ computations for the total spatial distance from $S_I$ to $ctr(B_x)$, and
$(n-1)$ computations for the distances from $ctr(B_x)$ to the centers of the other $(n-1)$ balls.
Similarly, when deriving the lower bound on $S_{I}$ and $B_{x}$, Inner-Triangle Distance Pruning 
only considers the distances between each pair of attendees in $%
S_{I} $, which can be computed incrementally and cached in early stages.
Therefore, each time a new attendee is added to $S_I$, Inner-Triangle Distance 
Pruning performs $(|S_I|-1+m)$ distance computations for $n$ balls.
Therefore, Outer-Triangle Distance Pruning and Inner-Triangle Distance Pruning are much more efficient than Activity Location Distance Pruning.

\subsection{Discussions} \label{discussion}

\noindent\textbf{User interests and existence of sponsors.}
We propose a generalized model to support the scenarios in terms of user interests. Let $\eta _{v,q}$ denote the 
\textit{interest measure} (i.e., how an individual $v$ prefers a candidate location $q$) of a person $v$ in an activity to be held at location $q$. A small interest measure $\eta _{v,q}$ implies that $v$ highly prefers the
activities to be associated with $q$. Similar to the spatial radius constraint
in SSGQ and MRGQ, a new interest constraint $\eta _{v,q}<h$ is
added to the two problems, where $h$ denotes the \textit{interest threshold}
of an activity. For a candidate member $v$ that prefers only karaoke studios
and bars, the interest measure from $v$ to coffee shops will be set to a
large value exceeding the threshold. Thus, $v$ in this case will never be
selected for an activity in $q$.

MRGQ and SSGQ can also flexibly handle the case when sponsors of the activity exist. 
Here, we describe a generalized graph model for the scenarios with sponsors.
The sponsors are represented by a set $S$ of new nodes in SSGQ and MRGQ. Here each
sponsor $s$ in $S$ is connected to a person $v$ if $s$ is correlated to 
$v$, e.g., $v$ is an employee, a former student, or a
regular customer of $s$. This link information can be acquired from the 
address directories, personal Facebook profiles, or customer
databases. To support SSGQ and MRGQ with sponsors, the set $S$ is added to the solution at the beginning of SSGS and MAGS. As such, these two
algorithms will automatically find a solution group with correlation to $S$ (i.e., the attendees that $S$ would like to sponsor). 
Moreover, if the activity locations are provided by a sponsor, such as
a chain restaurant, all branches of the chain restaurant group can be
included in the candidate location set $Q$. Note that the group size $p$
needs to be increased by $\left\vert S\right\vert $ in the scenarios with
sponsors, and the representatives of each sponsor can also be initially added to the
solution group.

%\noindent\textbf{Performance of MAGS for MRGQ in threshold graph.}
%In Graph Theory, analyzing the tractability of NP-hard graph
%problems in special graph classes is very important for theoreticians. Therefore, 
%we prove that MAGS can find the optimal solution in polynomial time in a special graph class, namely, \textit{threshold graph} \cite{MP95} in Appendix C.

\noindent\textbf{Dynamically changing user locations.}
To index user
locations, one approach is to employ the fundamental R-Tree. 
If user locations are changed frequently, however, this approach is likely to incur frequent
index updates or otherwise record outdated and inaccurate information.
A more promising approach is to exploit R-Tree extensions that are
specifically designed for dynamic environments, such as Time-Parameterized
R-Tree (TPR-Tree) [22] or an improved version of TPR-Tree, TPR*-Tree [23].
Similar to conventional R-Trees, TPR-Tree adopts Minimum Bounding Rectangles
(MBR) to hierarchically index the spatial objects (i.e., the locations of the
users). However, instead of recording objects' locations at individual
timestamps, TPR-Tree incorporates the velocity of each object to predict
their upcoming positions, and thus updates are only triggered when the velocity changes. 
This strategy significantly reduces the number of updates. In other words, the MBR of an object or a tree node is a function
of time.

More specifically, each dynamically changing location of a user in the
TPR-Tree is represented as 1) an MBR that denotes its extent at some
reference time (a system parameter), and 2) its current velocity vector. The
velocity vector of an MBR is represented by the largest velocity of an
object within the MBR in each direction. This strategy ensures that the MBR
always encloses the underlying objects in the future. With the above
information, the future MBRs and the objects' locations are not stored
explicitly, but are quickly computed based on the locations at the
reference time and the velocity vectors. Figure \ref{TPR_fig} presents an
example of the MBR and velocity vectors of three objects, i.e., $a$, $b$,
and $c$. The velocity vector of each object is shown as solid arrows,
and the velocity vector of the MBR is represented by dashed arrows.
Figure \ref{tpr_fig1} presents the locations of the objects and the
enclosing MBR at reference time $0$. In the next time slot shown in
Figure \ref{tpr_fig2}, the locations of the objects are changed, and the
enclosing MBR is enlarged to enclose all objects. As shown in this
example, TPR-Tree only stores the locations of the objects and the
corresponding MBR at the reference time, and the locations of the objects
with the enclosing MBR in the future can then be computed with the velocity
vectors.

\begin{figure}[tp]
\centering
\subfigure[][Boundaries at reference time $0$.] {\  \includegraphics[scale=0.27] {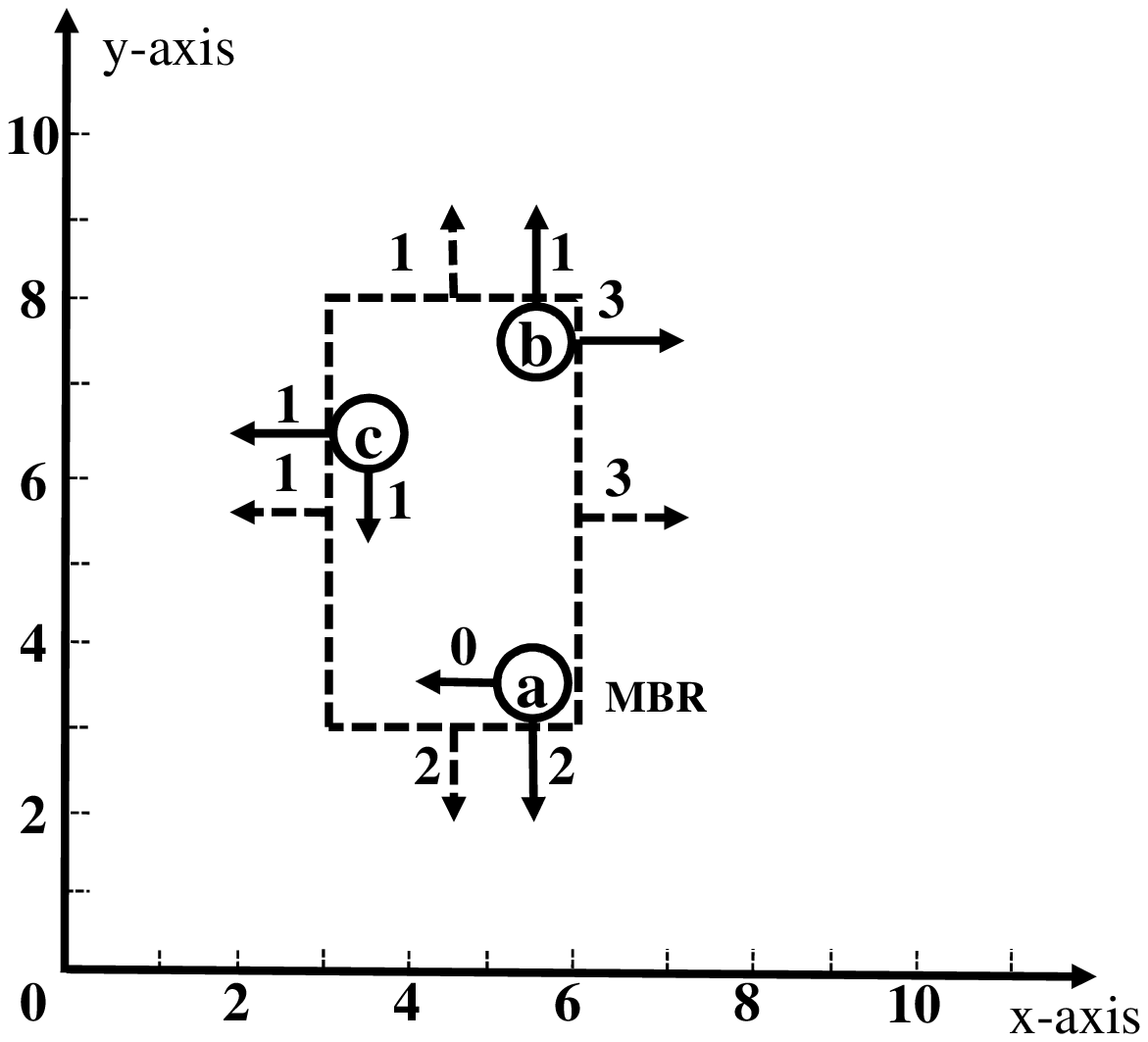}
\label{tpr_fig1} } \hspace{+20pt} 
\subfigure[][Boundaries at time $1$.] {\  \includegraphics[scale=0.27] {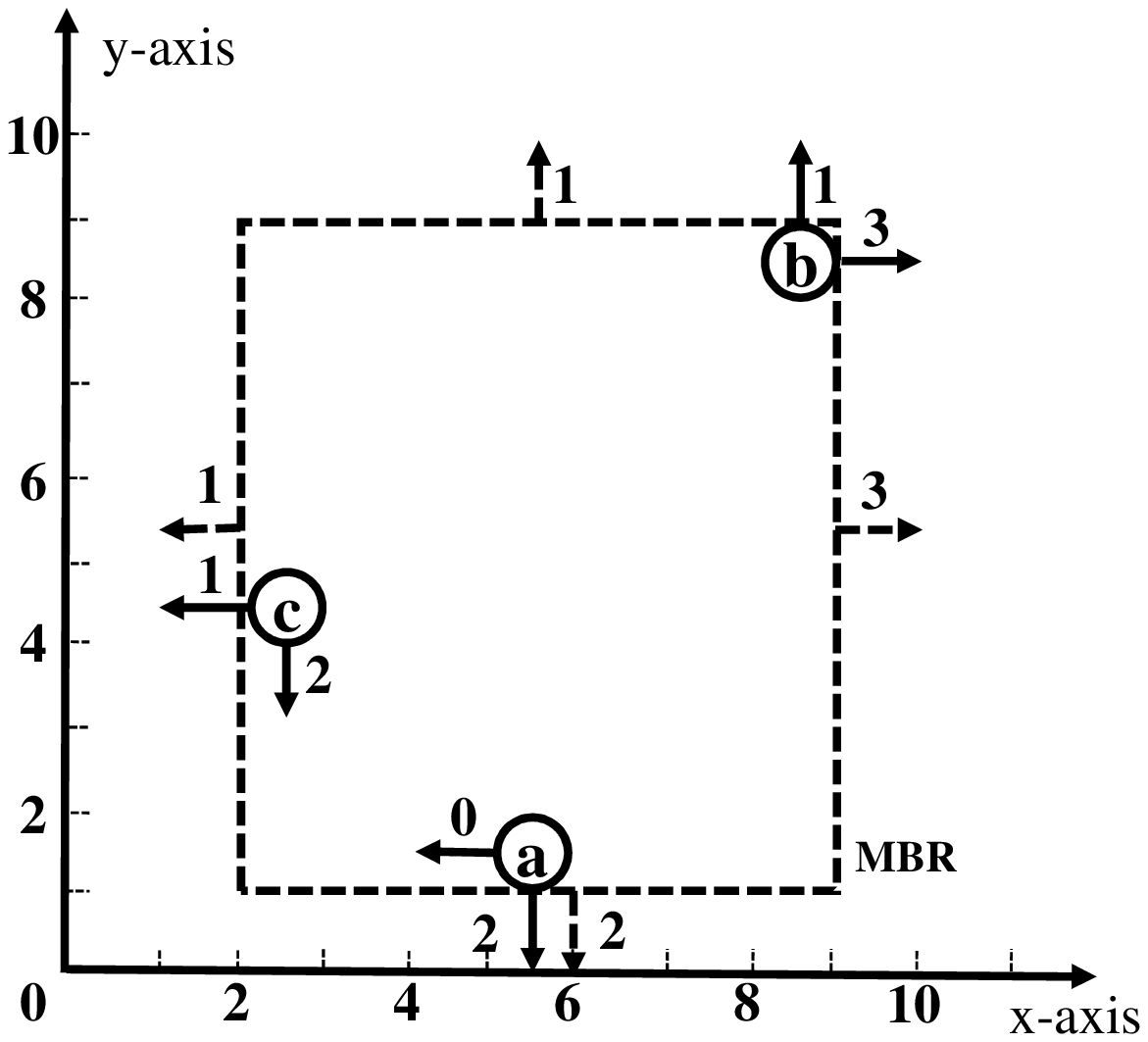}
\label{tpr_fig2} }
\caption{An example of the MBR and objects with velocity vectors.}
\label{TPR_fig}
\end{figure}

To support dynamic environments in this paper, we replace R-Tree with
TPR-Tree. Similar to the TPR-Tree problem, a parameter $t_{f}\geq 0$ is
included in SSGQ and MRGQ to ensure that the optimal group is returned for
dynamic environments from now to $t_{f}$ time units afterward. The
ordering and pruning strategies of SSGS and MAGS are performed in the same
way, except that the extents of the MBRs and the user locations have
to be calculated before SSGS and MAGS. On the other hand, the BallTree in
this paper only indexes the activity locations (e.g., restaurants, theaters,
etc.) and thus does not need to be updated frequently.

\section{Tractability of MRGQ in Threshold Graph}
In Graph Theory, analyzing the tractability of NP-Hard graph
problems in special graph classes is very important for theoreticians, e.g., [16][17].
Therefore, in addition to the inapproximibility of MRGQ
analyzed above, we also prove that MAGS can find the optimal
solution in polynomial time in a special graph class, namely \textit{%
threshold graph}, which is defined as follows [15]. 

\begin{definition}
\label{def_thresholdgraph} Given a threshold $\tau $ and a graph $G=(V,E)$
with a non-negative weight $w_{v}$ for each vertex $v\in V$, $G$ is called a
threshold graph if $w(U)=\sum_{v\in U}w_{v}\leq \tau $ holds for every
subset $U\subseteq V$ with every two nodes in $U$ sharing no edge between
them.
\end{definition}

Let $deg_{G}(v)$ denote the degree of $v\in V$. A \textit{degree partition} $%
D(V)$ of $V$ divides $V$ into $m+1$ non-overlapping subsets $%
D_{0},D_{1},..,D_{m}$, i.e., $\bigcup_{\forall i}D_{i}=V$. Each subset $D_{i}
$ includes all the vertices with degree $\delta _{i}$, and $\delta
_{i}>\delta _{j}$ if $i>j$. When $\tau =m$ and $w_{v}=deg_{G}(v)$, we
exploit the following theorem in the literature [15] to prove that 
MAGS can obtain the optimal solution in polynomial time in a threshold
graph.  

\begin{theorem}
\label{properties} [15] Let $\widehat{G}=(\widehat{V},\widehat{E})$ be a
threshold graph with degree partition $D(\widehat{V})$. For every pair of
vertices $u\in D_{i}$ and $v\in D_{j}$, $u$ is connected to $v$ in $\widehat{%
G}$ if and only if $i+j>m$. Therefore, every vertex in $D_{i}$ shares the
same neighbors.
\end{theorem}

We first describe the tie-breaking strategy for Socio-Spatial Ordering in
MAGS. Specifically, for any set $\widetilde{V}\subseteq S_{R}$ of vertices
with the same spatial distance to each query point $q_i$, the tie-breaking
strategy checks if the vertices in $\widetilde{V}$ satisfy Eq. (1) in
descending order of their vertex degrees. Therefore, the vertex $v\in 
\widetilde{V}$ with the maximum $deg_{G}(v)$ is first examined by Eq. (1)
and moved into $S_{I}$ if $v$ satisfies Eq. (1). Moreover, we employ a
pre-processing strategy to remove unqualified vertices from the input graph $%
\widehat{G}$, which works both when $\widehat{G}$ is a general graph or $%
\widehat{G}$ is a threshold graph. Given the social network $\widehat{G}$
with parameters $k$ and $p$, the pre-processing strategy iteratively removes
from $\widehat{G}$ every vertex $v$ with fewer than $p-k-1$ neighbors, along
with $v$'s incident edges. The iteration repeats until no more vertices can
be removed and produces a graph $G$. Any removed vertex $v$ cannot form a
feasible graph with other vertices in $\widehat{G}$, because even when all
the neighbors of $v$ (i.e., $N_{v}$) are included in the same group $F$,
there is still at least one vertex in $N_{v}\cup \{v\}$ with fewer than $%
p-k-1$ neighbors in $F$. The above process is called \textit{core
decomposition}, and the remaining graph $G$ is a maximal $(p-k-1)$-core
[18], where $G\subseteq \widehat{G}$ is the largest graph such that each
vertex has at least $p-k-1$ neighbors in $G$. The pre-processing strategy
can be done in $O(|\widehat{E}|)$ time [19]. Intuitively, if $|G|<p$, then there exists
no solution to MRGQ.

Note that $D_{j}$ is the degree partition on $\widehat{G}$, instead of $G$
after the core decomposition. Let $V(G)$ denote the set of vertices in $G$, 
and let $\widehat{j}$ denote the minimum $j$ such
that $D_{j}\cap V(G)\neq \varnothing $. 
In other words, $D_{0},D_{1},..,D_{%
\widehat{j}-1}$ are all removed during core decomposition. The following
theorem first compares $\widehat{G}$ and $G$.

\begin{theorem}
\label{node_identity} For input threshold graph $\widehat{G}$ and the graph $%
G$ after core decomposition, $\bigcup_{i=\widehat{j}}^{m}D_{i}=V(G)$ holds,
and the neighbors of every vertex in $D_{i}$ are also the neighbors of every
vertex in $D_{j}$ if $i<j$.
\end{theorem}

\begin{proof}
Apparently, $\bigcup_{i=\widehat{j}}^{m}D_{i}\supseteq V(G)$ since $G\subseteq 
\widehat{G}$, and thus we prove that $\bigcup_{i=\widehat{j}}^{m}D_{i}\subseteq V(G)$
by contradiction. Assume $u\in D_{r}$ for some $r\geq \widehat{j}$, but $u\notin V(G)$%
. Given a vertex $x$ in $D_{\widehat{j}}\cap V(G)$ and any vertex $y\in V(G)$ sharing
an edge with $x$, if $y\in D_{s},$ $\widehat{j}\leq s\leq m$, according to Theorem \ref%
{properties}, $\widehat{j}+s>m$ must hold; otherwise, $x$ and $y$ would not share an
edge. Since $\widehat{j}$ is the minimum number such that $D_{\widehat{j}}\cap V(G)\neq \varnothing 
$, $r+s>m$ must hold, implying that $u$ has an edge with $y$ because $r\geq \widehat{j}
$. From the definition of threshold graph, $u$ also share edges with all
neighbors of $x$. Therefore, the number of $u$'s neighbors in $V(G)$ is no
smaller than $x$ does, and $u$ should not be removed by the pre-processing
strategy because $x\in D_{\widehat{j}}\cap V(G)$. Therefore, this contradiction proves
that $\bigcup_{i=\widehat{j}}^{m}D_{i}=V(G)$ holds. Moreover, since vertex $u$ in $%
D_{i}$ is connected to $w$ in $D_{n}$ if and only if $i+n>m$ according to
Theorem \ref{properties}, vertex $v$ in $D_{j}$ is also connected to $w$ in $%
D_{n}$ if $j>i$ because $j+n>i+n>m$. Therefore, the neighbors of every vertex in $%
D_{i}$ are also the neighbors of every vertex in $D_{j}$ if $i<j$.
\end{proof}

In other words, core decomposition only trims the vertices with smaller
degrees, and if any vertex in $D_{i}$ is removed, the above theorem manifests
that all the other vertices in $D_{i}$ will be pruned as well. According to
Theorem \ref{properties}, every vertex in $D_{i}$ still shares the same
neighbors. In the following, we provide the theoretical result for MAGS in
Theorem \ref{opt_mags}. In Graph Theory, since each vertex $v$ of a
threshold graph is not associated with a spatial distance to $q_{i}\in Q$,
we first assume that $d_{v,q_{i}}=1$ for $\forall v\in V,q_{i}\in Q$.%
\footnote{%
The results are theoretically interesting since no social network belongs to
a threshold graph. On the other hand, the hardness result for a general
social network is presented in Theorem 1.}

\begin{theorem}
\label{opt_mags} Let $\widehat{G}=(\widehat{V},\widehat{E})$ be a threshold
graph with degree partition $\widehat{V}=D_{0}\cup D_{1}\cup \cdots \cup
D_{m}$ and $\beta =\max \{i:|D_{i}|+\cdots +|D_{m}|\geq p\}$. Given an $%
MRGQ(p,Q,k,t)$ for $\widehat{G}$, MAGS stops (either returning the optimal
solution or returning no solution) in polynomial time if $%
d_{v,q_{i}}=1,\forall v\in \widehat{V},q_{i}\in Q$.
\end{theorem}

\begin{proof}
We employ All-Pair Distance Ordering for MAGS here. If after core
decomposition, the resulting graph has fewer than $p$ vertices, i.e., $%
|G|<p$, there is no feasible solution, and MAGS stops in $O(|\widehat{E}|)$
time (i.e., only the core decomposition is performed). Otherwise, if $|G|\geq p$, we prove that there exists at least one
feasible solution, and we prove the theorem by examining two different cases
for $\beta $: $\beta >\lfloor \frac{m}{2}\rfloor $ and $\beta \leq \lfloor 
\frac{m}{2}\rfloor $.

(1) $\beta >\lfloor \frac{m}{2}\rfloor $. We prove that MAGS can find the
optimal solution by generating only $p$ nodes in the branch-and-bound tree.
From Theorem \ref{properties} and Theorem \ref{node_identity}, every two
vertices $u$ and $v$ in $\bigcup_{i=\lfloor \frac{m}{2}\rfloor +1}^{m}D_{i}$
are connected by an edge in $G$. Since $|D_{\beta }|+..+|D_{m}|\geq p$, the
first $p$ vertices selected and examined by Socio-Spatial Ordering and the tie-breaking strategy must
belong to $\bigcup_{i=\lfloor \frac{m}{2}\rfloor +1}^{m}D_{i}$ and form a
path of length $p$ from the root in the branch-and-bound tree. Since every
two vertices $u,v\in \bigcup_{i=\lfloor \frac{m}{2}\rfloor +1}^{m}D_{i}$ are
connected according to Theorem \ref{properties}, $F$ is a feasible group. 
%Since $G$ is a threshold graph, those $p$
%vertices lead to a feasible group $F$ of MRGQ because every two vertices $%
%u,v\in \bigcup_{i=\lfloor \frac{m}{2}\rfloor +1}^{m}D_{i}$ are connected.
In the $p$-th node generated in the branch-and-bound tree, $\langle
F,q_{ref}\rangle $ is a feasible solution in $G$, where $q_{ref}$ can be
derived by All-Pair Distance Ordering. Because $d_{v,q_{i}}=1$
here, $\forall v\in V,q_{i}\in Q$, $\langle F,q_{ref}\rangle $ is the
optimal solution, i.e., $q^{\ast }=q_{ref}$. Afterwards, the
distance pruning strategies of MAGS are performed exactly $p$ times (one
time for each branch-and-bound node except the leaf node, and one time for
the root of the branch-and-bound tree) to conclude that no further search is
required. Since $\langle F,q^{\ast }\rangle $ is the optimal solution in $G$%
, it will also be the optimal solution in $\widehat{G}$ due to $%
d_{v,q_{i}}=1,\forall v\in \widehat{V},q_{i}\in Q$.

(2) $\beta \leq \lfloor \frac{m}{2}\rfloor $. In this case, we prove that
MAGS can also find the optimal solution by generating exactly $p$ nodes in
the branch-and-bound tree. In MAGS, Socio-Spatial Ordering and the
tie-breaking strategy first move the vertices in $\bigcup_{i=\beta
+1}^{m}D_{i}$ into $S_{I}$ (vertices in $D_{i}$ are moved into $S_{I}$
earlier than those in $D_{j}$ if $i>j$), and then move the rest $p-|S_{I}|$
vertices from $D_{\beta }$ to $S_{I}$ to construct the group $F$. Moving
these $p$ vertices into $S_{I}$ creates the first $p$ nodes in the
branch-and-bound tree and builds up a path of length $p$ from the root.

In the following, we prove that $F$ is a feasible group in $G$. We first
examine the vertices in $F$ that are drawn from $\bigcup_{i=\beta }^{\lfloor 
\frac{m}{2}\rfloor }D_{i}$. Let $v\in F$ and $v\in \bigcup_{i=\beta
}^{\lfloor \frac{m}{2}\rfloor }D_{i}$. According to Theorem \ref{properties}
and Theorem \ref{node_identity}, all neighbors of $v$ in $G$ must belong
to $\bigcup_{i=\lfloor \frac{m}{2}\rfloor +1}^{m}D_{i}$. Since $G$ is a $%
(p-k-1)$-core, and the extracted subgraph $F$ contains all the vertices in $%
\bigcup_{i=\lfloor \frac{m}{2}\rfloor +1}^{m}D_{i}$, $v$ must have at least $%
p-k-1$ neighbors in $F$. By contrast, for every vertex $u\in F$ drawn from $%
\bigcup_{i=\lfloor \frac{m}{2}\rfloor +1}^{m}D_{i}$, since the neighbors of
every vertex in $D_{i}$ with $i\leq \lfloor \frac{m}{2}\rfloor $ are also
the neighbors of every vertex in $D_{j}$ with $j>\lfloor \frac{m}{2}\rfloor $
according to Theorem \ref{node_identity}, $u$ must have no fewer neighbors
than that of $v$ in $F$, where $v\in F$ is any vertex drawn from $%
\bigcup_{i=\beta }^{\lfloor \frac{m}{2}\rfloor }D_{i}$. From the above
description, a vertex $v\in F$ drawn from either $\bigcup_{i=\beta
}^{\lfloor \frac{m}{2}\rfloor }D_{i}$ or $\bigcup_{i=\lfloor \frac{m}{2}%
\rfloor +1}^{m}D_{i}$ must have at least $(p-k-1)$ neighbors in $F$.
Therefore, $F$ is a feasible group. Similar to case (1), $\langle
F,q_{ref}\rangle $ is the optimal solution, where $q_{ref}$ is obtained from
All-Pair Distance Ordering. Then, the distance pruning strategies conclude
that no further search is needed, and MAGS stops.

\textbf{Time Complexity. }In the following, we analyze the detailed time
complexity of constructing these $p$ nodes. The pre-processing strategy
takes $O(|\widehat{E}|)$ time. Before each of the $p$ nodes in the
branch-and-bound tree is constructed, MAGS extracts $v_{c}\in S_{R}$ and $%
q_{ref}\in Q_{I}$ from the lists $U_{R}$ and $U_{B}$. During the
examinations of $U_{R}$ and $U_{B}$, when the MBR $M_{i}\in U_{R}$ and ball $%
B_{j}\in U_{B}$ satisfying the score function in Eq. (7) are identified,
each of the three distance pruning strategies in Sec. 5.5 is performed once.
In the worst case, each time when APDO obtains $v_{c}$ and $q_{ref}$, each
internal node of R-Tree and BallTree is accessed. Therefore, $\max
\{|V|,|Q|\}$ times of extracting $M_{i}$ and $B_{j}$ satisfying Eq. (7) are
performed, and $\max \{|V|,|Q|\}$ times of each distance pruning strategy is
also performed. Since extracting $M_{i}$ and $B_{j}$ satisfying Eq. (7)
takes $O(|S_{I}||V||Q|)=O(p|V||Q|)$ time, and OTDP, ITDP and ALDP each takes 
$O(|Q|)$ time, the time complexity for MAGS to extract each pair of $%
v_{c}\in S_{R}$ and $q_{ref}\in Q_{I}$ is thus $O(\max \{|V|,|Q|\}(p|V||Q|))$%
. Therefore, for the $p$ nodes in the branch-and-bound tree, it takes $%
O(p^{2}\max \{|V|,|Q|\}(|V||Q|))$ time for extracting $v_{c}$ and $q_{ref}$
and performing the distance pruning strategies.

On the other hand, for the $p$ nodes in the branch-and-bound tree, $p$ times
of Eq. (1) checking and tie-breaking strategy are performed, which takes $%
O(p^{3}+p|V|\log |V|)$. Also, $p$ times of Familiarity Pruning (Eqs. (4),
(5)) are performed, which takes $O(p|V|^{2})$. Moreover, $p$ additional
times of $v_{c}$ and $q_{ref}$ extractions (with distance pruning strategies) are
performed to conclude that no further search is needed, which takes $%
O(p^{2}\max \{|V|,|Q|\}(|V||Q|))$ time. Therefore, the overall time
complexity of MAGS is $O(|\widehat{E}|)+O(p^{3}+p|V|\log
|V|)+O(p|V|^{2})+2\cdot O(p^{2}\max \{|V|,|Q|\}(|V||Q|))$. Since $|V|=O(|\widehat{V}|)$, where $\widehat{G}=(\widehat{V},%
\widehat{E})$ is the input graph before pre-processing. Therefore, the time
complexity is $O(|\widehat{E}|+p^{2}\max \{|\widehat{V}|,|Q|\}(|%
\widehat{V}||Q|))$.
\end{proof}

\vspace{-6pt}
\section{Experimental Results}

\label{Exp}

%\vspace{-5pt}
%\subsection{Experimental Setup}
\baselineskip=11.2pt
We implement SSGQ in Facebook and recruit 206 people from various backgrounds (e.g., students, and public and private sector workers) to compare solution quality and time overhead for answering SSGQ and MRGQ via manual coordination and our proposed algorithms. Each user completes 24 SSGQ tasks and 20 MRGQ tasks with the social graphs extracted from their social networks in Facebook, together with their spatial locations sampled from their Facebook Checkin records. 
%Due to the space constraints, the user study results are presented in Appendix M.

In addition to the real dataset collected from the 206 study participants, 
we evaluate the performance and the solution quality of SSGS, SSGMerge (the heuristic algorithm for SSGQ mentioned in Section \ref{SSGQ}) and MAGS using a large real dataset, \textit{DataSet\_4SQ}, 
obtained by crawling Foursquare \cite{YYLL11}, one of the most
representative LBSNs, for a month. \textit{DataSet\_4SQ} contains both the
social and spatial information of 153,577 individuals. Moreover, we also compare MAGS with two relevant algorithms, namely \textit{Geo-Social Circle of Friend Query (gCoFQ)} \cite{LSCH12} and \textit{p-Nearest Neighbor (pNN)}, to evaluate the solution quality and performance. In addition to \textit{DataSet\_4SQ}, we also evaluate MAGS for MRGQ on a large real dataset, \textit{DataSet\_Youtube} \cite{YL12}, which is a social network extracted from Youtube video-sharing website with 1,134,890 individuals.  
The activity location $q$ for SSGQ and $Q$ for MRGQ are randomly selected from \textit{DataSet\_4SQ}, and we measure 50 samples in each scenario. 
%Due to the space constraint,
%the user study of SSGQ, comparisons of SSGS and SSGMerge, comparisons of MAGS with other approaches, and the evaluations on Dataset\_Youtube are presented in \cite{OnlineVersion}.
%Our algorithms
%are implemented in an IBM 3650 server with Linux Ubuntu 4.2.4, two Quadcore
%Intel X5450 3.0 GHz CPUs, and 8 GB RAM.

%\begin{figure}[tp]
%\centering
%\subfigure[Query page.] {\  \hspace{-10pt} \includegraphics[scale=0.17]{App_a.eps} } 
%\subfigure[Result page.] {\ 
%\hspace{+5pt}
%\includegraphics[scale=0.17]{App_b.eps} } %\vspace{-15pt}
%\caption{Implementation of SSGQ in Facebook.}
%%\vspace{-10pt}
%\label{fig:FIG_APP}
%\end{figure}

\subsection{User Study}
We perform the user study with 24 and 20 tasks for SSGQ and MRGQ, respectively. 
The 24 tasks in the user study of SSGQ span various $p$ and network
sizes, where the spatial radius, $t$ is fixed to 10 km. Different $k$ are assigned in the first 12 tasks, while $k$ is not
specified in the other 12 tasks, to let each user freely select $p$
people for finding out the familiarity preferred by each person in
activities with different $p$.

%FIG7
%[FIGEXP] A - F
\begin{figure}[tp]
\centering
\subfigure[Coordination time.] {\  
\includegraphics[scale=0.15]{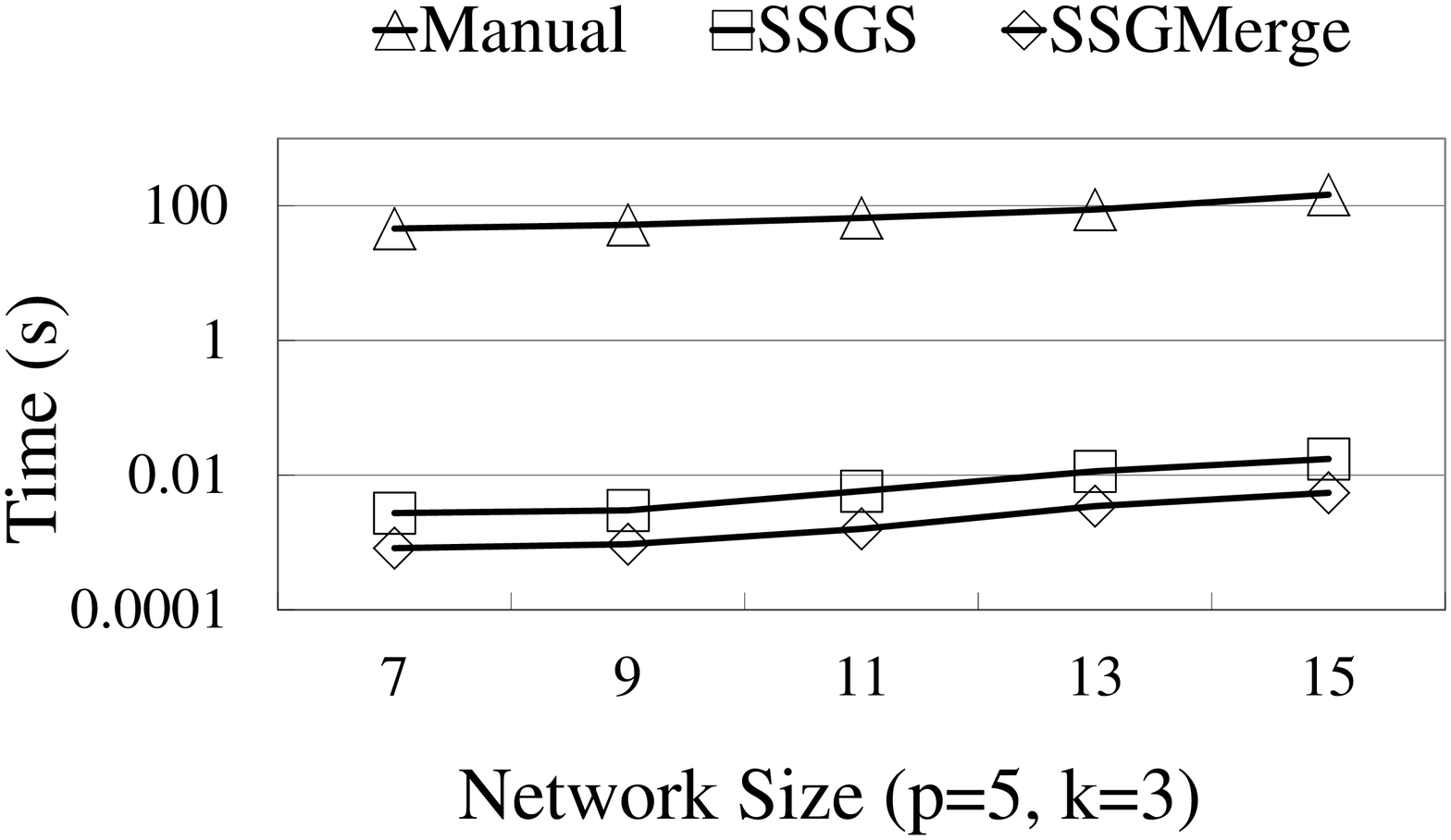} } 
\subfigure[Solution quality.] {\
\includegraphics[scale=0.15]{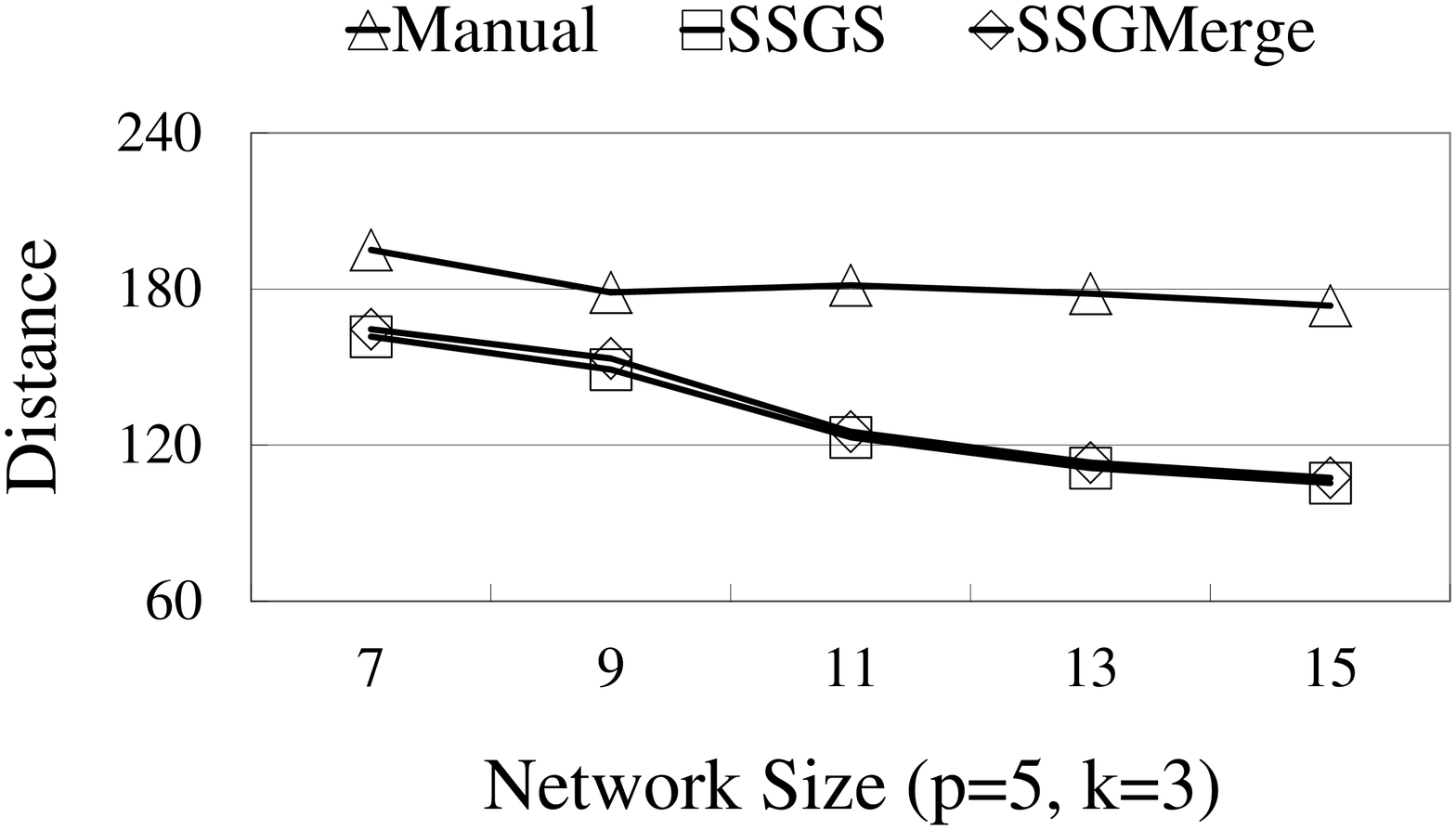} } 
\subfigure[Solutions with given $k$.] {\
\includegraphics[scale=0.15]{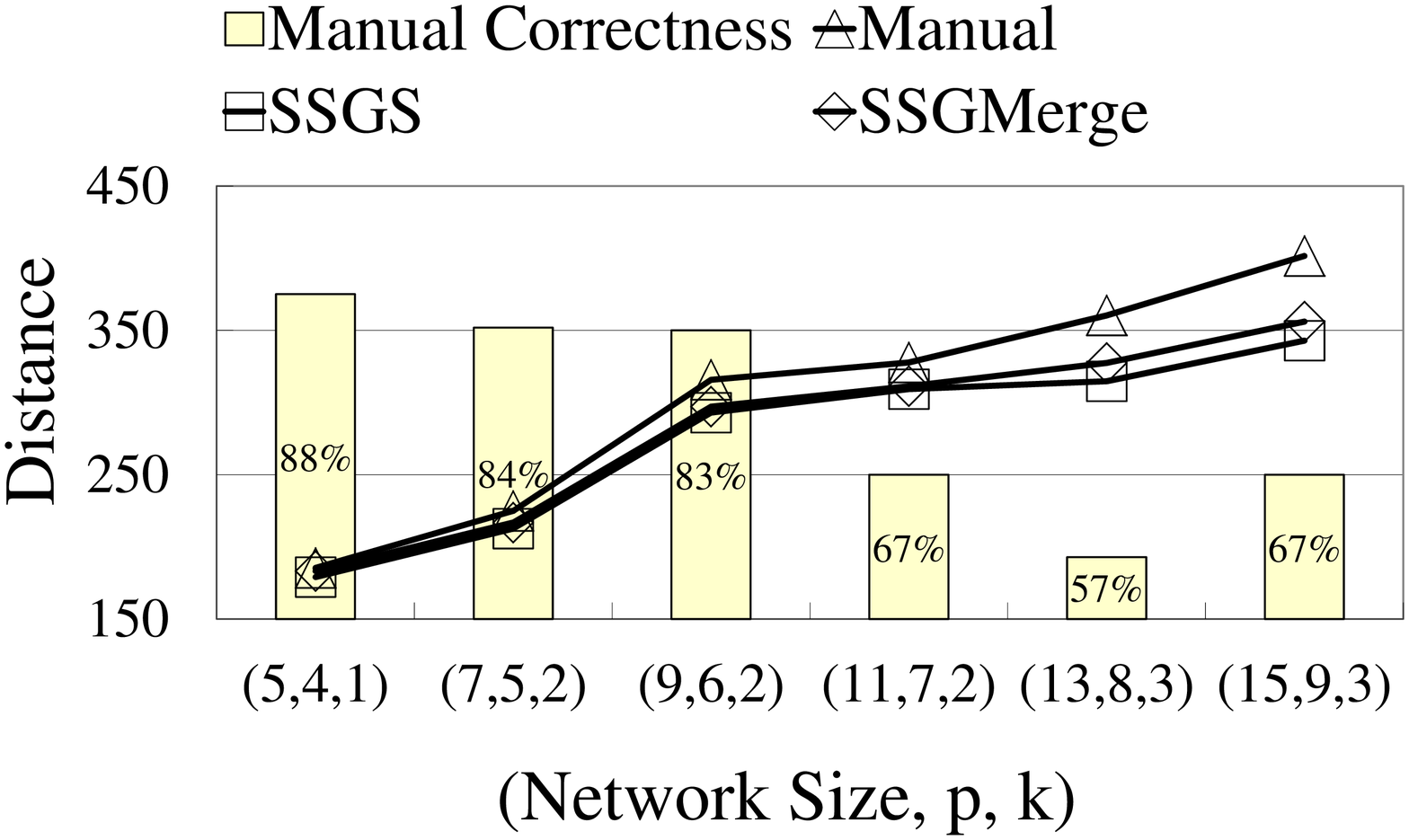} } 
\subfigure[Time without given $k$.] {\
\includegraphics[scale=0.15]{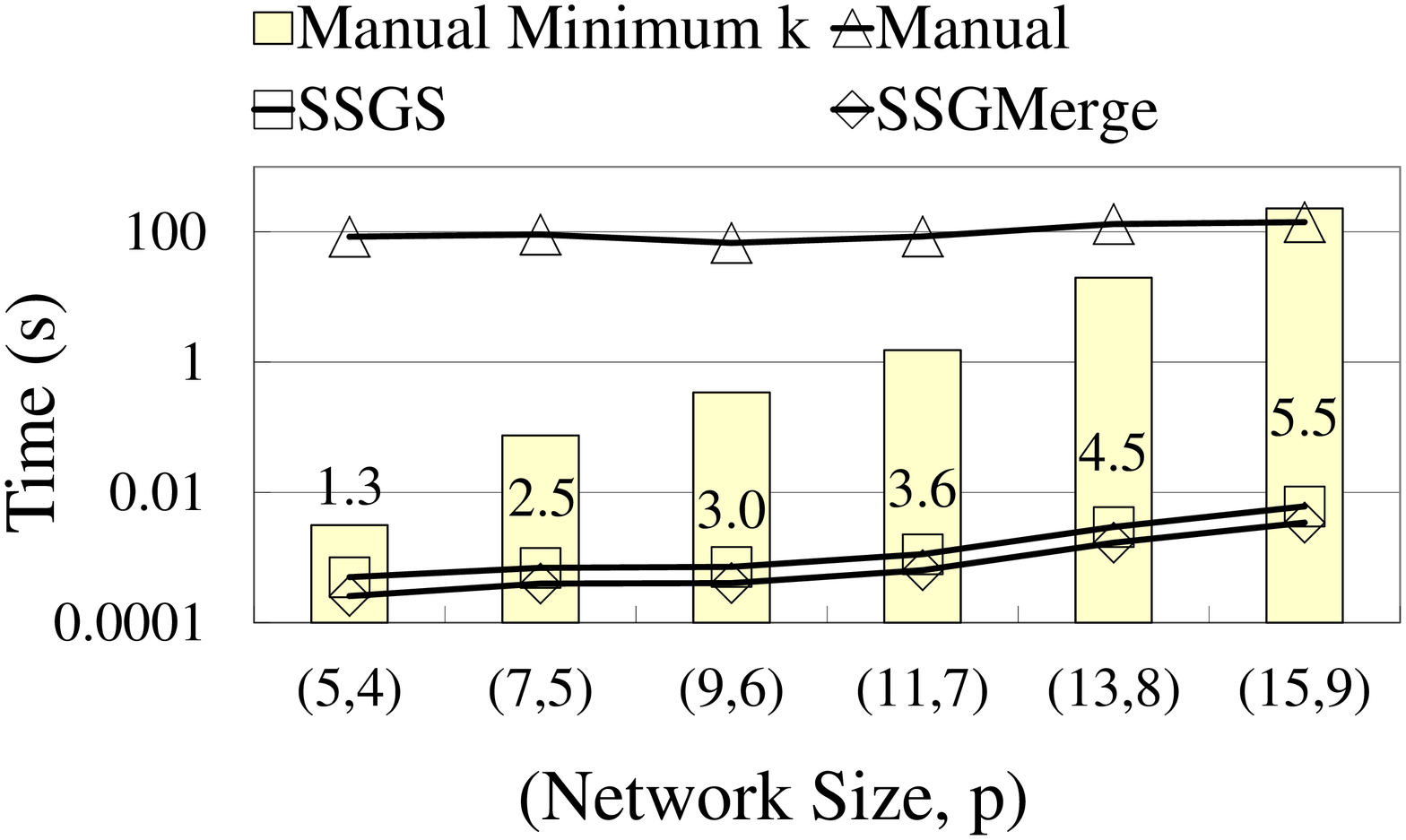} } 
\subfigure[Solutions without given $k$.] {\
\includegraphics[scale=0.15]{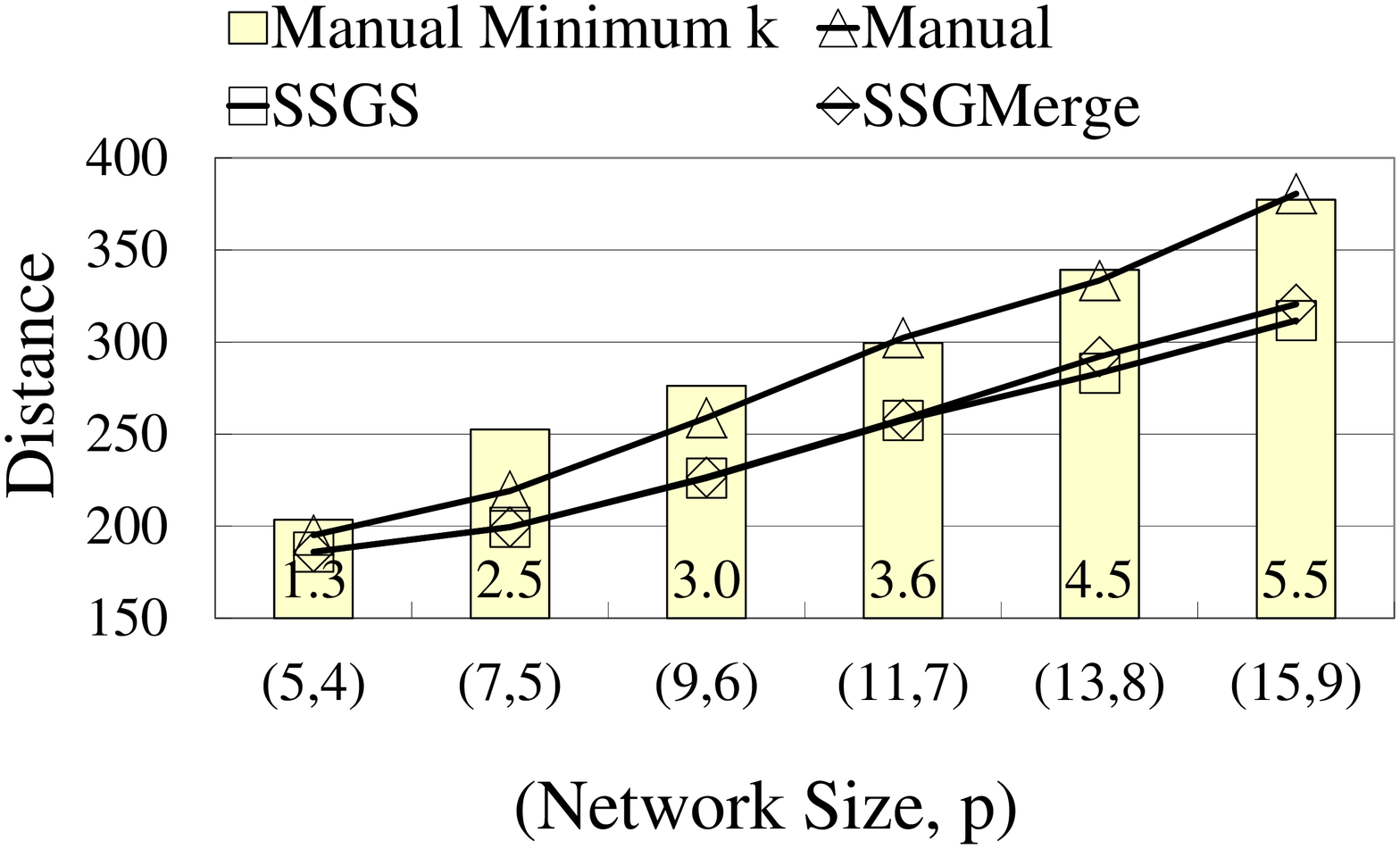} } 
\subfigure[Solutions with different $k$.] {\
\includegraphics[scale=0.15]{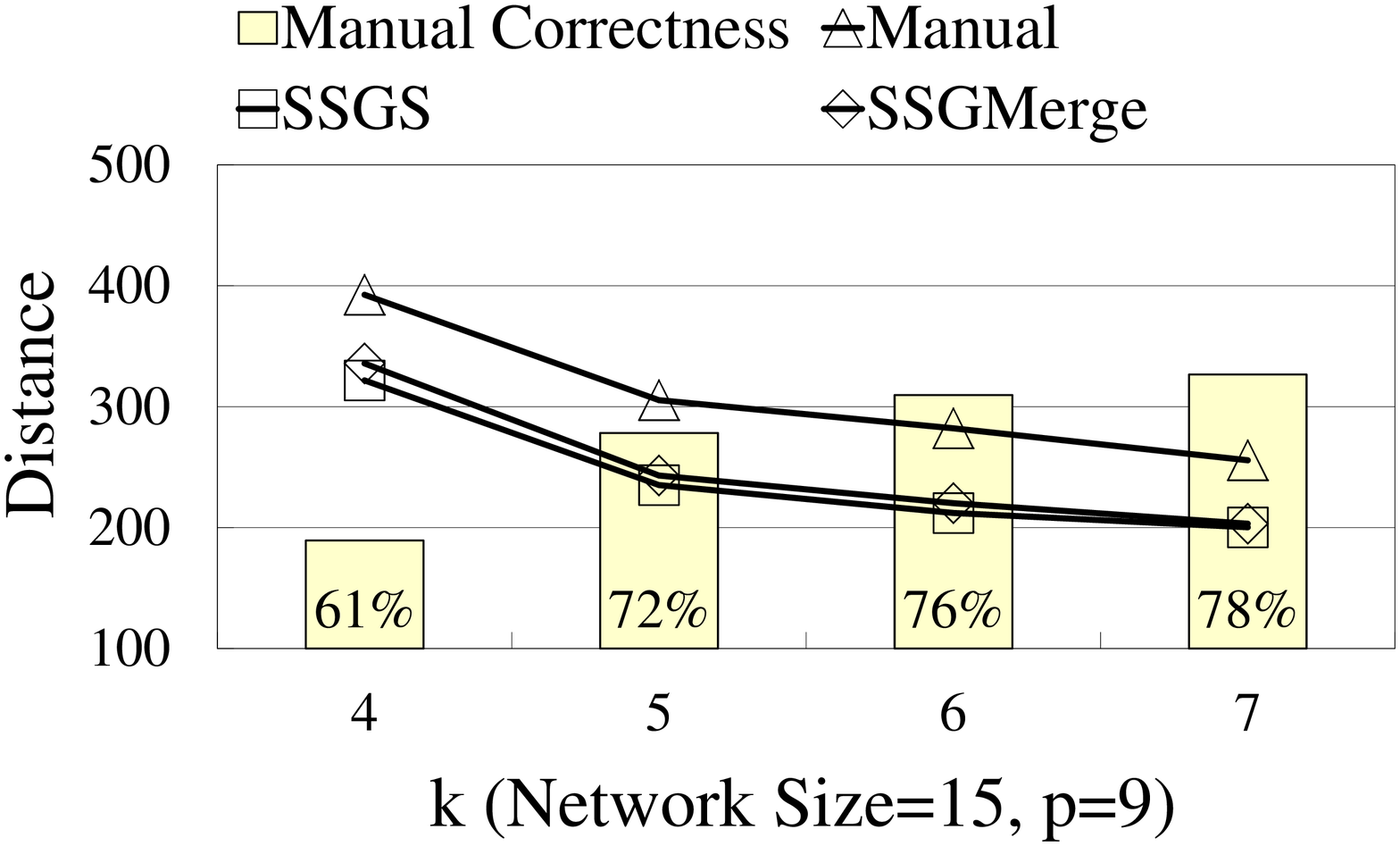} }%\vspace{-20pt}
\caption{Results of user study.}
%\vspace{-10pt}
\label{fig:FIG_EXPNEW}
\end{figure}

Figures \ref{fig:FIG_EXPNEW}(a)-(f) compare manual coordination,
SSGS and SSGMerge to answer SSGQ in the user study. Figure \ref{fig:FIG_EXPNEW}(a)
presents the time to find the solutions in different scenarios. The result
indicates that SSGQ is challenging for manual coordination, especially for a
large network size. In contrast, SSGS obtains the optimal solution with
less than 0.01 second, while SSGMerge obtains a near-optimal solution with much less time. 
Figure \ref{fig:FIG_EXPNEW}(b) with $p=5$ and $k=3$
demonstrates that the solutions from manual coordination require larger
spatial distance and thereby are not optimal. With a larger network size,
i.e., more friends nearby, it is easier to find a group of attendees with a
smaller total spatial distance to $q$. In addition, the solution quality in
Figure \ref{fig:FIG_EXPNEW}(c) shows that even in $p=5$, the solutions
obtained by manual coordination is not guaranteed to follow the familiarity
constraint, according to the correctness rate shown in Figure \ref%
{fig:FIG_EXPNEW}(c), because it is very challenging for a person to
jointly minimize the total spatial distance and ensure the familiarity
constraint. Moreover, the correctness rate drops dramatically as the network size
increases. On the other hand, as shown in Figure \ref{fig:FIG_EXPNEW}(b) and \ref{fig:FIG_EXPNEW}(c), 
SSGMerge obtains solutions which are very close to the optimal solution, this is because 
SSGMerge effectively utilizes the intermediate solutions to construct good solutions.

In Figures \ref{fig:FIG_EXPNEW}(d) and \ref{fig:FIG_EXPNEW}(e), we let each user freely select $%
p $ people to find out the familiarity preferred by each person in
activities with different $p$. The minimum $k$ here represents the smallest $%
k$ for each manual solution to follow the familiarity constraint. With this
parameter extracted from the manual solution, we regard it as an input
parameter for an SSGQ in the same social network. The results demonstrate
that SSGS and SSGMerge can find better solutions following the same $k$ with
a smaller time. In other words, even when a user does not specify $k$, it is
possible to analyze the previous manual coordination results and find out a
suitable $k$ for the user, such that SSGS and SSGMerge are able to find 
solutions in each query afterward.

Figure \ref{fig:FIG_EXPNEW}(f) with the network size as 15 and $p$ as 9 compares
the results of different $k$. As $k$ decreases, the correctness rate of
manual coordination drops because it becomes more difficult for a user to
find a tighter social group with the same number of attendees. Moreover, the
solution obtained by manual coordination is still worse than the solution of
SSGS and SSGMerge even with a loose requirement on social connectivity, i.e., a
large $k$.

These MRGQ tasks span various $p$, $k$, and $|Q|$, where $t$ is fixed to 10 km. In the user study, MAGS is equipped with APDO and all the proposed pruning strategies. We also compare the solution quality with an algorithm called \textit{GreedyManual (GM)}, which imitates the behavior of manual coordination. GM first finds the candidates within radius $t$ of each activity location and picks the activity location which has the largest number of candidates nearby.
Afterwards, if there exists a feasible group, GM returns it. Otherwise, it repeats the above procedure with
the remaining activity locations. 
%Due to the space constraints, the user study for SSGQ are presented in Appendix I.

\begin{figure}[tp]
\centering
\subfigure[Solution quality.] {\
\includegraphics[scale=0.21]{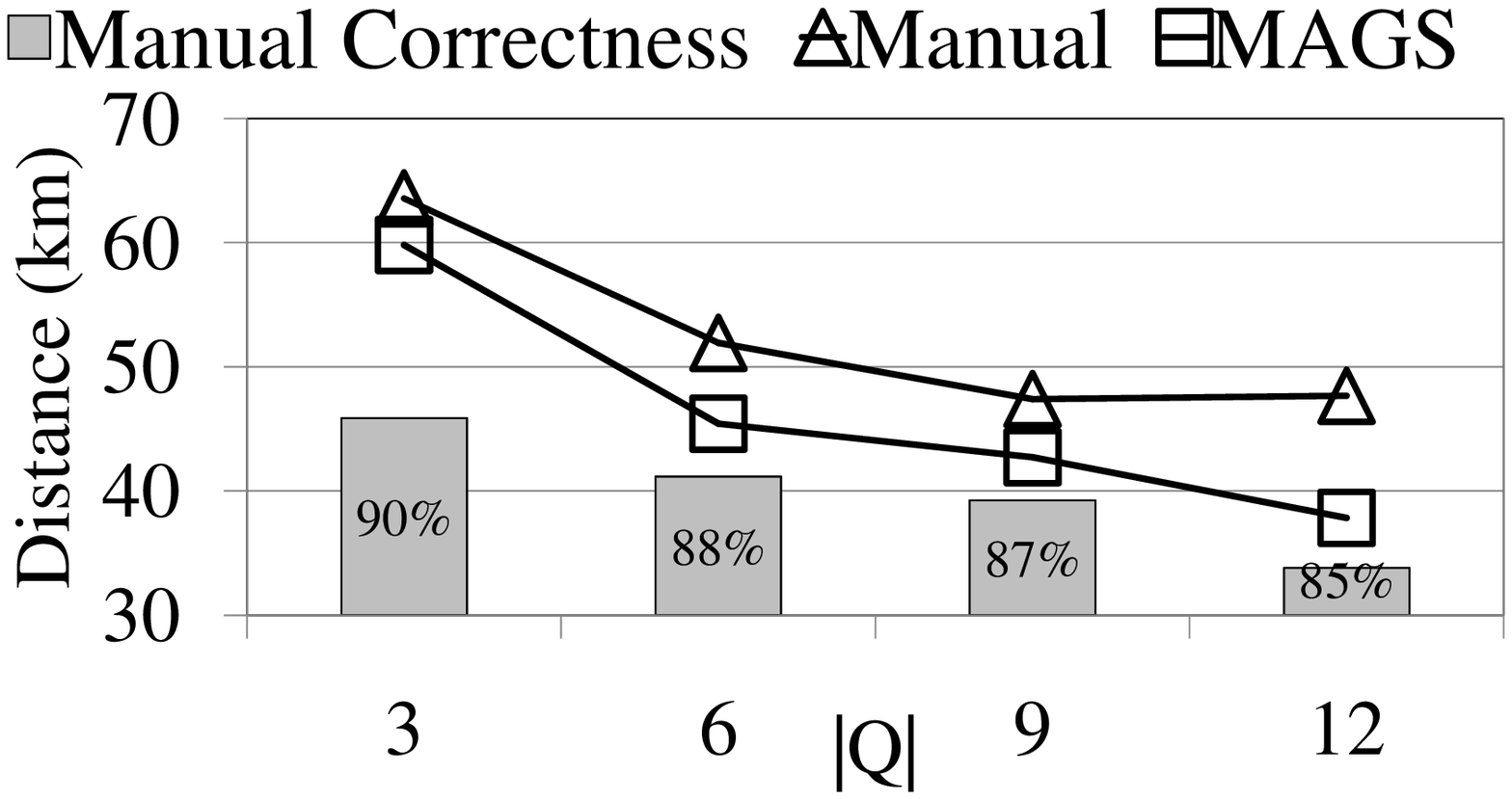} } 
\subfigure[Solutions with given $k$.] {\
\includegraphics[scale=0.21]{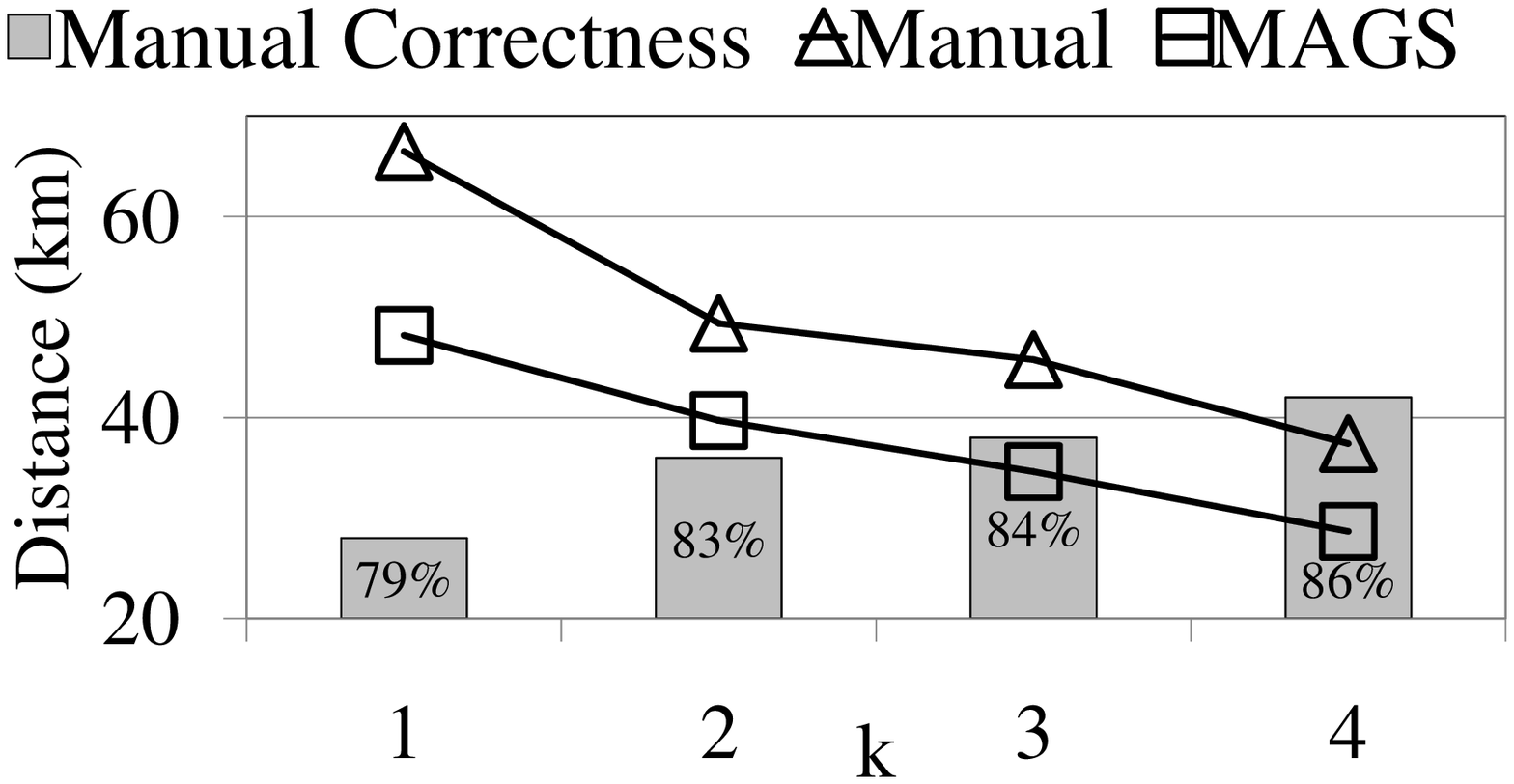} } 
\subfigure[Solutions without given $k$.] {\
\includegraphics[scale=0.21]{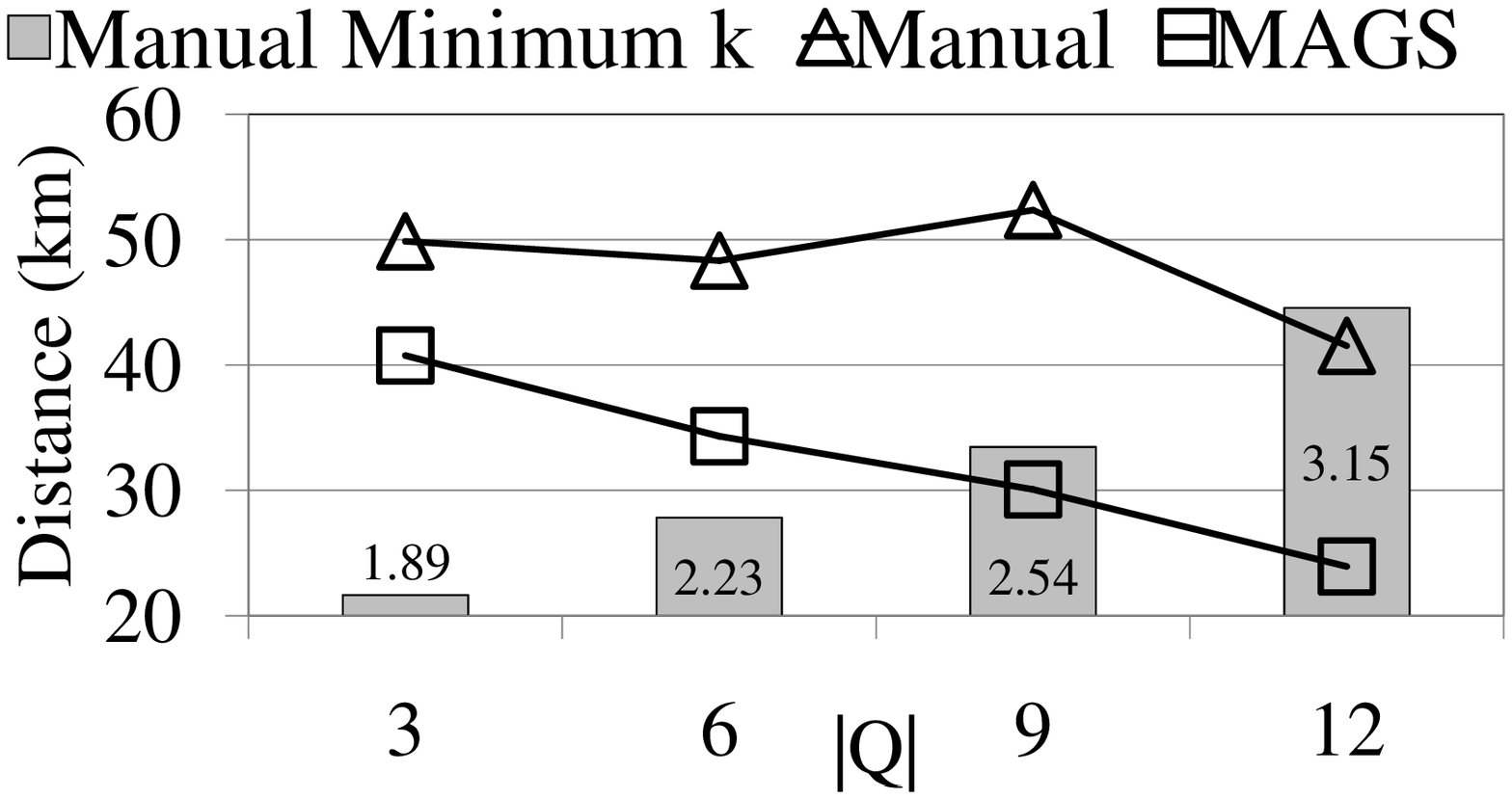} } 
\subfigure[Solutions with given $p$.] {\
\includegraphics[scale=0.21]{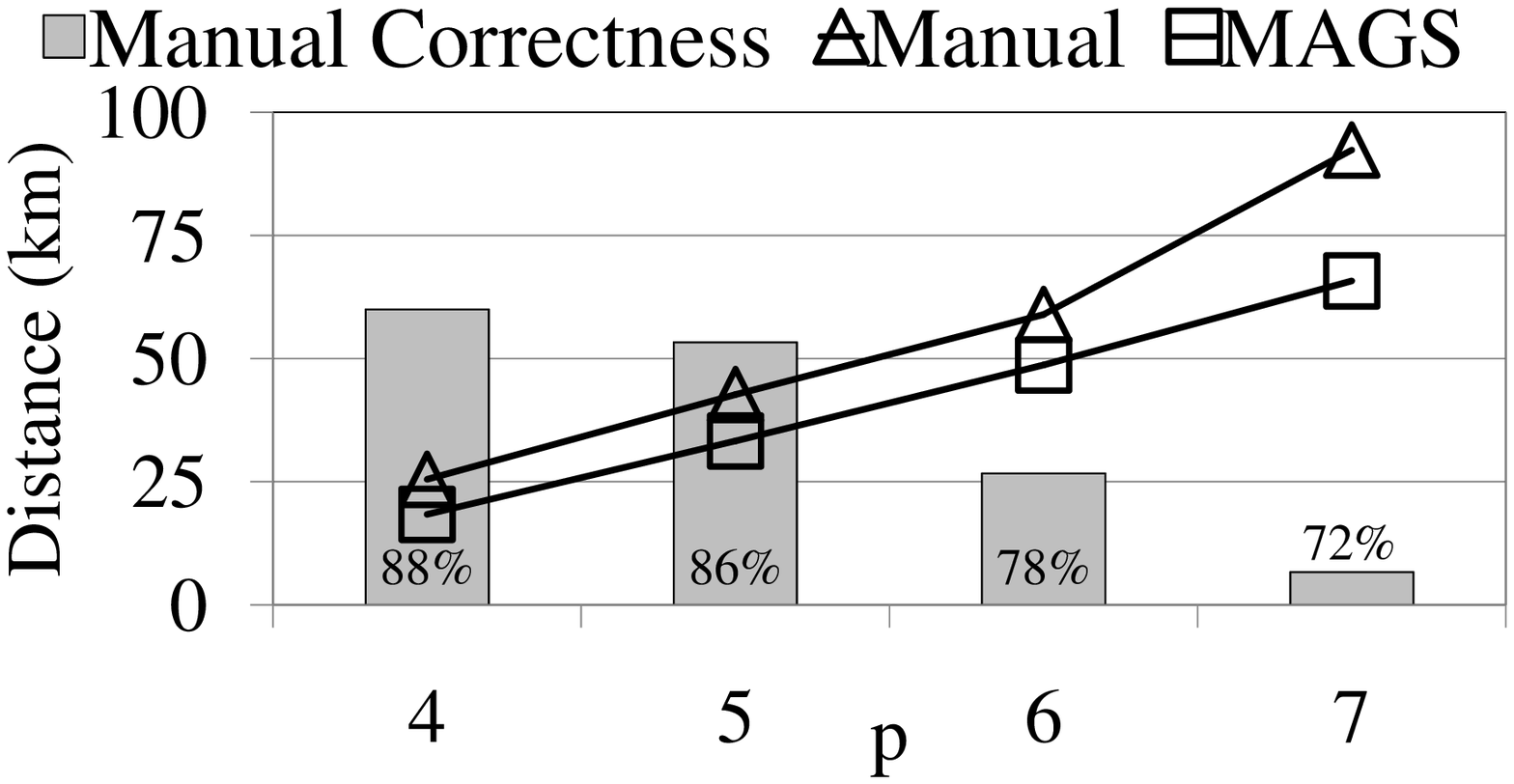} } 
\subfigure[Time without given $p$.] {\
\includegraphics[scale=0.21]{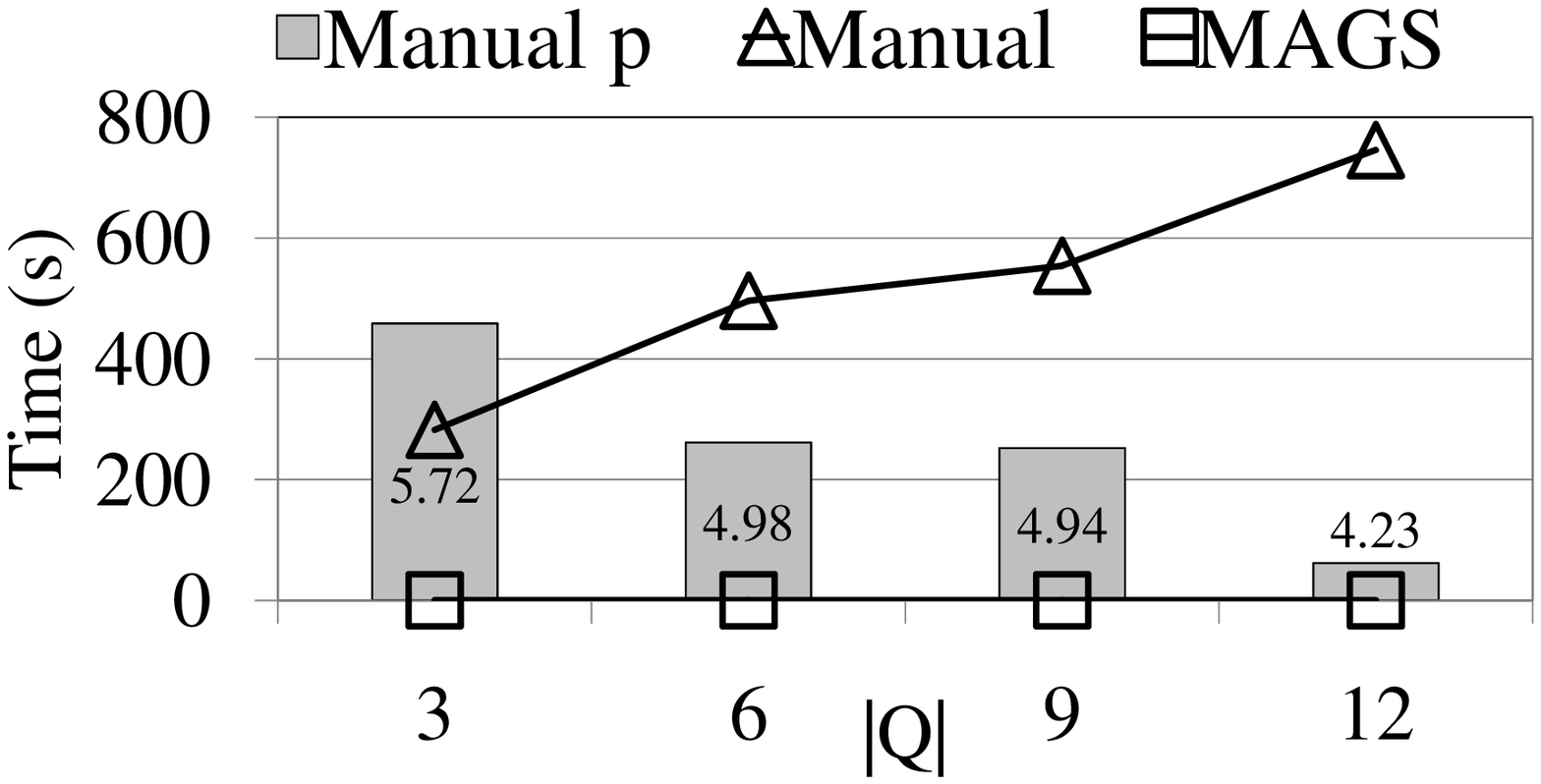} }
\subfigure[Time of activity organization.] {\  
\includegraphics[scale=0.21]{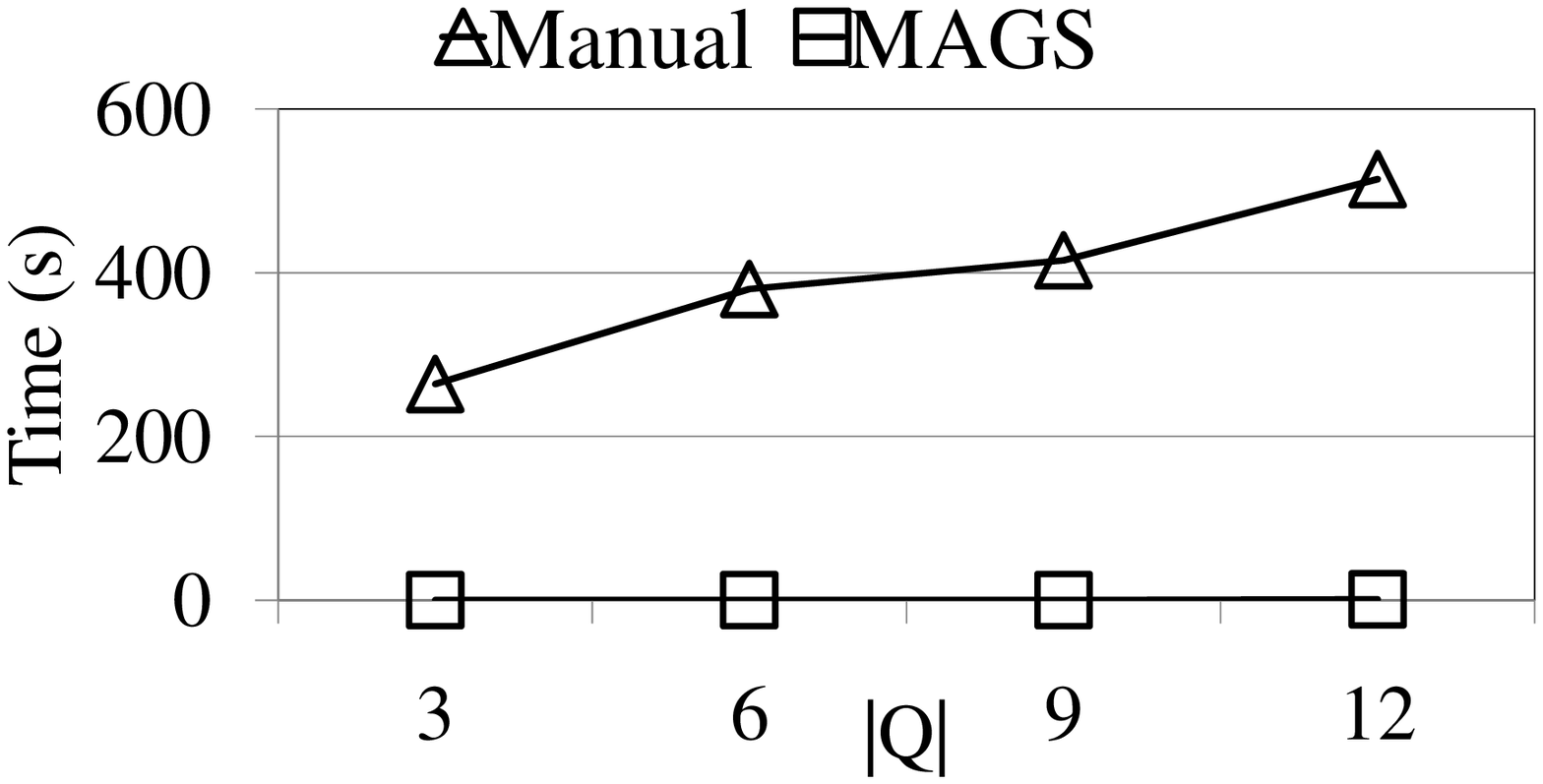} } 
\subfigure[Comparisons of GM.] {\  
\includegraphics[scale=0.21]{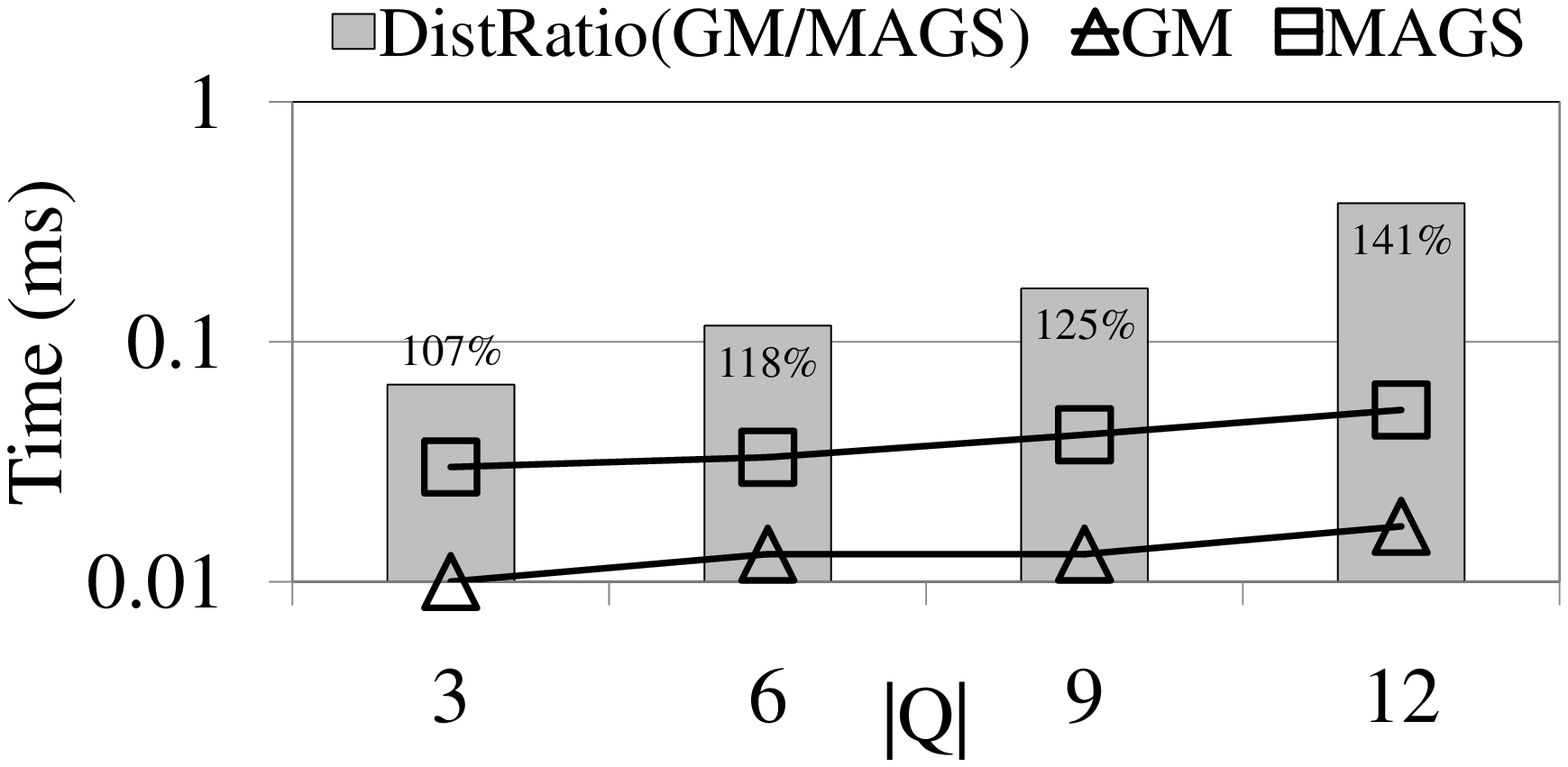} } 
%\subfigure[User satisfactory.] {\  
%\includegraphics[scale=0.21]{US_UserSatisfactory.eps} } 
\subfigure[User satisfaction.]{\ \label{fig_US_new}
\includegraphics[scale=0.14] {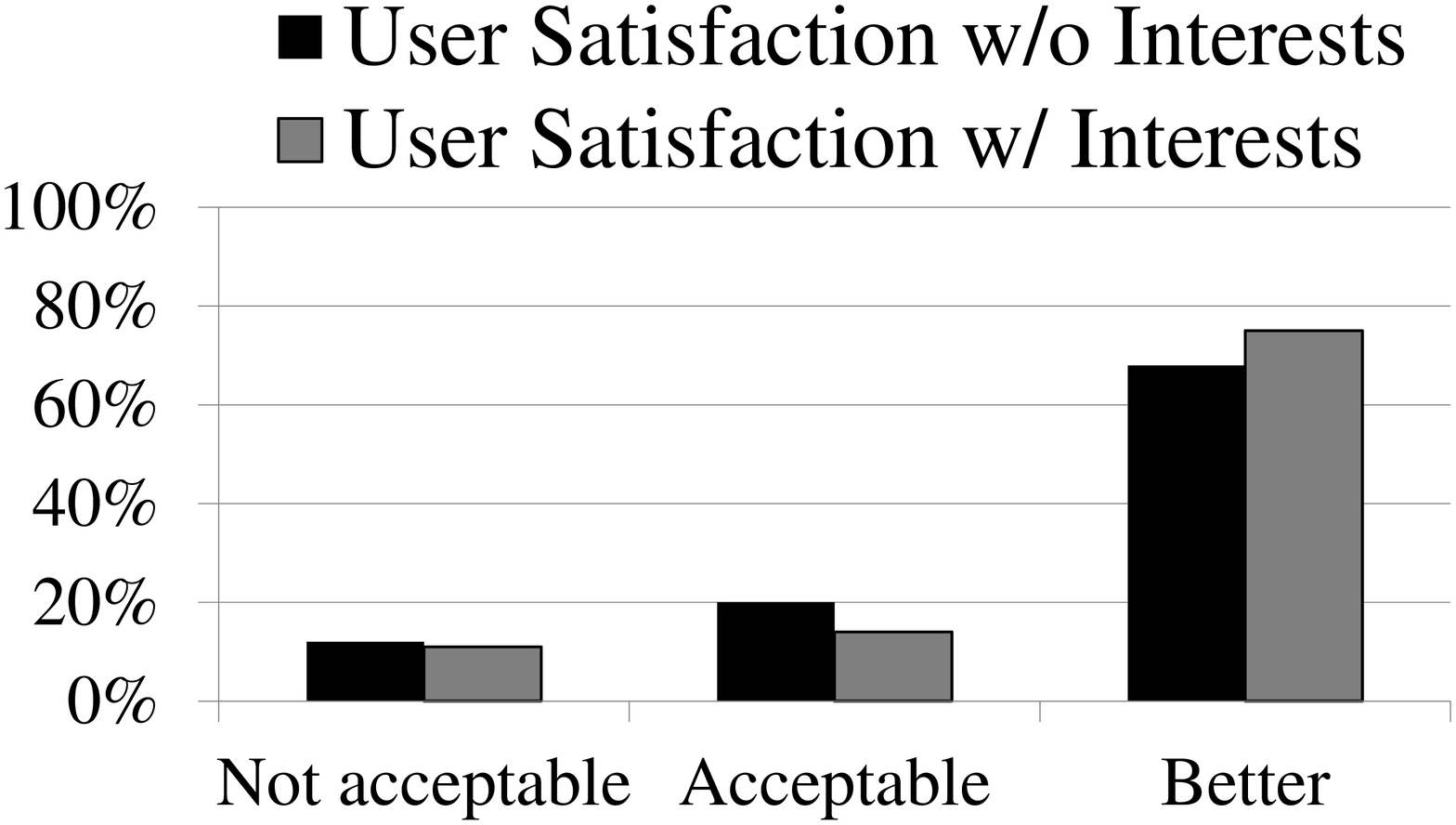} }
\caption{Results of user study.}
\vspace{-20pt}
\label{fig:FIG_US_MRGQ}
\end{figure}

Figures \ref{fig:FIG_US_MRGQ} compares manual coordination and
MAGS to answer MRGQ in the user study.  
Figure \ref{fig:FIG_US_MRGQ}(a) 
demonstrates that the solutions from manual coordination incur larger
spatial distance and thereby are not optimal. When the number of activity locations increases, 
it is easier to find a group of attendees and an activity location with a
smaller total spatial distance. In addition, the correctness rate in
Figure \ref{fig:FIG_US_MRGQ}(b) shows that even when $p=5$, the solutions
obtained by manual coordination are not guaranteed to follow the social
constraint, especially for a smaller $k$, because it is very challenging for a user to
jointly minimize the total spatial distance and ensure the social
constraint. In Figure \ref{fig:FIG_US_MRGQ}(c), we let each user freely select $5$ 
people and analyze the familiarity parameter preferred by each person in
activities. The minimum $k$ here represents the smallest 
$k$ to meet the familiarity constraint in user selection. With this
parameter extracted from the manual solution, we regard it as an input
parameter for an MRGQ query in the same social network. The results demonstrate
that users are difficult to handle small $k$ and large $|Q|$ due to the need to examine many more 
combinations. Thus, the distances
obtained by manual coordination are more deviated from the optimal solution obtained by MAGS.

Figures \ref{fig:FIG_US_MRGQ}(d) shows that as $p$ increases, the correctness rate and
solution quality of manual coordination significantly deteriorate because it becomes more difficult for a user to
find a tight social group. In Figure \ref{fig:FIG_US_MRGQ}(e), each user
can freely select any number of people for forming the group with $k=3$. 
Manual p in this figure indicates the average group
size measured in the user study. As $|Q|$ increases, the selected group size drops because it
becomes more challenging to find the optimal group. Moreover, users need much more time to find the
group when $|Q|$ grows, even with a small group size and a loose requirement on the social
connectivity, i.e., $p=\{4,5\}$ and $k=3$. Finally, Figure \ref{fig:FIG_US_MRGQ}(f)
presents the time spent to find the solutions in different scenarios. The result
indicates that MRGQ is challenging for manual coordination, especially for a
large number of potential candidate locations.

Figure \ref{fig:FIG_US_MRGQ}(g) compares the computation time and solution quality of GreedyManual (GM)
and MAGS. Although GM obtains the solutions within a smaller time, the solution quality is much worse than MAGS. 
This is because GM stops when a feasible group is obtained, which cannot effectively obtain the optimal solution. 

Since MRGQ can also consider the user interests (discussed in Section \ref{discussion}), we also compare the user satisfaction with or without considering user interests. 
We ask the users to choose 20 activity locations in MRGQ, where each location
is tagged as coffee shop, restaurant, bar, etc. The interest measure of each activity location 
$q_{i}$ to each user $u$ is specified as $\eta _{u,q_{i}}$ between 0 and 1 by the user. 
We let each user compare the groups selected by MAGS and the
user herself. Figure \ref{fig_US_new} with $p=7$ and $k=3$ compares
the user satisfaction with and without user interests incorporated. The
results manifest that $68\%$ and $75\%$ of the users agree that the groups
selected by MAGS outperform the manually selected groups before and after
incorporating the interests, respectively. Moreover, the increment of the
users that choose "Better" after incorporating the user interests mainly
come from those who previously chose "Acceptable". The results demonstrate
that incorporating user interests indeed improves the user satisfaction.

\vspace{-10pt}
\subsection{Performance Evaluation of Proposed Algorithms for MRGQ}

We evaluate the effectiveness and efficiency of the proposed algorithms for MRGQ. 
APDO and SRDO denote MAGS with All-Pair Distance Ordering and Single-Reference Distance Ordering (a simplified version of APDO, which is mentioned in Section \ref{DistOdrSOSA}), respectively, while
Socio-Spatial Ordering and Familiarity Pruning mentioned in Section \ref{Baseline} are also included.
In our experiments, unless specifically indicated, we set $k=4$, $p=8$, $|Q|=10,000$, and the maximum value of $t$ is 15 km.

\begin{figure}[tp]
\centering
\subfigure[Computation time.] {\
\includegraphics[scale=0.15]{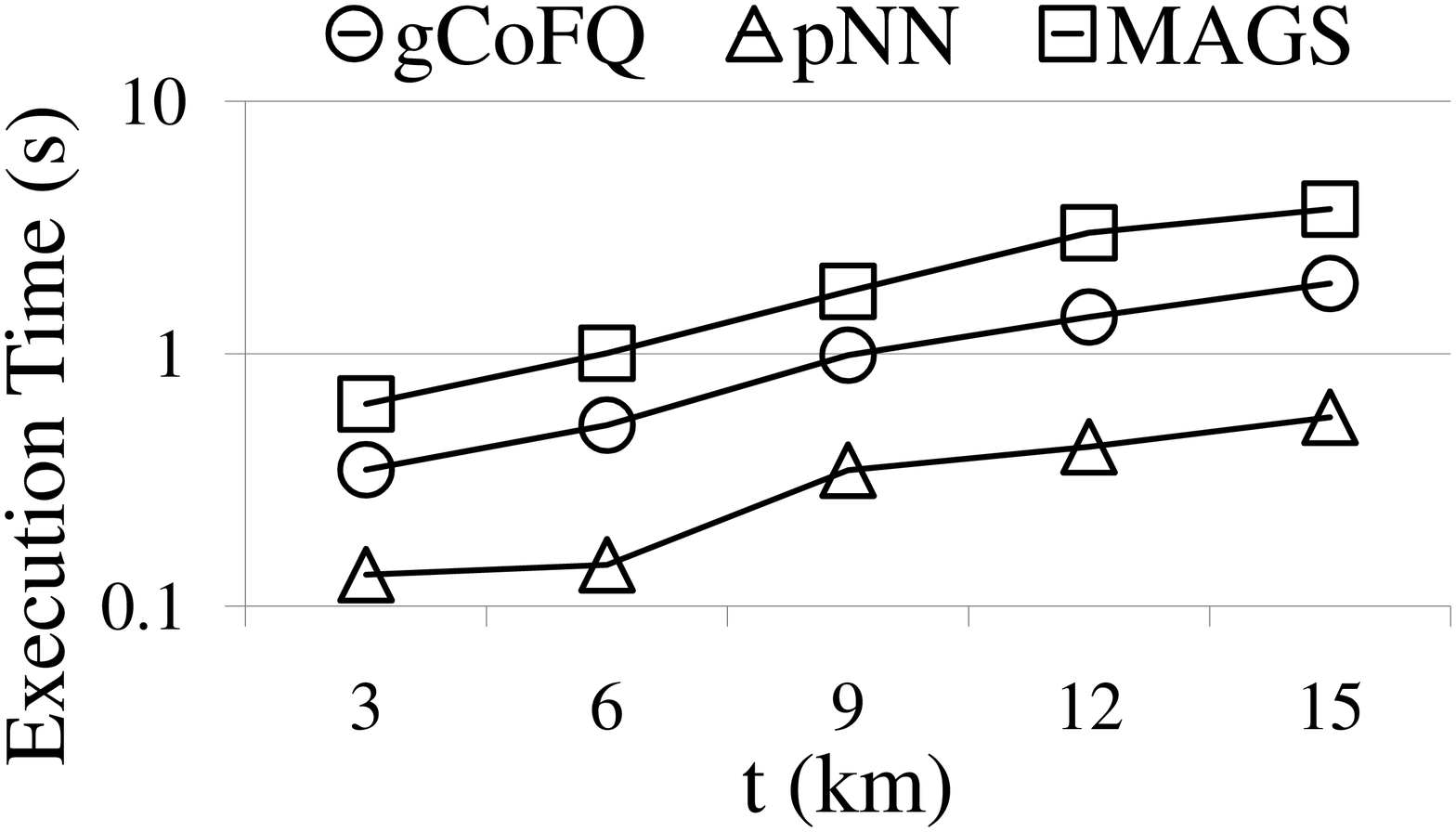} \label{FIG_Exp_RW_Time}} 
\subfigure[Solution quality.] {\
\includegraphics[scale=0.15]{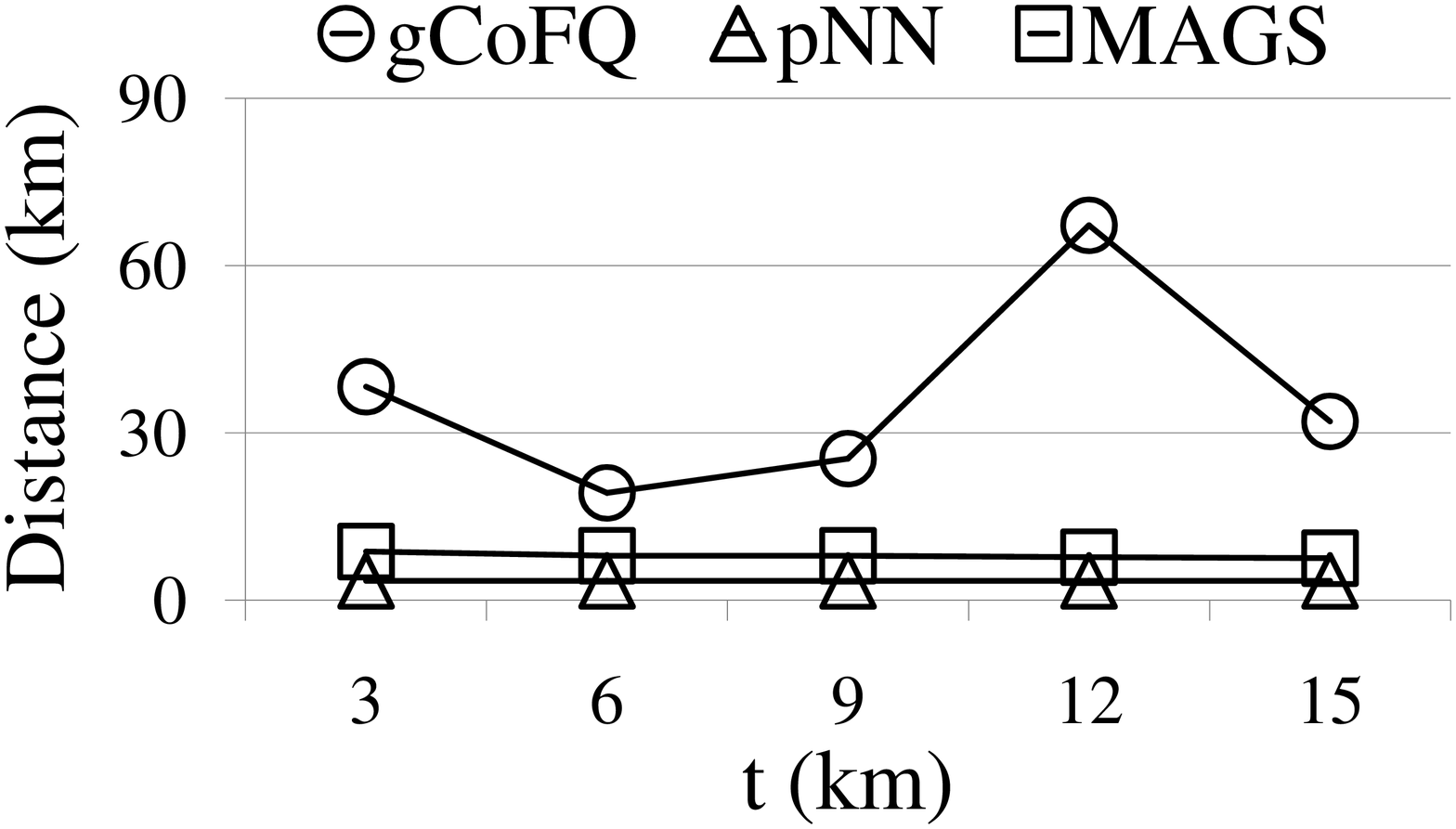} \label{FIG_Exp_RW_Sol}} 
\subfigure[Familiarity constraint.] {\
\includegraphics[scale=0.15]{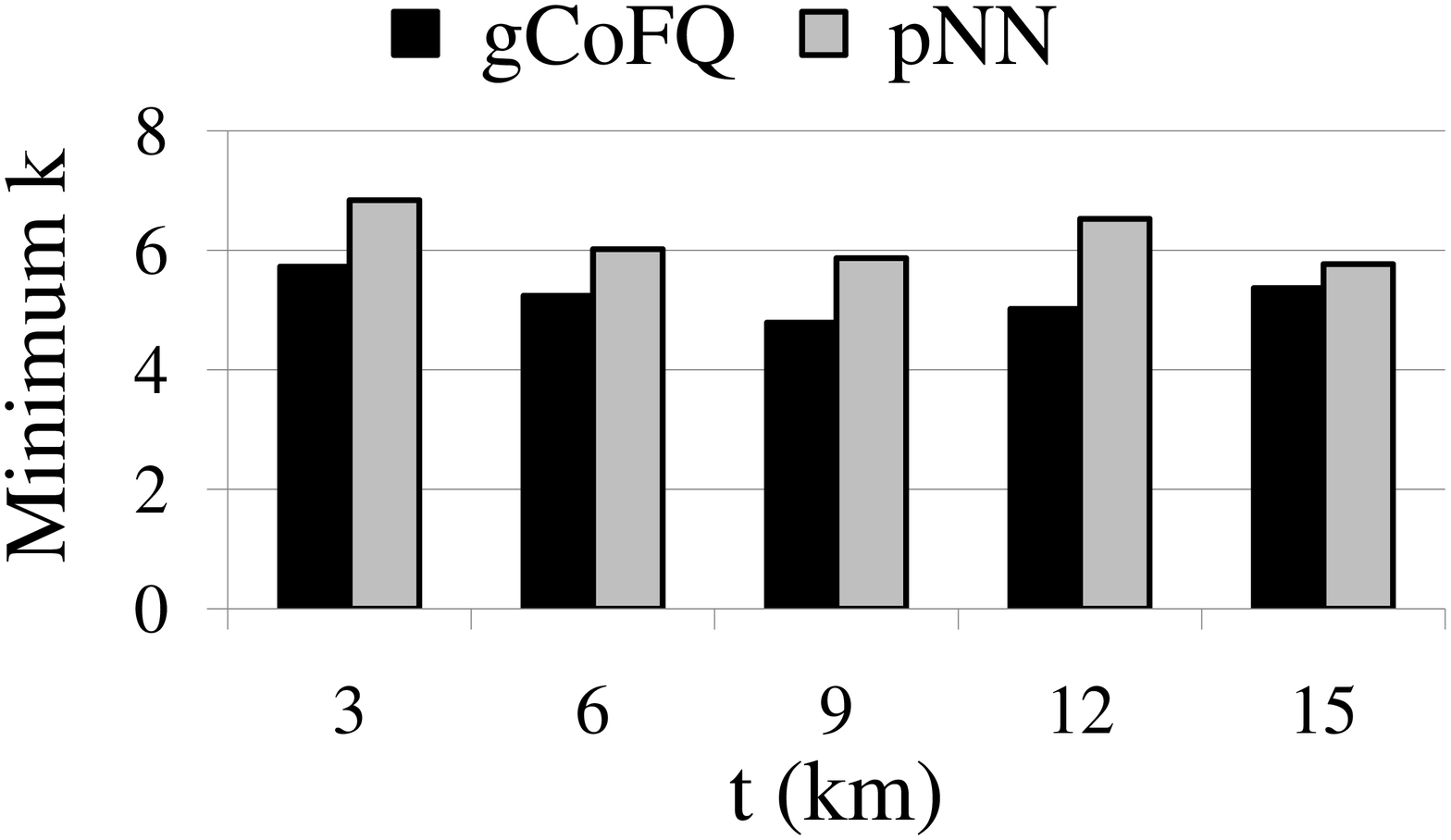} \label{FIG_RW_K}} 
\subfigure[Social diameter.] {\
\includegraphics[scale=0.15]{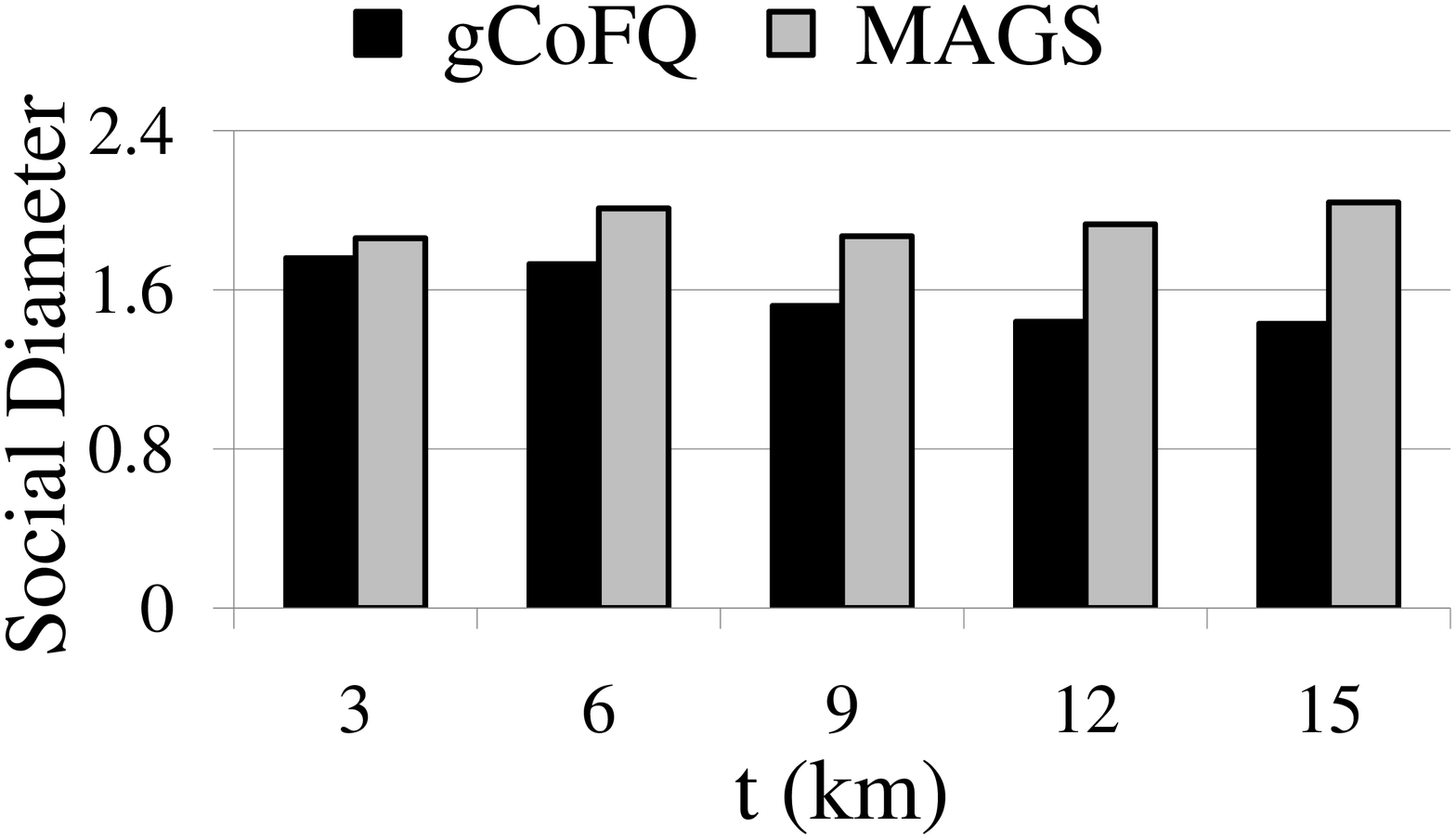} \label{FIG_RW_Dia}} 
\caption{Comparisons with relevant works on \textit{DataSet\_4SQ}.}
\vspace{-5pt}
\label{fig:FIG_RW}
\end{figure}

Figure \ref{fig:FIG_RW} first compares MAGS in MRGQ with the related works on \textit{DataSet\_4SQ}, where MAGS is equipped with APDO, Socio-Spatial Ordering and the proposed pruning strategies.
1) \textit{Geo-Social Circle of Friend Query (gCoFQ)} \cite{LSCH12} aims to find a group of $p$ people to minimize the linear combination of the social diameter and spatial diameter (maximum spatial distance between each pair of group members) of the selected group. In other words, there is no activity location in gCoFQ. In the experiments, gCoFQ is implemented to limit the spatial diameter within $2t$, and the nearest activity location after 
gCoFQ identifies the group is returned as the solution.
On the other hand, 2) \textit{pNN} extracts the group of $p$ members
along with their nearest activity location without considering the familiarity constraint. 
Figures \ref{FIG_Exp_RW_Time} and \ref{FIG_Exp_RW_Sol} compare the computation time and solution quality.
Although pNN obtains the group with the minimum time and distance, as shown in Figure \ref{FIG_RW_K}, 
the minimum $k$ of the obtained group (i.e., the minimum number of unfamiliar members each attendee has in the group) is
far from the specified $k$ value, i.e., $k=4$. In other words, the solution returned by pNN is not feasible to MGRQ. The solution quality of gCoFQ is worse than the other two algorithms
because gCoFQ does not examine activity locations during the group formation process, while Figure \ref{FIG_RW_K} shows that gCoFQ is also difficult to follow the familiarity constraint. In contrast, MAGS follows the familiarity constraint and can identify the optimal group along with the nearest activity location. 
Figure \ref{FIG_RW_Dia} compares the social diameter of the groups obtained by gCoFQ and MAGS, and the results of pNN are not able to be displayed because the groups obtained by pNN 
are usually disconnected. This figure manifests that, although MAGS is not designed to minimize the social diameter, the social diameter is still close to gCoFQ. 
%We also compare MAGS in MRGQ with other related works by varying different parameters, i.e., $k$, $p$, and $|Q|$, in Appendix M.

\begin{figure}[tp]
\centering
\subfigure[Different algorithms.] {\
\includegraphics[scale=0.15]{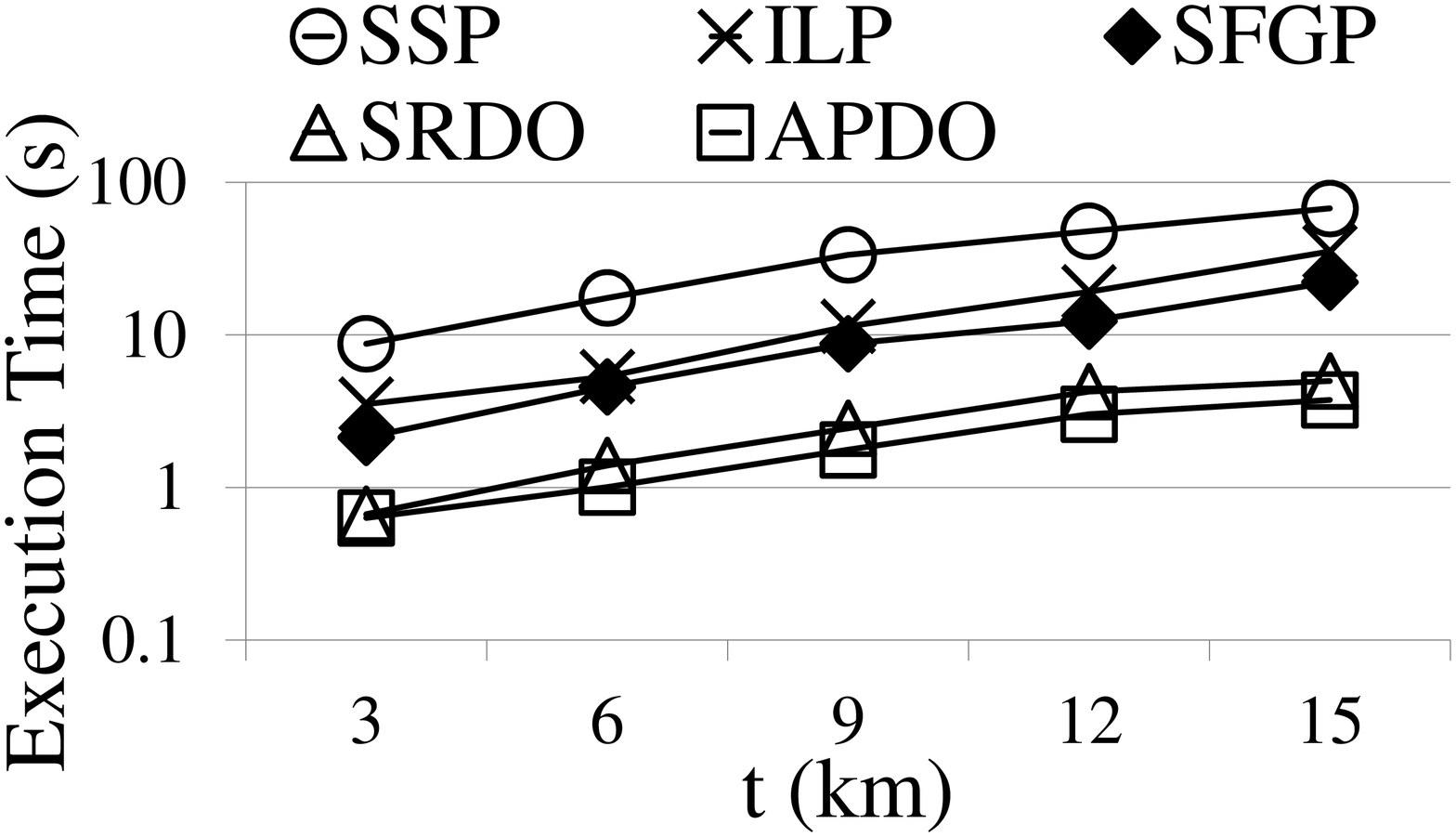} \label{FIG_Exp_Algo}} 
\subfigure[Impact of BallTree.] {\
\includegraphics[scale=0.15]{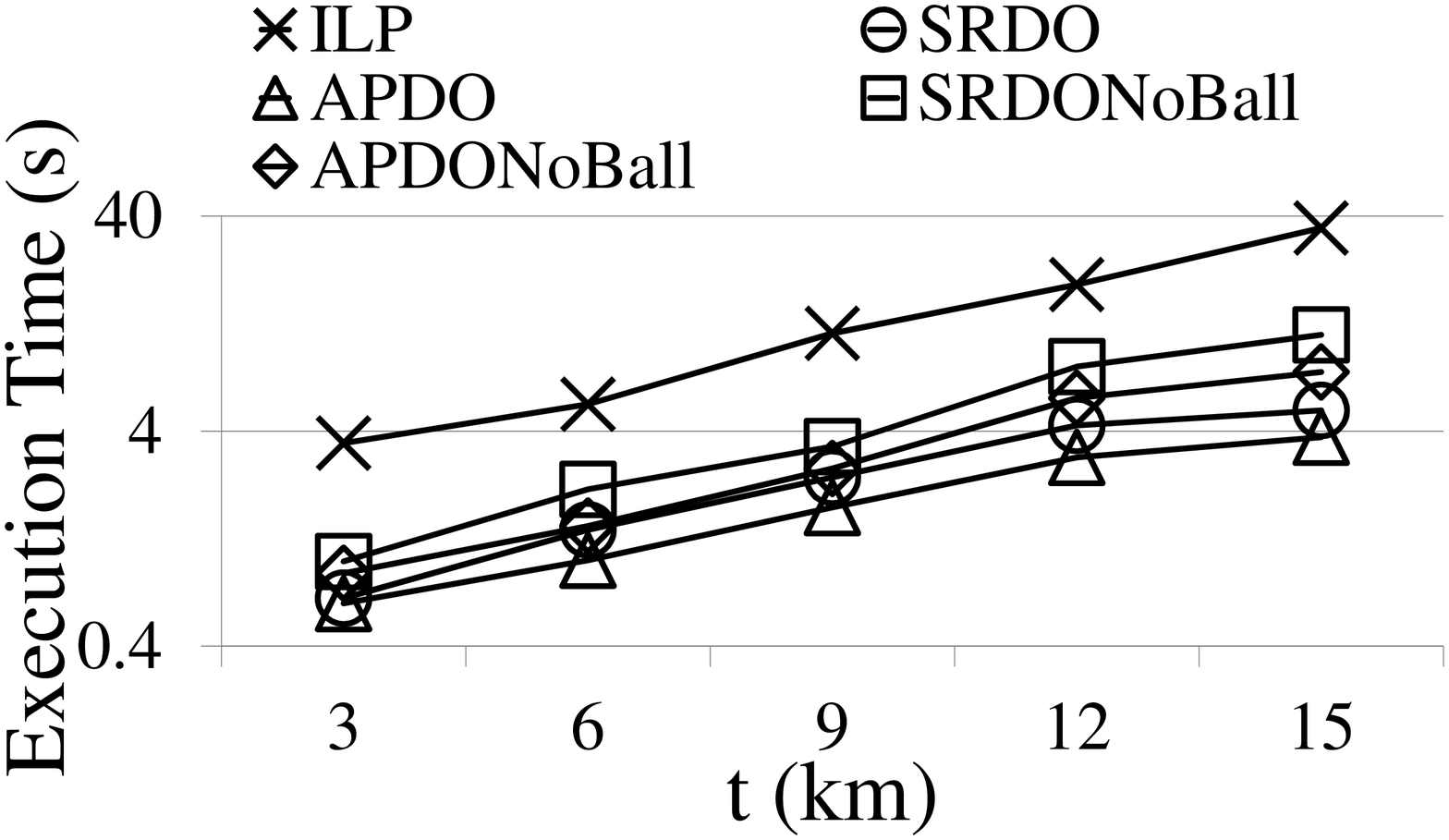} \label{FIG_Exp_Ball}} 
%\subfigure[Proximity-based Social Ordering.] {\
%\includegraphics[scale=0.15]{Exp_Diff_PSO.eps} \label{FIG_Exp_PSO}} 
%\subfigure[First feasible solution.] {\
%\includegraphics[scale=0.15]{Exp_Diff_FF.eps} \label{FIG_Exp_PSO_FF}} %\vspace{-5pt}
\subfigure[Distance Pruning.] {\
\includegraphics[scale=0.15]{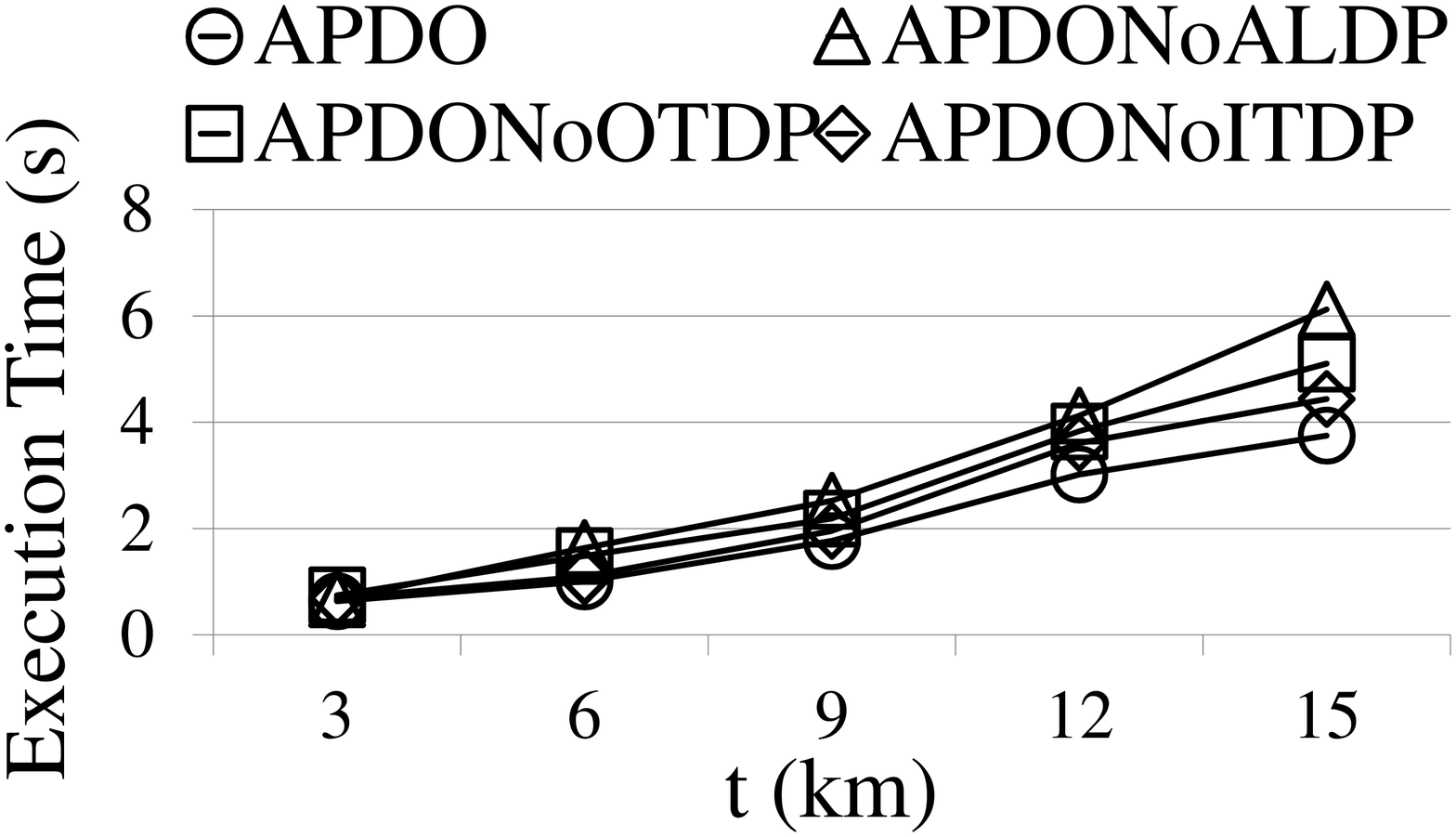} \label{FIG_Exp_DP}} 
\subfigure[Familiarity Pruning.] {\
\includegraphics[scale=0.15]{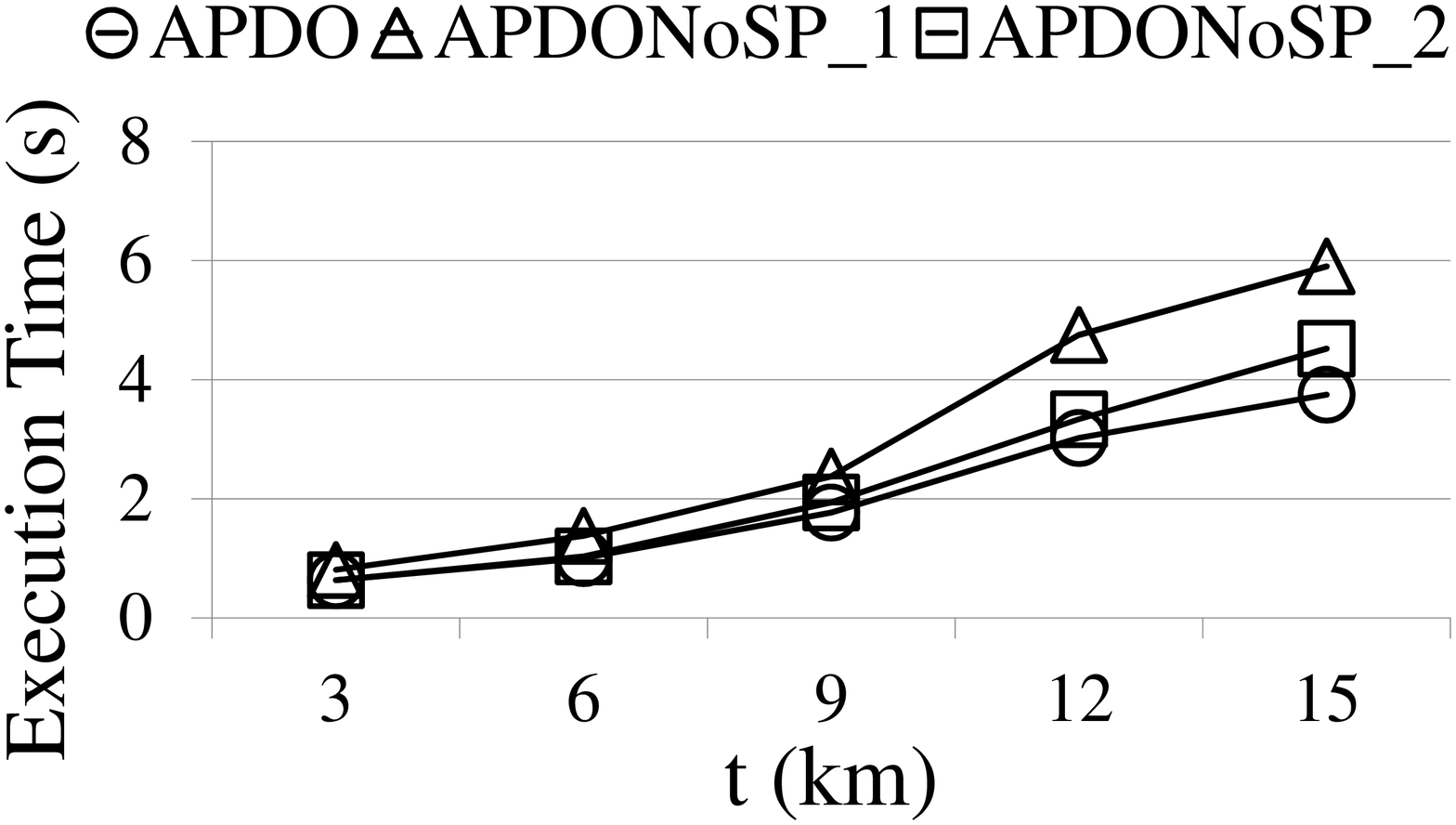} \label{FIG_Exp_SP}} 
\caption{Performance comparisons on \textit{DataSet\_4SQ}.}
\vspace{-5pt}
\label{fig:FIG_Exp_MAGS}
\end{figure}

Figures \ref{fig:FIG_Exp_MAGS} evaluates the efficiency of MAGS on \textit{DataSet\_4SQ}. Figure \ref{FIG_Exp_Algo} compares the computation time of the proposed algorithms with different values of $t$. Given its massive search space, SSP incurs the largest computation time as $t$ grows.
On the other hand, equipped with Socio-Spatial Ordering, BallTree, Distance Pruning, and
Familiarity Pruning, SRDO and APDO effectively reduce the time to acquire the optimal
solution and outperform Integer Linear Programming (ILP). Moreover, APDO requires the minimum computation time because it can effectively minimize
the total spatial distance from $S_I$ to $q_{ref}$ during each expansion of $S_I$.
Figure \ref{FIG_Exp_Ball} compares different algorithms with or without 
BallTree. Both SRDO and APDO require smaller computation time with BallTree and outperform ILP, 
since the proposed Outer-Triangle, Inner-Triangle and Activity Location Distance Pruning strategies are able to effectively remove redundant activity locations (within balls) at early stages.

%In addition to comparing different algorithms, we also evaluate the impact of Proximity-based Social Ordering (PSO) in different algorithms.
%Figure \ref{FIG_Exp_PSO} indicates that PSO can effectively reduce
%the computation time because it considers both spatial and social domains and thereby
%is able to guide an efficient search of the feasible solutions and the optimal
%solution. Figure \ref{FIG_Exp_PSO_FF} shows the
%computation time of the first feasible solution in different algorithms. The result indicates that
%PSO can effectively reduce the time for finding the first feasible solution. Therefore, the proposed 
%algorithms are capable of pruning more redundant groups at the early stages and reduce the computation time.

Figures \ref{FIG_Exp_DP} and \ref{FIG_Exp_SP} present the impact of the proposed pruning strategies, i.e., Outer-Triangle Distance Pruning (OTDP), Inner-Triangle Distance Pruning (ITDP), Activity Location Distance Pruning (ALDP), and Familiarity Pruning shown in Eqs. (\ref{MRGQ_SP_1}) and (\ref{MRGQ_SP_2}) in Section \ref{Baseline} (denoted as SP\_1 and SP\_2). The results manifest
that these pruning strategies effectively process the spatial and social relationships and indeed are critical for efficiently processing MRGQ. Moreover, the first Familiarity Pruning (SP\_1) is more powerful than the second one (SP\_2) since 
it derives a tighter upper bound on the number of people acquainted with each member in $S_I$.

\subsection{Performance Evaluation of Proposed Algorithms for SSGQ}

%FIG8
\begin{figure}[tp]
\centering
\subfigure[Different strategies.] {\
\includegraphics[scale=0.2]{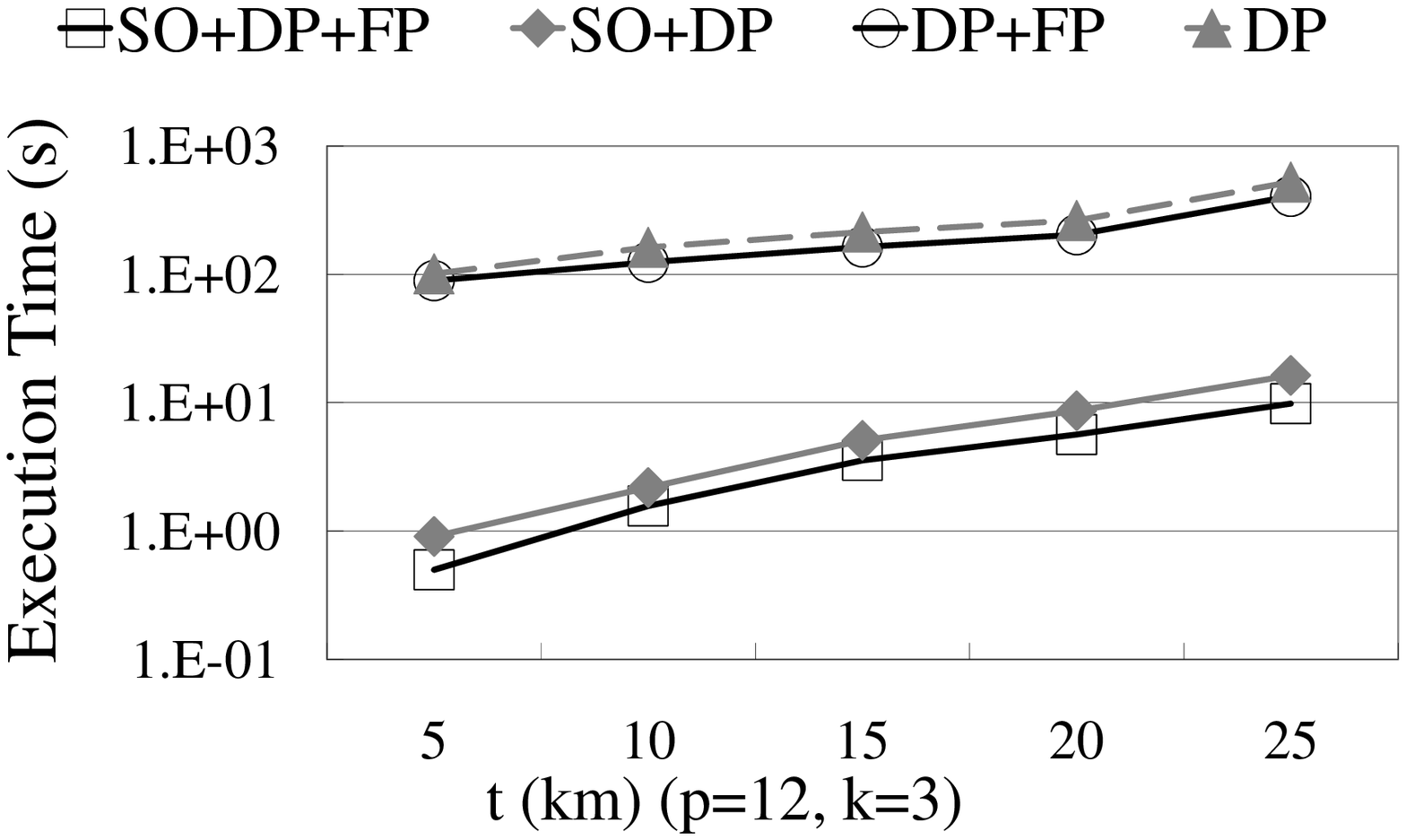} \label{FIG_Exp_SSGSelect_Diff_Pruning}} 
\subfigure[Socio-Spatial Ordering.] {\
\includegraphics[scale=0.2]{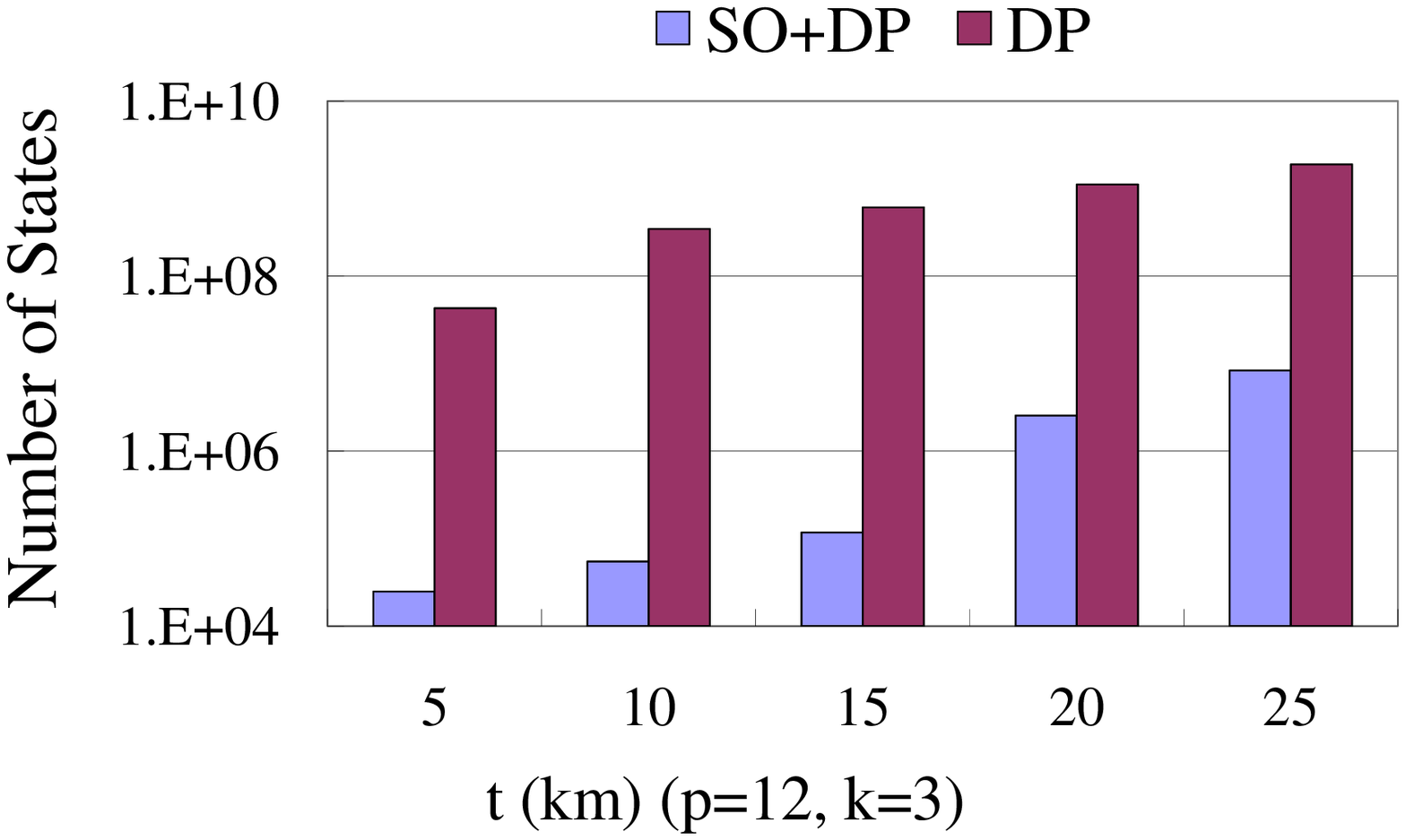} \label{FIG_Exp_SSGSelect_Diff_Pruning_State}} %\vspace{-5pt}
\subfigure[Varying $p$.] {\
\includegraphics[scale=0.2]{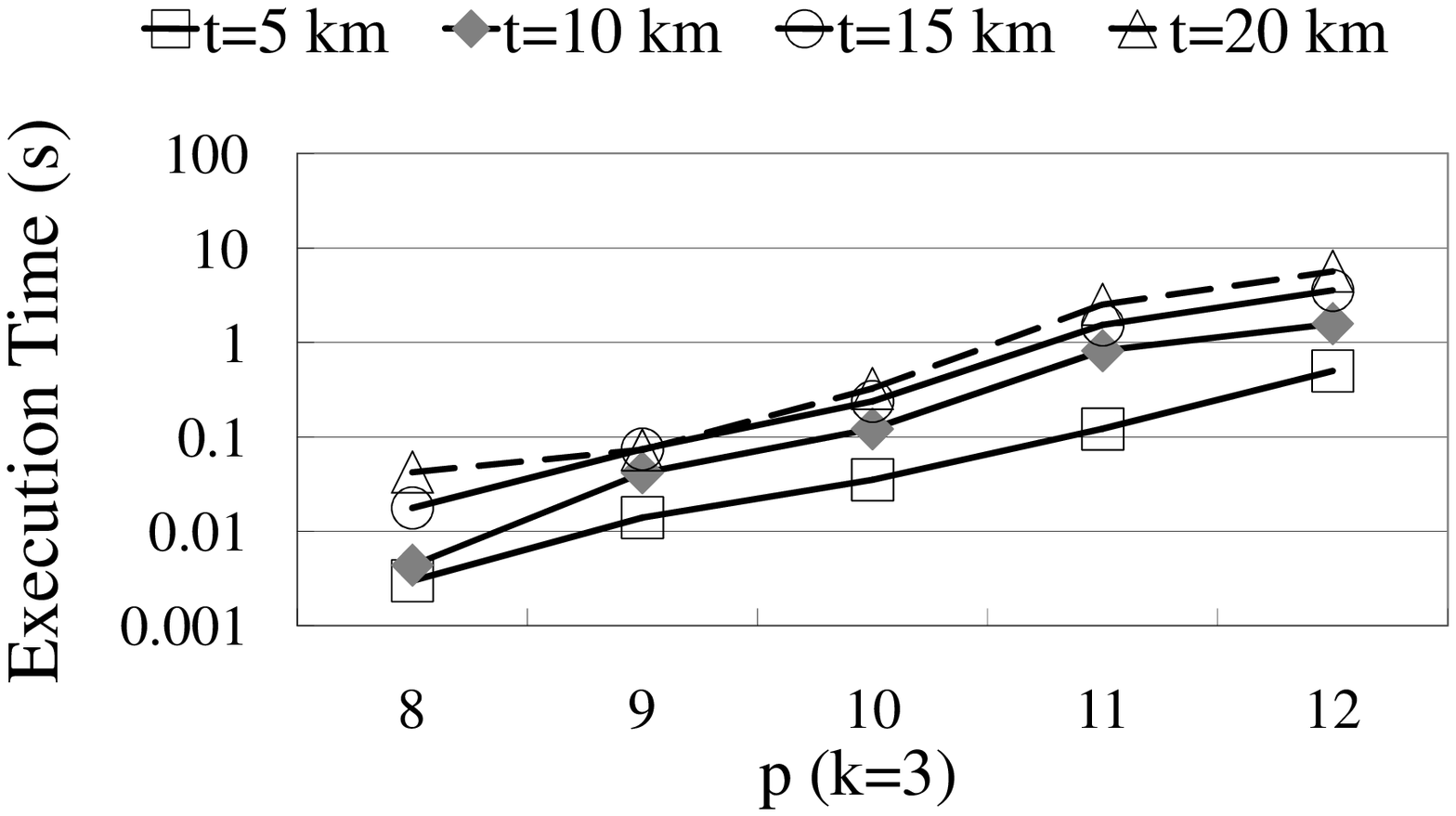} \label{FIG_Exp_Diff_P}} 
\subfigure[Varying $k$.] {\
\includegraphics[scale=0.2]{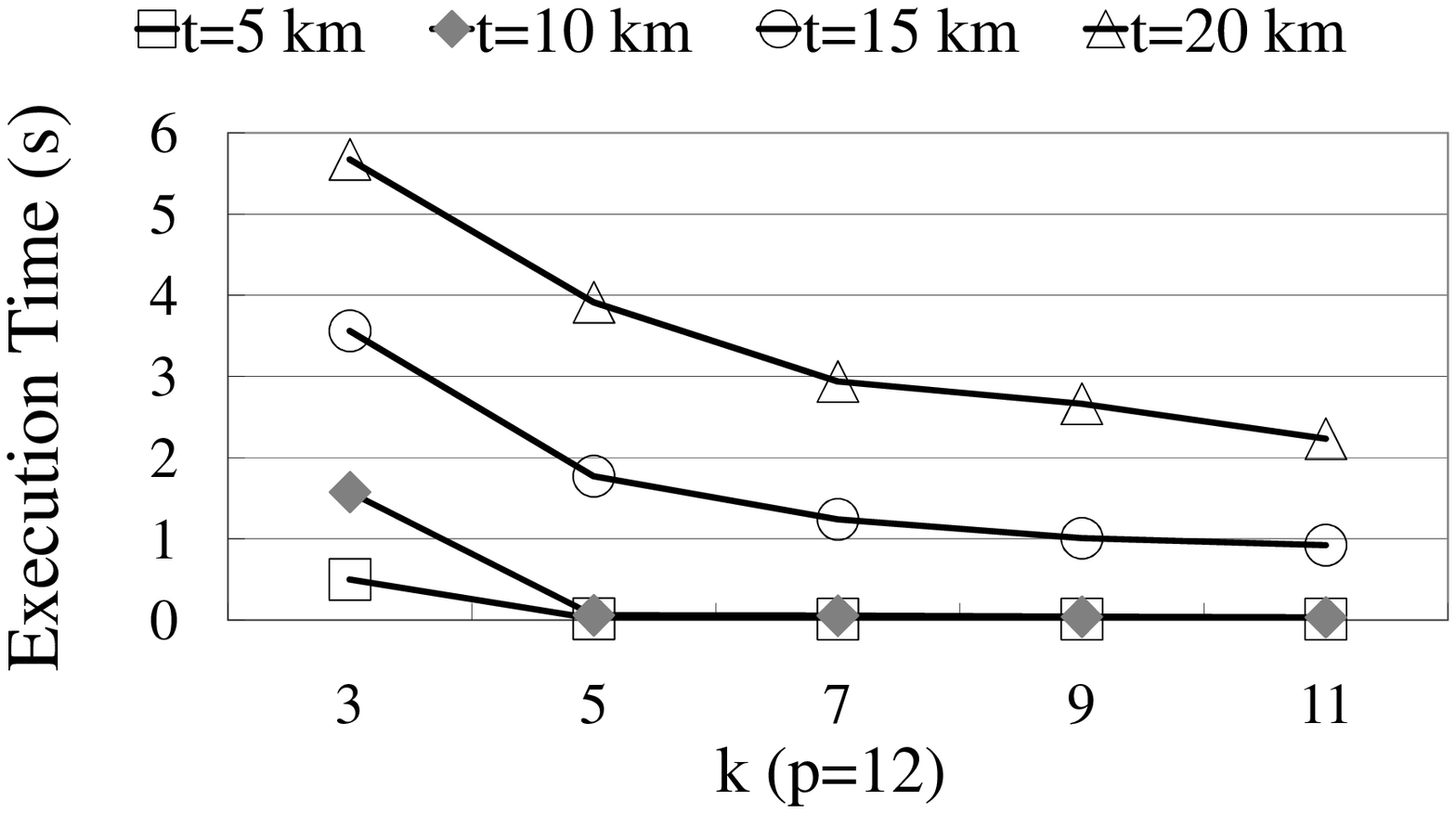} \label{FIG_Exp_Diff_K}} 
\caption{Performance comparisons on \textit{DataSet\_4SQ}.}
%\vspace{-10pt}
\label{fig:FIG_Exp_SSGSelect}
\end{figure}

We set the spatial radius, $t$, as 10km in this set of experiments, which is determined based on the user study. The study indicates that most of the users are willing to participate in impromptu activities within 10km from them.

For SSGMerge, we empirically set $\lambda $ within the range $50\leq \lambda \leq 800$,
because we observe that in this range, SSGMerge incurs small execution time while the obtained solutions achieve 
significant improvement over the straightforward approach, i.e., $i$-th feasible solution. 
On the other hand, $w$ should not be set too small, e.g., smaller than 10000, because the intermediate solutions 
in this case will not be able to include a sufficient number of distinct candidates, and thus limiting the possibility for constructing good solutions.

%In addition to comparing different algorithms, we also compare the impact of the ordering and pruning strategies for SSGS.
Figures \ref{fig:FIG_Exp_SSGSelect}(a) and \ref{fig:FIG_Exp_SSGSelect}(b) analyze the proposed strategies in Section 4, where SO, DP and FP denote Socio-Spatial Ordering, Distance Pruning and Familiarity Pruning, respectively.
The result indicates that Socio-Spatial Ordering (SO) is effective in reducing 
the execution time. This is because
Socio-Spatial Ordering considers both spatial and social domains and thereby
is able to guide the efficient search of the feasible solutions and the optimal
solution by exploring fewer states in the branch-and-bound tree, as shown
in Figure \ref{FIG_Exp_SSGSelect_Diff_Pruning_State}. 

Figures \ref{FIG_Exp_Diff_P} and \ref{FIG_Exp_Diff_K} compare SSGS with different parameter settings.
Figure \ref{FIG_Exp_Diff_P} indicates that the execution time increases as $p$ grows, 
because SSGS in this case needs to explore a larger search space to find the optimal
solution. On the other hand, Figure \ref{FIG_Exp_Diff_K} shows that a larger $k$ 
leads to a smaller execution time because it becomes easier to obtain
feasible solutions for Distance Pruning to trim the search space. 

%We further compare SSGS and SSGMerge with Integer Linear Programming for SSGQ (which is detailed in Appendix D), and also perform a detailed analysis on SSGMerge to examine its solution quality under different parameter settings. Please refer to Appendix O for comparisons with ILP for SSGQ, and Appendix P for detailed analysis on SSGMerge. 

\subsection{Comparisons with ILP for SSGQ}
\begin{figure}[tp]
\centering
\subfigure[Varying $t$.] {\
\includegraphics[scale=0.2]{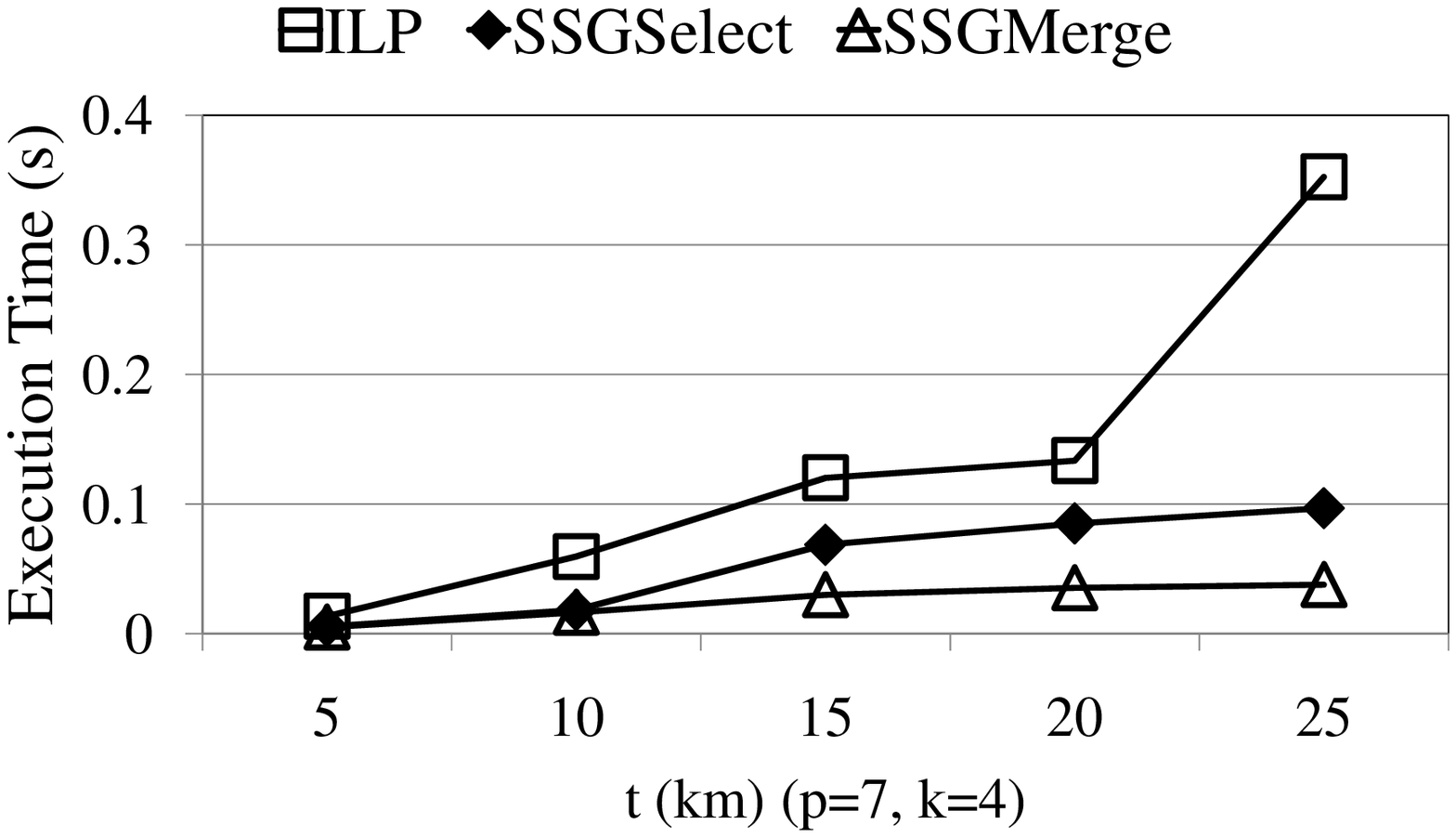} \label{FIG_Exp_ILP_Diff_V} }
\subfigure[Varying $p$.] {\
\includegraphics[scale=0.2]{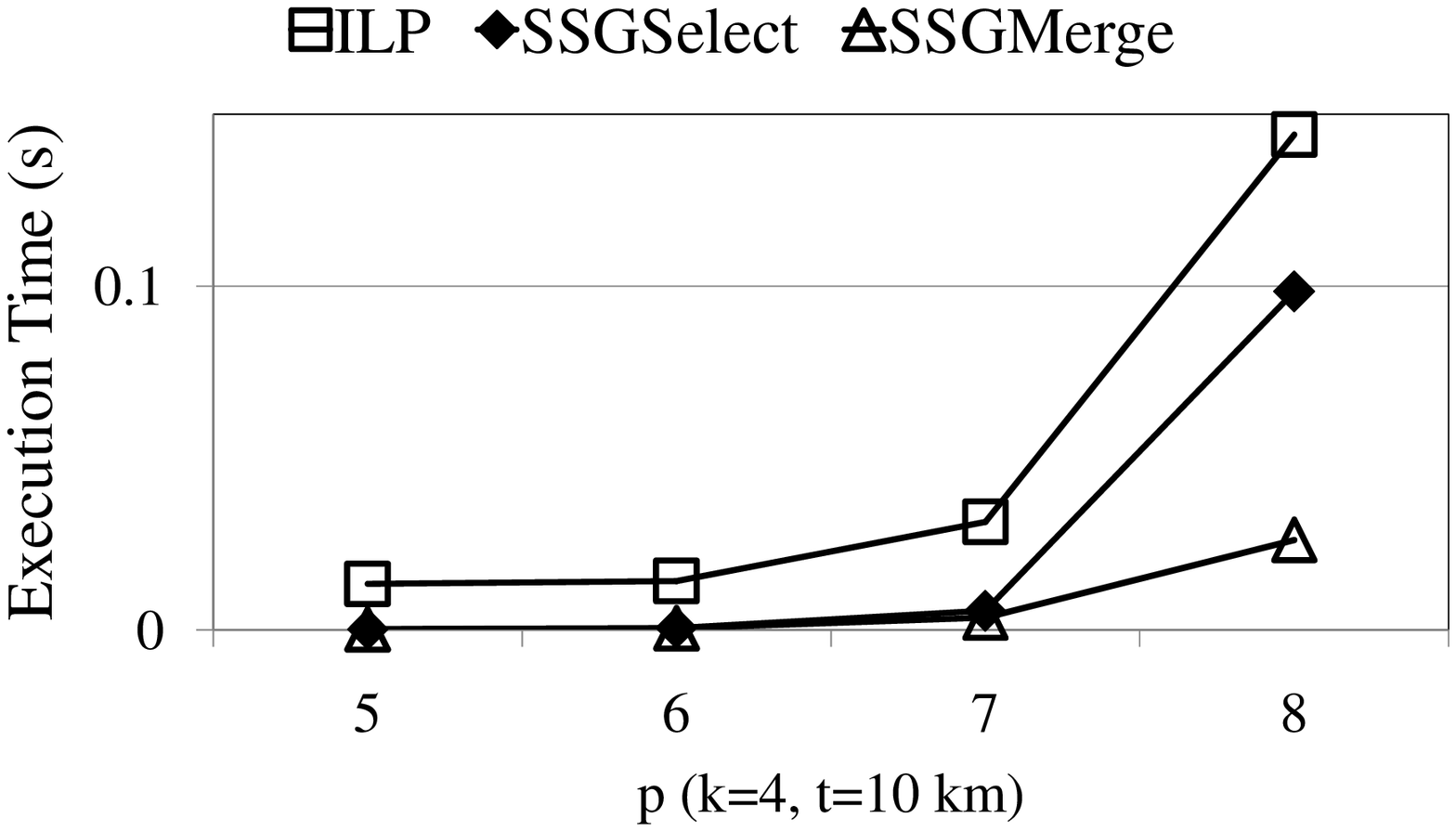} \label{FIG_Exp_ILP_Diff_P}}
\caption{Comparisons of the proposed algorithms with ILP for SSGQ.}
%\vspace{-10pt}
\label{FIG_Exp_ILP}
\end{figure}

Figure \ref{FIG_Exp_ILP} compares the performance of Integer Linear Programming (ILP) with SSGMerge and SSGS. 
In our experiments, a renown general-purposed commercial parallel optimizer, CPLEX [1], is adopted to find the optimal solutions with the proposed ILP formulation, while both SSGMerge and SSGS are single-thread programs. 
ILP here represents a baseline benchmark for examining the efficiency of the proposed algorithms.
Although SSGS and SSGMerge run in single-thread, 
they still outperform ILP.
This is because SSGMerge and SSGS carefully include effective pruning and ordering strategies.
Moreover, SSGMerge exploits the structure of intermediate solutions.
Therefore, SSGMerge achieves superior performance over SSGS and ILP.

\subsection{Performance Evaluation of SSGMerge}
\begin{figure}[tp]
\centering
\subfigure[Solution quality.] {\
\includegraphics[scale=0.2]{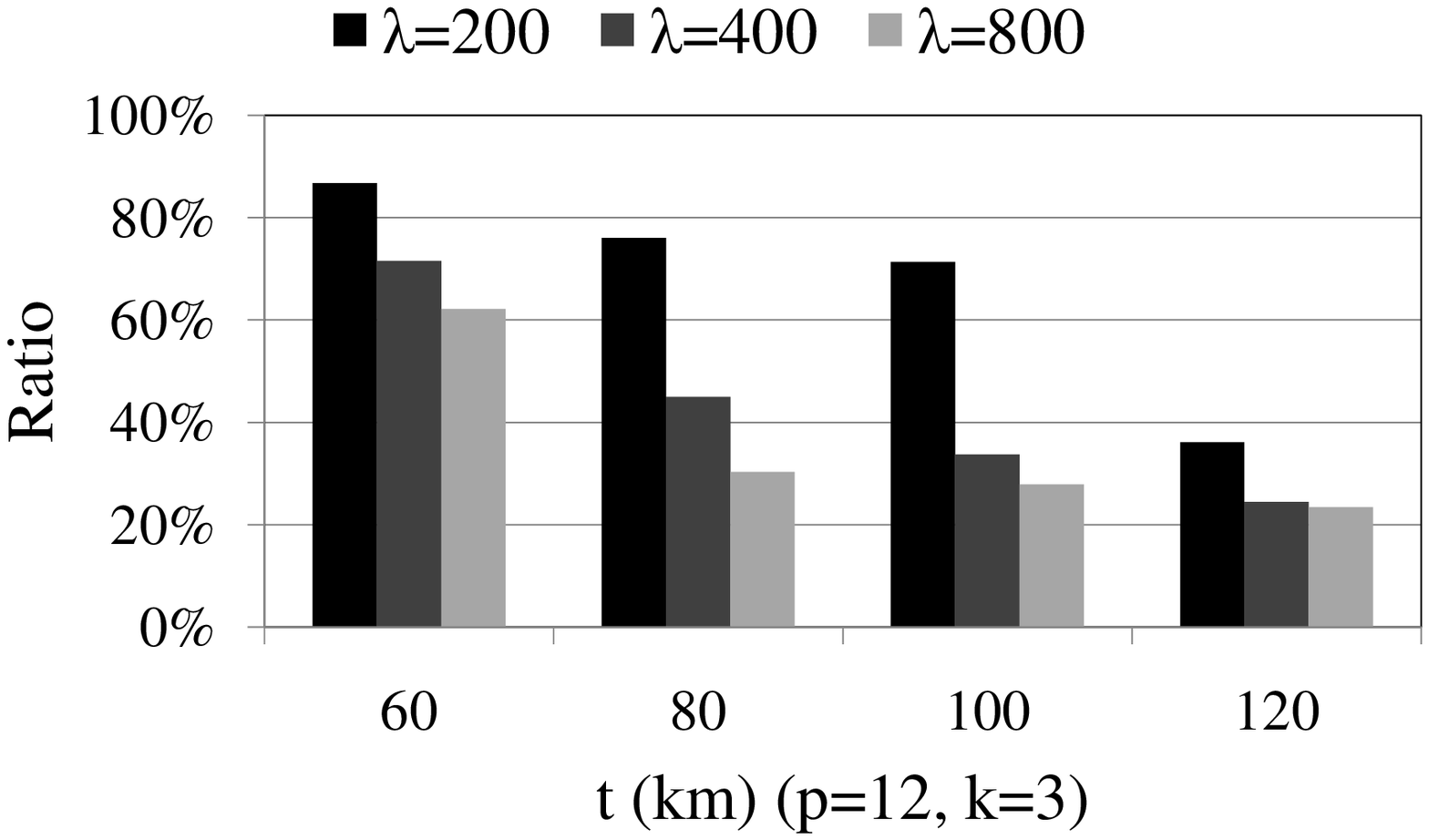} \label{FIG_Exp_Heur_Large_Sol} }
\subfigure[Execution time.] {\
\includegraphics[scale=0.2]{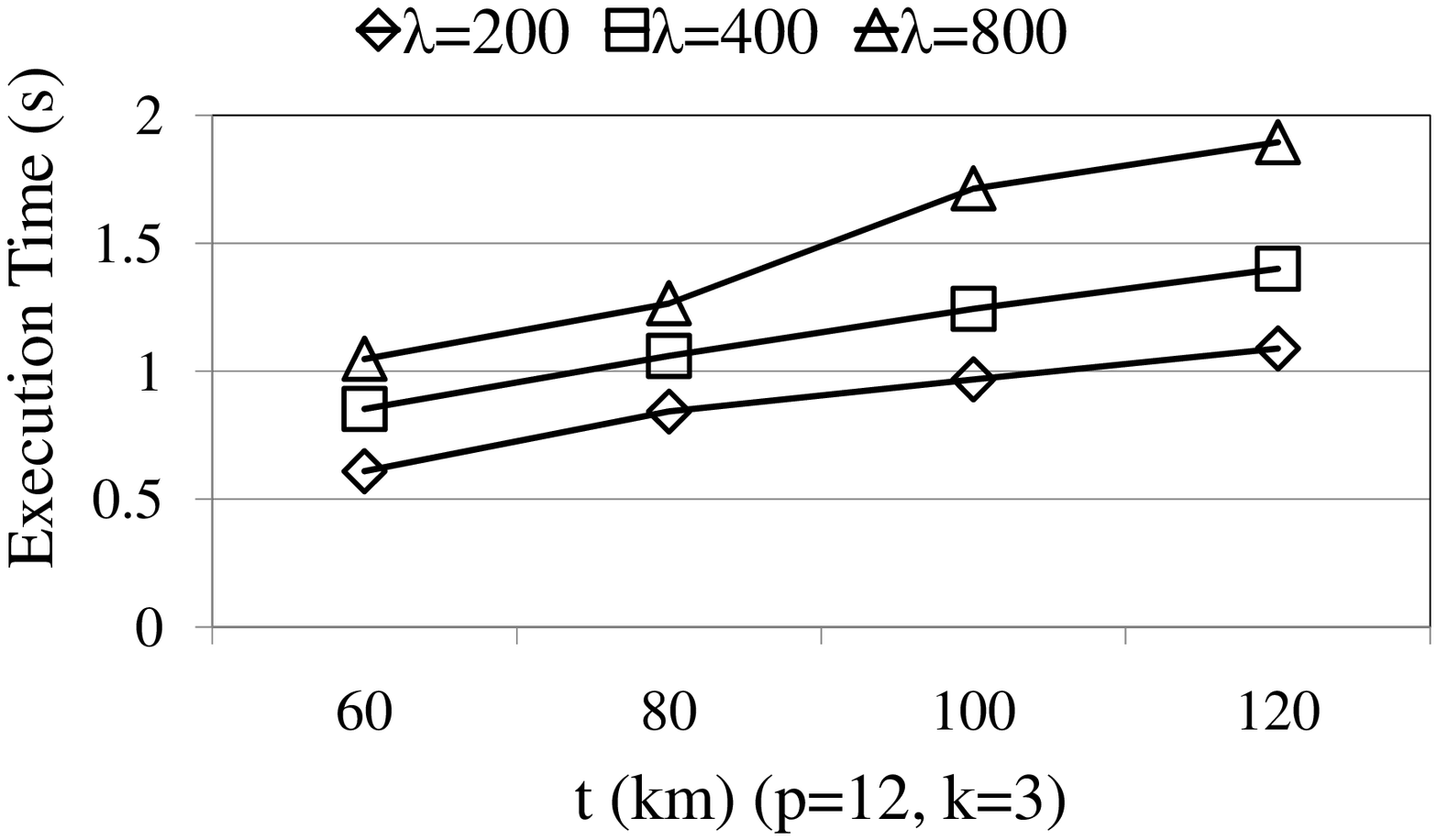} \label{FIG_Exp_Heur_Large_Time}}
\subfigure[Comparisons with OPT.] {\
\includegraphics[scale=0.2]{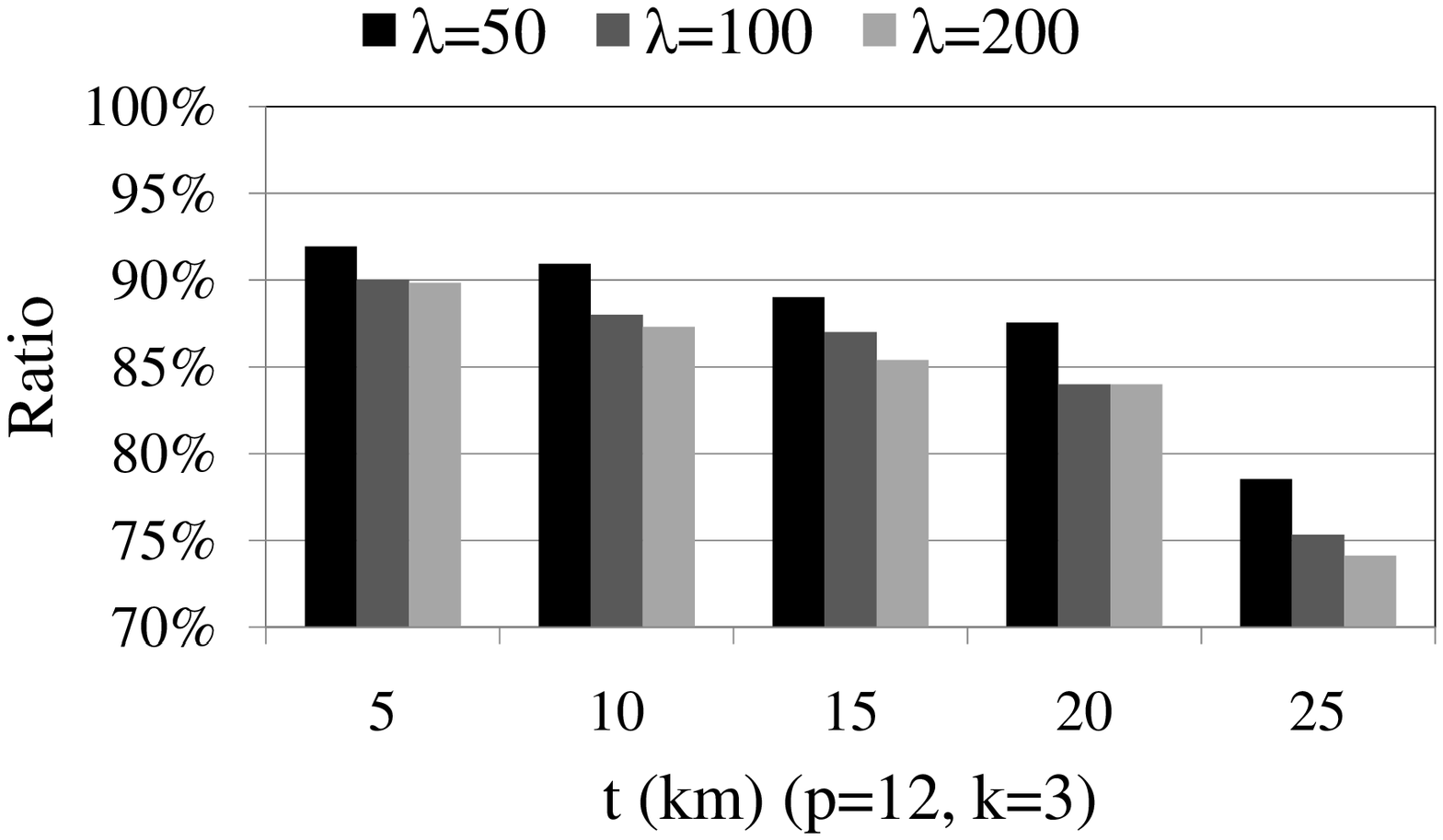} \label{FIG_Exp_Cmp_Heur_Sol}}
\subfigure[Comparisons with SSGS.] {\
\includegraphics[scale=0.2]{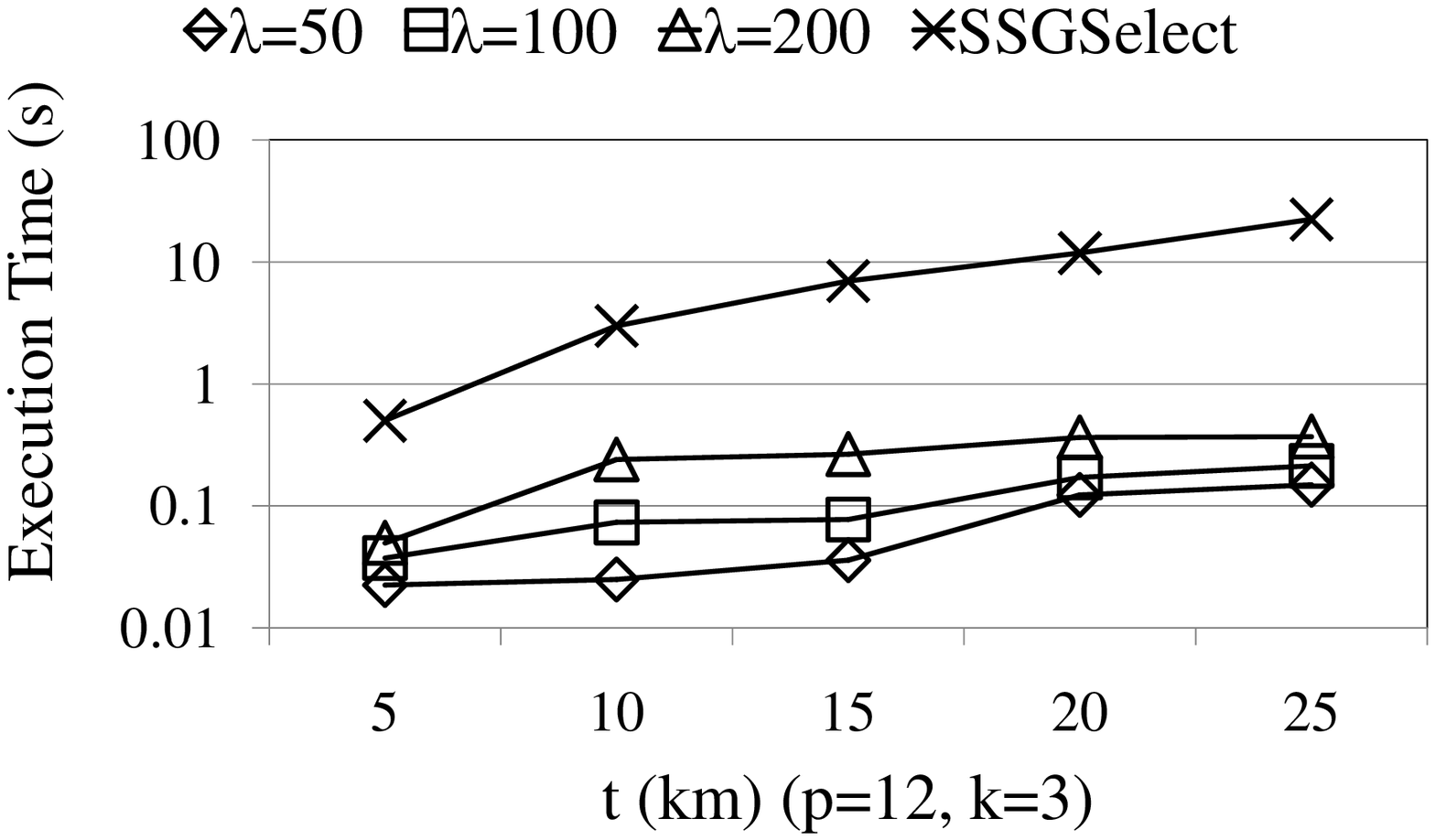} \label{FIG_Exp_Cmp_Heur_Time}}
\subfigure[Varying $w$.] {\
\includegraphics[scale=0.14]{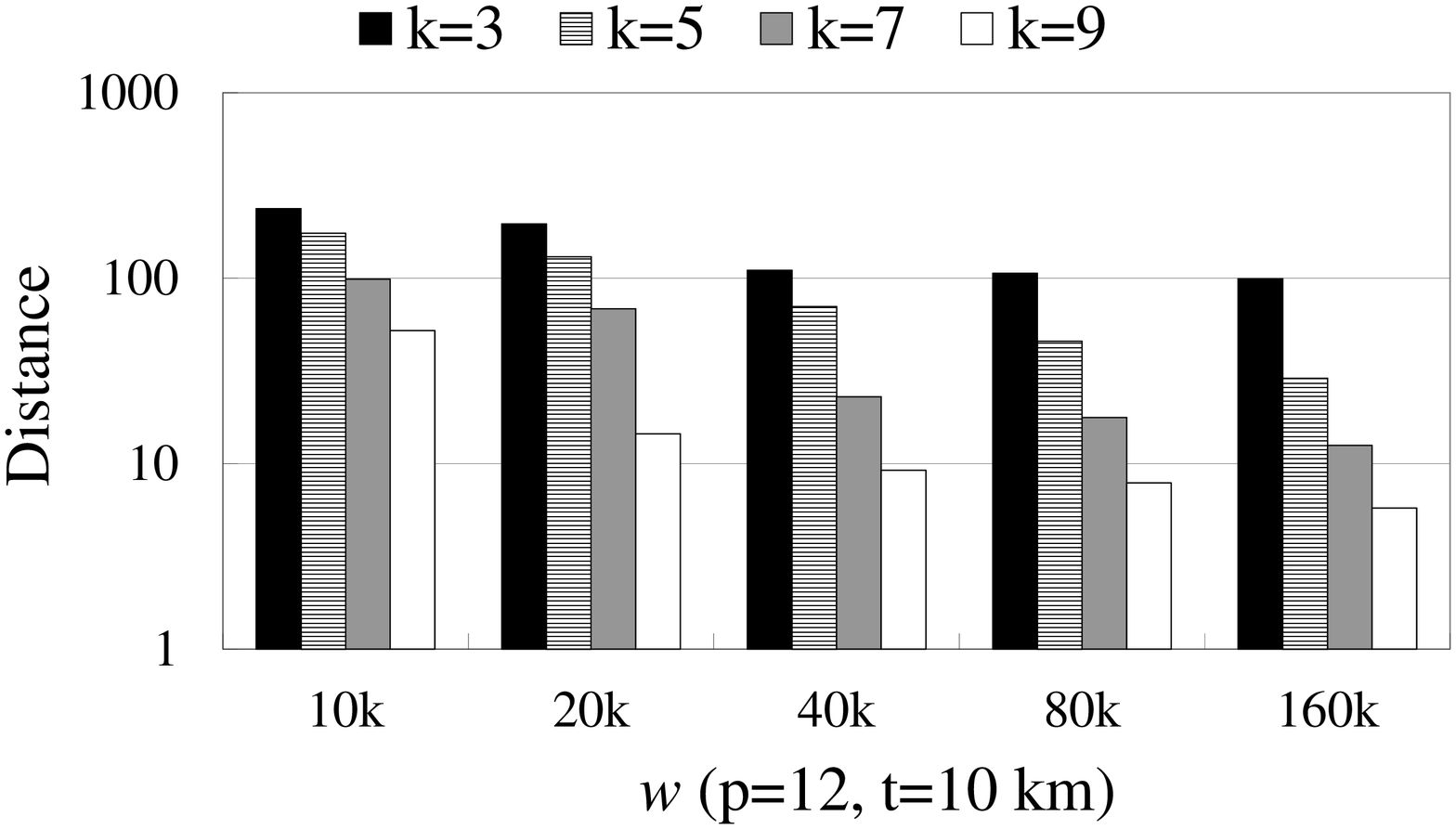} \label{FIG_Exp_Diff_T}}
\subfigure[Comparisons with the straightforward approach.] {\ 
\includegraphics[scale=0.14]{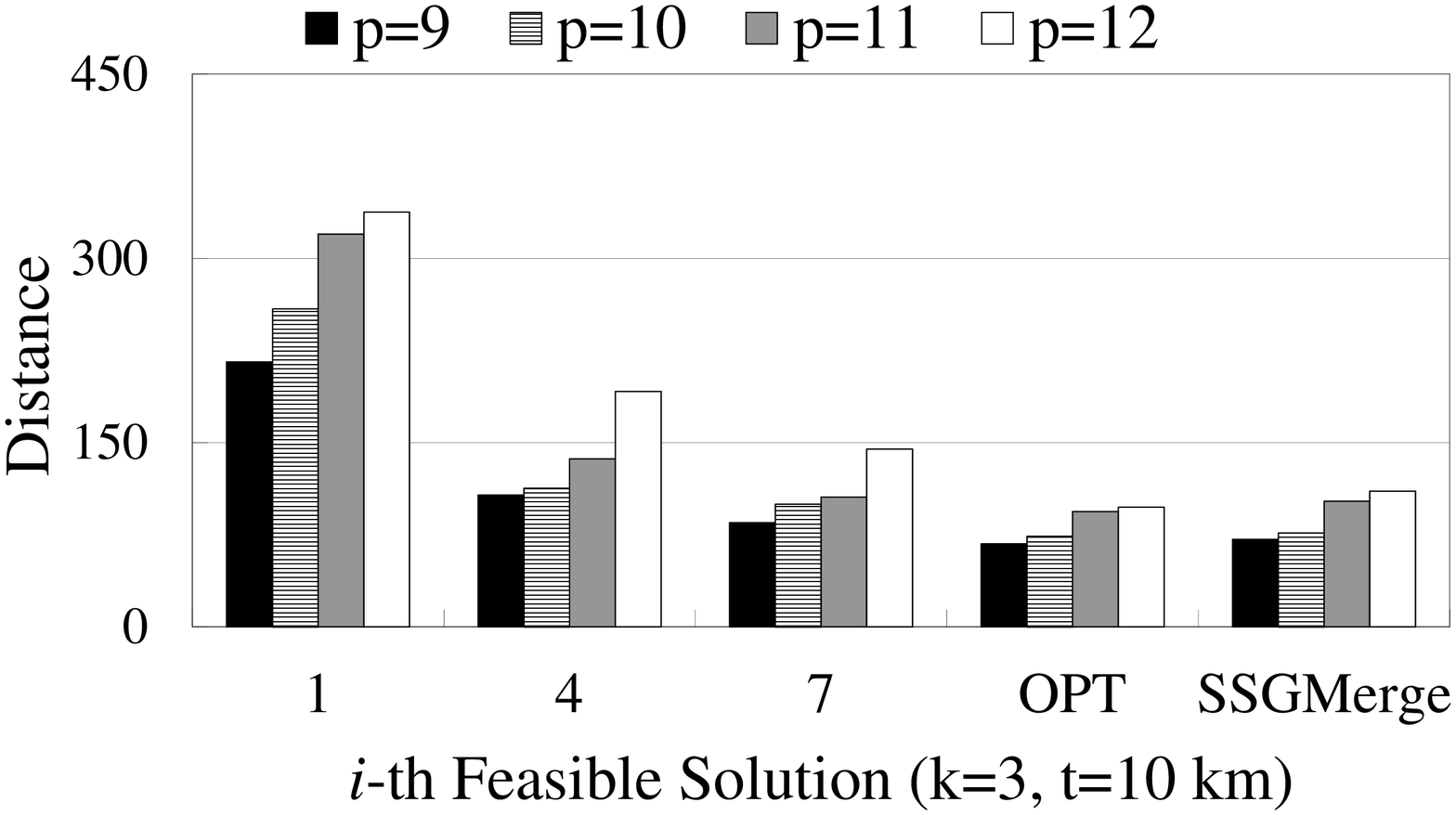} \label{FIG_Exp_Diff_I}}
\caption{Experimental results of SSGMerge.}
%\vspace{-10pt}
\label{FIG_Exp_Heur_Large}
\end{figure}

Figure \ref{FIG_Exp_Heur_Large} compares the solution quality and execution
time of SSGMerge with different settings, and the default value of $w$ is set to 20k.
In Figure \ref{FIG_Exp_Heur_Large_Sol} and \ref{FIG_Exp_Heur_Large_Time},
to compare the solution quality of SSGMerge and SSGS, we first measure the execution time of SSGMerge and then stop SSGS with the same length of time, and denote this solution as \textit{SSGTimeCut}.
%we compare SSGMerge with the straightforward approach with large $t$ settings to demonstrate the scalability of SSGMerge. 
%The straightforward approach, as mentioned in Section \ref{Heuri_Section}, simply returns
%the $i$-th feasible solution. 
%However, we avoid specifying $i$ directly since the time for SSGSelect to obtain the $i$-th feasible 
%solution may be different from the execution time of SSGMerge.
%Therefore, for fair comparisons, we run SSGSelect for the same time period as SSGMerge does 
%and regard the solution returned by SSGSelect as the baseline.}
Figure \ref{FIG_Exp_Heur_Large_Sol} shows the ratio
between SSGTimeCut and SSGMerge, i.e., the total spatial distance of solutions obtained by SSGMerge divided by the total spatial distance of solutions obtained by SSGTimeCut, with different $\lambda $. When 
$\lambda $ increases, SSGMerge can obtain better solutions since
it will examine more intermediate solutions and is more inclined to extract a better one.
Moreover, when $t$ grows, the improvement from SSGMerge becomes more significant.
This is because it becomes more difficult for SSGS to extract good feasible solutions at early stages, 
but SSGMerge can effectively merge existing intermediate
solutions to obtain good solutions. Figure \ref{FIG_Exp_Heur_Large_Time}
shows the execution time of SSGMerge. 
When $t$ grows, the execution time 
increases slowly. This is because the size of the state sets is fixed and
the extra computation of merging intermediate solutions incurred is thus limited.

Also, we compare the
solutions returned by SSGMerge with the optimal solutions returned by
SSGS in Figures \ref{FIG_Exp_Cmp_Heur_Sol} and \ref{FIG_Exp_Cmp_Heur_Time}.
Figure \ref{FIG_Exp_Cmp_Heur_Sol} displays the ratio of the optimal solution and
the solution obtained by SSGMerge, i.e., the optimal solution values divided by the solution values obtained by SSGMerge. The result manifests that the solutions obtained by SSGMerge are
close to the optimal solution. Figure \ref%
{FIG_Exp_Cmp_Heur_Time} compares the execution time of SSGMerge
and SSGS, where SSGMerge outperforms SSGS and the execution time of SSGMerge 
increases very slowly when $t$ grows.

Figure \ref{FIG_Exp_Diff_T} presents the solution quality of SSGMerge with different $w$.
Here we set $\lambda $ as 200. The total spatial distance decreases when $w$ grows, because with a larger $w$, 
SSGMerge can examine more distinct candidate attendees in different intermediate solutions, 
which enables SSGMerge to construct better solutions.

In Figure \ref{FIG_Exp_Diff_I}, 
we compare the solution quality for different $p$ of SSGMerge 
and the straightforward approach with specified $i$, 
i.e., the $i$-th feasible solution in SSGS. 
We first measure the execution time of SSGS to obtain different feasible solutions and then set 
proper $\lambda$ to stop SSGMerge with the same length of time, while \textit{OPT}
represents the optimal solution returned by SSGS. 
Figure \ref{FIG_Exp_Diff_I} shows that SSGMerge can obtain
solutions which are close to the optimal solution and outperform the straightforward approach. Moreover, 
although the solutions obtained by the straightforward approach converge quickly 
when $i$ increases, SSGMerge can still obtain much better solutions.
This is because SSGMerge each time combines a group with multiple attendees, 
which greatly reduces the time for expanding $S_{I}$ one by one. Moreover, Socio-Spatial Ordering
ensures that the early expanded groups incur small total spatial distances. Therefore, SSGMerge
can produce solutions which are close to the optimal solution.

\subsection{Comparisons of MRGQ with Relevant Works}
To compare with the state-of-the-art methods with different
parameters, we have conducted more experiments by varying $k$, $|Q|$, and $p$%
. The results are presented in Figure \ref{fig_exp_RW_new}. Figure \ref%
{fig_RW_K} compares the \textit{feasibility ratio} (i.e., the ratio of the
obtained solutions satisfying the familiarity constraint) of MAGS and the
other relevant approaches with different $k$. The proposed MAGS achieves
100\% of feasibility ratio with different $k$ because the proposed
Familiarity Pruning strategy effectively trims all the intermediate solutions
that will not satisfy the familiarity constraint at an early stage. As $k$
decreases, the feasibility ratios of gCoFQ and pNN drop because pNN does not
consider the social domain, while gCoFQ is designed to minimize a linear
combination of the social diameter and spatial distance of the group
members. It is worth noting that the feasibility ratio of pNN is low and
unacceptable even with a loose familiarity constraint (i.e., $k=6$) because
the groups returned by pNN are usually disconnected.

Figures \ref{fig_RW_Q} and \ref{fig_RW_Q_Fea} compare the solution
quality and feasibility ratio of different approaches. When the number $|Q|$
of candidate locations increases, all approaches can find the solutions with
smaller total spatial distances due to more potential good choices.
Nevertheless, the proposed MAGS outperforms gCoFQ in terms of solution quality
because gCoFQ does not examine activity locations but only tries to minimize
the maximum spatial distance between each pair of group members. Although
pNN acquires the groups with the minimum spatial distances, Figure \ref%
{fig_RW_Q_Fea} manifests that the feasibility ratio of pNN is very small
because the individuals who are closest to an activity location do not
satisfy the familiarity constraint in most cases. 

Figure \ref{fig_RW_P_Fea} presents the feasibility ratio with
different $p$. Given the familiarity constraint $k=4$, the proposed MAGS
always obtains the solutions satisfying the familiarity constraint (i.e.,
feasibility ratio is 100\%). However, when $p$ increases, the feasibility
ratios of gCoFQ and pNN drop. This is because more group members are
necessary to be connected to each other when $p$ is larger, and it is thus
more difficult for gCoFQ and pNN to follow the familiarity constraint. Moreover,
although gCoFQ is able to obtain the groups with small social diameters, a
few group members are still unacquainted with many other members. Therefore,
the feasibility ratio of gCoFQ is still not sufficient.

\begin{figure}[tp]
\centering
\subfigure[][Feasibility ratio of different $k$.] {\  \includegraphics[scale=0.15] {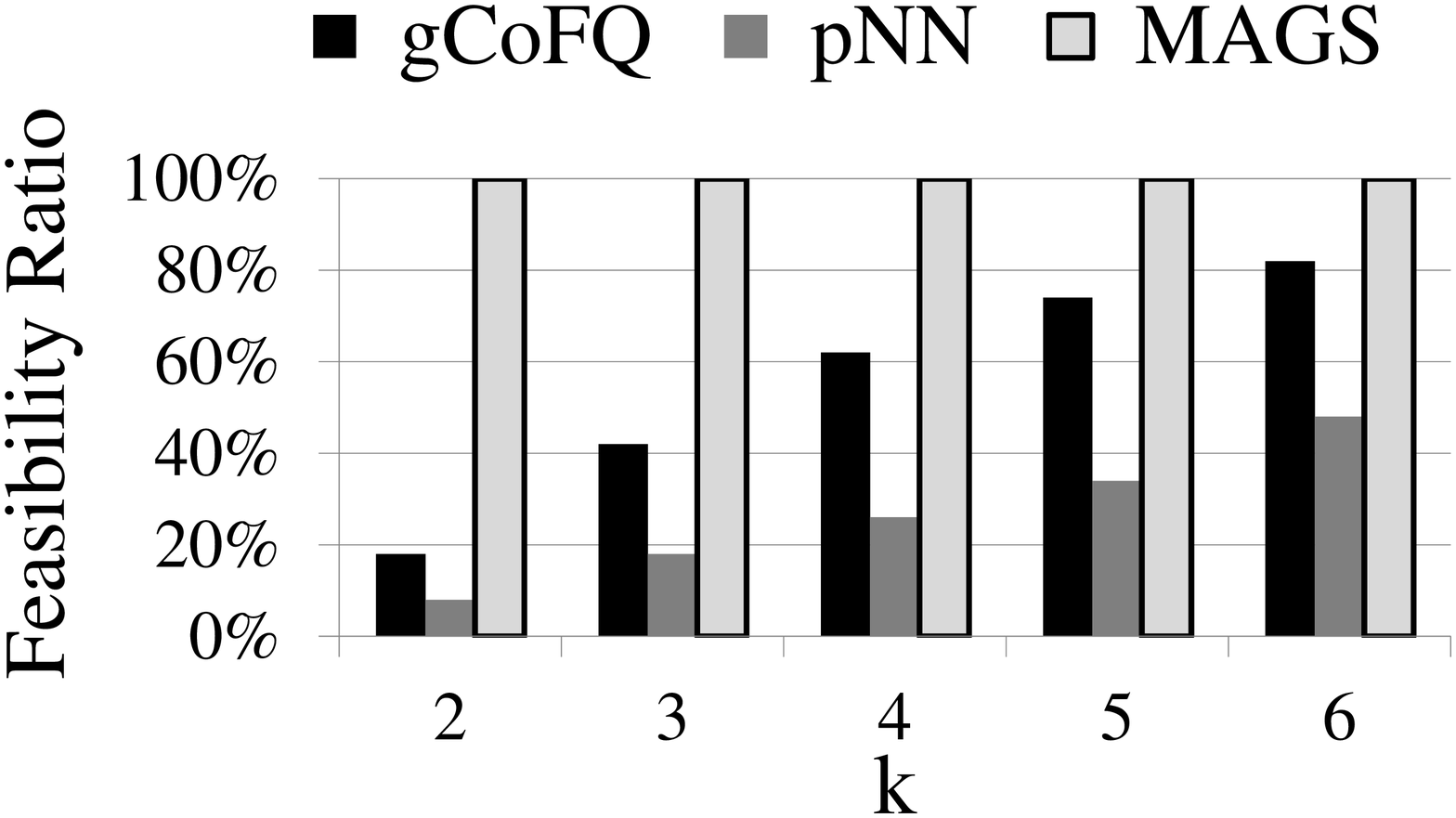}
\label{fig_RW_K} } 
\subfigure[][Solution quality of different $|Q|$.] {\  \includegraphics[scale=0.15] {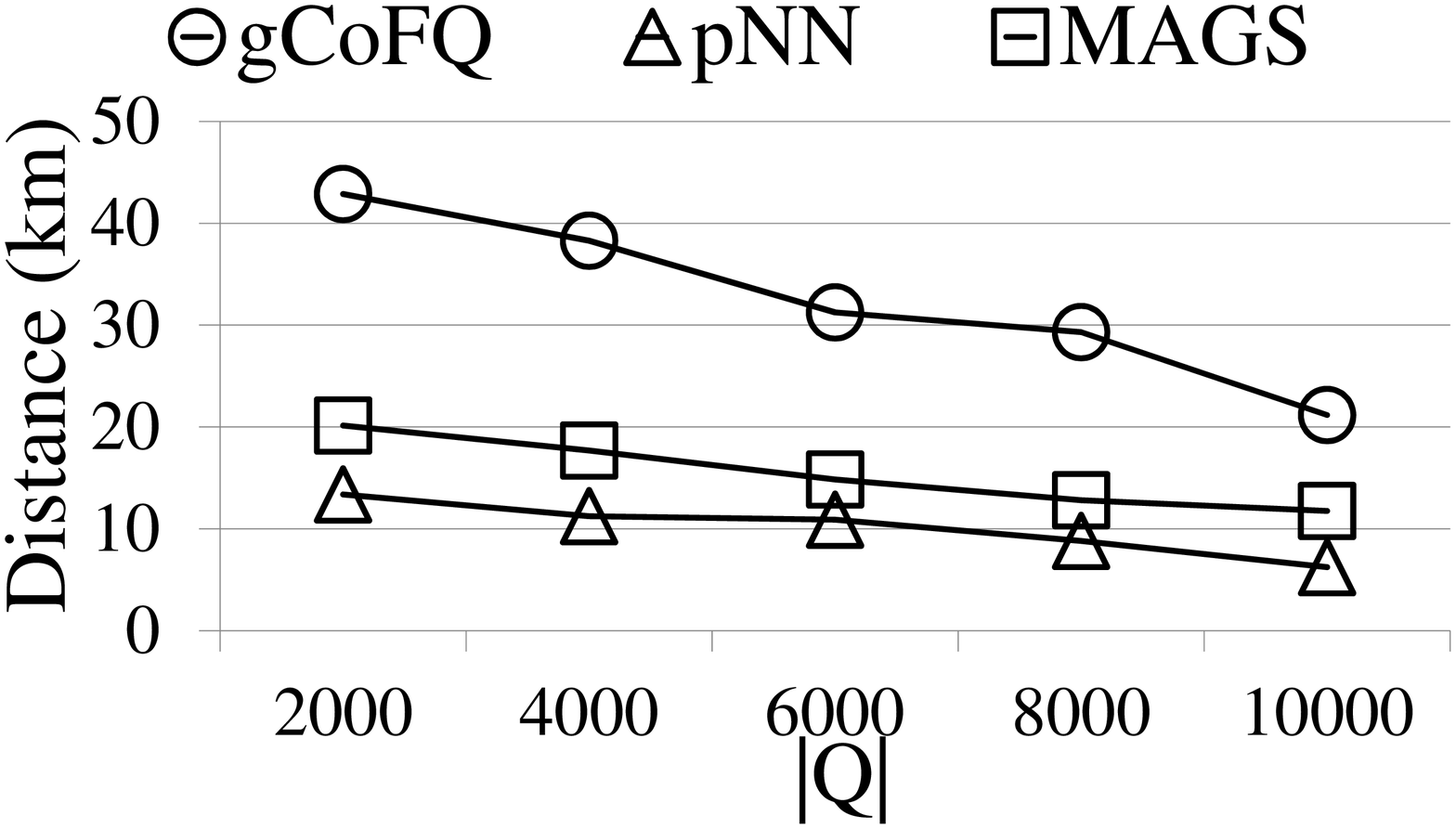}
\label{fig_RW_Q} } 
\subfigure[][Feasibility ratio of different $|Q|$.] {\  \includegraphics[scale=0.15] {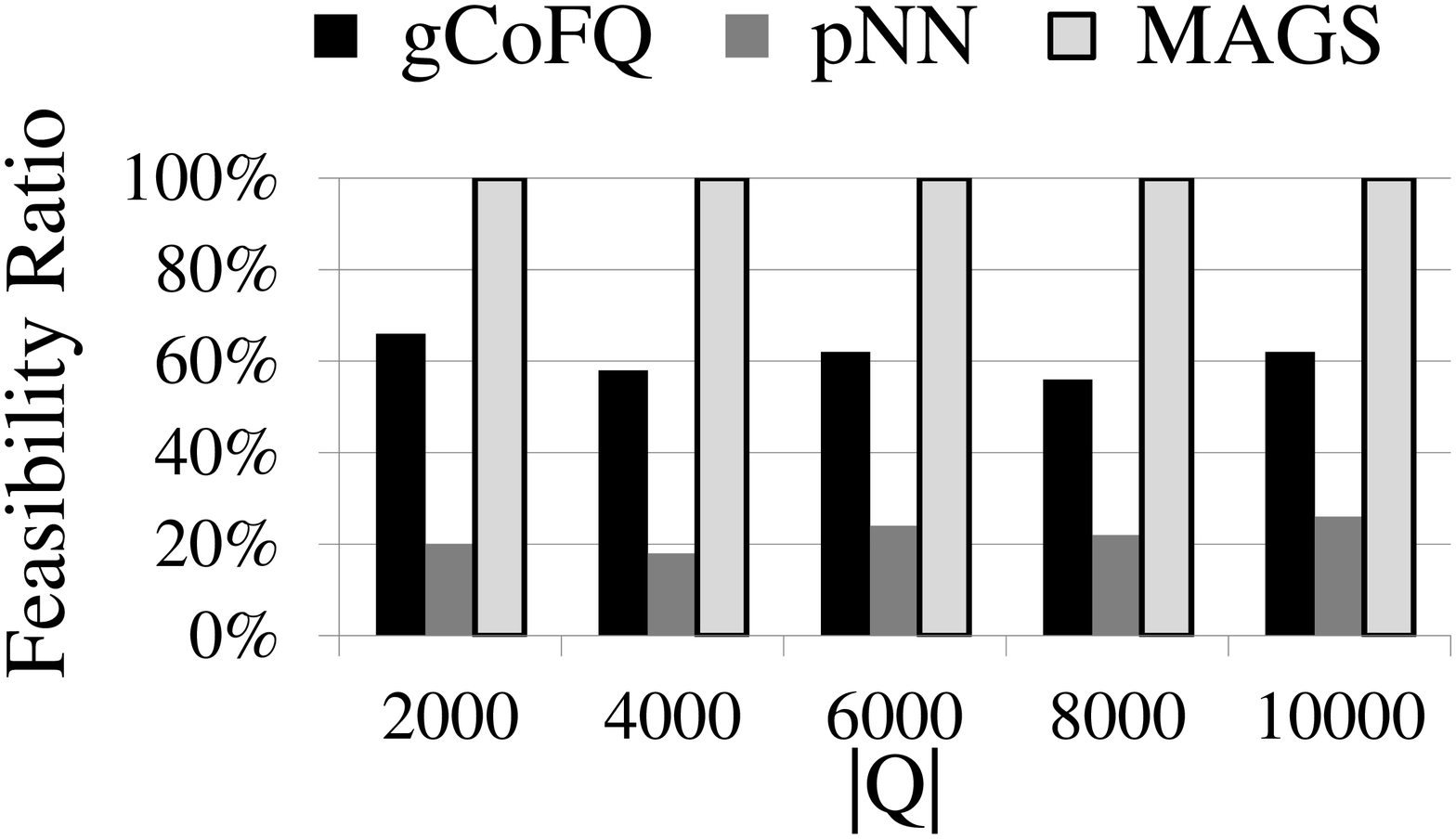}
\label{fig_RW_Q_Fea} } 
\subfigure[][Feasibility ratio of different $p$.] {\  \includegraphics[scale=0.15] {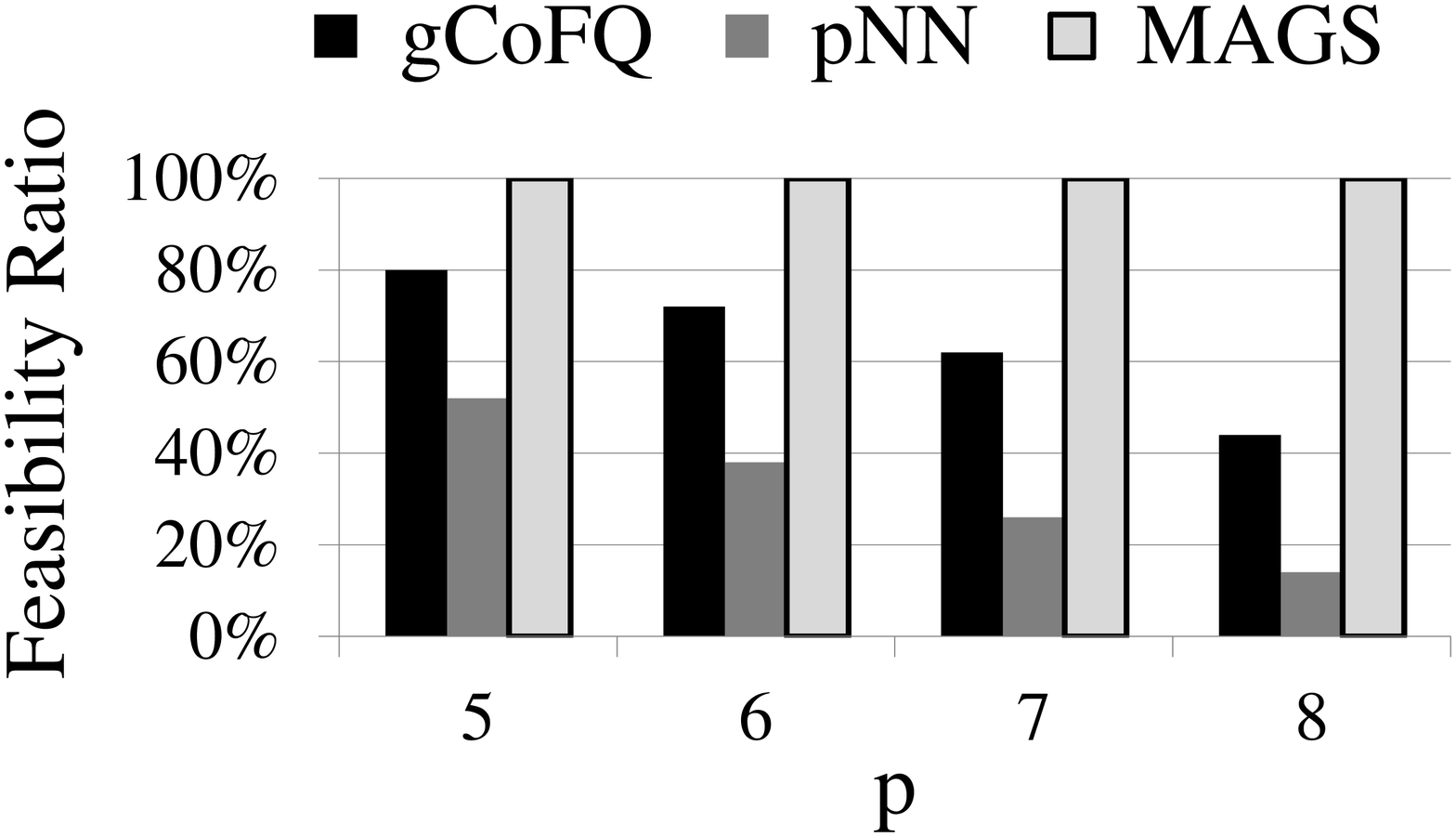}
\label{fig_RW_P_Fea} } 
\caption{Comparisons with relevant works with different parameters.}
\label{fig_exp_RW_new}
\end{figure}

\subsection{Experimental Results for MAGS in Dynamic Environments}
We perform experiments for MRGQ when the locations of the users are
dynamically changing. We generate the user trajectories according to [22] as follows. 
We first extract from $DataSet\_4SQ$ the locations visited by each individual, and each individual
is placed in one of its visited locations in equal probability with maximum speed of
0.75 km/min. The destination of each user is randomly
picked from other visited locations. During the first $1/6$ of the route, users accelerate from zero speed to 
their maximum speeds. During the middle $2/3$ of the route, they travel at their maximum speeds; and in the 
last $1/6$ of the route, they decelerate. When a user reaches her destination, a new destination
is assigned at random to her. 

We compare MRGQs with R*-Tree [24] and TPR-Tree [22] in our experiments by changing the ratio of moving users (i.e., a specific ratio of users are moving, and the rest are static), and measure the number of index updates for MRGQs in the two index structures. In addition, query execution time measures the time to process each query with the two index structures. In our experiments, we issue MRGQs randomly at 30 different time slices within 90 minutes from the start time. The query parameters are set as $t=9km$, $p=8$, $k=4$, $|Q|=10000$, and $t_f=0$. 

\begin{figure}[tp]
\centering
\subfigure[][Number of index updates.] {\  \includegraphics[scale=0.15] {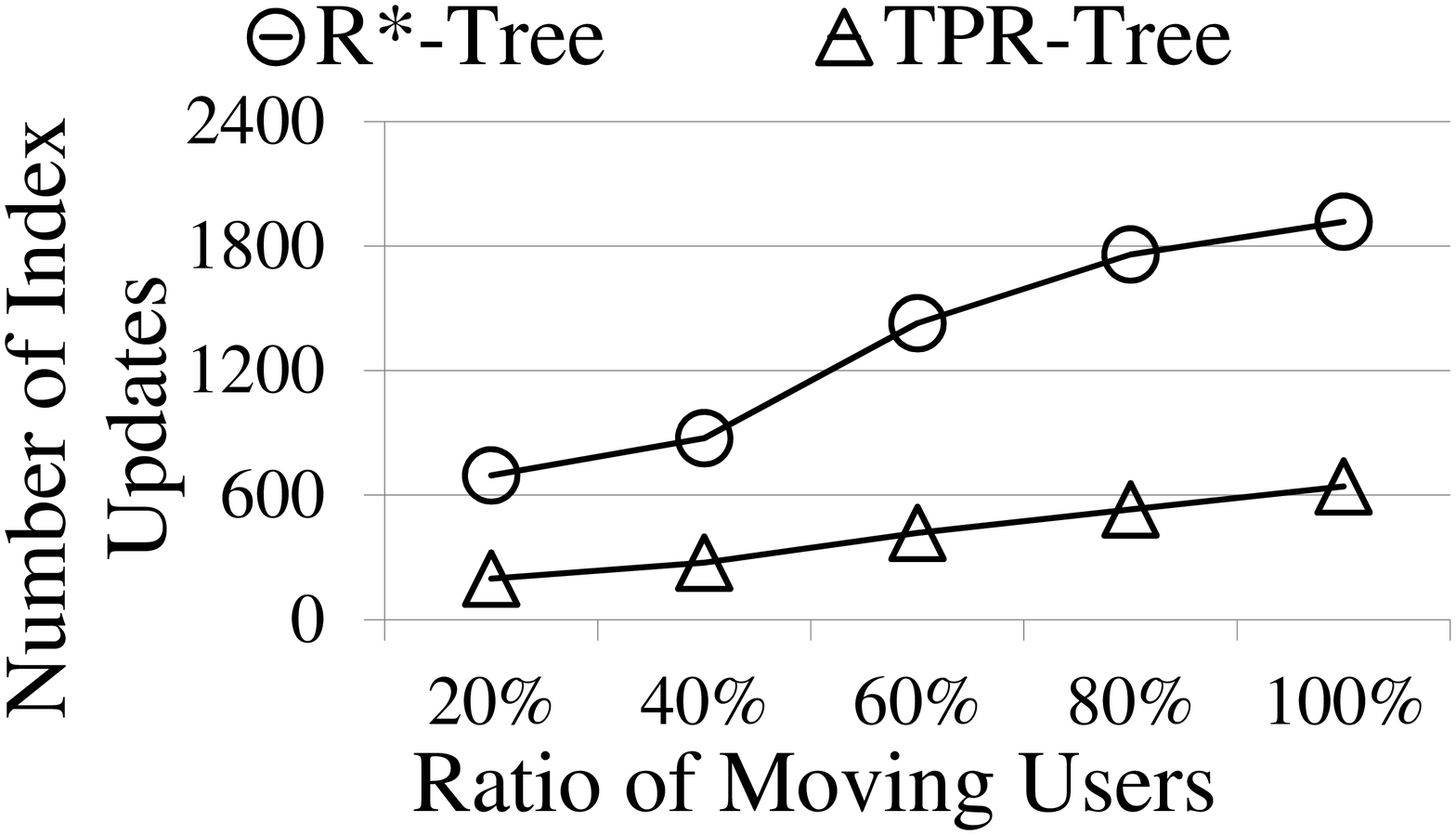}
\label{exp_indexupdate} } %\hspace{+20pt} 
\subfigure[][Query execution time.] {\  \includegraphics[scale=0.15] {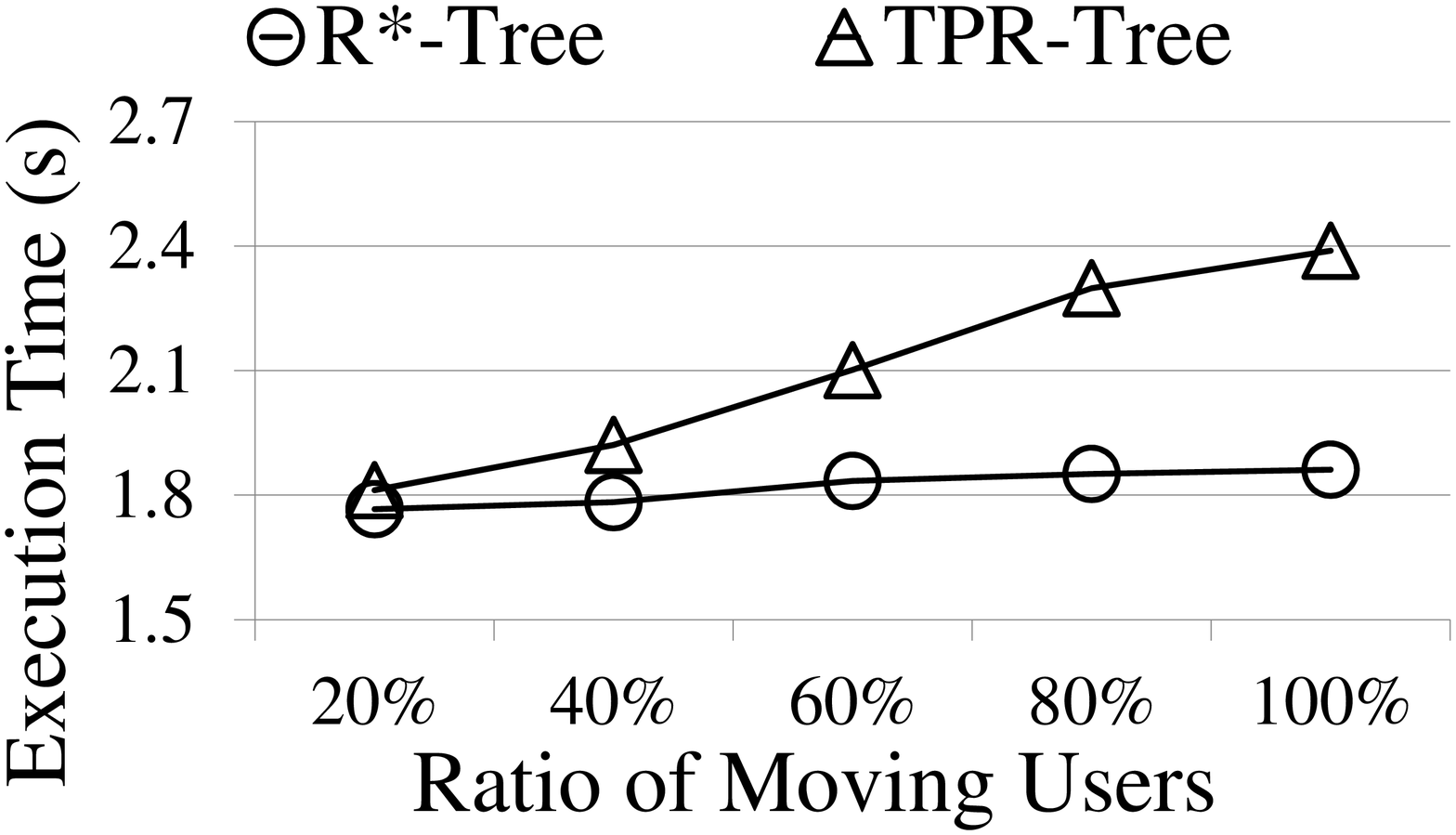}
\label{exp_exectime} }
\caption{Comparisons of R*-Tree and TPR-Tree.}
\label{TPR_exp}
\end{figure}

Figure \ref{exp_indexupdate} compares the number of index updates of R*-Tree and TPR-Tree with different ratios of moving users. As the ratio increases, the number of updates grows rapidly for R*-Tree. This is because when the number of moving users increases, more location updates occur, and R*-Tree updates the users locations, splits MBRs, and rebalances the index structure more frequently. In contrast, the number of index updates of TPR-Tree is small as compared to R*-Tree because TPR-Tree incorporates the velocity vectors for MBRs to avoid the frequent updates. Figure \ref{exp_exectime} compares the query execution time of R*-Tree and TPR-Tree with different ratios of moving users. Execution time of the MRGQs with R*-Tree and TPR-Tree are both small. The execution of MRGQs in R*-Tree is smaller than that in TPR-Tree because R*-Tree spends much time on updating the index structure and maintains smaller MBRs. The smaller MBRs in R*-Tree provide better index capability. In contrast, the execution time of TPR-Tree increases when the ratio of the number of moving users grows because in this case, the MBRs are larger and the index capability deteriorates. 

\subsection{Performance Comparisons for MAGS with Different Pruning Strategies}
We compare SRDO, APDONoOTDP, APDONoITDP, APDONoALDP, and APDO by changing parameters $k$
and $p$. Figure \ref{fig_exp_DP} presents the experimental results with the
default parameters $t=9$km, $p=8$, $k=4$, and $|Q|=10000$. When $k$ or $p$
increases, the computation time of APDO increases slowly. As shown, APDO
outperforms the other approaches, i.e., APDONoALDP, APDONoOTDP, APDONoITDP,
and SRDO because the proposed distance pruning strategies effectively avoid
redundant activity location examinations. When $k$ increases, many
approaches except APDO incur more computation time. This is because when $k$
becomes larger, Familiarity Pruning is less effective due to the loosen
familiarity constraint. Nevertheless, the proposed APDO with all the
distance pruning strategies is able to avoid redundant $S_{I}$ expansions.
Similarly, when $p$ increases, the search space also increases. APDO
with the pruning strategies can thus stop expanding $S_{I}$ earlier. Therefore,
APDO outperforms the other approaches. Finally, APDO explores many possible $%
q_{ref}$ and $v_{c}$ during the expansion of $S_{I}$ to
obtain good solutions much earlier than SRDO does. In other words, the distance
pruning strategies in APDO is very effective, which enables the proposed APDO to
significantly outperform SRDO.

\begin{figure}[tp]
\centering
\subfigure[][Distance pruning with different $k$ ($p=8$).] {\  \includegraphics[scale=0.15] {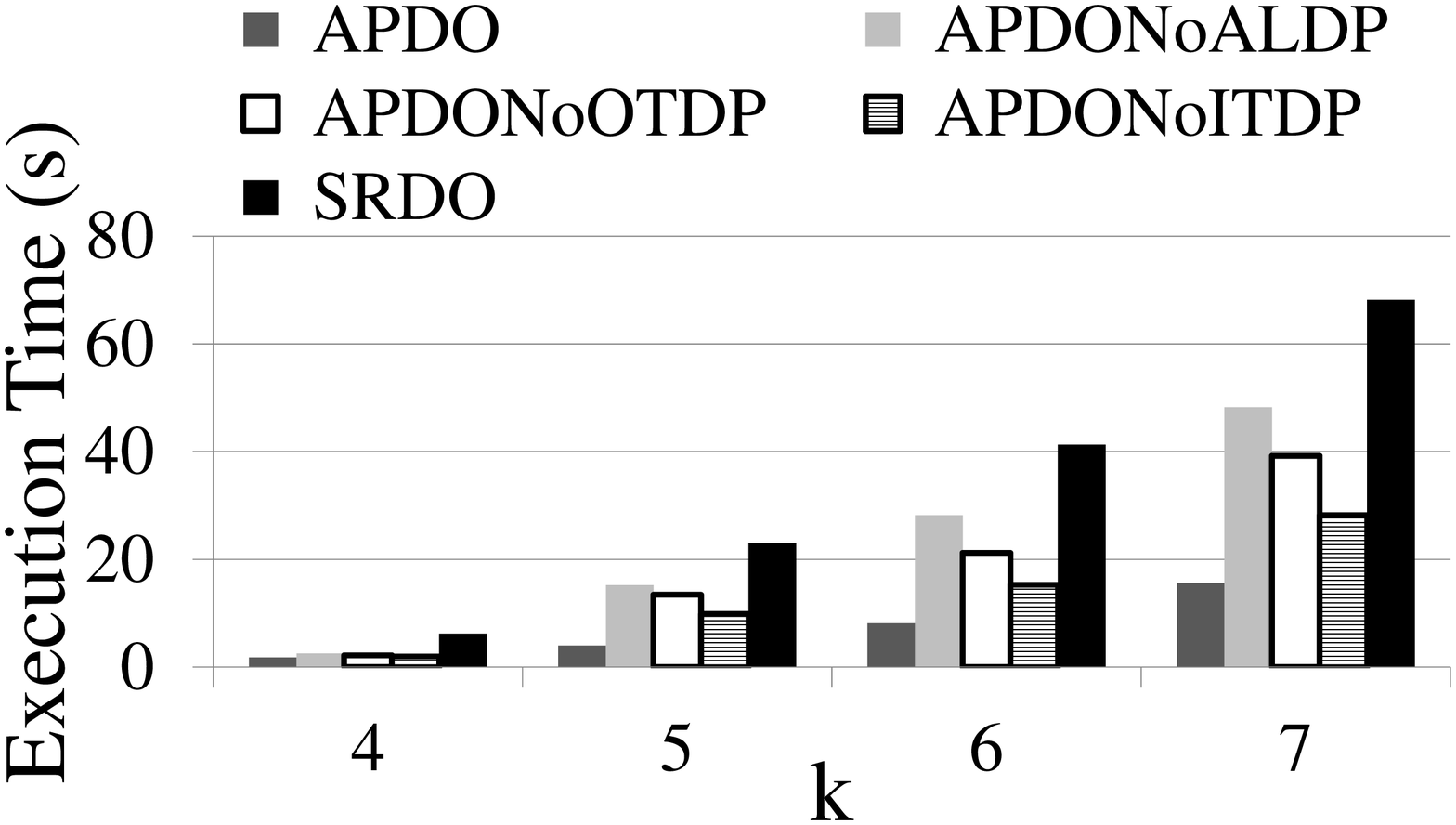}
\label{fig_DP_K} } 
%\hspace{+20pt}
\subfigure[][Distance pruning with different $p$ ($k=4$).] {\  \includegraphics[scale=0.15] {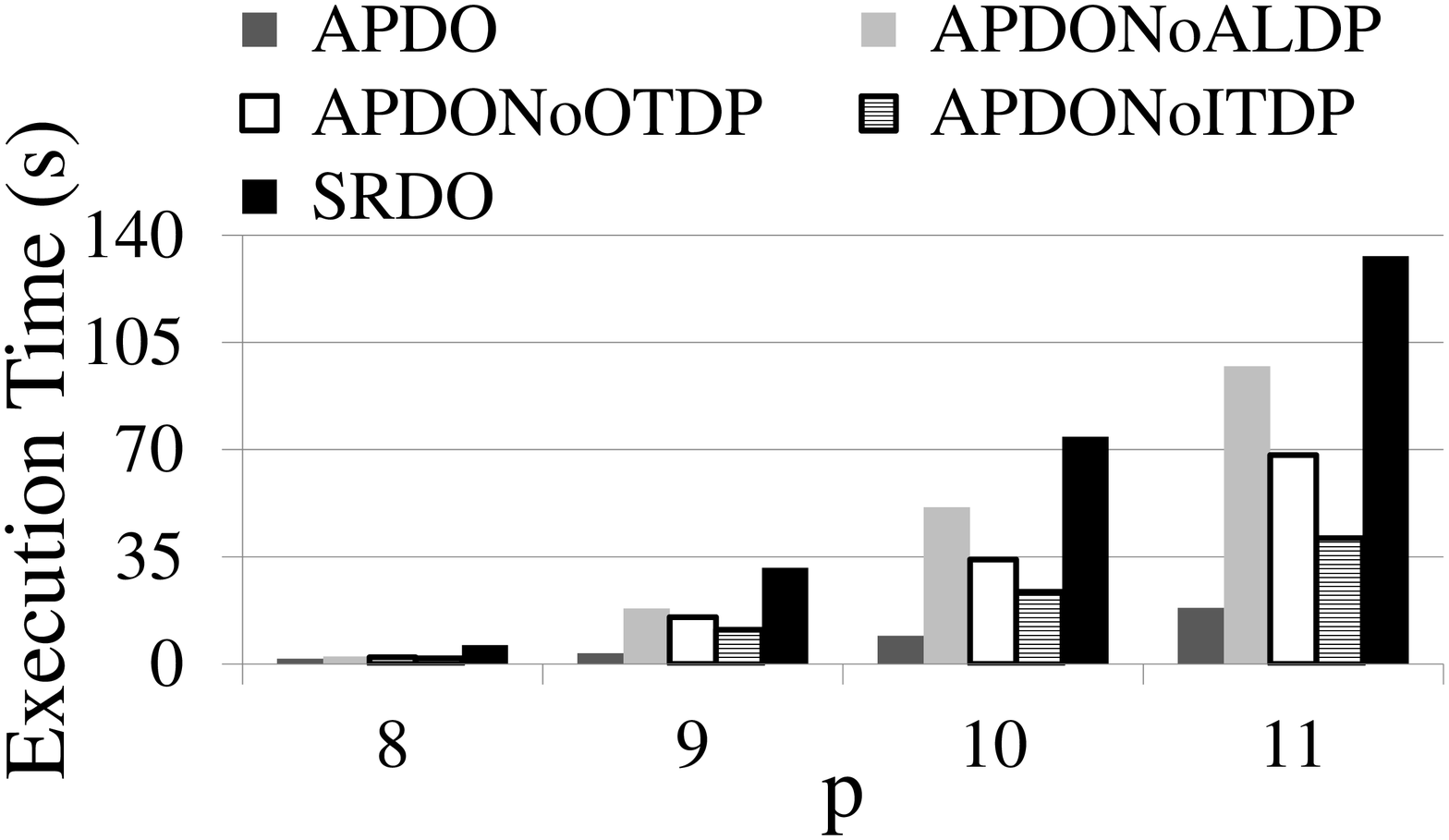}
\label{fig_DP_P} } 
\caption{Effectiveness of distance pruning strategies.}
\label{fig_exp_DP}
\end{figure}

\subsection{Experimental Results for MAGS on \textit{DataSet\_Youtube}}

We evaluate the proposed MAGS on \textit{DataSet\_Youtube}, which
is a social network extracted from Youtube video-sharing website with 1,134,890 individuals. Since 
there is no spatial information for this dataset, we randomly assign the spatial coordinates
to the individuals as her current location. In our experiments, we randomly extract the activity locations from 
\textit{DataSet\_4SQ}. In the following experiments, unless specifically indicated, we set $k=4$, $p=8$, $|Q|=10,000$, and the maximum value of $t$ is 15 km.

\begin{figure}[tbp]
\centering
\subfigure[Different algorithms.] {\
\includegraphics[scale=0.14]{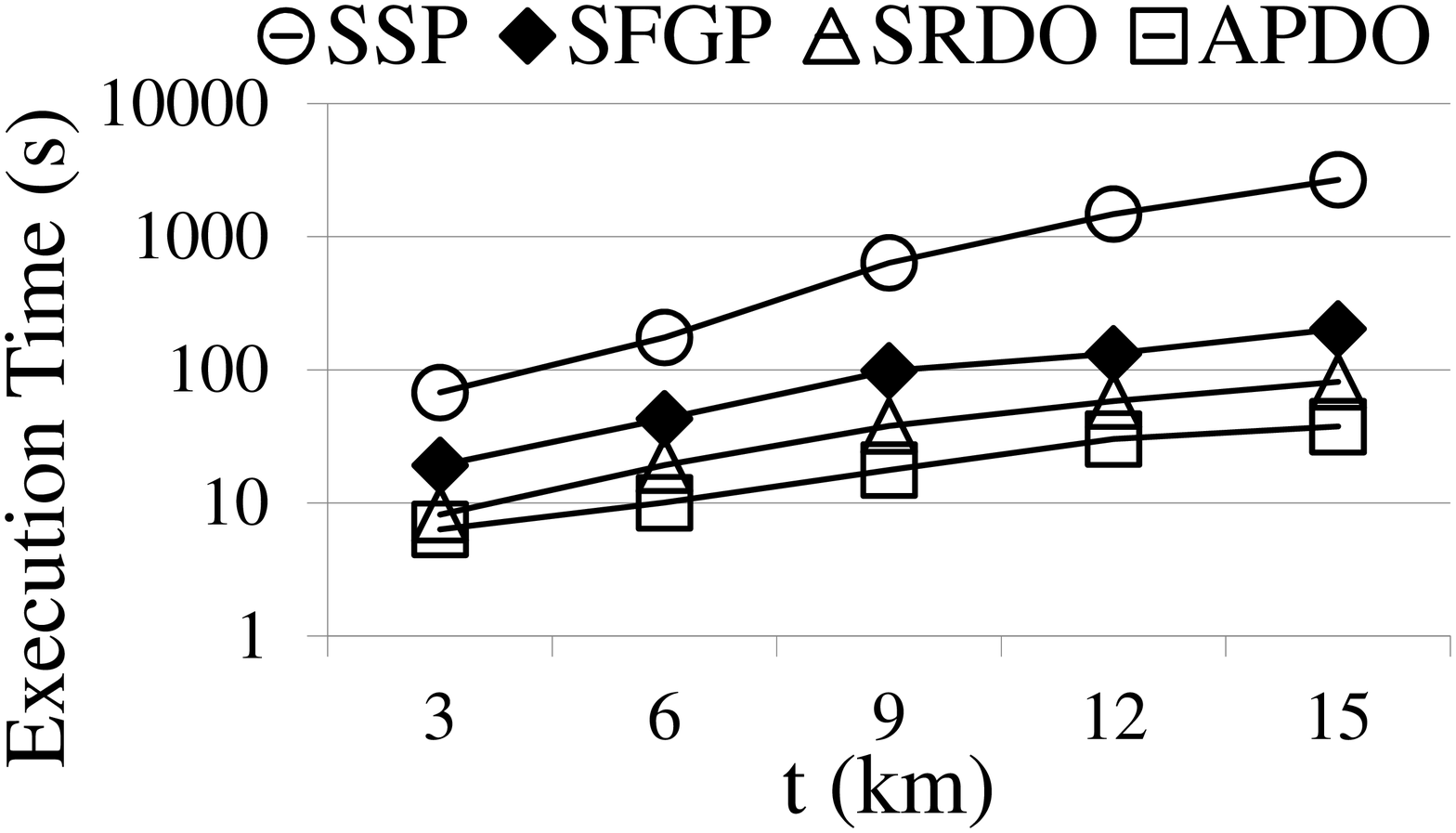} \label{FIG_YT_Algo}}
\subfigure[Varying $|Q|$.] {\
\includegraphics[scale=0.14]{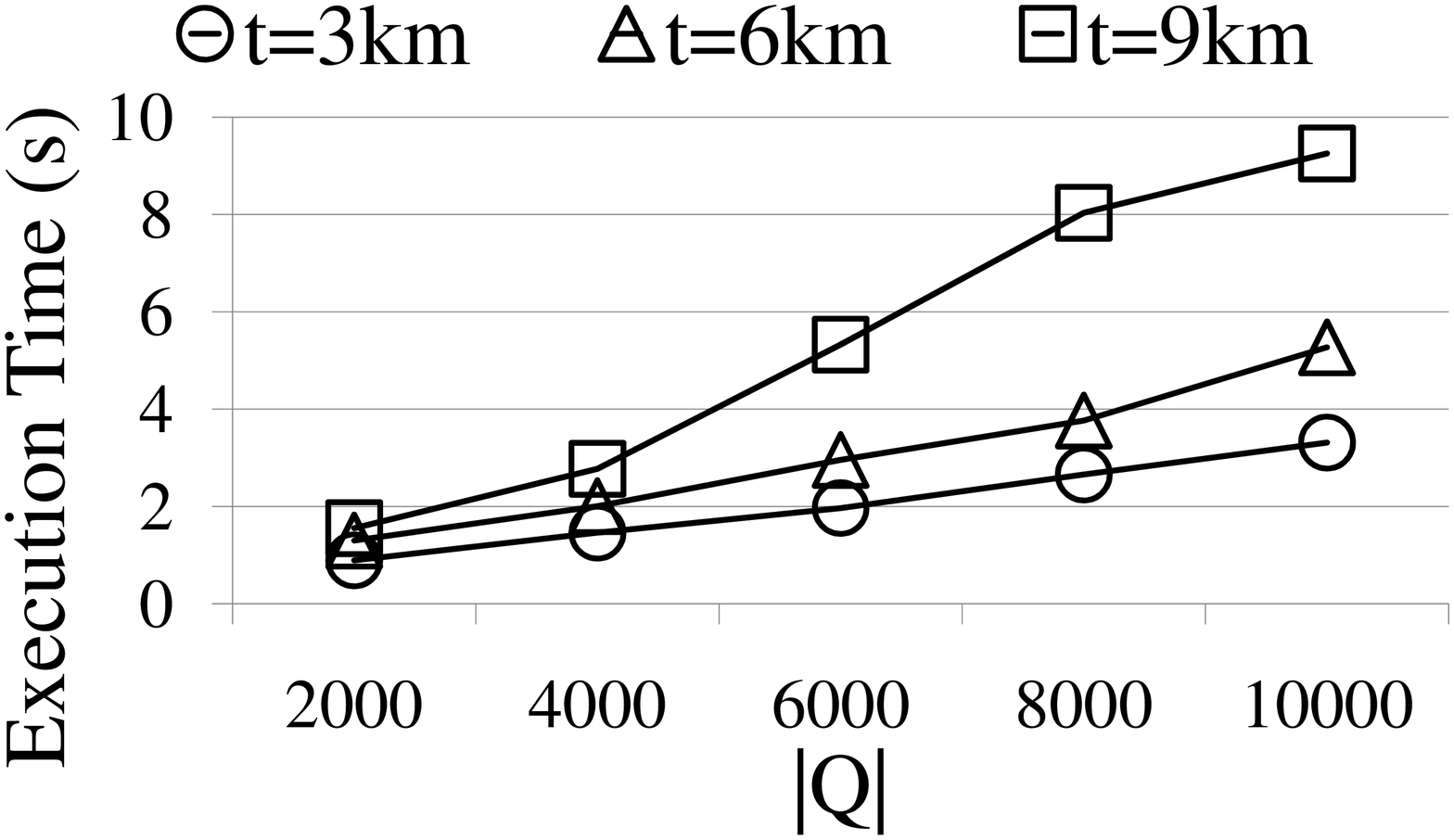} \label{FIG_YT_Q}} 
\subfigure[Varying $k$.] {\
\includegraphics[scale=0.14]{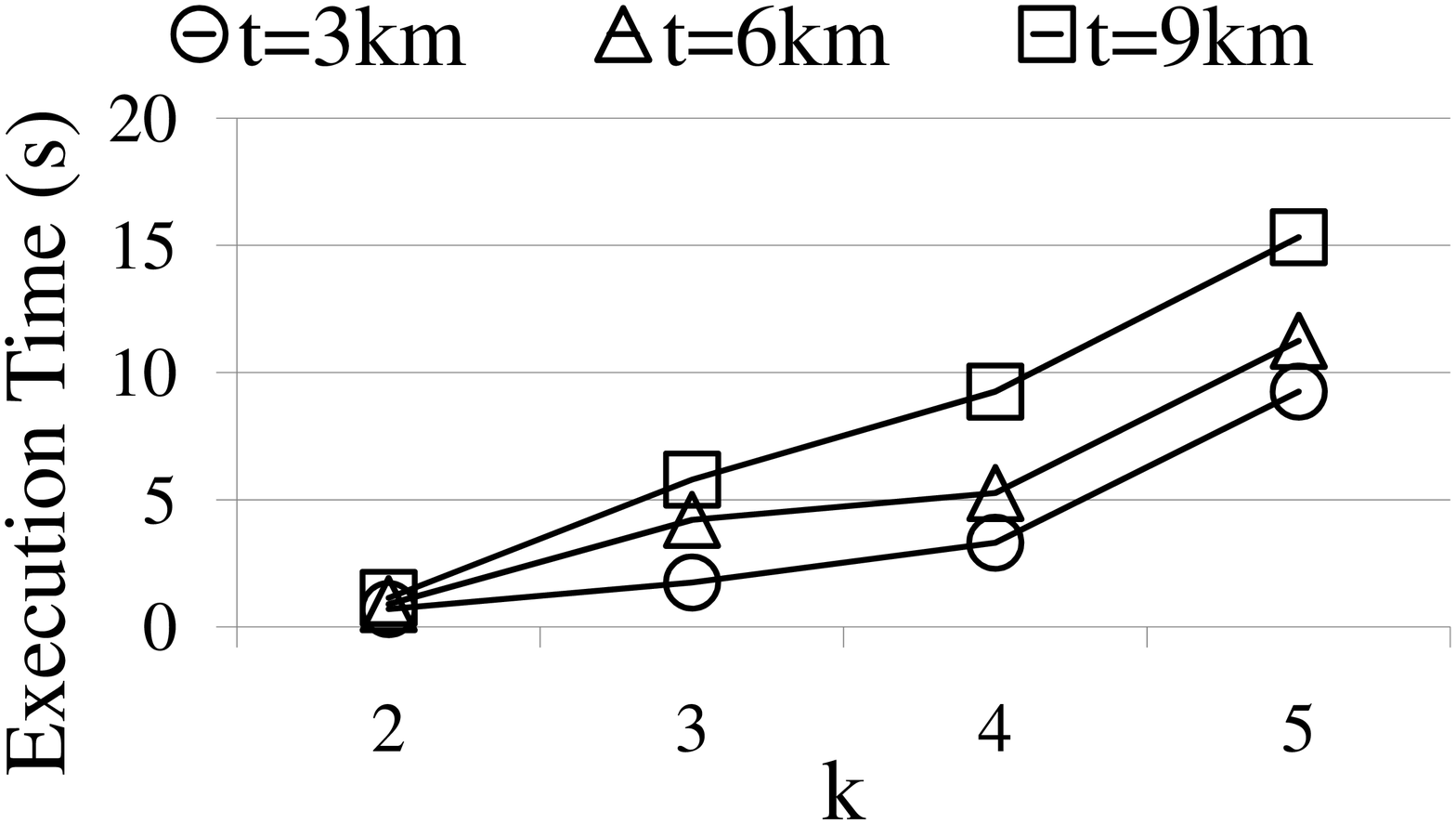} \label{FIG_YT_K}}
\subfigure[Varying $p$.] {\
\includegraphics[scale=0.14]{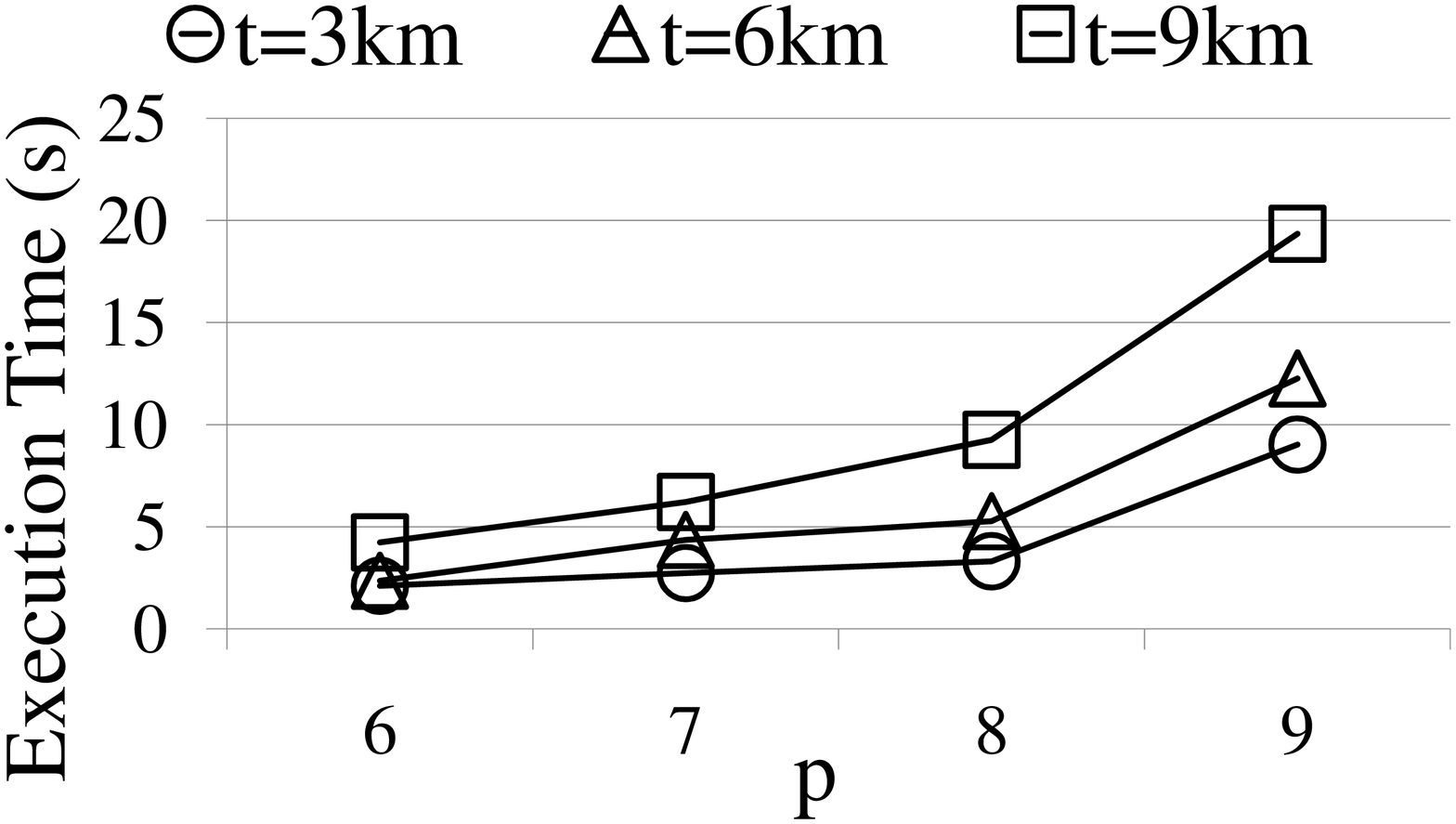} \label{FIG_YT_P}}
\caption{Sensitivity tests of MAGS on \textit{DataSet\_Youtube}.}
%\vspace{-20pt}
\label{FIG_YT}
\end{figure}

Figure \ref{FIG_YT} evaluates the proposed algorithms on \textit{DataSet\_Youtube}.
Figure \ref{FIG_YT_Algo} shows that, although \textit{DataSet\_Youtube} contains about 8 times the number of candidates as \textit{DataSet\_4SQ} does, the computation time of SRDO and APDO both
incur small computation time. This is because \textit{DataSet\_Youtube} is socially sparse (i.e., with an average degree 5.27), which enables the social pruning strategies to effectively remove redundant search space.
Figures \ref{FIG_YT}(b)-(d) compare APDO with different parameter settings. Since \textit{DataSet\_Youtube}
contains fewer spatially dense clusters, the computation time is limited with the distance pruning strategies. 
Moreover, Familiarity Pruning is able to quickly prune unqualified groups due to the large number of
low-degree nodes in the social graph, and the computation time is thus reduced.

\vspace{-10pt}
\section{Conclusion}

\label{Conclu} 
\baselineskip=11pt
To address the need of automatic
activity planning based on the social and spatial relationships of
attendees and activity locations, we define a new query, namely MRGQ, to jointly find the optimal set of attendees and the best activity location among multiple activity locations.
We also study a special case of MRGQ, namely SSGQ, which only features a single activity location. 
We show that processing MRGQ is NP-hard and inapproximable within any factor. We formulate MRGQ with Integer Linear Programming and propose an efficient algorithm, namely MAGS. In addition to indexing the candidate attendees in R-Tree, we propose to index the candidate locations in BallTree, and devise various ordering and pruning strategies based on the social and spatial relationships.
Experimental results show that the computation time required by single threaded MAGS is much smaller than using an IBM CPLEX parallel optimizer.
Moreover, we show that the problem of processing SSGQ is NP-hard and devise an efficient algorithm, namely SSGS, to process SSGQ. Various
strategies, including Distance Ordering, Socio-Spatial Ordering, Distance
Pruning, and Familiarity Pruning are proposed to prune redundant search
space and obtain the optimal solution efficiently.  
We also implement SSGQ in Facebook and conduct user studies for both SSGQ and MRGQ to demonstrate that the proposed algorithms significantly outperform manual coordination in terms of both solution quality and efficiency.

\appendices
\section{Pseudo Codes of the Proposed Algorithms}
The pseudo codes of the proposed algorithms, i.e., SSGS, SFGP, MAGS with Single-Reference Distance Ordering, and MAGS with All-Pair Distance Ordering, are presented as follows.

\begin{algorithm*}
\caption{SSGS algorithm}
\begin{algorithmic}[1]
\Require Graph $G=(V,E)$, location $l_v$ for each $v\in V$, the number of attendees $p$, activity location $q$, familiarity constraint $k$, and spatial radius $t$. The user locations $l_v,\forall v\in V$ are indexed by an R-Tree.
\Ensure Optimal group $F$.
\State $S_I\gets \varnothing$, $curDist\gets 0$, $F\gets \varnothing$, $D\gets \infty$, $\theta \gets k$
\State Employ R-Tree Range Query on $q$ to find the vertices within distance $t$ as $S_R$
\State \Call{FindGroup}{$S_I$,$S_R$,$curDist$}
\If{$D\neq \infty$}
  \State \textbf{output} $F$
\Else
  \State \textbf{output} "No Answer"
\EndIf

\Procedure{FindGroup}{$inS_I$,$inS_R$,$curDist$}
\State $S_I \gets inS_I$, $S_R\gets inS_R$
\While{$|S_I|+|S_R|\geq p$}
  \If{there is any unvisited vertex in $S_R$}
    \State \parbox[t]{\dimexpr\linewidth-\algorithmicindent}{Employ R-Tree distance browsing to extract from $S_R$ the next unvisited vertex $u$ \\ which has the minimum spatial distance to $q$; mark $u$ as visited\strut}
%    \If{$d_{u,q}>t$} \Comment{}
%      \State \textbf{break}
%    \EndIf
  \ElsIf{$\theta<p-1$}
    \State increase $\theta$ and mark the remaining vertices in $S_R$ as unvisited
  \Else
    \State \textbf{break}
  \EndIf

\If{$u$ satisfies the condition of Socio-Spatial Ordering in Eq. (1)}
  \State $S_I\gets S_I+\{u\}$, $S_R\gets S_R-\{u\}$, $curDist\gets curDist+d_{u,q}$
  \If{Familiarity Pruning in Eq. (2) or Distance Pruning in Eq. (3) is satisfied}
    \State \textbf{break}
  \ElsIf{$|S_I|<p$}
    \State \Call{FindGroup}{$S_I$,$S_R$,$curDist$}
  \Else
    \State $D\gets curDist$, $F\gets S_I$
    \State \textbf{break}
  \EndIf
\ElsIf{$\theta=p-1$}
  \State $S_R\gets S_R-\{u\}$
\EndIf

\EndWhile
\EndProcedure
\end{algorithmic}
\end{algorithm*}

\begin{algorithm*}
\caption{SFGP algorithm}
\begin{algorithmic}[1]
\Require Graph $G=(V,E)$, location $l_v$ for each $v\in V$, the number of attendees $p$, activity locations $Q$, familiarity constraint $k$, and spatial radius $t$. The user locations $l_v,\forall v\in V$ are indexed by an R-Tree.
\Ensure Optimal group $F$ and the corresponding activity location $q^{*}$
\State $S_I\gets \varnothing$, $S_R\gets V$, $Q_I\gets Q$, $F\gets \varnothing$, $D\gets \infty$, $\theta\gets k$
\State \parbox[t]{\dimexpr\linewidth-\algorithmicindent}{find $u\in S_R$ and $q_{ref}\in Q_I$ such that $u$ and $q_{ref}$ are the spatially closest pair\strut}
\State \Call{FindGroupAndLoc\_SFGP}{$S_I$,$S_R$,$Q_I$,$q_{ref}$}
\If{$D\neq \infty$}
  \State \textbf{output} $\langle F,q^{*}\rangle$
\Else
  \State \textbf{output} "No Answer"
\EndIf

\Procedure{FindGroupAndLoc\_SFGP}{$inS_I$,$inS_R$,$inQ_I$, $q_{ref}$}
\State $S_I \gets inS_I$, $S_R\gets inS_R$, $Q_I\gets inQ_I$
\While{$|S_I|+|S_R|\geq p$}
  \If{there is any unvisited vertex in $S_R$}
%    \If{$S_I=\varnothing$}
%      \State \parbox[t]{\dimexpr\linewidth-\algorithmicindent}{find $u\in S_R$ and $q_{ref}\in Q_I$ such that $u$ and $q_{ref}$ are the spatially closest pair\strut}
%    \Else
      \State \parbox[t]{\dimexpr\linewidth-\algorithmicindent}{Employ R-Tree distance browsing to extract from $S_R$ the next unvisited vertex $u$ \\ which has the minimum spatial distance to $q_{ref}$\strut} 
%    \EndIf
    \State mark $u$ as visited
  \ElsIf{$\theta<p-1$}
    \State increase $\theta$ and mark the remaining vertices in $S_R$ as unvisited
  \Else
    \State \textbf{break}
  \EndIf

\If{$u$ satisfies the condition of Socio-Spatial Ordering in Eq. (1)}
  \State $S_I\gets S_I+\{u\}$, $S_R\gets S_R-\{u\}$
  \If{Familiarity Pruning in Eq. (4) or Eq. (5) is satisfied}
    \State \textbf{break}
  \EndIf
  \ForAll{$q_i\in Q_I$}
    \If{$S_I$ and $q_i$ satisfies Distance Pruning in Eq. (3)}
      \State $Q_I\gets Q_I-\{q_i\}$
    \ElsIf{$d_{v,q_i}>t$}
      \State $Q_I\gets Q_I-\{q_i\}$
    \EndIf
  \EndFor
  \If{$Q_I=\varnothing$}
    \State \textbf{break}
  \EndIf
  \If{$|S_I|<p$} 
    \State \Call{FindGroupAndLoc\_SFGP}{$S_I$,$S_R$,$Q_I$,$q_{ref}$}
  \Else
    \If{$\min_{q_i\in Q_I}\sum_{v\in S_I}d_{v,q_i}<D$}
      \State $q^{*}\gets \arg\min_{q_i\in Q_I}\sum_{v\in S_I}d_{v,q_i}$
      \State $D\gets \sum_{v\in S_I}d_{v,q^{*}}$, $F\gets S_I$
    \EndIf    
    \State \textbf{break}
  \EndIf
\ElsIf{$\theta=p-1$}
  \State $S_R\gets S_R-\{u\}$
\EndIf

\EndWhile
\EndProcedure
\end{algorithmic}
\end{algorithm*}

\begin{algorithm*}
\caption{MAGS algorithm with APDO}
\begin{algorithmic}[1]
\Require Graph $G=(V,E)$, location $l_v$ for each $v\in V$, the number of attendees $p$, activity locations $Q$, familiarity constraint $k$, and spatial radius $t$. The user locations $l_v\in V$ are indexed by an R-Tree with root $M_0$. The activity locations $q_i\in Q$ are indexed by a BallTree with root $B_0$. 
\Ensure Optimal group $F$ and the corresponding activity location $q^{*}$
\State $S_I\gets \varnothing$, $S_R\gets V$, $Q_I\gets Q$, $F\gets \varnothing$, $D\gets \infty$, $\theta \gets k$
\State $\langle F,q^{\ast}\rangle \gets$ \Call{FindGroupAndLoc\_APDO}{$S_I$,$S_R$,$Q_I$,$B_0$}
\If{$D\neq \infty$}
  \State \textbf{output} $\langle F,q^{*}\rangle$
\Else
  \State \textbf{output} "No Answer"
\EndIf

\Procedure{FindGroupAndLoc\_APDO}{$inS_I$,$inS_R$,$inQ_I$,$inB_0$}
\State $S_I \gets inS_I$, $S_R\gets inS_R$, $Q_I\gets inQ_I$, $B_0\gets inB_0$
\While{$|S_I|+|S_R|\geq p$}
  \If{there is any unvisited vertex in $S_R$}
      \State $\langle v_{c}$,$q_{ref}\rangle\gets $ \textit{APDOandDistPruning($M_0$,$B_0$,$S_I$,$S_R$,$Q_I$)}, $u\gets v_{c}$
      \State \parbox[t]{\dimexpr\linewidth-\algorithmicindent}{let $Q_I$ be the set of leaf nodes in BallTree which are neither pruned \\ nor the descendants of a pruned ball\strut}
      \State mark $u$ as visited
  \ElsIf{$\theta<p-1$}
    \State increase $\theta$ and mark the remaining vertices in $S_R$ as unvisited
  \Else
    \State \textbf{break}
  \EndIf

\If{$u$ satisfies the condition of Socio-Spatial Ordering in Eq. (1)}
  \State $S_I\gets S_I+\{u\}$, $S_R\gets S_R-\{u\}$
  \If{Familiarity Pruning in Eq. (4) or Eq. (5) is satisfied}
    \State \textbf{break}
  \EndIf
  \ForAll{$q_i\in Q_I$}
    \If{$S_I$ and $q_i$ satisfies Distance Pruning in Eq. (3)}
      \State mark $q_i$ as pruned, $Q_I\gets Q_I-\{q_i\}$
    \ElsIf{$\exists v\in S_I$ such that $d_{v,q_i}>t$}
      \State mark $q_i$ as pruned, $Q_I\gets Q_I-\{q_i\}$
    \EndIf
  \EndFor
  \If{$Q_I=\varnothing$}
    \State \textbf{break}
  \EndIf
  \If{$|S_I|<p$} 
    \State \Call{FindGroupAndLoc\_APDO}{$S_I$,$S_R$,$Q_I$,$B_0$}
  \Else
    \If{$\min_{q_i\in Q_I}\sum_{v\in S_I}d_{v,q_i}<D$}
      \State $q^{*}\gets \arg\min_{q_i\in Q_I}\sum_{v\in S_I}d_{v,q_i}$
      \State $D\gets \sum_{v\in S_I}d_{v,q^{*}}$, $F\gets S_I$
    \EndIf
    \State \textbf{break}
  \EndIf
\ElsIf{$\theta=p-1$}
  \State $S_R\gets S_R-\{u\}$
\EndIf

\EndWhile
\EndProcedure
\end{algorithmic}
\end{algorithm*}

\begin{algorithm*}
\caption*{APDOandDistPruning}
\begin{algorithmic}[1]
%\Require An R-Tree indexing the location of each individual of $V$, and a BallTree indexing the activity locations $Q$. The pointers to the roots of R-Tree and BallTree are $M_0$ and $B_0$, respectively
%\Ensure Seed candidate $v_{seed}\in V$ and reference activity location $q_{ref}\in Q$
\Procedure{APDOandDistPruning}{$M_0$,$B_0$,$S_I$,$S_R$,$Q_I$}
\State let $U_R$ and $U_B$ be two lists
\State $M\gets M_0$, $B\gets B_0$
\State $U_R\gets M_0$, $U_B\gets B_0$
\State $v_{c}\gets \varnothing$, $q_{ref}\gets \varnothing$
\While{$M$ and $B$ are not both leaf nodes}
  \State \parbox[t]{\dimexpr\linewidth-\algorithmicindent}{pop MBR $M_i$ from $U_R$ and pop ball $B_j$ from $U_B$ such that \\ $\sum_{v\in S_{I}}MINDIST(v,B_{j})+ MINDIST(M_{i},B_{j})$ is minimum, and $B_{j}$ is not pruned\strut}
  \State $M\gets M_i$, $B\gets B_j$
  \If{$M$ and $B_x\gets B$ satisfy OTDP Lemma 2 and prune $B_y$ in $U_B$}
    \State remove $B_y$ from $U_B$, mark $B_y$ as pruned
  \ElsIf{ITDP in Lemma 4 prunes $B_x$ in $U_B$}
    \State remove $B_x$ from $U_B$, mark $B_x$ as pruned
  \ElsIf{ALDP in Lemma 5 prunes $B_x$ in $U_B$}
    \State remove $B_x$ from $U_B$, mark $B_x$ as pruned
  \EndIf
  \If{$M$ is not a leaf node}
    \ForAll{child MBR $M_i$ of $M$}
      \If{$M_i$ contains $l_v$ where $v\in S_R$}
        \State push $M_i$ into $U_R$
      \EndIf
    \EndFor
  \EndIf
  \If{$B$ is not a leaf node}
    \ForAll{child ball $B_j$ of $B$}
        \State push $B_j$ into $U_B$
    \EndFor
  \EndIf
\EndWhile
\State $v_{c}\gets M$, $q_{ref}\gets B$
\State \textbf{return} $\langle v_{c}$,$q_{ref}\rangle$
\EndProcedure
\end{algorithmic}
\end{algorithm*}


\begin{thebibliography}{99}
%\vspace{-5pt}
\baselineskip=9pt
\bibitem{CPLEX} CPLEX.
http://www-01.ibm.com/software/integration/\\optimization/cplex-optimizer/.

%\bibitem{OnlineVersion} C.-Y. Shen, D.-N. Yang, L.-H. Huang, W.-C. Lee, and M.-S. Chen. Online Version of Socio-Spatial Group Queries for Impromptu Activity Planning. http://arxiv.org/abs/1505.02681. 

\bibitem{RKV95} N. Roussopoulos, S. Kelley and F. Vincent. Nearest Neighbor
Queries. \textit{SIGMOD}, 1995.

\bibitem{LLT09} T. Lappas, K. Liu, and E. Terzi. Finding a Team of Experts
in Social Networks. \textit{KDD}, 2009.

\bibitem{LS10} C.-T. Li and M.-K. Shan. Team Formation for Generalized Tasks
in Expertise Social Networks. \textit{SocialCom}, 2010.

\bibitem{SG10} M. Sozio and A. Gionis. The Community-Search Problem and How
to Plan a Successful Cocktail Party. \textit{KDD}, 2010.

\bibitem{PSTM04} D. Papadias, Q. Shen, Y. Tao and K. Mouratidis. Group
Nearest Neighbor Queries. \textit{ICDE}, 2004.

\bibitem{TPS02} Y. Tao, D. Papadias and Q. Shen. Continuous Nearest Neighbor
Search. \textit{VLDB}, 2002.

\bibitem{YCLC11} D.-N. Yang, Y.-L. Chen, W.-C. Lee and M.-S. Chen. On
Social-Temporal Group Query with Acquaintance Constraint. \textit{VLDB},
2011.

\bibitem{LCHJ12}
W. Liu, W. Sun, C. Chen, Y. Huang, Y. Jian and K. Chen. Circle of Friend Query in Geo-Social Networks. 
\textit{DASFAA}, 2012.

%\bibitem{APP13} N. Armenatzoglou, S. Papadopoulos and D. Papadias. A General Framework for
%Geo-Social Query Processing. \textit{VLDB}, 2013.

\bibitem{GZLC09} Y. Gao, B. Zheng, W.-C. Lee, G. Chen. Continuous Visible
Nearest Neighbor Queries. \textit{EDBT}, 2009.

\bibitem{GZ09} Y. Gao and B. Zheng. Continuous Obstructed Nearest Neighbor
Queries in Spatial Databases. \textit{SIGMOD}, 2009.

\bibitem{HS99} G. Hjaltason and H. Samet. Distance Browsing in Spatial
Databases. \textit{TODS}, 1999.

%\bibitem{SO89} S. Omohundro. Five balltree construction algorithms. \textit{Technical Report}, 1989.

\bibitem{LSCH12} W. Liu, W. Sun, C. Chen, Y. Huang, Y. Jian and K. Chen. Circle of Friend
Query in Geo-Social Networks. \textit{DASFAA}, 2012.

\bibitem{SO89} S. Omohundro. Five Balltree Construction Algorithms. \textit{Technical Report}, 1989.

\bibitem{MP95}  N. Mahadev and U. Peled. Threshold Graphs and Related Topics. \textit{%
Elsevier}, 1995.

\bibitem{LZ05} D. Liu and X. Zhu. Circular Distance Two Labeling and the $\lambda$-Number for Outerplanar Graphs. 
\textit{SIAM J. Discrete Math}, 2005.

\bibitem{CK96} G. Chang and D. Kuo. The $L(2,1)$-Labeling Problem on Graphs. \textit{SIAM J. Discrete Math}, 1996.

\bibitem{S83} S. Seidman. Network Structure and Minimum Degree. \textit{Social Networks}, 1983.

\bibitem{BZ03} V. Batagelj and M. Zaversnik. An $O(m)$ Algorithm for Core Decomposition of Networks. \textit{CoRR}, 2003.

\bibitem{YYLL11} M. Ye, P. Yin, W.-C. Lee and D.-L. Lee. Exploiting
Geographical Influence for Collaborative Point-of-Interest Recommendation. \textit{SIGIR}, 2011.

\bibitem{YL12} J. Yang and J. Leskovec. Defining and Evaluating Network Communities based on Grouth-truth. 
\textit{ICDM}, 2012.

\bibitem{SJL00} S. Saltenis, C. Jensen, S. Leutenegger, and M. Lopez. Indexing the
Positions of Continuously Moving Objects. \textit{SIGMOD}, 2000.

\bibitem{TPR-Tree} Y. Tao, D. Papadias, and J. Sun. The TPR*-tree: an Optimized
Spatiotemporal Access Method for Predictive Queries. \textit{VLDB}, 2002.

\bibitem{BKS90} N. Beckmann, P. Kriegel, R. Schneider, and B. Seeger. The R*-tree: an Efficient and Robust Access Method for Points and Rectangles. \textit{SIGMOD}, 1990.

\end{thebibliography}
\end{document}